\documentclass{theoretics}
\title{Positional \texorpdfstring{$\omega$}{omega}-regular languages}
\ThCSshorttitle{Positional \texorpdfstring{$\omega$}{omega}-regular languages}

\ThCSauthor[1]{Antonio Casares}{antonio.casares@rptu.de}[0000-0002-6539-2020]
\ThCSauthor[2]{Pierre Ohlmann}{pierre.ohlmann@lis-lab.fr}[0000-0002-4685-5253]
\ThCSaffil[1]{University of Kaiserslautern-Landau, Germany}
\ThCSaffil[2]{CNRS, LIS, Aix-Marseille University, France}

\ThCSshortnames{A. Casares and P. Ohlmann}
 
\ThCSkeywords{Infinite duration games; Positionality; Strategy complexity; Parity automata}
 \ThCSthanks{We want to thank Hugo Gimbert, Damian Niwiński and Pierre Vandenhove for interesting discussions on the subject. This work was developed while Antonio Casares was at the University of Warsaw, supported by the Polish National Science Centre (NCN) grant ``Polynomial finite state computation'' (2022/46/A/ST6/00072). A preliminary version of this article appeared at LICS 2024~\cite{CasaresO24}.}%optional

\declaretheorem[style=standard,numbered=no,name={Global hypothesis\kern-.09em}]{globalHyp*}

\ThCSyear{2026}
\ThCSarticlenum{5}
\ThCSreceived{Dec 16, 2024}
\ThCSrevised{Dec 22, 2025} 
\ThCSaccepted{Jan 16, 2026}
\ThCSpublished{Feb 24, 2026}
\ThCSdoicreatedtrue

\makeatletter
\newcommand\IfRestateTF{%
  \ifx\label\thmt@gobble@label % or just compared to \@gobble
    \expandafter\@firstoftwo
  \else
    \expandafter\@secondoftwo
  \fi
}
\makeatother
\newcommand{\RestateRemark}{\IfRestateTF{{\normalfont\bfseries (Restated) }}{}}

\usepackage{mathtools}

\usepackage{xfrac}
\usepackage{relsize}

\usepackage{diagbox}
\usepackage{array} %to center content in tabular environments
\usepackage{makecell}

\usepackage{mathcommand}
\usepackage[notion, quotation, electronic]{knowledge}

\definecolor{ThCSdarkblue}{RGB}{25,34,64}
\knowledgestyle*{intro notion}{emphasize}
\knowledgestyle*{notion}{color=black}

\usepackage{enumitem}

\usepackage{tikz}
\usetikzlibrary{automata, positioning, arrows,arrows.meta,calc,decorations,decorations.markings,math,shapes.geometric}

\tikzset{
	>=stealth',%Black traingles
	-={stealth',ultra thick,scale=3} % makes the arrow heads bold
	node distance=1cm, % specifies the minimum distance between two nodes.
	every state/.style={thick}, % sets the properties for each node
	initial text=$ $, % sets the text that appears on the start arrow
}

\newcommand{\bfDescript}[1]{\textbf{#1}}

\makeatletter
\newcommand{\mylabelOne}[2]{#2\def\@currentlabel{(3')}\label{#1}}
\newcommand{\mylabelTwo}[2]{#2\def\@currentlabel{(3'')}\label{#1}}
\makeatother

\usepackage{colours}

\knowledge{notion}
| notion
| definition

\knowledge{notion}
| hyperlink
| hyperliens
| Hyperlinks
| hyperlinks

\knowledge{notion}
| strongly connected component
| SCC
| strongly connected components
| SCCs

\knowledge{notion}
| final SCC

\knowledge{notion}
| positive@SCC
| positive SCC

\knowledge{notion}
| negative@SCC
| negative SCC

\knowledge{notion}
| dead-ends
| dead-end

\knowledge{notion}
| recurrent

\knowledge{notion}
| transient
| transient vertex

\knowledge{notion}
| initial state
| initial@state
| initial states
| initial states@aut
| initial state@aut
| initial

\knowledge{notion}
| accepting@runAut
| accepting@aut
| accepting run@aut
| accepting runs@aut
| acceptance@runAut
| acceptance of runs@aut
| accepting run
| accepting@run

\knowledge{notion}
| rejecting@runAut
| rejecting run@aut
| rejecting run
| rejecting@run

\knowledge{notion}
| prefix-independent
| prefix-independence

\knowledge{notion}
| automaton
| (non-deterministic) 
| automata
| (non-deterministic) automaton
| $\oo $-automata
| automata over infinite words
| parity automata
| Parity automata
| parity automaton
| automaton@parity
| (non-deterministic) parity automaton
| parity@aut
| parity

\knowledge{notion}
| deterministic
| deterministic automaton
| deterministic automata
| Deterministic
| non-deterministic automata
| non-deterministic
| non-deterministic automaton
| non-determinism
| determinism 
| deterministic parity automata

\knowledge{notion}
| equivalent@automaton
| equivalent@automata
| equivalent@aut
| equivalence@aut
| equivalent to@aut
| equivalent to
| equivalent

\knowledge{notion}
| complete
| completeness
| completeness@aut

\knowledge{notion}
| run over $w$
| run@automaton
| run over
| run@aut
| over@run
| run $\rr $ over
| runs@aut
| run of@aut
| run over a word
| run over $w_kw_{k+1}\dots $
| run over $w^\oo $
| run over@aut
| run
| runs

\knowledge{notion}
| accepted@word
| accepts@aut
| accepted by@aut
| accepted
| accepted by
| accepted@aut
| rejected@aut
| accepting word
| rejected

\knowledge{notion}
| recognising
| recognise
| recognises
| recognising@automaton
| recognising@aut
| recognised
| recognises@aut
| language recognised
| languages recognised

\knowledge{notion}
| $\oo $-regular language
| $\oo $-regular languages
| $\oo $-regular@language
| $\oo $-regular automata
| $\oo $-regular
| $\oo $-regularity
| $\oo $-regular objective
| $\omega $-regular objectives
| $\omega $-regular languages
| non-$\oo $-regular
| $\oo $-regular conditions
| $\oo $-regular objectives
| $\omega $-regular

\knowledge{notion}
| $a$-self loop
| $v$-self loop
| $u$-self loop
| $v$-self loops
| self loops
| self loop

\knowledge{notion}
| history-deterministic
| HD-parity automaton
| history-determinism
| history-deterministic@aut
| HD@aut
| HD
| history-determinism@aut
| HD automaton
| HD automata
| (history-)deterministic
| history-deterministic automaton
| History-deterministic automata
| history-deterministic automata
| History-deterministic
| History-determinism

\knowledge{notion}
| resolver
| resolvers
| resolver@aut
| resolvers@aut
| resolve@aut

\knowledge{notion}
| sound@aut
| sound resolver@aut
| soundness@resolverAut
| sound@resolverAut
| sound resolver
| sound@resAut

\knowledge{notion}
| run induced by@aut
| induced by $\resolv$@aut
| induced run of $\resolv $
| run induced over $w$
| induced run
| induced run of $\resolv $ over $u_0$
| induced run of $\resolv $ over
| run induced by $\resolv $
| induced by@resolver
| run induced by@res
| run induced by@resolv
| run induced over $w$@res
| induced by@resolv
| induce@resolv
| induced by

\knowledge{notion}
| reachable using some sound resolver
| reachable using@resolver
| reachable using a sound resolver 
| reachable using this resolver
| reachable using $\resolv $
| reachable using the resolver $(r_0, \resolv )$
| reachable using $(r_0, \resolv )$ 
| reachable using $\resolv $@memory
| reachable@res

\knowledge{notion}
| memory structure@aut
| memory structure@resolver
| finite memory resolvers
| finite memory structures@resolv
| finite memory@res
| finite memory structure@resolv

\knowledge{notion}
| implements@resolver
| resolver implemented by
| implemented by@resolver
| implemented by@memRes
| implemented by@resMem
| implemented by@aut
| implemented by@resolvMem

\knowledge{notion}
| priority
| priorities
| colouring with priorities
| output priorities@aut

\knowledge{notion}
| automaton structure

\knowledge{notion}
| parity automaton on top of
| parity automaton on top
| automaton on top of
| on top
| on top of

\knowledge{notion}
| normalised
| normal form
| normalise
| form@normal
| normality
| Normality

\knowledge{notion}
| normalising
| normalisation

\knowledge{notion}
| objective
| objectives

\knowledge{notion}
| language
| languages

\knowledge{notion}
| winning conditions
| winning condition
| winning objective
| winning objectives

\knowledge{notion}
| game
| games
| Games
| infinite duration games
| jeux
| jeu

\knowledge{notion}
| $W$-game
| $W$-games
| $\locLang {W}{u}$-game
| $\locLangBuchi {W}{u}$-game
| $\locLang x q$-game
| $\Muller {\F }$-game
| $\Muller {\F }$-games
| $W_n$-game
| $\Lang {\A }$-game
| $\addNeutral {W}$-games

\knowledge{notion}
| play
| plays
| finite play

\knowledge{notion}
| Eve-game
| Eve-games
| Adam-games

\knowledge{notion}
| bipositional over (finite) Eve and Adam-games
| bipositional over finite Eve and Adam-games
| bipositional over Eve and Adam-games

\knowledge{notion}
| output@game

\knowledge{notion}
| winning@play

\knowledge{notion}
| losing@play
| losing
| loses

\knowledge{notion}
| subplay

\knowledge{notion}
| strategy
| strategies
| stratégies

\knowledge{notion}
| uniformly
| uniformly@strat
| uniform strategies
| uniform strategy

\knowledge{notion}
| wins
| wins $\G $ from $v$
| wins $\G \compositionAut \A $ from $(v, q_0)$
| win
| from@winGame
| winner
| wins $\G $ from
| winners
| wins from
| win from $v_0$
| winning from
| wins from $v_0$
| losing@strat
| winning from $v$
| stratégie gagnante
| wins $\G \compositionAut \A $ from $(v, q_\init )$
| win from

\knowledge{notion}
| win positionally
| positionally@win
| positionally

\knowledge{notion}
| consistent with the strategy
| consistent with
| consistent with@strat
| consistency with@strat

\knowledge{notion}
| winning region
| winning regions

\knowledge{notion}
| full winning region

\knowledge{notion}
| player
| controls
| controlled by

\knowledge{notion}
| Eve
| Eve's
| Eve-vertex
| controlled by Eve
| owners of the vertices

\knowledge{notion}
| Adam
| Adam-vertex
| controlled by Adam
| Adam's vertex
| Adam's

\knowledge{notion}
| winning strategy for Eve
| winning strategy
| winning@strat
| winning for Eve
| winning strategy from
| winning
| winning strategies
| winning from@strat
| winning strategy from $v$

\knowledge{notion}
| solve@game
| solve a game
| Solving@game
| Solving a game
| solving@game
| solving@games
| solved@games
| résolution de ces jeux
| résolution des jeux
| solving a game
| solving games

\knowledge{notion}  
| determined

\knowledge{notion}
| positional@strategy
| positional strategy
| positional@strat
| positional strategies
| positionnalité
| positionally@strat
| stratégies positionnelles
| positionnelle

\knowledge{notion}
| optimal strategy
| optimal strategies
| optimal (for Eve)
| optimal@strat
| playing optimally
| play optimally

\knowledge{notion}
| play optimally in $\G $ using a positional strategy
| play optimally using positional strategies
| jouer de manière optimale
| play optimally using@pos
| play optimally using a positional strategy
| play optimally@pos
| play optimally using

\knowledge{notion}
| positional@objective
| Positional
| positionality@half
| positionality
| positional
| positionnels

\knowledge{notion}  
| positional over $\mathcal {X}$ games
| positional over
| positional over@eve

\knowledge{notion}
| bipositional@objective
| bipositionality@objective
| bipositional
| bipositionality
| Bipositionality

\knowledge{notion}
| $\ee $-edges
| uncoloured edges

\knowledge{notion}
| $\ee $-free game
| $\ee $-free
| $\ee $-free@game
| $\ee $-free games
| $\ee $-free@games

\knowledge{notion}
| progress consistency
| progress consistent

\knowledge{notion}
| bi-progress consistent
| bi-progress consistency

\knowledge{notion}
| fully progress consistent
| fully progress consistent@sigAut
| full progress consistency

\knowledge{notion}
| $a$-transitions@out
| $a$-transition@out
| $a$-transitions@in
| $a$-transition@in
| $\NLetters {x}{q}$-transitions@in
| $a$-transition

\knowledge{notion}
| $y$-transitions@out
| $x$-transitions@out
| $x-1$-transitions@out
| $0$-transitions
| $y$-transitions
| ${>}x$-transitions
| $x-1$-transitions
| $x$-transitions
| $({\geq }x)$-transitions
| $({\geq }x-2)$-transitions
| $({<}x)$-transitions
| ${\geq }x$-transitions
| $({>}x)$-transitions@out
| $(x-1)$-transitions@out
| $(x-1)$-transitions

\knowledge{notion}
| deterministic over~$\DD '$
| deterministic over
| deterministic over $2$-transitions
| deterministic over transitions
| determinism over transitions
| determinism over

\knowledge{notion}
| accepted with priority $x$
| accepted with priority
| accepted with a priority
| accepted with an even priority
| accepted or rejected with priority
| accepted with
| accepted@byPriority

\knowledge{notion}
| is rejected with priority $x$
| is rejected with priority
| rejected with an odd priority
| rejected with priority
| rejected with a priority

\knowledge{notion}
| automaton with $\ee $-transitions
| $\ee $-transitions@aut
| $\ee $-letters
| $\ee $-transitions

\knowledge{notion}
| signature automaton
| signature automata
| signature
| Signature automata

\knowledge{notion}
| $d$-signature automaton

\knowledge{notion}
| $d$-structured signature automaton
| $(x-2)$-structured signature automaton
| $x$-structured signature automaton
| $(x-2)$-structured
| $x$-structured signature
| structured signature
| $(x-2)$-structured signature

\knowledge{notion}
| structured signature automaton
| structured
| structured signature automata

\knowledge{notion}
| reduced@sig
| reduction@sig
| reduced signature automata  
| reduced signature automaton
| reduced

\knowledge{notion}
| strongly reduced@sig
| refine this notion
| refine this notion@strReducedSig

\knowledge{notion}
| Local monotonicity
| local monotonicity
| local monotonicity of $({\geq } x)$-transitions
| local monotonicity of $({\geq }x)$-transitions

\knowledge{notion}
| $({<}x)$-safe separation

\knowledge{notion}
| $\SS $-graph
| $\SS $-graphs
| $\Sigma $-graphs
| graph
| $\SS \cup \{\ee \}$-graph

\knowledge{notion}
| monotone@graph
| monotone graph@univ
| monotone@univ
| monotonicity@graph
| Monotonicity

\knowledge{notion}
| $v$ satisfies@univ
| satisfies@univ
| satisfy@
| satisfies@univTree
| satisfies@tree
| satisfies@graph
| satisfy@univ

\knowledge{notion}
| totally ordered graph
| ordered@univ

\knowledge{notion}
| well-ordered graph
| well-ordered@graph
| well-ordered@univ
| well-order
| well-ordered

\knowledge{notion}
| $\SS $-tree
| $[0,\dots ,d]$-tree
| trees@univ
| $[0,d+1]$-tree

\knowledge{notion}
| root@univ
| root

\knowledge{notion}
| satisfies@treeUniv
| satisfy@treeUniv

\knowledge{notion}
| $(\kappa ,W)$-universal for trees
| $(\kappa ,\parity )$-universal for trees
| universality for trees
| $(\kk ,W)$-universal for trees
| $(\kappa ,W)$-universality for trees  
| tree@univ
| universal for trees

\knowledge{notion}
| morphism@univ
| morphism of $\SS $-graphs@univ

\knowledge{notion}
| $(\kappa ,W)$-universal  
| $(\kappa , W)$-universal
| universal
| universal graphs
| universal graph
| $(\kappa ,\parity )$-universal
| universality
| $(\kappa ,\parity )$-universality
| $(\kk ,W)$-universal
| Universal graphs
| Universal graph
| $(\aleph _0,W)$-universal

\knowledge{notion}
| universal graph $\UPar $ for the parity objective

\knowledge{notion}
| neutral letter
| neutral for $W$
| neutral letters

\knowledge{notion}
| total@preorder
| total preorder
| total preorders
| totally preordered
| total@preord
| total@order

\knowledge{notion}
| refinement
| refining
| refines
| refine

\knowledge{notion}
| collection of nested (total) preorders over a set $Q$
| collection of nested total preorders
| nested total preorders
| nested preorders
| nested@preorder
| collection of nested total preorders over a set $Q$

\knowledge{notion}
| congruence for $\DD '$
| congruences for
| congruence for
| congruency
| congruence for $y$-transitions
| congruence for $({>}x)$-transitions
| preserve this congruence
| preserve the congruence

\knowledge{notion}
| congruence
| congruences

\knowledge{notion}
| strong congruence for
| strong congruency
| strong congruence
| strong congruence for transitions

\knowledge{notion}
| congruence class
| equivalence class

\knowledge{notion}
| are monotone for@cong
| monotone for $\DD '$
| monotone@cong
| monotone for@cong
| are monotone for
| monotonicity@cong
| monotone for $\leqSig {x}$
| monotone for
| monotone for a preorder

\knowledge{notion}
|	(strictly) monotone@aut
| monotone@aut
| monotone automaton

\knowledge{notion}
| induced equivalence relation
| equivalence relation induced from a preorder 
| induced equivalence relation@preord
| equivalence relation

\knowledge{notion}
| residual of $L$ with respect to
| residuals
| residual
| Residuals

\knowledge{notion}
| residual classes
| residual class
| class@res
| class

\knowledge{notion}
| residual class@state
| class of states
| class of $q$

\knowledge{notion}
| equivalent@res

\knowledge{notion}
| associated to@residual
| associated to@res
| states associated to@res
| states of@res

\knowledge{notion}
| act uniformly
| uniformity
| uniform
| uniform over
| uniformity for
| $y_1$-uniformity
| are uniform over $\sim$-classes
| uniform@trans
| uniformity@trans
| are uniform
| behave uniformly
| uniformity@cong

\knowledge{notion}
| produces priority $y$ uniformly
| produces priority $x-1$ uniformly
| producing priority $x$ uniformly
| produce priority $x$ uniformly
| produces priority $x$ uniformly
| producing priority $x$ uniformly in  
| produces priority $x$ uniformly in

\knowledge{notion}
| automaton of residuals
| residual automaton

\knowledge{notion}
| redirected by the transformation

\knowledge{notion}
| preserves the residuals@transform

\knowledge{notion}
| leading to@class

\knowledge{notion}
| $[0,x]$-faithful
| $[0,y]$-faithful
| faithful
| faithfulness
| $[0,x-2]$-faithfulness
| $[0,x]$-faithfulness
| $[0,x-2]$-faithful
| $[0,x-1]$-faithful
| $[0,x-1]$-faithfulness
| $[0,x]$-faithful congruence
| faithful congruences
| $[0,x-1]$-faithful congruence
| $[0,x-2]$-faithful congruence
| $[0,y]$@faithful
| priority-faithful

\knowledge{notion}
| quotient of $\A $ by $\sim $
| quotient of $\A $ by
| quotient automaton
| quotient automaton@res
| quotient@aut

\knowledge{notion}
| $({\leq }x)$-quotient of $\A $ by $\sim $
| ${\leq }x$-quotient
| $({\leq }d)$-quotient by $\eqSig {d}$
| quotient automata@leq
| quotient of automata@leq
| quotient automaton@leq
| ${\leq }(x-1)$-quotient automata
| quotient@leq
| $({\leq }x)$-quotient
| ${\leq }(x-1)$-quotient automaton

\knowledge{notion}
| projection of $\rr $@quot
| projection@quot
| induces the run@resAut

\knowledge{notion}
| nice transformations
| nice transformation
| $\eqSig {x-2}$-nice transformation  
| $\sim $-nice transformation of $\A $ at level $x$
| $\eqSig {x-2}$-nice transformation at level $x-1$
| $\eqRes {x-2}$-nice transformation of $\A $ at level $x-1$
| $\eqRes {x-2}$-nice transformation at level $x-1$
| $\eqSig {x-2}$-nice transformation of $\A $ at level $x-1$
| $\eqSig {x}$-nice transformation of $\A $ at level $x$
| nice transformation at level $x$
| $\eqSig {x}$-nice transformation at level $x$

\knowledge{notion}
| semantically deterministic

\knowledge{notion}
| uniquely decodable alphabet
| uniquely decodable alphabets
| uniquely decodable

\knowledge{notion}
| prefix codes
| prefix code

\knowledge{notion}
| localisation of $W$ to~$\resClass {u}$
| localisation to residuals

\knowledge{notion}
| local automaton of the residual $\resClassState {q}$
| local automaton of the residual

\knowledge{notion}
| local automaton of the class $\classSig {x}{q}$
| local automaton of a $\eqSig {x}$-class
| local automata
| local automaton

\knowledge{notion}
| $({\leq}1)$-safe languages
| $({<}2)$-safe language@coB
| $({<}2)$-safe languages@coB
| safe language@coB
| safe languages@coB

\knowledge{notion}  
| $(\leq 1)$-safe languages
| safe languages
| safe language
| $({<}x)$-safe language
| ${<}x$-safe languages
| $({<}x)$-safe languages
| safe languages@sig  
| ${<}x$-safe language

\knowledge{notion}
|	$({<}2)$-safe component@coB
| safe component@coB
| safe components@coB

\knowledge{notion}  
| safe component
| safe components
| $({<}x)$-safe component
| ${<}x$-safe components
| ${<}y$-safe components
| ${<}1$-safe components
| ${<}(x-1)$-safe components
| $({<}x)$-safe components
| ${<}x$-safe component
| $({<}x-1)$-safe component
| $({<}x-1)$-safe components
| $({<})y$-safe component
| $({<}y)$-safe components
| safe component@sig

\knowledge{notion}
| safe path@coB
| safe run@coB
| safe@runCoB
| $1$-safe run@coB

\knowledge{notion}  
| safe path
| safe run
| safe@run
| $({<}x)$-safe@path
| ${<}x$-safe run
| $({<}x)$-safe path
| ${<}x$-safe@path
| ${<}x$-safe path
| ${<}x$-safe@run
| safe run@sig
| $(<x)$-safe path

\knowledge{notion}
| redundant@coB

\knowledge{notion}
| redundant@sig
| redundant@safec
| redundant@safeCentr

\knowledge{notion}
| closed objectives
| topologically closed
| closed@obj
| closed
| closed objective
| closeness@obj
| closed languages

\knowledge{notion}
| open objectives
| topologically open
| open@obj
| open
| open objective

\knowledge{notion}
| $\oo $-regular closed objectives
| $\omega $-regular closed
| $\oo $-regular closed
| $\oo $-regular closed objective

\knowledge{notion}
| $\oo $-regular open objectives
| $\oo $-regular open
| $\omega $-regular open
| open@reg

\knowledge{notion}
| concave
| Concavity

\knowledge{notion}
| super word@buchi
| super words@buchi

\knowledge{notion}
| super word
| super words
| super word for

\knowledge{notion}
| super letter@buchi
| super letters@buchi
| non-super letters@buchi

\knowledge{notion}
| super letter
| non-super letters
| super letters for
| super letters

\knowledge{notion}  
| neutral letters@polishSig
| neutral letter@sig
| neutral letters@sig

\knowledge{notion}
| polished@classBuchi
| polish@classBuchi
| polished in $\A $@buchi
| polished class@buchi

\knowledge{notion}
| polished@autBuchi
| polished automata@buchi
| polished automaton@buchi

\knowledge{notion}
| redirecting@buchi
| redirected@buchi
| redirected transitions@buchi

\knowledge{notion}
| polished@aut
| polish@aut
| $x$-polished@aut
| $x$-polished automaton
| $x$-polished automata
| polishing operation

\knowledge{notion}
| polished@class
| polished in $\A $ 
| polish@class
| $x$-polished@class
| polish
| $x$-polished
| $x$-polish@class
| polish@sig
| $x$-polished class
| polished@classSig

\knowledge{notion}
| redirected in $\A '$@sc
| redirected@sc
| redirected transition@safec
| redirected from@safec
| redirected@safec

\knowledge{notion}
| redirected transition@polish
| redirected transitions@polish
| redirected@polish

\knowledge{notion}  
| recurrent@class

\knowledge{notion}  
| transient@class
| transient classes

\knowledge{notion}
| safe centralised@coB
| safe centralisation@coB
| Safe centrality@coB
| safe centrality@coB

\knowledge{notion}
| safe centralised
| safe centralisation
| Safe centrality
| safe centrality
| $({<}x)$-safe centralise
| $({<}x)$-safe centralised
| $({<}x)$-safe centrality
| $({<}x)$-safe centralisation
| ${<}x$-safe centralise
| ${<}x$-safe centrality
| safe centralisation@sig

\knowledge{notion}
| homogeneous
| homogeneity

\knowledge{notion} 
| $({\leq }x)$-homogeneous 

\knowledge{notion}
| safe minimal@coB
| safe minimal
| safe minimality

\knowledge{notion}
| $1$-saturated
| $1$-saturation

\knowledge{notion}
| $(x-1)$-saturated
| $(x-1)$-saturation
| $x-1$-saturated
| $x-1$-saturation

\knowledge{notion}
| $\ee$-complete automata
| $\ee$-complete
| $\ee$-complete automaton
| $\ee $-completeness

\knowledge{notion}
  | $\ee $-completable
  
\knowledge{notion}
  | $\ee $-completion

\knowledge{notion}
| $x$-jump
| $x'$-jump
| $(x-2)$-jump
| $y$-jump
| $(y-2)$-jump
| $0$-jumps
| $(x-1)$-jump
| $(y-1)$-jump

\knowledge{notion}
| priority-closed
| priority-closed@aut
| priority-closure  
| closed by monotonicity
| closed by monotonicity@HM

\knowledge{notion}
| makes a reset
| reset@make

\knowledge{notion}
| reset-stable
| reset-stability 
| makes infinitely many resets
| makes finitely many resets

\knowledge{notion}
| Kopczyński@conjUnion
| Kopczyński's conjecture
| this conjecture@KopczUnion

\knowledge{notion}
| Ohlmann@conj
| Ohlmann's conjecture
| Neutral letter conjecture
| Neutral Letter Conjecture
| neutral letter conjecture

\knowledge{notion}
| games originating from logical formulas
| originating from logical formulas

\knowledge{notion}
  | parity objective
  | parity condition
  | parity@obj

\knowledge{notion}
  | B\"uchi automata
  | B\"uchi automaton
  | B\"uchi@aut

\knowledge{notion}
  | coB\"uchi automata
  | coB\"uchi
  | coB\"uchi@aut
  | coB\"uchi automaton

\knowledge{notion}
  | game graph
  | game graphs

\knowledge{notion}
  | subautomaton

\knowledge{notion}
  | trivial@SCC
  | non-trivial

\knowledge{notion}
  | B\"uchi transitions

\knowledge{notion}
  | coB\"uchi transitions

\knowledge{notion}
  | B\"uchi recognisable
  | B\"uchi@rec
  | Büchi recognisable

\knowledge{notion}
  | coB\"uchi recognisable
  | coBüchi recognisable

\knowledge{notion}
  | equivalent to@labelling
  | equivalent labelling

\knowledge{notion}
  | preorder
  | preorders
  | preordered
  
\knowledge{notion}
  | Büchi languages
  | B\"uchi language
  | B\"uchi objective

\newrobustcmd{\SOneS}{\ensuremath{\mathrm{S1S}}}

\newrobustcmd{\disjUnion}{\mathrel{\kl[\disjUnion]{\sqcup}}}
\knowledge{\disjUnion}{notion}
\newrobustcmd{\restSubsets}[2]{\kl[\restSubsets]{\restr{#1}{#2}}}
\knowledge{\restSubsets}{notion}
\knowledge{\partialF}{notion}
\newrobustcmd{\complSet}[1]{\kl[\complSet]{\overline{#1}}}
\knowledge{\complSet}{notion}

\newrobustcmd{\emptyword}{\kl[\emptyword]{\varepsilon}}
\knowledge{\emptyword}{notion}

\newrobustcmd{\Paths}{\kl[\Paths]{\mathsf{Paths}}}
\knowledge\Paths{notion}

\newrobustcmd{\first}{\kl[\first]{\mathsf{first}}}
\knowledge\first{notion}
\newrobustcmd{\last}{\kl[\last]{\mathsf{last}}}
\knowledge\last{notion}

\newcommand{\ent}{\text{enter}_{q}}
\newcommand{\sma}{\text{small}}
\newcommand{\neut}{\text{neutral}}
\newrobustcmd{\oddparity}{{\mathsf{oddParity}}}
\newrobustcmd{\prio}{{\mathsf{prio}}}
\newrobustcmd{\gtype}{{\mathsf{goodType}}}
\newrobustcmd{\type}{{\mathsf{type}}}

\newrobustcmd{\infOften}{\kl[\infOften]{\mathtt{Inf}}}
\knowledge{\infOften}{notion}
\newrobustcmd{\finOften}{\kl[\finOften]{\mathtt{Fin}}}
\knowledge{\finOften}{notion}
\newrobustcmd{\noOcc}{\kl[\noOcc]{\mathtt{No}}}
\knowledge{\noOcc}{notion}

\newcommand{\re}[1]{\xrightarrow{#1}}
\usetikzlibrary{calc,decorations.pathmorphing,shapes} %%

\newcounter{sarrow}
\newcommand\lrp[1]{%
	\stepcounter{sarrow}%
	\mathrel{\begin{tikzpicture}[baseline= {( $ (current bounding box.south) + (0,-0.5ex) $ )}]
			\node[inner sep=.5ex] (\thesarrow) {$\scriptstyle #1$};
			\path[draw,<-,decorate,
			decoration={zigzag,amplitude=0.7pt,segment length=1.2mm,pre=lineto,pre length=4pt}] 
			(\thesarrow.south east) -- (\thesarrow.south west);
	\end{tikzpicture}}%
}

\newcommand{\rp}[1]{\lrp{#1}}

\newrobustcmd\lrpResolver[2]{\kl[\lrpResolver]{
		\stepcounter{sarrow}%
		\mathrel{\hspace{1mm}\begin{tikzpicture}[baseline= {( $ (current bounding box.south) + (0,2.2mm) $ )}]
				\node[inner sep=.5ex] (\thesarrow) {$\scriptstyle #2$};
				\draw[<-, decorate,
				decoration={zigzag,amplitude=0.7pt,segment length=1.2mm,pre=lineto,pre length=4pt}] 
				(\thesarrow.south east) -- (\thesarrow.south west) node[below, pos=0.1, inner sep=1mm] {$\scriptstyle #1$};
			\end{tikzpicture}\hspace{0.7mm}}%
}}
\knowledge\lrpResolver{notion}

\newrobustcmd\lrpResolverMem[2]{\kl[\lrpResolverMem]{\lrpResolver{#1}{#2}}}
\knowledge\lrpResolverMem{notion}

\newrobustcmd\lrpAllResolver[2]{\kl[\lrpAllResolver]{
	\stepcounter{sarrow}%
	\mathrel{\hspace{1mm}\begin{tikzpicture}[baseline= {( $ (current bounding box.south) + (0,2.7mm) $ )}]
			\node[inner sep=.5ex] (\thesarrow) {$\scriptstyle #2$};
			\draw[<-, decorate,
			decoration={zigzag,amplitude=0.7pt,segment length=1.2mm,pre=lineto,pre length=4pt}] 
			(\thesarrow.south east) -- (\thesarrow.south west) node[below, pos=0.1, inner sep=1mm] {$\mathsmaller{\mathsmaller{\forall},\, #1}$};
	\end{tikzpicture}\hspace{0.7mm}}%
}}
\knowledge\lrpAllResolver{notion}

\newrobustcmd\lrpExistsResolver[2]{\kl[\lrpExistsResolver]{
		\stepcounter{sarrow}%
		\mathrel{\hspace{1mm}\begin{tikzpicture}[baseline= {( $ (current bounding box.south) + (0,2.7mm) $ )}]
				\node[inner sep=.5ex] (\thesarrow) {$\scriptstyle #2$};
				\draw[<-, decorate,
				decoration={zigzag,amplitude=0.7pt,segment length=1.2mm,pre=lineto,pre length=4pt}] 
				(\thesarrow.south east) -- (\thesarrow.south west) node[below, pos=0.1, inner sep=1mm] {$\mathsmaller{\mathsmaller{\exists},\, #1}$};
			\end{tikzpicture}\hspace{0.7mm}}%
}}
\knowledge\lrpExistsResolver{notion}

\newrobustcmd\lrpResolverWord[3]{\kl[\lrpResolverWord]{
		\stepcounter{sarrow}%
		\mathrel{\hspace{1mm}\begin{tikzpicture}[baseline= {( $ (current bounding box.south) + (0,2.7mm) $ )}]
				\node[inner sep=.5ex] (\thesarrow) {$\scriptstyle #3$};
				\draw[<-, decorate,
				decoration={zigzag,amplitude=0.7pt,segment length=1.2mm,pre=lineto,pre length=4pt}] 
				(\thesarrow.south east) -- (\thesarrow.south west) node[below, pos=0.1, inner sep=1mm] {$\scriptstyle #2,#1$};
			\end{tikzpicture}\hspace{0.7mm}}%
}}
\knowledge\lrpResolverWord{notion}

\newcommand\lrpE{\lrp{\phantom{w.}}}

\newrobustcmd{\colAut}{\kl[\colAut]{\mathsf{p}}}
\knowledge\colAut{notion}
\newrobustcmd{\size}[1]{\kl[\size]{|#1|}}
\knowledge\size{notion}
\newrobustcmd{\sizeAut}[1]{|#1|}
\newcommand{\init}{{\mathsf{init}}}

\newrobustcmd\Lang[1]{\kl[\Lang]{\mathcal{L}(#1)}}
\knowledge\Lang{notion}

\newrobustcmd{\initialAut}[2]{\kl[\initialAut]{#1_{#2}}}
\knowledge\initialAut{notion}

\newrobustcmd{\parity}{\kl[\parity]{\mathsf{parity}}}
\knowledge{\parity}{notion}
\newrobustcmd{\Buchi}[1]{\kl[\Buchi]{\mathsf{Buchi}(#1)}}
\newrobustcmd{\BuchiC}[2]{\kl[\BuchiC]{\mathsf{Buchi}_{#2}(#1)}}
\knowledge{\Buchi}[\BuchiC]{notion}
\newrobustcmd\resolv{\mathsf{r}}

\newcommand{\Eve}{\mathrm{Eve}}

\newcommand{\VAdam}{V_{\mathrm{Adam}}}
\newcommand{\VEve}{V_{\mathrm{Eve}}}

\newrobustcmd{\strat}{\mathsf{strat}}

\newrobustcmd{\winRegion}[2]{\kl[\winRegion]{\mathpzc{Win}_{#1}(#2)}}
\knowledge\winRegion{notion}

\newrobustcmd{\nextmoveResolver}{\kl[\nextmoveResolver]{\sigma}}
\knowledge\nextmoveResolver{notion}
\newrobustcmd{\transMem}{\kl[\transMem]{\mu}}
\knowledge\transMem{notion}

\newrobustcmd{\Safe}[1]{\kl[\Safe]{\mathsf{Safety}(#1)}}
\knowledge{\Safe}{notion}
\newrobustcmd{\Reach}[1]{\kl[\Reach]{\mathsf{Reach}(#1)}}
\knowledge{\Reach}{notion}

\newrobustcmd{\SigmaTwo}{\kl[\SigmaTwo]{\Sigma_2^0}}
\newrobustcmd{\SigmaThree}{\kl[\SigmaThree]{\Sigma_3^0}}
\knowledge{\SigmaTwo}[\SigmaThree]{notion}
\newrobustcmd{\PiTwo}{\kl[\PiTwo]{\Pi_2^0}}
\knowledge{\PiTwo}{notion}

\newrobustcmd{\BCSigma}{\kl[\BCSigma]{\mathcal{BC}(\SigmaTwo)}}
\knowledge{\BCSigma}{notion}

\newrobustcmd{\Res}[1]{\kl[\Res]{\mathsf{Res}(#1)}}
\knowledge{\Res}[\ResW]{notion}
\newrobustcmd{\ResW}{\kl[\ResW]{\mathsf{Res}(W)}}

\newrobustcmd{\lquot}[2]{\kl[\lquot]{\inv{#1}#2}}
\knowledge{\lquot}[\lquotW]{notion}
\newrobustcmd\lquotW[1]{\kl[\lquotW]{#1}^{-1}W}

\newrobustcmd{\resClass}[1]{\kl[\resClass]{[#1]}}
\knowledge{\resClass}{notion}
\newrobustcmd{\resClassState}[1]{\kl[\resClassState]{[#1]}}
\knowledge{\resClassState}{notion}

\newrobustcmd\eqRes[1]{\mathrel{\kl[\eqRes]\sim_{#1}}}
\knowledge{\eqRes}{notion}
\newrobustcmd\eqResState[1]{\mathrel{\kl[\eqResState]\sim_{#1}}}
\knowledge{\eqResState}{notion}

\newrobustcmd\leqRes{\mathrel{\kl[\leqRes]\leq}}
\newrobustcmd\lRes{\mathrel{\kl[\lRes]<}}
\knowledge{\leqRes}[\lRes]{notion}

\newrobustcmd\leqResState{\mathrel{\kl[\leqResState]\sqsubseteq}}
\newrobustcmd\lResState{\mathrel{\kl[\lResState]\sqsubset}}
\knowledge{\leqResState}[\lResState]{notion}

\newrobustcmd\autRes[1]{\kl[\autRes]\R_{#1}}
\knowledge{\autRes}{notion}

\newrobustcmd\safeCoB[1]{\kl[\safeCoB]{\mathsf{Safe}}_{<2}(#1)}
\knowledge{\safeCoB}{notion}
\newrobustcmd\safeSig[2]{\kl[\safeSig]{\mathsf{Safe}}_{<#1}(#2)}
\newrobustcmd\safeSigAut[3]{\kl[\safeSig]{\mathsf{Safe}}^{#3}_{<#1}(#2)}
\knowledge{\safeSig}[\safeSigAut]{notion}

\newrobustcmd\leqCoB{\mathrel{\kl[\leqCoB]\sqsubseteq_{2}}}
\newrobustcmd\nleqCoB{\mathrel{\kl[\nleqCoB]{\not\sqsubseteq_{2}}}}
\newrobustcmd\nlCoB{\mathrel{\kl[\nlCoB]\nless_{2}}}
\newrobustcmd\lCoB{\mathrel{\kl[\lCoB]\sqsubset_{2}}}
\knowledge{\lCoB}[\leqCoB|\nleqCoB|\nlCoB]{notion}

\newrobustcmd\safeComp[1]{\kl[\safeComp]S^{{<}#1}}
\knowledge{\safeComp}{notion}

\newrobustcmd\BLetters[2]{\kl[\BLetters]B_{[#2]_{#1}}}
\knowledge{\BLetters}{notion}
\newrobustcmd\NLetters[2]{\kl[\NLetters]N_{[#2]_{#1}}}
\knowledge{\NLetters}{notion}

\newrobustcmd\BLettersBuchi[1]{\kl[\BLettersBuchi]B_{#1}}
\knowledge{\BLettersBuchi}{notion}
\newrobustcmd\NLettersBuchi[1]{\kl[\NLettersBuchi]N_{#1}}
\knowledge{\NLettersBuchi}{notion}

\newrobustcmd\locAlphBuchi[1]{\kl[\locAlphBuchi]\SS_{[#1]}}
\knowledge{\locAlphBuchi}{notion}
\newrobustcmd\locLangBuchi[2]{\kl[\locLangBuchi]{#1_{[#2]}}}
\knowledge{\locLangBuchi}{notion}
\newrobustcmd\localAutResBuchi[1]{\kl[\localAutResBuchi]\A_{[#1]}}
\knowledge{\localAutResBuchi}{notion}

\newrobustcmd\locAlph[2]{\kl[\locAlph]\SS_{[#2]_{#1}}}
\knowledge{\locAlph}{notion}
\newrobustcmd\locLang[2]{\kl[\locLang]{W_{[#2]_{#1}}}}
\knowledge{\locLang}{notion}
\newrobustcmd\localAutSig[2]{\kl[\localAutSig]\A_{[#2]_#1}}
\knowledge{\localAutSig}{notion}

\newrobustcmd\flooreven[1]{\kl[\flooreven]\lfloor #1 \rfloor^{\text{even}}}
\knowledge{\flooreven}{notion}

\newrobustcmd\geqSig[1]{\mathrel{\kl[\geqSig]\sqsupseteq_{#1}}}
\newrobustcmd\gSig[1]{\mathrel{\kl[\gSig]\sqsupset_{#1}}}
\newrobustcmd\leqSig[1]{\mathrel{\kl[\leqSig]\sqsubseteq_{#1}}}
\newrobustcmd\nleqSig[1]{\mathrel{\kl[\nleqSig]{\not\sqsubseteq_{#1}}}}
\newrobustcmd\lSig[1]{\mathrel{\kl[\lSig]\sqsubset_{#1}}}
\newrobustcmd\eqSig[1]{\mathrel{\kl[\eqSig]\sim_{#1}}}
\newrobustcmd\neqSig[1]{\mathrel{\kl[\eqSig]\nsim_{#1}}}
\knowledge{\eqSig}[\neqSig|\classSig|\leqSig|\lSig|\gSig|\geqSig|\nleqSig]{notion}

\newrobustcmd\classSig[2]{\kl[\classSig]{[#2]_{#1}}}

\newrobustcmd\hierMon[1]{\kl[\hierMon]{\mathsf{HM}(#1)}}
\knowledge{\hierMon}{notion}

\newrobustcmd{\Verts}[1]{\kl[\Verts]{V(#1)}}
\knowledge{\Verts}{notion}
\newrobustcmd{\Edges}[1]{\kl[\Edges]{E(#1)}}
\knowledge{\Edges}{notion}

\newrobustcmd{\addNeutral}[1]{\kl[\addNeutral]{#1^\ee}}
\knowledge{\addNeutral}{notion}

\newrobustcmd\autLeq[3]{\kl[\autLeq]{\quotientTuned{#1}{\sim_{#2}}{{{\leq}#3}}}}
\knowledge{\autLeq}{notion}

\newrobustcmd\autGeq[2]{\kl[\autGeq]{\restr{#1}{{\geq}#2}}}
\knowledge{\autGeq}[\autGeqState]{notion}
\newrobustcmd\autGeqState[2]{\kl[\autGeqState]{\restr{\A_{#2}}{{\geq}#1}}}

\newrobustcmd\quotAut[2]{\kl[\quotAut]{\quotient{#1}{\sim_{#2}}}}
\knowledge{\quotAut}{notion}

\newrobustcmd\mnext{{\mathsf{next}}}
\newrobustcmd\inext{\kl[\inext]i_{\mathsf{next}}}
\knowledge{\inext}{notion}

\newrobustcmd\pick{\kl[\pick]{f}}
\knowledge{\pick}{notion}

\newrobustcmd\rankOrd{\kl[\rankOrd]{\mathsf{rank}}}
\knowledge{\rankOrd}{notion}

\newrobustcmd\pickP{\kl[\pickP]{f}}
\knowledge{\pickP}{notion}

\newrobustcmd\pickMax{\kl[\pickMax]{f}}
\knowledge{\pickMax}{notion}

\newrobustcmd\leqHM[1]{\mathrel{\kl[\leqHM]\leq_{#1}}}
\newrobustcmd\lHM[1]{\mathrel{\kl[\leqHM]<{#1}}}
\newrobustcmd\geqHM[1]{\mathrel{\kl[\leqHM]\geq_{#1}}}
\newrobustcmd\gHM[1]{\mathrel{\kl[\leqHM]>_{#1}}}
\knowledge{\leqHM}[\geqHM]{notion}

\newrobustcmd\classHM[2]{\kl[\classHM]{[#2]_{#1}}}
\knowledge{\classHM}{notion}

\newrobustcmd\UPar{\kl[\UPar]U_{\hspace{-0.85mm}\mathsf{parity}}}
\knowledge{\UPar}{notion}
\newrobustcmd\UAut{\kl[\UAut]U_{\! \A}}
\knowledge{\UAut}{notion}
\newrobustcmd\TPar{\kl[\TPar]{T_{\mathsf{parity}}}}
\knowledge{\TPar}{notion}
\newrobustcmd\phiPar{\kl[\phiPar]{\phi_{\mathsf{parity}}}}
\knowledge{\phiPar}{notion}
\newrobustcmd\cleq{\mathrel{\kl[\cleq]{\preccurlyeq}}}
\knowledge{\cleq}{notion}

\newrobustcmd\ext{\kl[\ext]{{\mathsf{ext}}}}
\knowledge{\ext}{notion}

\newrobustcmd\UTop[1]{\kl[\UTop]{#1^{\top}}}
\knowledge{\UTop}{notion}

\newrobustcmd\countLetters[2]{\kl[\countLetters]{|#1|_{#2}}}
\knowledge{\countLetters}{notion}

\newrobustcmd\WInfMuller{\kl[\WInfMuller]{W_{\mathsf{fin}}}}
\knowledge{\WInfMuller}{notion}

\newrobustcmd\WOccParity{\kl[\WOccParity]{W_{\mathsf{OccParity}}}}
\knowledge{\WOccParity}{notion}

\mathchardef\hyphen=45 %Decimal
\usepackage[super]{nth}

\newcommand\restr[2]{{% we make the whole thing an ordinary symbol
		\left.\kern-\nulldelimiterspace % automatically resize the bar with \right
		#1 % the function
		\littletaller % pretend it's a little taller at normal size
		\right|_{#2} % this is the delimiter
}}

\newcommand{\littletaller}{\mathchoice{\vphantom{\big|}}{}{}{}}

\newcommand{\ts}{\textsuperscript}

\newrobustcmd\inv[1]{#1^{-1}}
\newcommand{\quotient}[2]{{\raisebox{0.1em}{$#1$\hspace{-1mm}}\left/\hspace{-0.5mm}\raisebox{-.1em}{{\scriptsize$#2$}}\right.}}

\newcommand{\quotientTuned}[3]{\quotient{#1}{\begin{smallmatrix*}[l]
		\phantom{.}\\[-0.9mm]
		\mathlarger{#2}\\[-0.2mm]
		\hspace{-1.6mm}\mathsmaller{\mathsmaller{\mathsmaller{#3}}}
\end{smallmatrix*}}}

\newcommand{\tand}{\text{ and }}
\newcommand{\tor}{\text{ or }}
\newcommand{\tin}{\text{ in }}
\newcommand{\tif}{\text{ if }}
\newcommand{\tow}{\text{ otherwise}}
\newcommand{\tst}{\text{ such that }}

\DeclareMathAlphabet{\mathpzc}{OT1}{pzc}{m}{it}

\newrobustcmd{\NN}{\mathbb{N}}
\newrobustcmd{\ZZ}{\mathbb{Z}}
\newrobustcmd{\QQ}{\mathbb{Q}}
\newrobustcmd{\RR}{\mathbb{R}}
\newrobustcmd{\CC}{\mathbb{C}}
\newrobustcmd{\WW}{\mathbb{W}}

\newrobustcmd{\I}{\mathcal{I}}
\newrobustcmd{\F}{\mathcal{F}}
\newrobustcmd{\D}{\mathcal{D}}
\newrobustcmd{\N}{\mathcal{N}}
\newrobustcmd{\G}{\mathcal{G}}

\renewcommand{\L}{\mathcal{L}}
\newrobustcmd{\M}{\mathcal{M}}
\newrobustcmd{\Q}{\mathcal{Q}}
\newrobustcmd{\C}{\mathcal{C}}
\newrobustcmd{\A}{\mathcal{A}}
\newrobustcmd{\B}{\mathcal{B}}
\newrobustcmd{\Z}{\mathcal{Z}}
\newrobustcmd{\R}{\mathcal{R}}
\newrobustcmd{\T}{\mathcal{T}}
\newrobustcmd{\U}{\mathcal{U}}
\newrobustcmd{\W}{\mathcal{W}}

\renewcommand{\O}{\mathcal{O}}
\renewcommand{\S}{\mathcal{S}}

\newrobustcmd{\kk}{\kappa}
\newrobustcmd{\uu}{\upsilon}
\newrobustcmd{\dd}{\delta}

\renewcommand{\ss}{\sigma}

\newrobustcmd{\rr}{\rho}

\renewcommand{\aa}{\alpha}

\newrobustcmd{\bb}{\beta}

\newrobustcmd{\oo}{\omega}
\newrobustcmd{\pp}{\varphi}

\newrobustcmd{\ee}{\varepsilon}

\renewcommand{\SS}{\Sigma}
\newrobustcmd{\GG}{\Gamma}
\newrobustcmd{\DD}{\Delta}

\knowledgerenewmathcommand\nu{\cmdkl{\LaTeXnu}}
\knowledgenewmathcommand\nuAcd{\cmdkl{\LaTeXnu}}
\knowledgerenewmathcommand\eta{\cmdkl{\LaTeXeta}}
\begin{document}

\maketitle              % typeset the header of the contribution
\begin{abstract}
 In the context of two-player games over graphs, a language $L$ is called positional if, in all games using $L$ as winning objective, the protagonist can play optimally using positional strategies, that is, strategies that do not depend on the history of the play.
 In this work, we describe the class of parity automata recognising positional languages, providing a complete characterisation of positionality for $\omega$-regular languages. 
 As corollaries, we establish decidability of positionality in polynomial time,   finite-to-infinite and 1-to-2-players lifts, and show the closure under union of prefix-independent positional objectives, answering a conjecture by Kopczyński in the $\omega$-regular case.
\end{abstract}
 
%
%

%\noindent This document contains hyperlinks. Each occurrence of a "notion" is linked to its ""definition"". On an electronic device, the reader can click on words or symbols (or just hover over them on some PDF readers) to see their definition.

\section{Introduction}\label{sec:intro}

\subsection{Context: Strategy complexity in infinite duration games}

We study games in which two antagonistic players, that we call Eve and Adam, take turns in moving a token along the edges of a given (potentially infinite) edge-coloured directed graph. 
Vertices of the graph are partitioned between Eve and Adam; when the token reaches a vertex, its owner
chooses where to move next.
This interaction goes on in a non-terminating mode, producing an infinite path in the graph called a "play". The winner of such a play is determined according to a language of infinite sequences of colours~$W$, called the "objective" of the game; plays producing a sequence of colours in~$W$ are winning for Eve, and plays that do not satisfy the "objective"~$W$ are winning for the opponent Adam.

One of the central applications of "games" on graphs is the problem of reactive synthesis: given a system interacting with its environment and a formal specification, we aim to design a controller that
guarantees the specification is met.
The interaction between the system and the environment can be modelled as a game where a winning strategy corresponds to a correct implementation of a controller~\cite{BL69Strategies,Thomas95SynthesisGames, BloemCJ2018Handbook}.

In this context, a crucial parameter is the complexity of "strategies" required by the players to "play optimally". Games admitting simple strategies are both easier to solve algorithmically, and the controllers obtained for them can be represented succinctly~\cite{BL69Strategies}.

\subparagraph{Positional strategies.} The simplest "strategies" are "positional@@strat" ones, those that depend only on the current vertex, and not on the history of the "play".
In this work, we are interested in the following question: 
Given a fixed "objective" $W$, can players always "play optimally using positional strategies" in all "games" with "winning objective" $W$? 
If the answer is affirmative for a single player ("Eve") we say that $W$ is "positional"\footnote{Sometimes in the literature the term ``half-positional'' or ``Eve-positional'' is used to stress the asymmetric nature of this notion.}; if it is affirmative for both players, we say that $W$ is "bipositional". 
Also, it might be relevant to consider the question for subclasses of games, in particular, for finite games, or for 1-player games.

\subparagraph{Bipositionality.} The class of "bipositional" "objectives", both over finite and infinite games, is already well understood.
A characterisation of "bipositionality" over finite games was obtained by Gimbert and Zielonka~\cite{GimbertZielonka2005Memory}, using two properties called \emph{monotonicity} and \emph{selectivity}. An important and useful corollary of their result is what is commonly known as a \emph{1-to-2-player lift}: an objective $W$ is "bipositional" over finite games if and only if both players can "play optimally using positional strategies" in finite 1-player games.

Over infinite games, a very simple and elegant characterisation of "bipositionality" was given by Colcombet and Niwiński for "prefix-independent" objectives~\cite{CN06}: a "prefix-independent" objective $W$ is "bipositional" if and only if it is the "parity objective".
In particular, these objectives are necessarily "$\oo$-regular".
No such characterisation is known for non-"prefix-independent" objectives (although a generalisation of this result for finite memory without the "prefix-independent" assumption is studied in~\cite{BRV22OmegaRegMemory}).
%Both the results of Gimbert and Zielonka and of Colcombet and Niwiński have been generalised to finite memory determinacy by Vandenhove and co-authors~\cite{Vandenhove23Thesis, BRORV20FiniteMemory, BRV22OmegaRegMemory}.

\subparagraph{Positionality.} Although "positionality" is arguably more relevant than "bipositionality" in the context of reactive synthesis (the controller is built based on "Eve's" "strategies"), much less is known for this class.
During the 90s, "positionality" of some central objectives was proved, notably of "parity@@obj"~\cite{EmersonJutla91Determinacy} and Rabin languages~\cite{Klarlund94Determinacy},
but the first thorough study of "positionality" was conducted by Kopczyński in his PhD thesis~\cite{Kop08Thesis}.
There, he provides some sufficient conditions for "positionality" (which were generalised in~\cite{BFMM11HalfPos}) and introduces an important set of conjectures that have greatly influenced research in the area in recent years (see~\cite[Sect.~6]{CFH22TenYears},~\cite[pg.~2 and ex.~8]{BFRV23Regular},~\cite[pg.~55]{Ohlmann21PhD},~\cite{BCRV22HalfPosBuchi,Casares2021Chromatic,Kozachinskiy22Chromatic,Kozachinskiy24EnergyGroups,Kozachinskiy24InfSeparation,OS24Sigma2} for works discussing some of his conjectures).
However, no general characterisation was found for "positionality".

Recently, Ohlmann made a step forward by characterising  "positionality" via "monotone@@graph" "universal graphs"~\cite{Ohlmann21PhD, Ohlmann23Univ}.
While this characterisation is a valuable tool for proving "positionality", it is not constructive and does not directly yield decidability results.
Also, Ohlmann's result comes with a caveat: necessity of the existence of "universal graphs" for "positional" objectives is only guaranteed for those containing a "neutral letter" (a letter that does not change membership to $W$ after its removal).
He conjectures that this restriction is not essential, as the addition of a "neutral letter" to any objective should not break "positionality".

\subparagraph{\texorpdfstring{$\oo$}{omega}-regular languages.} A central class of languages over infinite words is the class of "$\oo$-regular languages", which admits several alternative definitions: these are the languages "recognised" by "deterministic" "parity automata", by "non-deterministic" "B\"uchi automata", definable using $\oo$-expressions, or using monadic second order logic~\cite{Buchi1962decision, McNaughton1966Testing, Mostowski1984RegularEF}.

One of the main contributions of Kopczyński was to show decidability of "positionality" over finite games for "prefix-independent" "$\oo$-regular objectives"~\cite[Theorem~2]{Kop07OmegaReg}.
His procedure works by enumerating all possible games where positionality might fail (up to a sufficiently large size); it runs in $\O(n^{\O(n^2)})$ time, where $n$ is the size of a deterministic parity automaton recognising the objective, and does not reveal much about the structure of automata "recognising" "positional" languages. 

Regarding "positionality" over arbitrary games and for non-"prefix-independent" "objectives", characterisations have been found for some subclasses of "$\oo$-regular" "objectives".
For "closed objectives" (objectives recognised by safety automata), "positionality" was characterised by Colcombet, Fijalkow and Horn in 2014~\cite{ColcombetFH14PlayingSafe}. 

Recently, a characterisation of "positionality" for languages "recognised" by "deterministic" "B\"uchi automata" was provided by Bouyer, Casares, Randour and Vandenhove~\cite{BCRV22HalfPosBuchi} (see also Proposition~\ref{prop-warm:char-Buchi-all}).
As a corollary, they establish polynomial-time decidability of "positionality" for  "deterministic" "B\"uchi automata".
However, the conditions they provide are not necessary for "positionality" in general, for instance, for languages "recognised" by "coB\"uchi automata".
%Our objective is then:
%\begin{center}
%	\fbox{ \begin{tabular}{c}
%			Characterise the class of "positional" "$\oo$-regular" languages,\\
%			both over finite and infinite games.
%		\end{tabular}}
%\end{center} 

\subparagraph{Finite-to-infinite and 1-to-2-player lifts.} As mentioned above, a consequence of Gimbert and Zielonka's result~\cite{GimbertZielonka2005Memory} is that, in order to check "bipositionality" over finite games, it suffices to check whether players can "play optimally" in 1-player games.
Recently, generalisations of 1-to-2-player lifts have been studied in the setting of finite memory by Kozachinskiy~\cite{Kozachinskiy22Mildly} and Vandenhove~\cite{Vandenhove23Thesis,BRORV20FiniteMemory, BRV22OmegaRegMemory}.
Vandenhove conjectures that if $W$ is "positional" over "Eve-games" (resp. over finite games), then $W$ is "positional" over all games~\cite[Conjecture~9.1.1]{Vandenhove23Thesis}.
This conjecture has been shown to hold in the case of languages "recognised" by "deterministic" "B\"uchi automata"~\cite{BCRV22HalfPosBuchi}.

\subparagraph{Closure under union.} One of the recurring themes in Kopczyński's PhD thesis~\cite{Kop08Thesis} is the following question.

\AP \begin{conjecture}[{""Kopczyński's conjecture""~\cite[Conjecture~7.1]{Kop08Thesis}}]\label{conj:Kopcz-Union}
	Let $W_1,W_2\subseteq \SS^\oo$ be two "prefix-independent" "positional" objectives. Then $W_1\cup W_2$ is "positional".
\end{conjecture}

%\AP We refer to this question as "Kopczyński's conjecture".
Very recently, Kozachinskiy~\cite{Kozachinskiy24EnergyGroups} disproved this conjecture, but only for "positionality" over \emph{finite games}. Also, the counter-example he gives is not "$\oo$-regular".
On the positive side, "Kopczyński's conjecture" is known to hold in some subclasses of "$\oo$-regular" "objectives": Muller objectives~\cite{Zielonka1998infinite}, "concave" objectives~\cite{Kopczynski2006Half} and objectives "recognised" by "deterministic" "B\"uchi automata"~\cite{BCRV22HalfPosBuchi}, as well as for the family of $\SigmaTwo$ objectives (objectives "recognised" by infinite "coB\"uchi automata")~\cite{OS24Sigma2}. "Kopczyński's conjecture" and this latter result have been generalised to the setting of finite memory~\cite[Section~6.3]{CO25LMCS}. 
Solving "Kopczyński's conjecture" over infinite games is one of the driving open questions for the field.

\subsection{Contributions and organisation}

Our main contribution is a characterisation of "positionality" for "$\oo$-regular languages", stated in Theorem~\ref{th-reslt:MainCharacterisation-allItems}. 
We propose a syntactic description of a family of "deterministic" "parity automata", so that any automaton in this class "recognises" a "positional" language, and any "positional" language can be "recognised" by such an automaton. 
In fact, we describe two slightly different such families, called, respectively, "fully progress consistent" "signature automata" and "$\ee$-completable" automata. These families offer distinct advantages and complement our intuitions on "positionality".

From this characterisation, we derive multiple corollaries that address the majority of open questions related to "positionality" in the case of "$\oo$-regular languages":

\begin{enumerate}
	\item \bfDescript{Decidability in polynomial time.} Given a "deterministic" "parity automaton" $\A$, we can decide in polynomial time whether $\Lang{\A}$ is "positional" or not (Theorem~\ref{th-reslt:decid-poly}).
	
	\item \bfDescript{Finite-to-infinite and 1-to-2-players lift.} An "$\oo$-regular" objective $W$ is "positional" over arbitrary games if and only if it is "positional" over finite, "$\ee$-free" "Eve-games" (Theorem~\ref{th-half:lifts}). This answers a question raised by Vandenhove~\cite[Conjecture~9.1.1]{Vandenhove23Thesis}. 
	
	\item \bfDescript{Closure under union.} The union of two "$\oo$-regular" "positional" objectives is "positional", provided that one of them is "prefix-independent" (Theorem~\ref{th-reslt:union-PI}). This solves a stronger variant of "Kopczyński's conjecture" in the case of "$\oo$-regular languages".
	
	\item \bfDescript{Closure under addition of a neutral letter.} If $W$ is "$\oo$-regular" and "positional", the objective obtained by adding a "neutral letter" to $W$ is "positional" too (Theorem~\ref{th-reslt:neutral-letter}). This solves "Ohlmann's conjecture" in the case of "$\oo$-regular languages".
\end{enumerate}

We obtain some additional results for classes of objectives that are not necessarily "$\oo$-regular". We relax the "$\oo$-regularity" hypothesis in two orthogonal ways.

\begin{enumerate}\setcounter{enumi}{4}
	\item \bfDescript{Characterisation of bipositionality of all objectives.} We extend the characterisation of "bipositionality" of Colcombet and Niwiński~\cite{CN06} to all objectives, getting rid of the "prefix-independence" assumption (Theorem~\ref{th-bi:bipositional}).

	\item \bfDescript{Characterisation of positionality of closed and open objectives.} 
	We characterise "positionality" for "closed@@obj" and "open@@obj" objectives. We also obtain as corollaries 1-to-2 players lifts and closure under addition of a "neutral letter" for these classes of objectives.
\end{enumerate}

\subsubsection*{Technical tools} We would like to highlight some technical tools that take primary importance in our proofs.

\textbf{Universal graphs.} In general, showing that a given objective is "positional" can be challenging, as we need to show that \emph{for every game} "Eve" can "play optimally using positional strategies".
Ohlmann's characterisation using "monotone@@graph" "universal graphs" provides a clear path to prove "positionality" (see Proposition~\ref{prop-prelim:univ-graphs}).
We rely on this result to show that "parity automata" satisfying the syntactic conditions imposed in Theorem~\ref{th-reslt:MainCharacterisation-allItems} do indeed "recognise" "positional" languages.
	
\textbf{History-deterministic automata.} "History-deterministic" automata are a model in between "deterministic" and "non-deterministic" ones; we refer to~\cite{BL23SurveyHD,Kupferman22UsingPast} for detailed expositions on them.
	Although the statements of our results do not mention "history-determinism", they appear naturally in two different parts of our proofs:
	\begin{itemize}
		\item Establishing necessity of the syntactic conditions from our main characterisation requires a very fine control of the structure of automata. We develop a technique for decomposing automata, for which we need to use and generalise the methods introduced by Abu Radi and Kupferman~\cite{AK22MinimizingGFG} for the minimisation of "HD" "coB\"uchi automata". %This allows us to greatly simplify the structure of automata, at the cost of introducing "history-determinism".
		
		\item To prove that these conditions are sufficient, we construct a "monotone@@graph" "universal" graph from a "signature automaton". To facilitate this process, we first ``saturate'' automata, adding as many transitions as possible without modifying the languages they "recognise". This procedure generates non-determinism, but preserves "history-determinism", the key property that allows us to prove "universality" of the obtained graph.
	\end{itemize}
	
	We believe that this use of "history-determinism" showcases their usefulness and canonicity.
	%Our results are not aimed at understanding the succinctness of  expressiveness properties of "history-deterministic" automata; rather, they serve as a tool. This approach reveals that "history-determinism" offers a natural framework, enabling us to carry out various proofs that would otherwise be unattainable using exclusively "deterministic" automata.
	
	%Indeed, we do not study "history-determinism" for the sake of understanding the succinctness of expressivity of the model, they appear just as a tool, because they turn out to be a more natural and canonical model that allows to carry out many proofs that would be impossible to do using exclusively "deterministic" automata.

	\textbf{Normal form of parity automata.} In our central proof, we rely on a "normal form" of parity automata, as defined in~\cite[Section~6.2]{CCFL24FromMtoP}.
	Automata in "normal form" present a set of properties that simplify manipulating them and reasoning about their runs.
	We make consistent use of these properties in our combinatorial arguments.
	This "normal form" is commonly used in the literature and applied in areas such as the study of "history-deterministic" "coB\"uchi" "automata"~\cite{KS15DeterminisationGFG,AK22MinimizingGFG,EhlersSchewe22NaturalColors} or automata learning~\cite{BohnLoding23DetParityFromExamples}. 
	
  \textbf{Congruences for parity automata.} Since the beginning of the theory of finite automata, the notion of "congruence" has played a fundamental role~\cite{Arnold85Congruence,Saec90Saturating, MS97Syntactic}.
	Here, we propose a notion of "congruences" for parity automata that make it possible to build quotient automata that are compatible with the acceptance condition. This newly introduced vocabulary allows us to formalise the details of the proof of Theorem~\ref{th-reslt:MainCharacterisation-allItems} in a simpler way.
	We believe that it will be useful for the study of "parity automata" in other contexts.

\subsubsection*{Organisation of the paper}
%Given the technical complexity of this paper, the decisions regarding its presentation have been done with the utmost focus on clarity, aiming to help the reader understand the contributions, as well as the newly introduced definitions and techniques.
%To this end, we have incorporated many examples and intuitive explanations, including a dedicated warm-up section. 
%Consequently, the paper's length has expanded, and some aspects of its organisation might have deviated from conventional choices.
After introducing some general definitions and terminology used throughout the paper, we begin Section~\ref{sec:char-half-pos} by stating the characterisation result (Theorem~\ref{th-reslt:MainCharacterisation-allItems}) and its main consequences, without providing formal details about the technical concepts appearing in its statement.  Section~\ref{sec:warm-up} is a warm-up for the definitions used in the main characterisation and for the techniques used in its proof. We gradually introduce conditions that are necessary for "positionality", obtaining partial results and providing numerous examples along the way.
Section~\ref{sec:proof} contains the most technical part of the paper. We introduce the notions of "signature automata", "full progress consistency" and "$\ee$-complete automata" appearing in the statement of Theorem~\ref{th-reslt:MainCharacterisation-allItems}, and we give a proof of it. Nevertheless, most details in the proof of necessity are relegated to Appendix~\ref{appendix:proofs-necessity}.
In Section~\ref{sec:decision} we provide two conceptually different polynomial-time decision methods for deciding "positionality". 
Sections~\ref{sec:bipositional} and~\ref{sec:open-closed} contain, respectively, the two last contributions of the paper: a characterisation of "bipositionality" for all objectives and a characterisation of "positionality" for "open" and "closed" objectives. 
%The proofs appearing in these latter sections are much simpler.

\section{Preliminaries}\label{sec:prelim}

We introduce definitions and notations used throughout the paper. First we introduce in Section~\ref{subsec-prel:games} "games" and "positionality", as well as the more technical notion of "universal graphs" used in our proofs for "positionality". In Section~\ref{subsec-prel:automata} we introduce definitions about "parity automata" and notions about "congruences" for them.

The reader who does not plan to get into the more technical details of Sections~\ref{sec:proof} and~\ref{sec:decision} can skip Subsections~\ref{subsec-prel:universalG} and~\ref{subsec-prel:congruences} from this preliminaries. Also, the hyperlinks on words should help the reader to easily refer to the definitions.
	
\subsection{Games and positionality}\label{subsec-prel:games}	

\subsubsection{Games on graphs}
\subparagraph{\texorpdfstring{$\SS$}{Sigma}-graphs.} \AP A ""$\SS$-graph"" $G$ is given by a (potentially infinite) set of vertices $V$ together with a set of coloured directed edges $E\subseteq V\times \SS \times V$.
We write $v \re c v'$ to refer to an edge in $G$ with source $v$, target $v'$, and colour $c$.
This notation naturally extends to finite and infinite paths.
%Often, we define "$\SS$-graphs" simply by describing the set of transitions $v \re c v'$ that occur, and without formally giving names to edges.
\AP The size of a graph $G$ is defined to be the cardinality of $V$.

\begin{globalHyp*}
	\AP We assume throughout the paper that "$\Sigma$-graphs" do not contain ""dead-ends"", that is, every vertex has at least one outgoing edge. (This assumption is useful when considering infinite paths.)
\end{globalHyp*}

\subparagraph{Games.} \AP A ""game"" is an edge-coloured "graph" together with a set of winning sequences of colours and a partition of the vertices into those controlled by a player named ""Eve"" and her opponent, named ""Adam"".
\AP Formally, it is represented by a tuple $\G = \left(V,E, \VEve, W\right)$, where $G=(V,E)$ is a "$\SS\cup\{\ee\}$-graph" (called the ""game graph""),  $\VEve$ is the set of vertices owned by "Eve" and $W\subseteq \SS^\oo$ is the ""winning objective"".
Letter $\ee \notin \SS$ is an additional letter used to represent ""uncoloured edges""; we impose that no infinite path in $G$ is composed exclusively of $\ee$-edges.
\AP Games not containing "uncoloured edges" are called ""$\ee$-free"".
\AP We let $\VAdam = V\setminus \VEve$ be the vertices controlled by "Adam".
\AP An ""Eve-game"" is a "game" $\G$ in which all the vertices are controlled by "Eve", that is, $V = \VEve$.
\AP A~"game" having $W$ as "winning objective" is called a ""$W$-game"".

Unless stated otherwise, we take the point of view of player "Eve"; expressions as ``winning'' will implicitly stand for ``winning for Eve'', and "strategies" will be defined for her.

\AP In this paper, the words ``language'' and ``""objective""'' are synonyms.

\subparagraph{Plays.} 
In a game, players move a pebble from one vertex to another for an infinite amount of time.
\AP The player who owns the vertex $v$ where the pebble is placed chooses an edge  $v\re{c}v'$ and the pebble travels through this edge to its target, producing colour $c$. In this way, they produce a path $\rr=v_0\re{c_0}v_1\re{c_1}v_2\re{c_2}\dots\in E^\oo$, that we call a ""play"". 
\AP Such a "play" is ""winning@@play"" (for "Eve") if the sequence $w\in \SS^\oo$ obtained by removing from $c_0c_1c_2\dots$ the occurrences of $\ee$ belongs to $W$. We say that it is ""losing@@play"" (or \emph{winning for Adam}) on the contrary.

%\AP For convenience, for every vertex $v\in V$ we consider the ""empty path from $v$"", written $\lambda_v$. 
\AP We let $\intro*\Paths(\G)$ be the set of finite paths in $\G$; these are either non-empty sequences in $E^+$ or a vertex $v\in V$ (encoding the empty path starting in that vertex).

\subparagraph{Strategies and winning regions.} 
\AP A ""strategy"" (for "Eve") is a function $\strat: \Paths(\G) \to E$,  that tells "Eve" which move to choose after any possible "finite play".
\AP We say that a "play" $\rr\in E^\oo$ is ""consistent with the strategy"" $\strat$ if after each finite prefix $\rr'$ of $\rr$ ending in a vertex controlled by "Eve", the next edge in $\rr$ is $\strat(\rr')$.
\AP  We say that the strategy $\strat$ is ""winning from@@strat"" a vertex $v\in V$ if all infinite "plays" starting in $v$ "consistent with" $\strat$ are "winning@@play". If such a "strategy" exists, we say that "Eve" ""wins $\G$ from $v$"".
Strategies for "Adam" are defined symmetrically.

\AP The ""winning region"" of a game $\G$, written $\intro*\winRegion{\Eve}{\G}$, is the set of vertices $v\in V$ such that "Eve" "wins" $\G$ from $v$.
\AP We say that a "strategy" is ""optimal (for Eve)"" if it is "winning@@strat" from all vertices in $\winRegion{\Eve}{\G}$. (Note that Eve always has an "optimal strategy".)

% \begin{remark}\label{rmk-p0-prelims:optimal-strategy-exists}
% 	"Eve" always has an "optimal strategy", that is,
% 	there is a strategy $\strat_{\mathsf{opt}}$ that is "winning from" all $v\in \winRegion{\Eve}{\G}$.
% \end{remark}

\subparagraph{Determinacy.} We say that a game $\G$ is ""determined"" if either "Eve" or "Adam" have a "winning strategy from $v$", for every vertex $v$.
In this work, all "games" will be "determined", as by Martin's theorem~\cite{Martin75BorelDet} games using Borel "objectives" are "determined", and all "objectives" that we will consider (for instance, all "$\omega$-regular objectives") are Borel.

%\begin{theorem}[Martin’s theorem of Borel determinacy~\cite{Martin75BorelDet}]
%	Games using Borel "winning objectives" are "determined".
%\end{theorem}

\subparagraph{Graphical representation of games.}
We use circles to represent vertices controlled by "Eve" and squares to represent those controlled by "Adam". 
We will allow ourselves to consider games with edges labelled by finite words $w=w_1w_2\dots w_n\in \SS^*$. Formally, such transitions will stand for a sequence of $n$ transitions, with $n-1$ intermediate vertices. 
We represent this kind of transitions by a wiggly arrow.
%We point out that, however, in "$\ee$-free" games we cannot consider edges labelled by the "empty word".
We will also use this notation for infinite words: for $w\in \SS^\oo$ we write $v\lrp{w\phantom{.}}$ for an infinite sequence of edges labelled with the letters of $w$ starting from $v$. In this case, the resulting "game graph" is necessarily infinite.

\subsubsection{Positionality}

\subparagraph*{Positional strategies.}	
We say that a "strategy" $\strat: \Paths(\G) \to E$ is ""positional@@strategy"" if there exists a mapping $\ss\colon \VEve\to E$ such that for every finite "play" $\rr=v_0\re{c_0}\dots \re{c_{n-1}} v_n$  ending in a vertex $v_n$ controlled by "Eve" we have:
\[ \strat(\rr) = \ss(v_n). \] 
That is, a "strategy" is "positional@@strat" if the choice of the next transition only depends on the current position, and not on the history of the path.

\AP We say that "Eve" (resp. Adam) can ""win positionally"" from a subset $A\subseteq V$ if there is a "positional strategy" that is "winning" from any vertex in $A$.
\AP We say that "Eve" (resp. Adam) can ""play optimally in $\G$ using a positional strategy"" if she can "win positionally" from her "winning region".

\subparagraph*{Positional objectives.}	
An "objective" $W\subseteq \SS^\oo$ is ""positional"" if for every "$W$-game", "Eve" can "play optimally using positional strategies".\footnotemark{}  We say that $W$ is ""bipositional@@objective"" if both $W$ and $\SS^\oo\setminus W$ are "positional", or, equivalently, if both "Eve" and "Adam" can "play optimally using positional strategies" in "$W$-games".
\AP If $\mathcal{X}$ is a subclass of $W$-games (notably, finite, "$\ee$-free" and "Eve-games"), we say that $W$ is  ""positional over $\mathcal{X}$ games"" if for every "$W$-game" in $\mathcal{X}$, "Eve" can "play optimally using positional strategies". 
The same terminology is used for "bipositionality".

\begin{remark}
	\AP Our notion of "positionality" uses what sometimes are called ""uniform strategies"", that is, we require that a single "positional strategy" suffices to win independently of the initial vertex. 
	This notion is strictly stronger than the non-uniform version in which we allow to use different "strategies" depending on the initial vertex.
	Said differently, the fact that Eve always has "optimal strategies" does not hold if we restrict to "positional" strategies, see Figure~\ref{fig-p0-prelim:game-non-uniform} for an example.
	Nevertheless, we note that if $\varepsilon$-edges are allowed, or for "prefix-independent" objectives, both notions of "positionality" coincide, as we can always add a vertex controlled by Adam from which he picks the starting position. 
\end{remark}

\footnotetext{As in other definitions, the notion of "positionality" depends not only on the set $W$, but also on the set of colours $\SS$. As the set of colours will always be clear from the context, we omit $\SS$ in the notations.}

\begin{figure}
	\centering
	\includegraphics[scale=1.5]{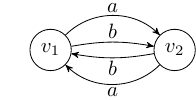}
	\caption{Consider the game above, where "Eve" controls both vertices $v_1$ and $v_2$. Let $W = ab(a+b)^\oo$ be the "winning condition" of the game, that is, "Eve" "wins" if the play starts by $ab$. She has two "positional strategies" $\strat_1$ and $\strat_2$ "winning" from $v_1$ and $v_2$, respectively. However, no "positional strategy" is winning from the entire "winning region" $\{v_1,v_2\}$.}
	\label{fig-p0-prelim:game-non-uniform}
\end{figure}

\subsubsection{Universal graphs}\label{subsec-prel:universalG}
We now introduce "universal graphs", which will serve as our main tool for deriving "positionality" results.

\subparagraph{Morphisms of \texorpdfstring{$\SS$}{Sigma}-graphs.} \AP Given two "$\SS$-graphs" $G=(V,E)$ and $G'=(V',E')$, a ""morphism of $\SS$-graphs@@univ"" $\phi$ from $G$ to $G'$ is a map $\phi:V\to V'$ such that for each edge $v \re c v' \tin G$, it holds that $\phi(v) \re c \phi(v')$ defines an edge in $G'$.
We write $\phi:G \to G'$ to denote that $\phi$ is a "morphism@@univ".

\subparagraph{Universality.} \AP Given a "$\SS$-graph" $G$, a vertex $v$ of $G$ and an "objective" $W \subseteq \SS^\omega$, we say that ""$v$ satisfies@@univ"" $W$ in $G$ if for any infinite path $v \lrp{w\phantom{.}} \tin G$, it holds that $w \in W$.
\AP Given a cardinal $\kappa$, a graph $U$ is ""$(\kappa,W)$-universal"" if all graphs $G$ of size $<\kappa$ admit a "morphism@@univ" $\phi:G \to U$ such that any vertex $v$ that "satisfies@@univ" $W$ in $G$ is mapped to a vertex $\phi(v)$ that "satisfies@@univ" $W$ in $U$.

\subparagraph{Monotonicity.}
\AP A ""totally ordered graph"" (resp. ""well-ordered graph"") is a graph $G$ together with a total order (resp. well-order) $\leq$ on its vertex set $V$.
\AP Such a graph is called ""monotone@@univ"" if
\[
u\leq v, \, v'\leq u' \tand u \re c u' \tin G \; \implies \; v \re c v' \tin G.
\]
We often note the conditions on the left by $v \geq u \re c u' \geq v'$.

We now state our main tool for proving "positionality".

\begin{proposition}[{\cite[Theorem 3.1]{Ohlmann23Univ}}]\label{prop-prelim:univ-graphs}
	Let $W \subseteq \SS^\oo$ be an "objective". If for
	all cardinals~$\kappa$ there exists a "$(\kappa, W)$-universal" "well-ordered@@univ" "monotone@@univ"  graph, then $W$ is
	"positional" over all "games".
\end{proposition}

\subparagraph{Universality for trees.}
%It is often more convenient to work with trees.
\AP A ""$\SS$-tree"" is a "$\SS$-graph" $T$ with a distinguished vertex $t_0$, called the ""root@@univ"", and such that every vertex $t$ of $T$ admits a unique path from the root.
Since graphs (and in particular trees) are assumed without "dead-ends", trees are always infinite.
\AP We say that a tree $T$ ""satisfies@@treeUniv"" $W$ if its root $t_0$ "satisfies@@univ" $W$ in $T$.

\AP We say that a graph $U$ is ""$(\kappa,W)$-universal for trees"" if all "trees@@univ" $T$ of size $<\kappa$ which "satisfy@@treeUniv" $W$ admit a "morphism@@univ" $\phi\colon T \to U$ mapping the "root@@univ" $t_0$ to a vertex $\phi(t_0)$ that "satisfies@@univ" $W$ in $U$. 

\AP Given an "ordered@@univ" "$\SS$-graph" $U$, we let $\intro*\UTop{U}$ be the "$\SS$-graph" obtained by adding a fresh vertex $\top$, maximal for the order of the graph, with transitions $\top \re a v$ for every $a\in \SS$ and every vertex $v$ of the graph.
The following useful result follows directly from the proof of~\cite[Theorem~3.1]{Ohlmann23Univ} (see also~\cite[Theorem~3.1]{CO25LMCS}).

\begin{lemma}\label{lemma-prelim:univ-for-trees}
	Let $W \subseteq \SS^\omega$ be an "objective" and $\kappa$ a cardinal.
	If $U$ is a "well-ordered@@graph" "monotone graph@@univ" that is "$(\kappa,W)$-universal for trees", then $\UTop{U}$ is "well-ordered@@graph" "monotone@@univ" "$(\kappa,W)$-universal"  (for graphs).
\end{lemma}

Therefore, thanks to Proposition~\ref{prop-prelim:univ-graphs}, building graphs that are "universal for trees" suffices to prove positionality.

\subparagraph{Universal graph for the parity objective.}
As an important example, we give a "universal graph" for the "parity objective"; it is implicit in the works of Emerson and Jutla~\cite{EmersonJutla91Determinacy} and Walukiewicz~\cite{Walukiewicz96}. In the latter, the term \emph{signatures}  was used to name tuples of ordinals ordered lexicographically (term first used in~\cite{StreettEmerson1989}).
Such a tuple is meant to count, for each odd priority, how many times it is seen before a stronger (even or odd) priority.

\AP
\begin{example}[Universal graph for the parity objective]\label{ex-warm:graph-parity}
	Consider the ""parity objective"" over $[0,d]$, (we assume $d$ even, and use min-parity):
	\[
	\intro*\parity_{} = \{w \in \{0,\dots,d\}^\oo \mid \liminf w \text{ is even}\}.
	\]
	Fix a cardinal $\kappa$. We define a graph $\intro*\UPar$ having as set of vertices tuples $(\lambda_1,\lambda_3,\dots,\lambda_{d-1}) \in \kappa^{d/2}$ that we consider ordered lexicographically. This is indeed a "well-order". We let its edges be:
	\[
	(\lambda_1,\dots,\lambda_{d-1}) \re x (\lambda'_1,\dots,\lambda'_{d-1}) \iff \begin{cases}
		(\lambda'_1,\dots,\lambda'_{x-1}) \leq (\lambda_1,\dots,\lambda_{x-1})  &\tif x \text{ is even}, \\
		(\lambda'_1,\dots,\lambda'_{x}) < (\lambda_1,\dots,\lambda_{x}) &\tow. 
	\end{cases}
	\]
	Where the order between truncated tuples as on the right is also the lexicographic one.
	A representation of the graph $\UPar$ appears in Figure~\ref{fig-prel:universal-parity}.
	
	Clearly, $\UPar$ is "monotone@@univ". 
	We show in Lemma~\ref{lemma-prelim:infinite-paths-in-UPar} below that all vertices in $\UPar$ "satisfy@@univ"~$\parity$. Lemma~\ref{lemma-prelim:UPar-Universal} states that $\UPar$ is  "$(\kappa,\parity)$-universal for trees", so, by Lemma~\ref{lemma-prelim:univ-for-trees}, $\UTop{\UPar}$ is a "well-ordered@@graph" "monotone@@univ" "$(\kappa,\parity)$-universal" graph.
\end{example}

\begin{figure}
	\centering
	\includegraphics[scale=1.4]{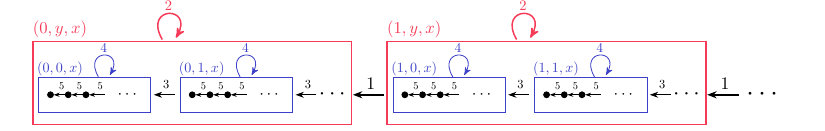}
	\caption{"Universal graph" $\UPar$ for the "parity objective" over "priorities" $[0,5]$.
		Vertices are ordered from left to right.
		Edges between two boxes $B_1 \re{x} B_2$ represent that there are edges $v_1\re{x} v_2$ for all $v_1\in B_1$ and all $v_2\in B_2$. 
		Edges obtained by monotonicity are not all represented: if $v\re{x}v'$ and $v''\leq v'$, then $v\re{x}v''$ too; for example, by reading colour $5$ from a vertex $v$ one can go to any vertex strictly on the left of $v$.
		Edges coloured $0$ are not depicted in the figure: they appear between every pair of vertices.
		The label of a box represents the forms of the names of vertices inside it.}
	\label{fig-prel:universal-parity}
\end{figure}

\begin{lemma}\label{lemma-prelim:infinite-paths-in-UPar}
	Every infinite path in $\UPar$ "satisfies@@graph" the "objective" $\parity$.
\end{lemma}
\begin{proof}
	Consider an infinite path $\rr = (\lambda_1^1,\dots,\lambda_{d+1}^1) \re{w_1} (\lambda_1^2,\dots,\lambda_{d+1}^2) \re{w_2} \dots$ in $\UPar$. Let $x$ be the minimal priority appearing infinitely often in $\rr$, which we assume odd for contradiction.
	Then, from some position, no priority $<x$ is read, thus the sequence of prefixes $(\lambda_1^i,\dots,\lambda_{x}^i)$ is decreasing and moreover strictly decreases in infinitely many places.
	This contradicts well-foundedness of the lexicographical order over tuples of ordinals.
\end{proof}

\begin{lemma}\label{lemma-prelim:UPar-Universal}
	The graph $\UPar$ is "$(\kappa,\parity)$-universal for trees".
\end{lemma}
\begin{proof}
	Take a "tree@@univ" $T$ of size $<\kappa$ which "satisfies@@treeUniv" $\parity$; note that by "prefix-independence", all vertices in $T$ "satisfy@@univ" $\parity$.
	We aim to construct a "morphism@@univ" $\phi:T \to \UPar$.
	
	Fix a vertex $t \in T$ and an odd priority $y \in \{1,3,\dots,d-1\}$.
	Then, in any path $t \lrp{w\phantom{.}}$ in $T$ there are only finitely many occurrences of $y$ before a smaller priority appears.
	\AP We may define an ordinal $\intro*\rankOrd_y(t)$ capturing the number of such occurrences: $\rankOrd_y(t)$ satisfies that, if $t'$ is an $x$-sucessor of $t$ (that is, $t\re{x}t'$), then:
	\begin{itemize}
		\item $\rankOrd_y(t')\leq \rankOrd_y(t)$, if $y<x$, and
		\item $\rankOrd_y(t')< \rankOrd_y(t)$ if $x=y$.
	\end{itemize}  %; formally, $\lambda_y(v)$ is the rank of the (well-founded) tree obtained by restricting $T$ to vertices $t'$ such that either $t \lrp{} t' \lrp{}$ blablabla inutile, 
	It can be verified that $\phi:v \mapsto (\rankOrd_1(v),\rankOrd_3(v),\dots,\rankOrd_{d-1}(v))$ defines a "morphism@@univ" from $T$ to $\UPar$.
\end{proof}

\subsection{Automata over infinite words}\label{subsec-prel:automata}

\subsubsection{Parity automata}

%\subparagraph*{Parity automata.}
\AP A ""(non-deterministic) parity automaton"" over the alphabet $\SS$ is represented by a tuple $\A = (Q,\Sigma, q_\init, \DD, \intro*\colAut)$, where $Q$ is a finite set of states, $\SS$ is a set of letters called the input alphabet (possibly infinite), $q_\init$ is the ""initial state"", $\DD\subseteq Q\times \SS \times Q$ is a set of transitions, and $\colAut\colon \DD \to [d_{\min},d_{\max}]$ 
%is a function assigning numbers to transitions; we refer to this numbers as ""priorities"". 
\AP where $[d_{\min},d_{\max}]\subseteq \NN$ is a finite subset of numbers that we refer to as ""priorities"".
We write $q\re{a:x}q'$ to indicate that there is a transition $e=(q,a,q')\in\DD$ with $\colAut(e)=x$.
%We write $q\re{a:x}q'$ to denote the transition $(q,a,x,q')\in\DD$.
\AP We refer to transitions of a "parity automaton" labelled with "priority" $x\in \NN$ as ""$x$-transitions@@out"". Similarly, we refer to transitions having input letter $a\in \SS$ as ""$a$-transitions@@in"". The difference between the two uses of the term should be clear from the context.
\AP An "automaton" $\A' = (Q',\Sigma', q'_\init, \DD', \colAut')$ is a ""subautomaton"" of $\A$ if $Q'\subseteq Q$, $\DD'\subseteq \DD$ and $\colAut'$ is the restriction of $\colAut$ to $\DD'$.

\AP For a state $q\in Q$, we write $\intro*\initialAut{\A}{q}$ for the automaton obtained by setting $q$ as "initial state".

\AP We assume in all the paper that all automata are ""complete"", that is if for every $q\in Q$ and $a\in \Sigma$, there is at least one transition $q\re{a:s}$. 

\AP A ""strongly connected component"" (shortened SCC) of an automaton $\A$ is a maximal set of states $S\subseteq Q$ such that any pair of states in $S$ are interreachable. We say that a "SCC" is ""trivial@@SCC"" if it is a singleton.
\AP A state $q$ is ""recurrent"" if it belongs to some "non-trivial" "SCC", and ""transient"" otherwise.

\AP An ""automaton structure"" $\S$ is an "automaton@@parity" without a colouring function $\colAut$, and $\A$ is a ""parity automaton on top of"" $\S$ if it has been obtained by defining a colouring $\colAut$ on $\S$.
%Formally, an automaton on top of $\S$ is a tuple $(\S,\colAut)$, with $\colAut\colon \DD \to [d_{\min},d_{\max}]$.

\subparagraph*{Runs and recognisability.}
\AP A  ""run over"" an infinite word $w=a_0a_1a_2\dots\in \SS^\oo$ in $\A$ is a path
\[\rr = q_\init \re{a_0:x_0}q_1\re{a_1:x_1}q_2\re{a_2:x_2}\dots\in \DD^\oo.\] 
\AP It is ""accepting@@run"" if 
\[ \min \{x\in \NN \mid x = x_i \text{ for infinitely many } i\}	\; \text{ is even},   \]
 and ""rejecting@@run"" if the minimal "priority" produced infinitely often is odd (note that we use the min-parity condition).
\AP A word $w\in \SS^\oo$ is  ""accepted@@word"" by $\A$ if there exists an "accepting@@runAut" "run over $w$".
\AP The ""language recognised"" by an automaton $\A$ is the set 
\[ \intro*\Lang{\A}= \{ w\in \Sigma^\oo \mid w \text{ is "accepted@@word" by } \A \}.\]
\AP Two "automata" "recognising" the same language are said to be ""equivalent@@aut"".
\AP A language is called ""$\oo$-regular"" if it can be "recognised" by a "parity" automaton.

%\subparagraph{Priority accepting a word.} 
\AP For an even "priority" $x$, we say that a word $w\in \SS^\oo$ can be ""accepted with priority $x$"" in $\A$ if there exists a "run over $w$" such that the minimal "priority" produced infinitely often is $x$.
\AP For an odd priority $x$, we say that $w$ ""is rejected with priority $x$"" if in every "run over $w$"
 the minimal "priority" produced infinitely often is $x$.
Note that in a non-deterministic automaton, not all rejected words need to have a well-defined rejecting priority.

\begin{remark}[Transition-based acceptance]
	We emphasise that in our definition, the acceptance condition is put over the \emph{transitions} of the automaton. This will be a crucial element in our characterisation. We refer to~\cite[Chapter~VI]{Casares23Thesis} for further discussions on the comparison between transition-based and state-based automata.
\end{remark}

\subparagraph*{\texorpdfstring{B\"uchi}{Büchi} and \texorpdfstring{coB\"uchi}{coBüchi} automata.}
\AP A ""B\"uchi automaton"" is a "parity automaton" using $[0,1]$ as its set of "priorities". %In this case, transitions carrying priority $0$ are called ""B\"uchi transitions"".
\AP "Parity automata" using $[1,2]$ as set of "priorities" are called ""coB\"uchi@@aut"". %and transitions carrying "priority" $1$ are called ""coB\"uchi transitions"".
\AP We say that a language $W$ is ""B\"uchi recognisable"" (resp. ""coB\"uchi recognisable"") if it can be "recognised" by a "deterministic" "B\"uchi automaton"  (resp. "deterministic" "coB\"uchi automaton"). We note that these classes are incomparable and strict subclasses of the "$\oo$-regular languages".

\AP For $u\in\SS^*$, we write $\intro*\infOften(u)$ (resp. $\intro*\finOften(u)$) to denote the language of infinite words containing infinitely often (resp. finitely often) the factor $u$. We note that these languages are "B\"uchi@@rec" and "coB\"uchi recognisable", respectively.
We also write $\intro*\noOcc(u)$ for the language of infinite words avoiding any occurrence of the factor $u$.

\subparagraph{Determinism and homogeneity.} 
\AP We say that an "automaton" $\A$ is ""deterministic"" if for every $q\in Q$ and $a\in \Sigma$, there is only one $a$-transition $q\re{a:x}$ outgoing from $q$.
\AP Let $\DD'\subseteq \DD$ be a subset of transitions of an "automaton"~$\A$. We say that $\A$ is ""deterministic over~$\DD'$"" if the restriction of $\A$ to $\DD'$ is "deterministic", that is, for each letter $a$, there is a most one outgoing "$a$-transition@@out" in $\DD'$ from each state. %if $q\re{a}p\in \DD'$, then no other "$a$-transition@@out" outgoing from $q$ is in $\DD'$. 
\AP Any "parity" automaton admits an "equivalent" "deterministic" one~\cite{Mostowski1984RegularEF}.

\AP We say that a "parity automaton" $\A$ is ""homogeneous"" if for every state $q\in Q$ and letter $a\in \SS$, if $q\re{a:x}p$ is a transition in $\A$, then any other "$a$-transition@@in" from $q$ produces "priority"~$x$.
%For $0\leq x\leq d$, we say that $\A$ is ""$({\leq}x)$-homogeneous"" if the previous property is satisfied only for priorities $\leq x$: $a$-transitions outgoing from $q$ either all have the same priority $\leq x$, or they all have priority $>x$.

\begin{remark}
	Let $\A$ be a "homogeneous" "parity automaton" that is "deterministic over" transitions producing priority $x$. If $q\re{a:x}p$ is a transition in $\A$, then there is no other outgoing "$a$-transition@@out" from $q$.
\end{remark}

\subparagraph{Notations for paths.} 
For two states $q, p$ of a "parity automaton" $\A$ and a finite word $w\in \SS^*$, we write $q\lrp{w:x}p$ if there exists a path from $q$ to $p$ labelled $w$ such that the minimal priority appearing on it is $x\in \NN$.
We write $q\lrp{w:\geq x}p$ (resp. $q\lrp{w:> x}p$) to denote that there exists such a path producing no "priority" $< x$ (resp. $\leq x$). This is possibly the empty path $q\re{\ee}q$, producing no priority.
We use similar notations for $\leq x$ and $< x$.
We generalise these notations for infinite paths: for an infinite word $w\in \SS^\oo$ we write $q\lrp{w:x}$ if there exists an infinite path from $q$ labelled $w$ such that the minimal priority seen on it is $x$.

We may apply this notations to "non-deterministic" automata -- hence the use of an existential quantification -- though in most cases we will work with "deterministic" ones.

\subparagraph*{History-deterministic automata.}
Let $\A= (Q,\Sigma, q_\init, \DD,\colAut)$ be a (non-deterministic) "parity automaton".
\AP A ""resolver@@aut"" for $\A$ is a function $\resolv\colon \DD^* \times \SS \to \DD$ such that, for all words $w = a_0a_1\dots \in \SS^\oo$, the sequence $e_0e_1\dots \in \DD^\oo$, called the ""run induced by@@aut"" $\resolv$ over $w$ and defined by $e_i = \resolv(e_0\dots e_{i-1}, a_i)$, is actually a "run over" $w$ in $\A$.
\AP We write $q_\init \intro*\lrpResolver{\resolv}{w:x} q$ to denote that the "run induced by@@aut" $\resolv$ over $w$ produces $x$ as minimal priority and lands in $q$.

\AP We say that the resolver is ""sound@@aut"" if it satisfies that, for every $w\in \Lang{\A}$, the "run induced by@@aut" $\resolv$ over $w$ is an "accepting run@@aut". 
In other words, $\resolv$ should be able to construct an "accepting run@@aut" in $\A$ letter-by-letter with only the knowledge of the word so far, for all words in $\Lang \A$.

\AP An "automaton" $\A$ is called ""history-deterministic"" (shortened HD) if there is a "sound resolver@@aut" for it. "History-deterministic" automata are sometimes called good-for-games in the literature, we refer to~\cite{BL23SurveyHD} for a discussion on the relation between these notions and a survey on them.

\subparagraph*{Normal form of parity automata}
Let $\A = (Q,\Sigma, q_\init, \DD, \colAut)$ be a "parity automaton". We say that a labelling $\colAut'\colon \DD \to [d_{\min}',d_{\max}']$ is ""equivalent to@@labelling"" $\colAut$ over $\A$ if for every cycle $\ell\subseteq \DD$, $\min \colAut(\ell)$ is even if and only if $\min \colAut'(\ell)$ is even.
We say that a "SCC" of a "parity automaton" is ""positive@@SCC"" if the minimal "priority" appearing on it is even, and ""negative@@SCC"" otherwise.

\begin{definition}[{Normal form~\cite{CCFL24FromMtoP}}]\label{def-prelim:normality-recall}%\label{lemma-prelim:normality-recall}
	\AP A "parity automaton" $\A$ is in ""normal form"" if it holds that for every pair of states $q, p$ in a same "positive SCC" (resp. "negative SCC"), whenever there is a path $q\lrp{w:x} p$ producing $x$ as minimal "priority", then, for every $0\leq y\leq x$ (resp. $1\leq y\leq x$), there is a returning path $p\lrp{w':y}q$  producing $y$ as minimal priority.\footnote{This notion can be refined in the natural way to fit automata not using priority $0$ at all (for instance "coB\"uchi automata"). We refer to~\cite{CCFL24FromMtoP} for formal details.}
	%\footnote{Note that this definition is not suited when working with automata not using priority 0: for instance a coB\"uchi automaton may not have an equivalent coB\"uchi labelling in normal form due to the SCCs where only priority 2 occurs. When working with such automata, for instance in Section~\ref{subsec-warm:coBuchi} we will need to slightly adjust the definition. We refer to }
\end{definition}

That is, if $\A$ is in "normal form", the restriction of $\A$ to priorities $\geq x$ consists in a disjoint union of "strongly connected components". Moreover, if priority $y>x$ appears in one of these "SCCs", then all priorities between $x$ and $y$ appear on it.

Every "parity automaton" admits an "equivalent labelling" so that the obtained automaton is in "normal form", and this labelling is unique (except for the labelling of edges changing of "SCC", for which no condition is imposed).
\AP This labelling is the one assigning to each transition the smallest possible priority~\cite[Theorem~6.27]{CCFL24FromMtoP}.
Moreover, it can be computed in polynomial time.
We refer to this process as the ""normalisation"".

\begin{proposition}[{\cite{CartonMaceiras99RabinIndex}}]\label{prop-prelim:normalF-polytime}
	Given a "parity automaton" $\A$, we can compute in polynomial time an "equivalent labelling" defining an automaton in "normal form".
\end{proposition}

\subparagraph{Automata with \texorpdfstring{$\ee$}{epsilon}-transitions.}
\AP An ""automaton with $\ee$-transitions"" is defined just as an "automaton" over the alphabet $\SS \cup \{\ee\}$, where $\ee \notin \SS$ is a distinguished letter.
The language of an automaton $\A$ with $\ee$-transitions is the set of words $w \in \SS^\omega$ such that there exists $w' \in (\SS \cup \{\ee\})^\omega$ which is "accepted by" $\A$ and such that $w$ is obtained from $w'$ by removing all occurrences of the letter $\ee$.

\subparagraph*{Alphabets of words.}
As an important element of our main proof, we will need to consider automata whose transitions are labelled from an alphabet $A \subseteq \SS^+$ of finite words.
Such an automaton defines a language $L \subseteq A^\omega$ which we would like to see as a language $L \subseteq \SS^\omega$; however, this may pose a problem if a word $w \in \SS^\omega$ admits several decompositions in $A^\omega$.

\AP We say that a set $A\subseteq \SS^+$ is a ""uniquely decodable alphabet"" if any word $w\in \SS^\omega$ admits a unique decomposition as elements of $A$: for any infinite sequences $a_1,a_2,\dots$ and $a'_1,a'_2,\dots$ of elements of $A$, if $a_1 a_2 \dots = a'_1 a'_2 \dots$ then $a_i=a'_i$ for all $i$.

\AP We will only consider alphabets of words $A \subseteq \SS^+$ which are ""prefix codes"": if $a \in A$, then no proper prefix of $a$ belongs to~$A$.
It is an easy check that these are "uniquely decodable", and therefore one may indeed see a language $L \subseteq A^\omega$ as $L \subseteq \SS^\omega$.

\subsubsection{Congruences and monotone preorders over automata}\label{subsec-prel:congruences}
%\addcontentsline{toc}{subsection}{Congruences and monotone preorders over automata}
	
	\subparagraph{Equivalence relations and preorders.}
	\AP We will use $\sim_X$ to denote different equivalence relations, and $[q]_X$ to denote the ""equivalence class"" of an element $q$ (which is usually a state in an automaton).
	
	\AP A ""preorder"" $\sqsubseteq_X$ is a binary relation that is reflexive and transitive. (We reserve the symbol $\sqsubseteq$ for preorders over states of automata.)
	\AP The ""equivalence relation induced from a preorder"" $\sqsubseteq_X$ is the relation defined as:
	\[ q \sim_X q'  \; \iff\; q \sqsubseteq_X q' \tand q' \sqsubseteq_X q.\]
	\AP Given a "preorder" $\sqsubseteq_X$, we always write $\sim_X$ for the induced equivalence relation, and simply write $\sqsubseteq$ for the induced order over equivalence classes, for instance we may write $[q]_X \sqsubseteq [q']_X$.
	\AP A "preorder" is ""total@@preorder"" if every pair of elements are comparable.

	\AP Let $\R_1$ and $\R_2$ be two binary relations over a set $A$ (usually "preorders" or equivalence relations). We say that $\R_1$ is a ""refinement"" of $\R_2$ if for all $q,p\in A$, $q\mathrel{\R_1} p$ implies $q\mathrel{\R_2} p$.
	We note that if $\sqsubseteq_1$ is a "preorder" "refining" $\sqsubseteq_2$, then the "induced equivalence relation" $\sim_1$ "refines"~$\sim_2$.

	\subparagraph{Congruences, uniformity and monotonicity.}
	Let $\A$ be a (possibly "non-deterministic") "automaton" over $\SS$ with states $Q$ and transitions $\DD$.
	Let $\sim$ be an equivalence relation over $Q$ and let $\DD'\subseteq \DD$ be a subset of transitions (usually $\DD'$ will be the set of transitions using a given "priority").
	\AP We say that transitions of $\DD'$ ""are uniform over $\sim$-classes"" if for all $q\sim q'$ and $a \in \SS$, if $q \re {a}p\in \DD'$ then all "$a$-transitions@@out" $q' \re{a}$ are in $\DD'$.
	\AP We say that $\sim$ is a ""congruence for"" $\DD'$ (or that transitions in $\DD'$ \emph{preserve $\sim$}) if for all $q\sim q'$ and $a \in \SS$, if $q \re {a}p\in \DD'$ then there exists $q' \re{a}p' \in \DD'$, and for all such transitions $p\sim p'$. If $\DD'=\DD$, we just say that $\sim$ is a ""congruence"".
	\AP We say that $\sim$ is a ""strong congruence for"" $\DD'$ if, moreover, we have the equality $p= p'$ for transitions as above.
	
	\begin{remark}
		If $\A$ is "deterministic" and $\sim$ is a "congruence for" $\DD'$, then these transitions "are uniform over $\sim$-classes".
	\end{remark}
	
	\begin{figure}
		\centering
		\includegraphics[scale=1.5]{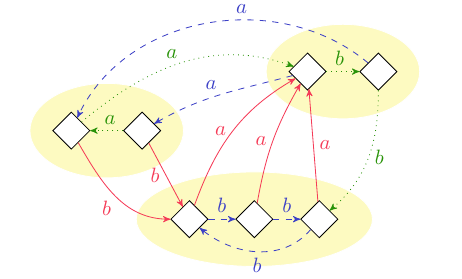}
		\caption{Representation of the notions of "uniformity", "congruence" and "strong congruence". We picture an automaton with three equivalence classes, each of them represented by a yellow bubble. Green-dotted transitions are "uniform" over the classes, but the relation is not a "congruence" for them. The relation is a "congruence for" blue-dashed transitions, and a "strong congruence for" red-solid transitions.}
		%\label{fig-open:aut-reset-stable}
	\end{figure}

	Let $\sqsubseteq$ be a "preorder" over $Q$.
	\AP We say that transitions in $\DD'$ ""are monotone for"" $\sqsubseteq$ if for all $q \sqsubseteq q'$ and $a \in \SS$, if $q \re a p\in\DD'$ then there exists $q' \re a p'\in \DD'$ and for all such transitions, $p\sqsubseteq p'$.
	\AP Transitions in $\DD'$ are said strictly monotone for $\sqsubseteq$ if, moreover, whenever $q < q'$,  $q \re a p\in\DD'$ and $q'\re a p'\in \DD'$, we have $p<p'$.
	If $\DD'=\DD$, we simply say that $(\A,\sqsubseteq)$ is ""(strictly) monotone@@aut"".

	All these properties can equivalently be stated with words $w \in \SS^*$ instead of letters $a \in \SS$.

	\begin{remark}
		If transitions in $\DD'$ are "monotone for a preorder" $\sqsubseteq$, then its "induced equivalence relation" is a "congruence for $\DD'$".
	\end{remark}

	\subparagraph{Quotient by a congruence.} Let $\A$ be an automaton and let $\sim$ be a "congruence" over its set of states $Q$. 
	\AP We define the ""quotient of $\A$ by $\sim$"" to be the "automaton structure" $\intro*\quotAut{\A}{}$ given by:
	\begin{itemize}
		\item The set of states are the $\sim$-classes.
		\item There is a transition $[q]\re{a}[p]$ if there are $q'\in [q]$, $p'\in [p]$ such that $q'\re{a}p'$ in $\A$.
		\item The "initial state" is $[q_\init]$.
	\end{itemize}
	
	We note that if $\sim$ comes from a "monotone@@cong" preorder, the obtained "automaton structure" $\quotAut{\A}{}$ with the induced order over the classes is "monotone@@aut".
	
	\begin{remark}
		The quotient $\quotAut{\A}{}$ is a "deterministic" "automaton structure".
	\end{remark}

	\AP A "run over a word" $w$ in $\A$, $\rr = q_0\re{w_0} q_1\re{w_1}\dots$ naturally induces a "run over" $w$ in $\quotAut{\A}{}$, $[q_0]\re{w_0} [q_1]\re{w_1}\dots$, that we call the ""projection of $\rr$@@quot"" in the quotient automaton. 
	
	\begin{lemma}\label{lemma-prelim:induced-run-quotient}
		Let $\sim$ be a "congruence" in $\A$.
		Any run in $\quotAut{\A}{}$ is the "projection@@quot" of some "run" in $\A$.
	\end{lemma}
	\begin{proof}
		Let $[q_0]\re{w_0} [q_1]\re{w_1}\dots$ be a "run" in $\quotAut{\A}{}$. We build the desired run in $\A$ recursively. For the base case, it suffices to take $p_0\in [q_0]$ to be the "initial state" of $\A$ (which belongs to $[q_0]$ by definition of the "initial state" of $\quotAut{\A}{}$).
		Suppose that $p_0\re{w_0} p_1\re{w_1}\dots p_k$ has already been built, with $p_i\in [q_i]$. By definition of the "quotient automaton", there are $q_k'\in [q_k]$ and $q_{k+1}'\in [q_{k+1}]$ with $q_k'\re{w_k} q_{k+1}'$. By the definition of a "congruence", there is a transition $p_k\re{w_k} p_{k+1}$ and $p_{k+1}\in [q_{k+1}]$.
	\end{proof}

\subparagraph{Notations on paths in automata with a congruence.}
Let $\sim$ be a "congruence" over the states of a "parity automaton" $\A$. We write $[q]\re{a:x} [p]$ if for all $q'\sim q$, every "$a$-transition" from $q'$ is of the form $q'\re{a:x}p'$ with $p'\sim p$.
We extend this notation to paths $[q]\lrp{w:x} [p]$ and for outputs ${\leq}x$, ${<}x$, ${\geq}x$ and ${>}x$ in the natural way.
\begin{remark}
	If $\sim$ is a "congruence for" transitions producing priority $x$, $q\lrp{w:x} p$ implies $[q]\lrp{w:x} [p]$.
\end{remark}

Let $\A$ be a "history-deterministic" automaton with initial state $q_\init$ and let $\resolv$ be a "resolver" for it. 
We recall that we write $q_\init\lrpResolver{\resolv}{u:x}q$ to denote the "run induced by@@res" $\resolv$ over $u$. We use the same conventions as above regarding outputs with the symbols ${\leq}x$, ${<}x$, ${\geq}x$ and ${>}x$.

\AP For $u_0\in \SS^*$, we write $q\intro*\lrpResolverWord{\resolv}{u_0}{w:x}p$ if:
\begin{itemize}
	\item $q_\init\lrpResolver{\resolv}{u_0}q$, and
	\item the "induced run of $\resolv$" over $u_0w$ ends in $p$ and produces $x$ as minimal "priority" in the part of the run corresponding to $w$.
\end{itemize}
\AP We write $q \intro*\lrpExistsResolver{\resolv}{w:x}p$ if $q\lrpResolverWord{\resolv}{u_0}{w:x}p$ for some $u_0\in \SS^*$.
\AP We write $q \intro*\lrpAllResolver{\resolv}{w:x}p$ if, for any word $u_0\in \SS^*$ such that $q_\init\lrpResolver{\resolv}{u_0}q$, we have $q\lrpResolverWord{\resolv}{u_0}{w:x}p$.

If $\sim$ is a "congruence" in $\A$, we write $[q] \lrpAllResolver{\resolv}{w:x}[p]$ if, for any word $u_0\in \SS^*$ such that $q_\init\lrpResolver{\resolv}{u_0}q'\in[q]$, we have $q'\lrpResolverWord{\resolv}{u_0}{w:x}p'\in [p]$.
We avoid using this notation for paths quantified existentially, as we consider that the corresponding semantics are not as intuitive.

\subsubsection{Residuals and semantic determinism}
%\addcontentsline{toc}{subsection}{Residuals}
	\subparagraph{Residuals of a language.}	
	Let $L\subseteq \SS^\oo$ be a language of infinite words and let $u\in \SS^*$. 
	\AP We define the ""residual of $L$ with respect to"" $u$ by
	\[
	\intro*\lquot{u}{L} = \{ w\in \SS^\oo \mid uw\in L\}.
	\] 
	\AP We denote $\intro*\Res{L}$ the set of "residuals" of $L$, which we will always order by inclusion.
	% This induces a partial "preorder" $\leq_L$ over $\SS^*$ defined by
	% \[
	% u \leq_L u' \iff \lquot{u}{L} \subseteq \lquot{u'}{L}.
	% \]
	\AP This induces an "equivalence relation" $\intro*\eqRes{L}$ over $\SS^*$ given by the equality of "residuals".
	\AP The corresponding equivalence classes $\intro*\resClass{u}=\{u'\in \SS^*\mid \lquot{u}{L} = \lquot{u'}{L} \}$ are called ""residual classes"".
	%, and we sometimes write jus as $\resClass{u}$ when $L$ is clear from context. 
	We write $\resClass{u} \intro*\leqRes \resClass{u'}$ if $\lquot{u}{L} \subseteq \lquot{u'}{L}$ (and $\lRes$ if this inclusion is strict).
	
%	Intuitively, we interpret the order on $\Res{L}$ as a measure of how far we are of belonging to $L$.
%	A word $u\in \SS^*$ such that $\resClass{u}$ is big satisfies that it has many continuations in $L$, while if $\resClass{u}$ is small, it has few of them.
%	

	\begin{remark}\label{rmk-prelim:left-quot-omega-regular}
		If $L$ is "$\oo$-regular", $\Res{L}$ is finite, and for all $u\in \SS^*$, $\lquot{u}{L}$ is also "$\oo$-regular". Contrary to the case of finite words, there are non "$\oo$-regular" languages with a finite set of "residuals".
	\end{remark}
	
	We now state a key monotonicity property for residuals; its proof is a direct check.
	
	\begin{lemma}\label{lemma-prelim:monotonicity-residuals}
		For any language $L \subseteq \SS^\oo$ and for any finite words $u,u',w \in \SS^*$, if $\resClass{u} \leq \resClass{u'}$ then $\resClass{uw} \leq \resClass{u'w}$. In particular, if $\resClass{u} = \resClass{u'}$ then $\resClass{uw} = \resClass{u'w}$.
	\end{lemma}

	\subparagraph{Prefix-independence.}	\AP We say that a language $L\subseteq \SS^\oo$ is ""prefix-independent""\footnote{In some parts of the litterature, prefix-independence is referred to as shift-invariance.} if for all $w\in \SS^\oo$ and $u\in \SS^*$,  $uw\in L$ if and only if $w\in L$.
	Equivalently, $L$ is "prefix-independent" if and only if $\Res{L}$ is a singleton.

	\subparagraph{Semantic determinism.} 
	\AP We say that an "automaton" $\A$ is ""semantically deterministic"" if for all state $q\in Q$, letter $a\in \SS$ and transitions $q\re{a}p_1$ and $q\re{a}p_2$, it is satisfied that $\Lang{\initialAut{\A}{p_1}} = \Lang{\initialAut{\A}{p_2}}$, where   $\initialAut{\A}{p}$ is the automaton obtained by setting $p$ as "initial state". 
	
	\begin{lemma}[\cite{KS15DeterminisationGFG}]
		Any "history-deterministic" "automaton" contains an "equivalent@@aut" "semantically deterministic" and "history-deterministic" "subautomaton". Moreover, for "parity automata", this "subautomaton" can be computed in polynomial time.
	\end{lemma}

	\begin{globalHyp*}%{\ref{part:p4}}
		We assume in the whole paper that "history-deterministic" automata are "semantically deterministic".
	\end{globalHyp*}
	
	We refer to \cite{AK23SemDet} for more details on "semantically deterministic" automata.
	
	\subparagraph{Residual associated to a state.} 
	%\AP For $\A$ an "automaton" "recognising" $L$ and $q$ a state of $\A$, we will write $\intro*\lquotState{q}{L}$ to denote $\Lang{\initialAut{\A}{q}}$, and 
	\AP We write $q\intro*\eqResState{\A} q'$ if $\Lang{\initialAut{\A}{q}}= \Lang{\initialAut{\A}{q'}}$ (and drop the subscript if $\A$ is clear from the context). The ""class of $q$"", written $\intro*\resClassState{q}$, is the set of states equivalent to $q$.
	
	\AP If $\A$ is "semantically deterministic" and $q$ is reachable, $\Lang{\initialAut{\A}{q}}$ coincides with 
	$\lquot{u}{L}$ for any word $u\in \SS^*$ leading to $q$ from the "initial state" of $\A$. In that case, we say that the "class of states" $\resClassState{q}$ is ""associated to@@residual"" the "residual class" $[u]$.
	%We note that this definition does not depend on the choice of $u$.
	\AP The inclusion of these languages induces a "preorder" $\intro*\leqResState_\A$ on the states of $\A$ (with its corresponding equivalence relation). By Lemma~\ref{lemma-prelim:monotonicity-residuals}, if $\A$ is "semantically deterministic", the relation $\eqResState{\A}$ is a "congruence" and the "preorder" $\leqResState_{\A}$ makes $\A$ a "monotone automaton".

\begin{remark}\label{rmk-prelim:sem-det-iff-residuals-is-congruence}
	An "automaton" $\A$ without unreachable states is "semantically deterministic" if and only if $\eqResState{\A}$ is a "congruence".
\end{remark}

	\subparagraph{Automaton of residuals.} Let $L\subseteq \SS^\oo$ be a language of infinite words.
	\AP The ""automaton of residuals"" of $L$ is a "deterministic" "automaton structure" $\intro*\autRes{L}$ over $\SS$  defined as follows:
	\begin{itemize}
		\item The set of states is the set of "residual classes" of $\Res{L}$:  $Q = \{\resClass{u} \mid u\in \SS^*\}$.
		\item The "initial state" is $\resClass{\ee}$.
		\item For each state $\resClass{u}$ and letter $a\in \SS$, it contains the transition $\resClass{u}\re{a}\resClass{ua}$.
	\end{itemize}

	We will be interested in the question of whether we can define a "parity@@aut" or "B\"uchi automaton" "on top of" the $\autRes{L}$ so that the obtained "automaton" "recognises"~$L$.

	The states of $\autRes{L}$ are ordered by the inclusion of "residuals". By Lemma~\ref{lemma-prelim:monotonicity-residuals},  transitions of $\autRes{L}$ are "monotone@@cong" for this order.
	
	\begin{remark}\label{rmk-prelim:residual-aut-as-quotient}
		For any "semantically deterministic" automaton $\A$ "recognising" $L$, the "automaton of residuals" $\autRes{L}$ coincides with the "quotient of $\A$ by" the "congruence"~$\eqResState{\A}$.
	\end{remark}

%\subparagraph*{Concavity.}
%%%Antonio: Removed. Added small def in an example of warm up
%For comparison purposes, we define "concave" objectives. However, this notion will play no role in our characterisation or our proofs.
%"Concavity" was introduced by Kopczyński~\cite{Kop08Thesis} as an approximation to "positionality": "prefix-independent" "concave" "objectives" are "positional" over finite games~\cite[Theorem~4.7]{Kop08Thesis}, but the converse does not hold.
%This notion was further studied for non-"prefix-independent" objectives in~\cite{BFMM11HalfPos}.
%
%\AP We say that an objective $W\subseteq \SS^\oo$ is ""concave"" if for all pairs of sequences of finite words $u_1,u_2,\dots \in \SS^+$ and $w_1,w_2,\dots \in \SS^+$ we have:
%\[ u_1u_2\dots \notin W \; \tand \; w_1w_2\dots \notin W \;\; \implies \;\;  u_1w_1u_2w_2\dots \notin W .\]
%That is, the shuffle of any pair of "losing" words is a "losing" word.

\section{Positionality of \texorpdfstring{$\omega$}{omega}-regular objectives: Statement of the results}\label{sec:char-half-pos}

In this section, we state the central result of the paper and its consequences: a characterisation of "deterministic" "parity automata" "recognising" "positional" "$\oo$-regular languages" (Theorem~\ref{th-reslt:MainCharacterisation-allItems}). 
The statement of the theorem uses terminology that will be formally introduced in Section~\ref{sec:proof}, here we just provide some intuitive explanations. 

\subsection{Characterisation of positionality for \texorpdfstring{$\omega$}{omega}-regular objectives}\label{subsec-reslt:struct-char}
		
		We state our main characterisation theorem. 
		Items  are ordered following the sequence of logical implications in its proof (with the exception of~\ref{item-th:ee-complete-forall}).

		\begin{theorem}\label{th-reslt:MainCharacterisation-allItems}
			Let $W\subseteq \SS^\oo$ be an "$\oo$-regular" "objective". The following are equivalent:
			\begin{enumerate}[ref=(\theenumi)]
				\item\label{item-th:halfPosFiniteEve} $W$ is "positional" over finite "$\ee$-free" "Eve-games".

				%\item\label{item-th:stronglyStructured} There is a "deterministic" "fully progress consistent@@sigAut" "structured" "signature automaton" "recognising" $W$.
				
				\item\label{item-th:signature} There is a "deterministic" "fully progress consistent" "signature automaton" "recognising"~$W$.
				
				\item\label{item-th:ee-complete} There is a "deterministic"  "$\ee$-completable" "parity automaton" "recognising"~$W$.
				 
				\item[\mylabelOne{item-th:ee-complete-HD}{3'.}] %\label{item-th:ee-complete-HD}  
				There is a "history-deterministic" "$\ee$-complete" parity automaton "recognising"~$W$.
	
				\item[\mylabelTwo{item-th:ee-complete-forall}{3''.}]%\label{item-th:ee-complete-forall} 
				Any ("non-deterministic") parity automaton "recognising" $W$ is "$\ee$-completable".
		
				\item\label{item-th:universalGraph} For all cardinals $\kappa$, there is a "well-ordered@@graph" "monotone@@univ" "$(\kappa,W)$-universal" graph.
				
				\item\label{item-th:halfPosArbitrary} $W$ is "positional" over all "games" (potentially infinite and containing "$\ee$-edges").
			\end{enumerate}
		\end{theorem}

		This is an automata-oriented characterisation of positionality: we identify two classes of "deterministic" "parity automata" ("fully progress consistent" "signature" and "$\ee$-completable"), such that any "positional" language can be "recognised" by automata in these classes. Each of them presents some formal advantages that make
		them suitable for different kinds of proofs. 
		%On one hand, signature  automata can be built recursively from any given automaton recognising a positional language. On the other, $\ee$-completable automata are closer to monotone universal graphs, allowing to prove positionality of a language recognised by such an automaton. We explain these notions in Section
		"Signature automata" are "parity automata" with a very restricted syntactic structure: for all "priorities" $x$ there is a "total preorder" $\leqSig{x}$ over the states, such that these refine one another and satisfy some local "monotonicity@@cong" properties (see Section~\ref{subsec:def-signature} for the precise definition, which is quite involved).
		Our main technical contribution is to show that any "positional" "$\oo$-regular" "objective" $W$ can be "recognised" by a "signature automaton" (implication from \ref{item-th:halfPosFiniteEve} to~\ref{item-th:signature}). 
		This is achieved by applying a number of transformations to a given "parity automaton", until obtaining an automaton with all the desired structural properties.
		%Our proof goes by induction on the "priorities"; a lot of structure is required for the induction to go through.
		The final automaton satisfies a further more global property -- necessary for "positionality" -- that we call "full progress consistency": words making a strict progress in the automaton with respect to some of the preorders must be "accepted@@aut" if repeated infinitely often.
		
		Next, we prove that "deterministic" "fully progress consistent" "signature" automata are in fact "$\ee$-completable": one may add $\ee$-transitions along a tree structure without augmenting their language, as illustrated in Figure~\ref{fig:epsilon-completable-automaton} (see Section~\ref{subsec:from-signature-to-HP} for a formal definition).
		This corresponds to the implication from~\ref{item-th:signature} to~\ref{item-th:ee-complete}, Item~\ref{item-th:ee-complete-HD} follows immediatly.
		%To conclude the chain of implications in Theorem~\ref{th-reslt:MainCharacterisation-allItems}, 
		We then show how to obtain "well-ordered" "monotone@@univ" "universal" graphs from "history-deterministic" "$\ee$-complete" automata (implication from~\ref{item-th:ee-complete-HD} to~\ref{item-th:universalGraph}), which is fairly straightforward, and obtain the implication~\ref{item-th:universalGraph}~$\Rightarrow$~\ref{item-th:halfPosArbitrary} thanks to Proposition~\ref{prop-prelim:univ-graphs}.

		\begin{figure}
			\begin{center}
				\includegraphics[width=\linewidth]{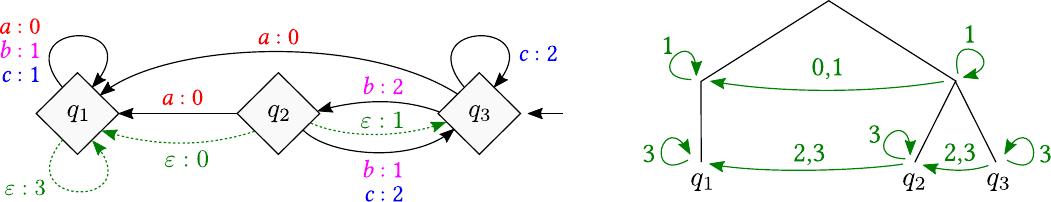}
				\caption{On the left, a deterministic automaton recognising the "positional" language $\infOften(a) \vee (\noOcc(a) \wedge \finOften(bb))$. On the right, a representation of an "$\ee$-completion" of the automaton: we can add "$\ee$-transitions" indicated by the tree, that is: $q_2, q_3 \re{\ee:0,1} q_1$, $q_2 \re{\ee:1} q_3$, $q_3 \re{\ee:1} q_2$, $q_3\re{\ee:2,3} q_2 \re{\ee:2,3} q_1$, and $q\re{\ee:x}q$ for all $q$ and odd $x$.
				Some of these are represented on the left as dotted arrows.}
				\label{fig:epsilon-completable-automaton}
			\end{center}
		\end{figure}
		
		This shows the equivalence of all the statements of Theorem~\ref{th-reslt:MainCharacterisation-allItems}, except for Item~\ref{item-th:ee-complete-forall}:
		\emph{any} "parity automaton" "recognising" a "positional" language is "$\ee$-completable" (including "non-deterministic" ones).
		The proof of this result, included in Corollary~\ref{cor-dec:ee-closure-nd}, relies on the equivalence between Items~\ref{item-th:halfPosFiniteEve} and~\ref{item-th:halfPosArbitrary} and its consequences (mainly Theorem~\ref{th-reslt:union-PI} about closure under union).
		We do not know whether the existence of a "non-deterministic" "$\ee$-complete" automaton suffices to prove "positionality".
		
		% \begin{proposition}\label{prop-reslt:ee-complete-equivalences}
		% 	Let $W$ be an "$\oo$-regular" "objective".
		% 	The following are equivalent:
		% 	\begin{enumerate}
		% 		\item There is a "deterministic" "$\ee$-completable" automaton "recognising"~$W$.
		% 		\item There is a "history-deterministic" "$\ee$-complete" automaton "recognising"~$W$.
		% 		\item $W$ is "positional" over all "games".
		% 		\item Any ("non-deterministic") automaton "recognising" $W$ is "$\ee$-completable".
		% 	\end{enumerate}
		% \end{proposition}
\subsection{Main consequences on positionality}\label{subsec-reslt:consequences}
		
		We now discuss consequences of Theorem~\ref{th-reslt:MainCharacterisation-allItems}.
		
		\paragraph*{Decidability of positionality in polynomial time}
		
		\begin{theorem}\label{th-reslt:decid-poly}
			Given a "deterministic" "parity automaton" $\A$, we can decide in polynomial time whether $\Lang{\A}$ is "positional".
		\end{theorem}
		
		We will give two proofs for Theorem~\ref{th-reslt:decid-poly}, both of which are detailed in Section~\ref{sec:decision}.
		The first proof applies the procedure turning a "deterministic" parity automaton into a "signature" "automaton"; if some step of the procedure fails, then the "objective" is not "positional".
		The second proof is more direct and builds up on another consequence of Theorem~\ref{th-reslt:MainCharacterisation-allItems}, namely the closure under union (see discussion below).

		\paragraph*{Finite-to-infinite and 1-to-2 player lifts}

		The following result simply restates the implication \ref{item-th:halfPosFiniteEve} $\implies$ \ref{item-th:halfPosArbitrary} from Theorem~\ref{th-reslt:MainCharacterisation-allItems}.
		
		\begin{theorem}\label{th-half:lifts}
			If an "$\oo$-regular" objective is "positional over" finite, "$\ee$-free" "Eve-games", then it is "positional" over all games (potentially infinite and containing $\ee$-edges).
		\end{theorem}%			

		\paragraph*{Closure under union of prefix-independent positional languages}
		
		We now show that "Kopczyński's conjecture" holds for "$\oo$-regular languages": "prefix-independent" "positional" languages are closed under union. In fact, we show a stronger result: it suffices to suppose that only one of the objectives is "prefix-independent".
		
		\begin{theorem}\label{th-reslt:union-PI}
			Let $W_1,W_2\subseteq \SS^\oo$ be two "positional" "$\oo$-regular" objectives, and suppose that $W_1$ is "prefix-independent". Then, $W_1\cup W_2$ is "positional".
		\end{theorem}
		
		In order to obtain this theorem, we use the 1-to-2-players lift stated in Theorem~\ref{th-half:lifts}. The result from Theorem~\ref{th-reslt:union-PI} can be easily obtained for "Eve-games" (one only needs to be careful with the definition of a "uniform strategy"), so it suffices then to apply the lift to get the result for all types of games.
		
		\begin{lemma}\label{lemma-half:union-PI-EveGames}
			Let $W_1,W_2\subseteq \SS^\oo$ be two objectives that are "positional over@@eve" "Eve-games", and suppose that $W_1$ is "prefix-independent". Then, $W_1\cup W_2$ is "positional over@@eve" "Eve-games".
		\end{lemma}

		\begin{proof}
			Let $\G$ be an "Eve-game" using $W_1\cup W_2$ as "winning condition".
			We show that "Eve" has a "positional strategy" that "wins from" any vertex of her "winning region". %\footnote{If both $W_1$ and $W_2$ were "prefix-independent",it would suffice that from each vertex $v$ in her "winning region", "Eve" has a "positional" strategy $\strat_v$ winning from $v$, and then apply Lemma~\ref{}. As we do not assume that both objectives are "prefix-independent", we need to justify that a single "positional strategy" suffice to win from all the "winning region".}
			We let $\G_1$ be the "game" with the same "game graph" than $\G$ and $W_1$ as "winning condition". Consider "Eve's" "winning region" in this game, $\winRegion{\Eve}{\G_1}$. By "positionality" of $W_1$, she has a "positional strategy" $\strat_1$ ensuring to produce paths labelled with $W_1$ from states in $\winRegion{\Eve}{\G_1}$.		
			Moreover, by "prefix-independence" of $W_1$, there is no path leading to $\winRegion{\Eve}{\G_1}$ from a vertex that is not in this "winning region".
			
			We let $\G_2$ be the "game" with $\G\setminus \winRegion{\Eve}{\G_1}$ as "game graph", and using $W_2$ as "winning condition". By "positionality" of $W_2$, "Eve" has a "positional strategy" $\strat_2$ for this game that is winning from $\winRegion{\Eve}{\G_2}$.
			
			We consider the "positional strategy" $\strat$ in $\G$ that coincides with $\strat_1$ over $\winRegion{\Eve}{\G_1}$ and coincides with $\strat_2$ over $\G_2$. It is clear that this "strategy" is "winning@@strat" from vertices in $\winRegion{\Eve}{\G_1}\cup \winRegion{\Eve}{\G_2}$. We show that these are all vertices from which "Eve" can win $\G$, so $\strat$ is an "optimal strategy".
			
			\begin{claim}
				$\winRegion{\Eve}{\G} = \winRegion{\Eve}{\G_1} \cup \winRegion{\Eve}{\G_2}$.
			\end{claim}
			\begin{subproof}
				Let $v$ be a vertex in $\G$ such that "Eve" "wins from" it. A "strategy" in an "Eve-game" is just an infinite path in the "game graph". Let therefore $\rr_v$ be an infinite path from $v$ labelled with a word $w\in W_1\cup W_2$. 
				Suppose that $v\notin \winRegion{\Eve}{\G_1}$. In particular $w\notin W_1$, so $w\in W_2$. As there is no path leading to $\winRegion{\Eve}{\G_1}$ from a vertex that is not in this region, the path $\rr_v$ is contained in $\G \setminus\winRegion{\Eve}{\G_1}$, which is the "game graph" of  $\G_2$. Therefore, $v\in \winRegion{\Eve}{\G_2}$.
			\end{subproof}
			This finishes the proof.\qedhere
		\end{proof}
		
		"Kopczyński's conjecture" and its stronger version in which only one of the objectives is supposed to by "prefix-independent" remain open for arbitrary (non "$\oo$-regular") objectives.

		\paragraph*{Closure of positionality under addition of neutral letters}
		As mentioned in the introduction,  Ohlmann recently characterised "positional" objectives by means of the existence of "universal graphs"~\cite{Ohlmann23Univ}. One direction (stated in Proposition~\ref{prop-prelim:univ-graphs}) holds for any objective: if $W$ admits "well-ordered" "monotone@@graph" "universal graphs", then it is "positional". To obtain the converse, the proof proposed by Ohlmann requires a further hypothesis: $W$ has to contain a "neutral letter", that is, a letter that can be removed from any word without modifying the membership in $W$. 
		In his work, he left open the problem of whether adding a "neutral letter" preserves "positionality". This is a central question in the theory of "positionality", as it would imply that "universal graphs" completely characterise "positionality" without any further hypothesis on the objectives.
		This question is almost\footnote{The only difference is that in "$\ee$-free" games we assume that there are no infinite paths composed exclusively of $\ee$-edges, whereas these may appear in "$\addNeutral{W}$-games".} equivalent to the one raised by Kopczy{\'n}ski in his PhD thesis~\cite[Section~2.5]{Kop08Thesis}: if $W$ is "positional" over "$\ee$-free" games, is it "positional" over all games? We introduce these notions formally for completeness.
		
		\AP Let $W\subseteq \SS^\oo$ be an objective. A letter $c\in \SS$ is ""neutral for $W$"" if, for all $w_1,w_2,\dots \in \SS^+$ and $n_1,n_2,\dots \in \NN$:
		\begin{itemize}
			\item $c^{n_1}w_1c^{n_2}w_2\dots \in W \: \iff \: w_1w_2\dots \in W$, and
			\item $w_1c^\oo \in W \iff \lquotW{w_1}\neq \emptyset$.
		\end{itemize}
		
		\AP Given an objective $W$, we let $\intro*\addNeutral{W}$ denote the unique objective obtained by adding a fresh "neutral letter" $\ee$ to $W$.

		\begin{proposition}[{\cite{Ohlmann23Univ}}]\label{prop-reslt:ohlmann-neutral-letters}
			Let $W\subseteq \SS^\oo$. The objective $\addNeutral{W}$ is "positional" if and only if for all cardinals $\kappa$ there is a "well-ordered" "monotone@@univ" "$(\kappa,W)$-universal" graph.
		\end{proposition}

		\begin{conjecture}[{""Neutral letter conjecture""~\cite{Ohlmann23Univ}}]\label{conj-half:neutral-letter}
			For every "positional" objective $W$, the objective $\addNeutral{W}$ is "positional".
		\end{conjecture}
		
		Our characterisation (Item~\ref{item-th:universalGraph} in Theorem~\ref{th-reslt:MainCharacterisation-allItems}), together with Proposition~\ref{prop-reslt:ohlmann-neutral-letters}, answers this question in the case of "$\oo$-regular" objectives.
		
		\begin{theorem}\label{th-reslt:neutral-letter}
			Let $W\subseteq \SS^\oo$ be an "$\oo$-regular objective". If $W$ is "positional", then $\addNeutral{W}$ is "positional".
			Also, $W$ is "positional" over "$\ee$-free" games if and only if $W$ is "positional" over all games.
		\end{theorem}

%\subsection{Minimisation of automata recognising positional languages}

\section{Warm-up: Illustrating ideas on restricted classes of languages}\label{sec:warm-up}

The goal of this section is to give a gentle introduction to the techniques and ideas which are used in the proof of our main result (implication from \ref{item-th:halfPosFiniteEve} to \ref{item-th:signature} in Theorem~\ref{th-reslt:MainCharacterisation-allItems}).
Readers who prefer to go directly to the statement and proofs of the general case may skip this section.

We single out four crucial properties that a "parity automaton" "recognising" a "positional" objective should satisfy; the general characterisation will consist in a generalisation of these.
\begin{itemize}
\item "Positional" objectives have "residuals" totally ordered by inclusion (inducing a total order on the corresponding "residual classes" of the states of automata).
\item This order should satisfy a semantic property called "progress consistency".
\item Transitions with "priority" $0$ "preserve the congruence" induced by the residuals (if $q\eqResState{} p$ and $q\re{a:0}$, then $p\re{a:0}$).

\item In each "congruence class", states which are interreachable using paths avoiding priorities $\leq 1$ have comparable "$({\leq}1)$-safe languages" (defined below).
\end{itemize}

%\AP The complexity of an "$\oo$-regular language" can be measured by the number of "priorities" that is required to recognise it with a "deterministic" automaton (its ""parity index"").
We propose to study four restricted classes of "$\omega$-regular objectives" that allow us to isolate these different points, namely, "closed", "open", "Büchi recognisable" and "coBüchi recognisable" "objectives".
Considering objectives in these four classes allows us to illustrate the necessity of the four properties above, and the techniques we use to derive them.
In each case, we state a characterisation of "positionality" and give a full proof of necessity, which is the more difficult direction.
These characterisation and proof techniques are generalised to all "$\omega$-regular languages" in our main inductive proof of necessity (Section~\ref{subsec:from-HP-to-signature}).

We moreover incorporate in this section many examples illustrating our results and the ideas in our proofs.

\subsection{Closed objectives and total order on the residuals}\label{subsec-warm:closed}

We now discuss the first property announced above: "residuals" of "positional" objectives are totally ordered by inclusion. The necessity of this condition holds even for non "$\oo$-regular" objectives.

\paragraph*{Residuals of positional objectives are totally ordered}

\begin{lemma}\label{lemma-warm:total-order-residuals-nec}
	If an objective $W\subseteq \SS^\oo$ is "positional", then $\Res{W}$ is totally ordered by inclusion.
\end{lemma}

\begin{proof}
   We show the contrapositive.
   Suppose that $W$ has two incomparable "residuals", $\lquotW{u_1}$ and $\lquotW{u_2}$.
   Take $w_1\in \lquot{u_1}{W}\setminus\lquot{u_2}{W}$ and $w_2\in \lquot{u_2}{W}\setminus\lquot{u_1}{W}$.
   Stated differently, we have
   
   \begin{tabular}{l l}
	   \centering
	   $u_1w_1\in W$, & $u_1w_2\notin W$,\\ 
	   $u_2w_1\notin W$, & $u_2w_2\in W$. 
   \end{tabular} 
   
   \vspace{2mm}
   Consider the (infinite) "Eve-game" $\G$ represented in Figure~\ref{fig-warm:game-residuals}.
   "Eve" "wins" $\G$ from $v_1$ and $v_2$: if a "play" starts in $v_i$, for $i=1,2$, she just has to take the path labelled $w_i$ from $v_{\mathsf{choice}}$. However, she cannot win from both $v_1$ and $v_2$ using a "positional strategy". Indeed, such a "positional strategy" would choose one transition $v_{\mathsf{choice}} \lrp{w_i}$, and the "play" induced when starting from $v_{j}$, $j\neq i$, would be "losing@@play".
\end{proof}
\begin{figure}
	\centering
	\includegraphics[scale=1.5]{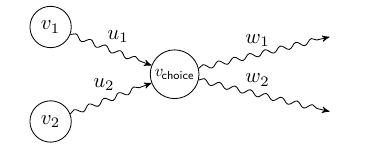}
	\caption{A game $\G$ in which "Eve" cannot "play optimally using positional strategies" if $\Res{W}$ is not totally ordered, as in the proof of Lemma~\ref{lemma-warm:total-order-residuals-nec}. \qedhere}
	\label{fig-warm:game-residuals}
\end{figure}

\paragraph*{Closed objectives}

Let $\SS$ be a set of letters and $L \subseteq \SS^*$ be a language of finite words. 
\AP The safety objective associated to $L$ is defined by
\[
\intro*\Safe{L}= \{ w\in \SS^\oo \mid w \text{ does not contain any prefix in } L\}.
\]

\AP An "objective" $W$ is ""topologically closed"" if $W=\Safe{L}$ for some $L\subseteq \SS^*$. This terminology is justified since "objectives" of the form $\Safe{L}$ are exactly the closed subsets of $\SS^\oo$ for the Cantor topology (see for example~\cite{Thomas1991AutomataOI}).

\begin{remark}
	An objective $W=\Safe{L}$ is "$\oo$-regular" if and only if $L$ is a regular language of finite words if and only if $\Res{W}$ is finite.
	We refer to this class as ""$\oo$-regular closed objectives"".
\end{remark}
%\begin{remark}
%	\AP An objective $W$ is "$\oo$-regular closed" if and only if it can be "recognised" by a special class of "deterministic" $\Weak{1}$-automata, that we call ""safety automata"".
%	In a safety automaton, some states are defined to be ""rejecting@@states"" and a "run" is accepting if and only if it does never visit a "rejecting state".
%	Moreover, a minimal "safety automaton" for $L$ has one state per "residual".
%\end{remark}

%Closed $\omega$-regular languages $L$, admit a canonical safety automaton $\A$ defined as follows.
%\begin{itemize}
%\item States of $\A$ are residuals, with $[\emptyword]$ being the initial state and the empty residual (if it exists, meaning if $L$ is not $\SS^*$) being the only rejecting state.
%\item Transitions are given by $[u] \re a [ua]$ for $u \in \SS^*$ and $a \in \SS$.
%\end{itemize}
%It is a direct check that $\L(\A) = L$.

It turns out that for "$\omega$-regular closed" objectives the converse of Lemma~\ref{lemma-warm:total-order-residuals-nec} holds.
This was first established in~\cite{ColcombetFH14PlayingSafe}.

\begin{proposition}[Positionality of closed objectives~\cite{ColcombetFH14PlayingSafe}]\label{prop-warm:char-closed}
	Let $W\subseteq \SS^\oo$ be an "$\oo$-regular closed objective".
	Then, $W$ is "positional" if and only if $\Res{W}$ is totally ordered by inclusion.
\end{proposition}

Thus, "residuals" encode the information needed to decide whether an "$\omega$-regular closed" objective is "positional".
We do not include a proof of sufficiency in this warm-up; a proof for all (non-necessarily "$\oo$-regular") "closed objectives" is given in Theorem~\ref{th-open:pos-closed}.
However, a much subtler understanding is needed for non-"closed objectives", as witnessed by the example below.

\begin{example}[Non-positional open objective]\label{ex-warm:not-pos-open}
	Consider the non-"closed objective" 
	\[ W = \{w\in \SS^\oo \mid w \text{ contains the factor } aa \}. \]
	Its three "residuals" are totally ordered by inclusion:
	\[ \lquotW{\ee} \subseteq \lquotW{a} \subseteq \lquotW{(aa)}. \]
	However, it is not "positional", as witnessed by the "game" in Figure~\ref{fig-warm:game-reach-aa}.
\end{example}
\begin{figure}
	\centering
	\includegraphics[scale=1.5]{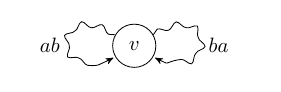}
	\caption{A game $\G$ in which "Eve" cannot produce the factor $aa$ "positionally@@strat".}
	\label{fig-warm:game-reach-aa}
\end{figure}

\subsection{Open objectives and progress consistency}\label{subsec-warm:open}

We now introduce "progress consistency", a semantic property of the order of "residuals" which is necessary for "positionality". For "$\oo$-regular open" objectives, this property, together with the total order of "residuals" is also sufficient.

\paragraph*{Progress consistency}

\begin{definition}[Progress consistency]
	\AP An objective $W\subseteq \SS^\oo$ is ""progress consistent"" if for all $u,w\in \SS^*$:
	\[ \resClass{u} \lRes \resClass{uw} \implies uw^\oo \in W. \] 
\end{definition}
Intuitively, a "progress consistent" objective satisfies that whenever we read a word that makes some strict progress with respect to the order of the "residuals", by repeating this word we produce a sequence in $W$.

We remark that the objective $\Reach{\SS^*aa}$ from Example~\ref{ex-warm:not-pos-open} is not "progress consistent", as the word $ba$ makes progress from residual $\lquotW{\ee}$, but $(ba)^\oo\notin W$.

Let us establish necessity of progress consistency for "positional" objectives. 

\begin{lemma}[Necessity of progress consistency]\label{lemma-warm:prog-cons-nec}
	Any "positional" objective is "progress consistent".
\end{lemma}
\begin{proof}
	We show the contrapositive of the statement. Let $W$ be an objective that is not "progress consistent", that is, there are $u,w\in \SS^*$ such that $\resClass{u} \lRes \resClass{uw}$ and $uw^\oo \notin W$.
	Let $w'\in \lquotW{(uw)}  \setminus \lquotW{u}$.
	Consider the game $\G$ depicted in Figure~\ref{fig-warm:game-prog-cons}.	
	\begin{figure}
		\centering
		\includegraphics[scale=1.5]{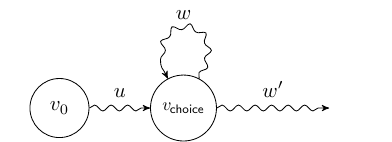}
		\caption{A game $\G$ in which "Eve" cannot "play optimally using positional strategies" if $W$ is not "progress consistent", as in the proof of Lemma~\ref{lemma-warm:prog-cons-nec}.}
		\label{fig-warm:game-prog-cons}
	\end{figure}

	"Eve" "wins" game $\G$ from vertex $v_0$ by producing the "play"
	\[ v_0\lrp{u\phantom{.}}v_\mathsf{choice} \lrp{w\phantom{.}}v_\mathsf{choice} \lrp{w'\phantom{.}}.\]	
	However, she cannot "win positionally" from $v_0$ since "positional strategies" produce either $uw^\omega$ or $uw'$, and both of these words are "losing".
\end{proof}

\begin{remark}
	The previous lemma applies, in particular, to "$\oo$-regular closed" objectives. We did not need to add "progress consistency" as a hypothesis in Proposition~\ref{prop-warm:char-closed}, as this property is granted for "closed@@obj" objective by Lemma~\ref{lemma-prelim:monotonicity-residuals}.
\end{remark}

We are now ready to move on to the characterisation of "$\oo$-regular open" objectives.

\paragraph*{Open objectives}

We now study the dual of "closed objectives", namely, "open@@obj" ones.
Let $L\subseteq \SS^*$. 
\AP The reachability objective associated to $L$ is defined by
\[
\intro*\Reach{L}= \{ w\in \SS^\oo \mid w \text{ contains a prefix in } L\}.
\]

\AP An objective $W$ is ""topologically open"" if $W=\Reach{L}$ for some $L\subseteq \SS^*$. (These are the open subsets of $\SS^\oo$ for the Cantor topology.)
\AP Similarly to the previous subsection, we define the class of ""$\oo$-regular open objectives"" as those that are both "open@@obj" and "$\oo$-regular".

\begin{remark}
	An open objective $W=\Reach{L}$ is "$\oo$-regular" if and only if $L$ is a regular language of finite words if and only if $\Res{W}$ is finite. 
\end{remark}

%\begin{remark}
%	\AP A language $W$ is "$\oo$-regular open" if and only if it can be "recognised" a "deterministic" $\Weak{1}$-automata that we call ""reachability automata"". In a reachability automaton, some states are defined to be ""accepting@@states"" and a "run" is accepting if and only if it eventually visits an "accepting state".
%	Moreover, a minimal "reachability automaton" for $W$ has one state per "residual".
%\end{remark}

Let us state a characterisation of "positionality" for "$\omega$-regular open" objectives.
Characterisations for the full classes of "open" and "closed" objectives (without "$\oo$-regularity" assumptions) will be obtained in Section~\ref{sec:open-closed}.

\begin{proposition}[Positionality for open objectives]\label{prop-warm:char-open}
	An "$\oo$-regular open" objective $W$ is "positional" if and only if it is "progress consistent" and its set of "residuals" $\Res W$ is totally ordered.
\end{proposition}

Necessity follows from combining Lemmas~\ref{lemma-warm:total-order-residuals-nec} and~\ref{lemma-warm:prog-cons-nec}; we omit a proof of sufficiency in this warm-up.
In particular, we obtain the following corollary of Propositions~\ref{prop-warm:char-closed} and~\ref{prop-warm:char-open}.

\begin{corollary}\label{cor-warm:open-bipositional}
	Any "positional" "$\oo$-regular open" objective is "bipositional".
\end{corollary}

 We now give an example of an objective that satisfies the requirement from the previous proposition.

\begin{example}[Positional open objective]
	Consider the "$\omega$-regular open" objective
	\[
		W_n = \Reach {(a\SS^*)^n}.
	\]
	\AP It was introduced (for $n=2$) in~\cite[Lemma~13]{BFMM11HalfPos} as an example of a "bipositional" "objective" which is not ""concave"" (that is, no shuffle of two words outside $W$ belongs to $W$; see~\cite[Def~4.2]{Kop08Thesis} for details).
	Its "residuals" are given by
	\[
	 \lquotW{\ee} \subsetneq \lquotW{a} \subsetneq \lquotW{(aa)} \subsetneq \dots \subsetneq \lquotW{(a^n)}=\SS^\oo,
	\]
	which are totally ordered.
	Moreover, for any "residual class" $[a^i]$ with $i<n$, we have $[a^i] <[a^iu]$ if and only if $u$ contains the letter $a$, in which case $a^iu^\omega \in W$.
	Therefore, $W$ is "progress consistent", so we conclude that it is "bipositional".
\end{example}

%Therefore, we have seen that "positionality" of $W$ enforces some constraints on its "residuals", namely, they should be totally ordered (Lemma~\ref{lemma-warm:total-order-residuals-nec}) and "progress consistent" (Lemma~\ref{lemma-warm:prog-cons-nec}).

Many natural examples of objectives are in fact "prefix-independent"; for those, the two conditions about the "residuals" above are trivially satisfied.
Yet, this does not suffice to guarantee their positionality.
We continue our introductory exploration with objectives recognised by "deterministic" "B\"uchi automata".

\subsection{\texorpdfstring{B\"uchi}{Büchi} recognisable objectives: Uniformity of \texorpdfstring{$0$}{0}-transitions}\label{subsec-warm:Buchi}

Our goal in this section is to present another property of "positional" "$\omega$-regular objectives", namely, that they can be  "recognised" by a "deterministic" "parity automaton"~$\A$ in which "$0$-transitions" are "uniform over" each "residual class":
\begin{center}
For any $q \eqRes{\A} q'$ and $a \in \SS$, if $q \re{a:0}$ then $q' \re{a:0}$.
\end{center}
In our main induction (Section~\ref{subsec:from-HP-to-signature}), we will derive a similar property for all even "priorities".
To illustrate the technique, we now only focus on the case of "B\"uchi recognisable" languages, which helps alleviate some of the technicalities while preserving the important ideas behind the proof.
On the way, we characterise "positionality" for these objectives, reobtaining the main result of~\cite{BCRV22HalfPosBuchi}. %(giving a slightly simplified proof).
%We note that the proof we present here is a slight simplification of the one in~\cite{BCRV22HalfPosBuchi}.

The proof is split into two parts: first, we focus on the "prefix-independent" case, and then reduce to it.

\paragraph*{Prefix-independent \texorpdfstring{B\"uchi}{Büchi} recognisable objectives}

\AP It is well-known that ""Büchi languages"" are "positional"~\cite{EmersonJutla91Determinacy}, that is, those of the form
 \[ \AP \intro*\BuchiC{B}{\SS} =\{ w\in \SS^\oo \mid \text{ letters of } B \text{ appear infinitely often in } w\}. \] 
We prove now that these are the only "positional" "prefix-independent" "B\"uchi recognisable" objectives.
Our proof considerably simplifies that of~\cite{BCRV22HalfPosBuchi}.

\begin{proposition}[{\cite[Proposition~11]{BCRV22HalfPosBuchi}}]\label{prop-warm:char-Buchi-PI}
	A "prefix-independent" "B\"uchi recognisable" $W$ is "positional" if and only if it is a "B\"uchi language".
\end{proposition} 

In particular, Proposition~\ref{prop-warm:char-Buchi-PI} tells us that there is a "B\"uchi automaton" with just one state "recognising" $W$. In this automaton,  "$0$-transitions" are trivially "uniform@@trans".

%We now concentrate on proving Proposition~\ref{prop-warm:char-Buchi-PI}.

\subparagraph{Super words and super letters.} 
\AP We say that $u\in \SS^+$ is a ""super word@@buchi"" (for $W$) if, for every $w\in \SS^\oo$, if $w$ contains $u$ infinitely often as a factor, then $w\in W$. 
If $u$ is a letter, we say that it is a ""super letter@@buchi"".
Let $\intro*\BLettersBuchi{W}\subseteq \SS$ be the set of "super letters@@buchi" for $W$. It is clear that $\BuchiC{\BLettersBuchi{W}}{\SS}\subseteq W$. We will show that, if $W$ is "positional", this is in fact an equality.

%\begin{example}[Language not determined by super letters]
%	Take $\SS = \{a,b\}$ and consider the language $L = \{w\in \SS^\oo \mid \text{ factor } aa \text{ appears infinitely often in } w\}$. A "B\"uchi automaton" with two states, $q_1$, $q_2$ is depicted in Figure~\ref{fig:}. We consider the trivial equivalence relation: $q_1\sim q_2$.
%	We note that $L$ does not admit any "super letter" in $\SS$, so $\NLetters{0}{q_1}=\{a,b\}$. However, the word $aa\in \NLetters{0}{q_1}^*$ is a "super word" for $L$.
%\end{example}

\begin{lemma}[Existence of super letters]\label{lemma-warm:existence-super-letter}
A non-empty "prefix-independent" "positional" objective $W$ "recognised" by a "deterministic" "B\"uchi automaton" admits a "super letter@@buchi".
\end{lemma}

One may easily deduce Proposition~\ref{prop-warm:char-Buchi-PI} from Lemma~\ref{lemma-warm:existence-super-letter}: the restriction $W'$ of $W$ to "non-super letters@@buchi" is a "prefix-independent" "B\"uchi recognisable" "positional" objective which contains no "super letter@@buchi". Thus, Lemma~\ref{lemma-warm:existence-super-letter} tells us that $W' = \emptyset$ and therefore $W=\BuchiC{\BLettersBuchi{W}}{\SS}$.

%So we now prove Lemma~\ref{lemma-warm:existence-super-letter}.

\begin{proof}[Proof of Lemma~\ref{lemma-warm:existence-super-letter}]
Fix a non-empty "prefix-independent" "positional" objective $W$ "recognised" by a "deterministic" "B\"uchi automaton" $\A$, and assume without loss of generality that $\A$ in "normal form" and strongly connected, which can be assumed by "prefix-independence" (thus, every state can be chosen initial).
Since $W$ is non-empty, note that $\A$ must contain a transition with "priority" $0$.
We will use the following observation.

\begin{claim}[Super words in B\"uchi automata]\label{claim-warm:char-super-words}
	A word $w\in \SS^+$ is a "super word@@buchi" if and only if for all states $q$ of $\A$, "priority" $0$ appears on the path $q\lrp{w:0}$.
	This is in particular the case if $w$ is a letter.
\end{claim}
\begin{subproof}
	By "normality" of $\A$, if there is $q$ such that $q\lrp{w:1}q'$ then there is a word $w'\in \SS^*$ labelling a returning path $q'\lrp{w':1}q$. Therefore, $(ww')^\oo\notin W$, so $w$ is not a "super word@@buchi".
	The converse implication is clear, since each time word $w$ is read, the automaton produces "priority"~$0$.
\end{subproof}

We now prove existence of super words.
	
\begin{claim}[Existence of super words]\label{cl:existence-super-words}
	There is a "super word@@buchi" for $\Lang{\A}$.
\end{claim}
\begin{subproof}
	We note that, as $\A$ is strongly connected and contains some "priority" $0$, for each state $q$ there is a finite word that produces "priority" $0$ when read from $q$. 
	We let $\{q_1,q_2,\dots,q_k\}$ be an enumeration of the states of $\A$ and recursively define $k$ finite words $w_1,w_2,\dots, w_k\in \SS^*$ satisfying:
	\[ q_i \lrp{w_1w_2\dots w_{i-1}} q' \lrp{w_i:0} q''. \]
	In words, for $i\in{1,\dots,k}$, reading $w_1w_2\dots w_i$ from $q_i$, produces priority $0$.
	This implies that $w_1w_2\dots w_i$ produces "priority" $0$ when read from any $q_j$ for $j\leq i$.
	Therefore, $w=w_1w_2\dots w_k$ produces priority $0$ when read from any state of $\A$, so by Claim~\ref{claim-warm:char-super-words}, $w$ is a "super word@@buchi".
\end{subproof}

We now prove that, by "positionality" of $W$, "super words@@buchi" can be chopped into smaller super words.
This implies Lemma~\ref{lemma-warm:existence-super-letter} by repeatedly chopping a "super word@@buchi" obtained from Claim~\ref{cl:existence-super-words} until obtaining a "super letter@@buchi".

\begin{claim}[Chopping super words]\label{cl:chopping-super-words}
	Let $w=w_1w_2 \in \SS^+$ be a "super word@@buchi".
	Then either $w_1$ or $w_2$ is a "super word@@buchi".
\end{claim}
\begin{subproof}
	Suppose by contradiction that neither $w_1$ nor $w_2$ are "super words@@buchi". Then, by Claim~\ref{claim-warm:char-super-words}, there are states $q_1$ and $q_2$ such that $q_1\lrp{w_1:1}q_1'$ and $q_2\lrp{w_2:1}q_2'$.
	By "normality", we obtain returning paths $q_1'\lrp{u_1:1}q_1$ and $q_2'\lrp{u_2:1}q_2$.
	Therefore, $(w_1u_1)^\oo\notin W$ and $(w_2u_2)^\oo\notin W$.
	We consider the "game" $\G$ depicted in Figure~\ref{fig-warm:game-chop-super-words}.
	"Eve" can  "win" this game, as alternating the two self loops she produces the word $(u_1w_1w_2u_2)^\oo$, which belongs to $\Lang{\A}$ since $w_1w_2$ is a "super word@@buchi".
	However, "positional strategies" in this game produce either $(w_1u_1)^\oo$ or $(w_2u_2)^\oo$, both losing.
	This contradicts the "positionality" of $\Lang{\A}$.		
\end{subproof}
\qedhere
\begin{figure}
	\centering
	\includegraphics[scale=1.5]{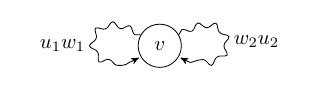}
	\caption{A game $\G$ in which "Eve" can win by forming the "super word@@buchi" $w_1w_2$ infinitely often, but in which she cannot "win" using a "positional strategy".\qedhere}
	\label{fig-warm:game-chop-super-words}
\end{figure}

\end{proof}

\paragraph*{Positionality for \texorpdfstring{B\"uchi}{Büchi} recognisable objectives}

We now state a characterisation of "positionality" for all "B\"uchi recognisable" objectives, without assuming "prefix-independence".

\begin{proposition}[{Positionality of Büchi recognisable objectives~\cite[Theorem~10]{BCRV22HalfPosBuchi}}]\label{prop-warm:char-Buchi-all}
	Let $W\subseteq \SS^\oo$ be a "Büchi recognisable" "objective". Then, $W$ is "positional" if and only if:
	\begin{itemize}
		\item $\ResW$ is totally ordered,
		\item $W$ is "progress consistent", and
		\item $W$ can be "recognised" by a "B\"uchi automaton" "on top of" the "automaton of residuals".
	\end{itemize}
\end{proposition}

\begin{example}[{Positional Büchi recognisable objective~\cite[Example~7]{BCRV22HalfPosBuchi}}]\label{ex-warm:Buchi-a-aa}
	Over the alphabet $\SS = \{a,b\}$, let
	\[ W = \infOften(a) \cup \Reach{aa}, \]
	that is, a word $w\in \SS^\oo$ belongs to $W$ if either it contains letter $a$ infinitely often, or it contains the factor $aa$ at some point. 
	This objective has three different "residuals", 
	\[
	[\ee] \lRes [a] \lRes [aa]=\SS^\oo.
	\]
	Figure~\ref{fig-warm:positional-Buchi-objective} depicts a deterministic "B\"uchi automaton" defined on top of the "residual automaton" of $W$.
	It is easy to verify that this objective is "progress consistent", so by Proposition~\ref{prop-warm:char-Buchi-all}, it is a positional objective.
%	
%	The "positionality" of this objective cannot be shown by applying existing "positionality" criteria~\cite{Kop08Thesis,BFMM11HalfPos}, nor by applying known characterisations of "bipositionality"~\cite{GimbertZielonka2005Memory}, as it is simply not "bipositional".
\end{example}

\begin{figure}
	\centering
	\includegraphics[scale=1.5]{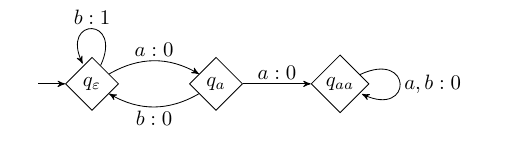}
	\caption{"B\"uchi automaton" "recognising" the objective $W = \infOften(a) \cup \Reach{aa}$.}
	\label{fig-warm:positional-Buchi-objective}
\end{figure}

As earlier, we focus on explaining the necessity of the conditions from Proposition~\ref{prop-warm:char-Buchi-all}, and we omit a proof of sufficiency in this warm-up. We already know that the two first conditions are necessary (Lemmas~\ref{lemma-warm:total-order-residuals-nec} and~\ref{lemma-warm:prog-cons-nec}). 
We now present the techniques used to obtain the necessity of the third condition.

\paragraph*{Uniformity of \texorpdfstring{$0$}{0}-transitions for positional objectives}

Our objective is to derive the following result.

\begin{lemma}[Uniform behaviour of $0$-transitions]\label{lemma-warm:congruence-0-letters}
	Let $W\subseteq \SS^\oo$ be a "positional" "Büchi recognisable" language.
	There is a "deterministic" "B\"uchi automaton" $\A$ "recognising" $W$ such that for every pair $q\eqResState{\A} q'$ of equivalent states and for every letter $a$, transition $q\re{a}$ produces "priority" $0$ if and only if  transition $q'\re{a}$ produces "priority" $0$.
\end{lemma} 

Necessity of the conditions from Proposition~\ref{prop-warm:char-Buchi-all} easily follows. In fact, what the previous lemma tells us is that we can take the "quotient automaton@@res" $\quotAut{\A}{\A}$ and assign priorities to its transitions consistently.

%\begin{proof}[Proof of necessity in Proposition~\ref{prop-warm:char-Buchi-all}]
%	We have already seen that "positionality" entails a total order over "residuals" as well as "progress consistency" (Lemmas~\ref{lemma-warm:total-order-residuals-nec} and~\ref{lemma-warm:prog-cons-nec}).
%	Take an automaton $\A$ "recognising" $W$ and satisfying the property of Lemma~\ref{lemma-warm:congruence-0-letters}.
%	Define a "B\"uchi automaton" over the "automaton of residuals" of $W$ by assigning priority $0$ to a transition $[u] \re a [u']$ if and only if reading $a$ produces priority $0$ from states in the "residual class" "associated to@@res" $[u]$ in $\A$.
%	Then priorities produced along a "run over" $w \in \SS^\oo$ in this automaton are the same as those in $\A$, so it "recognises" $W$.
%\end{proof}

The techniques we now introduce for proving Lemma~\ref{lemma-warm:congruence-0-letters} will be extended in our main induction (Section~\ref{subsec:from-HP-to-signature}).
The idea is to reduce to the "prefix-independent" case, captured by Proposition~\ref{prop-warm:char-Buchi-PI}. For this, we associate a "prefix-independent" language, over an ad-hoc alphabet, to each "residual" of the objective under consideration.

For the rest of this part of the section, we fix a "positional" objective $W$ "recognised" by a "deterministic" "B\"uchi automaton" $\A$.

\subparagraph{Localising to a residual.}   
\AP For each "residual class" $\resClass{u}$ of $W$, we define the local alphabet at $\resClass{u}$ as:
\[
	\intro*\locAlphBuchi{u} = \{w\in \SS^+ \mid \resClass{uw} = \resClass{u} \text{ and for any proper prefix } w' \text{ of } w, \resClass{uw'} \neq \resClass{u} \}.
\]

Note that, if it is non-empty, $\locAlphBuchi{u}$ is a "prefix code", and therefore it is a "uniquely decodable alphabet".
Note that in general $\locAlphBuchi{u}$ may be infinite, however this is completely harmless in this context, and we will freely allow ourselves to talk about automata over infinite alphabets.\footnotemark{} Also, $\locAlphBuchi{u}$ is possibly empty; in the following definitions we assume that this is not the case.
\footnotetext{Note that in a finite automaton $\A$ over an infinite alphabet $\SS$, there are finitely many classes of letters such that two letters from the same class admit exactly the same transitions in $\A$. 
We can let $\SS_\mathrm{fin}$ be the set of equivalence classes of $\SS$ for this relation, and let $\A_\mathrm{fin}$ be the induced automaton over $\SS_\mathrm{fin}$.
It holds, that $w\in \Lang{\A}$ if and only if $w_\mathrm{fin}\in \Lang{\A_\mathrm{fin}}$ where $w_\mathrm{fin}$ is the projection of $w$ to $\SS_\mathrm{fin}$ .}

\AP Seeing words in $\locAlphBuchi{u}^\omega$ as words in $\SS^\omega$, define the ""localisation of $W$ to~$\resClass{u}$"" to be the objective
\[
	\intro*\locLangBuchi W u = \{ w \in \locAlphBuchi{u}^\omega \mid uw \in W\}.
\]
Observe that $\locLangBuchi W u$ is "prefix-independent". Moreover it is "positional": any "$\locLangBuchi{W}{u}$-game" in which "Eve" could not "play optimally using positional strategies" would provide a counterexample for the "positionality" of $W$.

For a state $q$ of a "B\"uchi automaton" $\A$ "recognising" $W$, define $\locAlphBuchi{q}$ and $\locLangBuchi{W}{q}$ in the natural way: $\locAlphBuchi{q} = \locAlphBuchi{u}$ and $\locLangBuchi{W}{q}=\locLangBuchi{W}{u}$ for $u$ a word that reaches $q$ from the "initial state".
Observe that a word $w\in \SS^*$ belongs to $\locAlphBuchi{q}^*$ if and only if it connects states in the class~$\resClassState{q}$. Elements in $\locAlphBuchi{q}$ are those that do not pass twice through this class.
We remark that $\locAlphBuchi{q} \neq \emptyset$ if and only if $q$ is a "recurrent" state (it belongs to some non-trivial "SCC").

\AP Let $q$ be a "recurrent" state. The ""local automaton of the residual $\resClassState{q}$"" is the "B\"uchi automaton" $\intro*\localAutResBuchi{q}$ defined as:
\begin{itemize}
	\item The set of states is $\resClassState{q}$.
	\item The initial state is arbitrary.
	\item For $w\in \locAlphBuchi{q}$, $q\re{w:x}q'$ if $q\lrp{w:x}q'$ in $\A$.
\end{itemize}
The language of $\localAutResBuchi{q}$ is $\locLangBuchi W u$, thus $\locLangBuchi W u$ is "Büchi recognisable".
Therefore, Proposition~\ref{prop-warm:char-Buchi-PI} yields that $\locLangBuchi{W}{u}$ is a "B\"uchi language": there exists a set $\BLettersBuchi{[u]} \subseteq \locAlphBuchi{u}$ such that $\locLangBuchi{W}{u} = \BuchiC{\BLettersBuchi{[u]}}{\locAlphBuchi{u}}$.
\AP We let $\intro*\NLettersBuchi{[u]} = \SS_{[u]} \setminus \BLettersBuchi{[u]}$ be the set of "non-super letters", and extend these notations to states of $\A$ by putting $\BLettersBuchi{[q]} = \BLettersBuchi{[u]}$ and $\NLettersBuchi{[q]}=\NLettersBuchi{[u]}$ where $u$ is any word leading from the initial state to $q$.

\subparagraph{Polished automata.} \AP For a "recurrent" state $q$, we say that a "residual class" $\resClassState{q}$ is ""polished in $\A$@@buchi"" if: 
\begin{enumerate}
	\item For all $q_1,q_2\in \resClassState{q}$, there is a word $u\in \NLettersBuchi{[q]}^*$ such that $q_1\lrp{u:1}q_2$.
	\item For every $q'\in \resClassState{q}$ and every word $u\in \NLettersBuchi{[q]}$, reading $u$ from $q'$ produces "priority" 1. 
\end{enumerate}
Stated differently, $\resClassState{q}$ is "polished@@classBuchi" if the restriction of $\localAutResBuchi{q}$ to transitions labelled with letters in $\NLettersBuchi{[q]}$ is strongly connected and does not contain any transition with "priority" $0$.

\AP We say that the automaton $\A$ is ""polished@@autBuchi"" if all its "residual classes" are "polished@@classBuchi".

\begin{remark}
	If $\A$ is "polished@@autBuchi" and $q$ is a "transient" state, then, it is the only state in its "residual class": $\resClassState{q} = \{q\}$.
	\AP We will apply the term ""recurrent@@class"" (resp. ""transient@@class"") to a class $\resClassState{q}$ if $q$ is "recurrent" (resp. "transient"). This is well defined by the previous comment. 
\end{remark}

\begin{lemma}[Obtaining a "polished@@autBuchi" automaton]\label{lemma-warm:polished-automaton}
	Any "positional" "Büchi recognisable" language $W$ can be "recognised" by a "polished@@autBuchi" "deterministic" "B\"uchi automaton".
\end{lemma}

\begin{proof}
	Let $\A$ be a "deterministic" "B\"uchi automaton" "recognising" $W$. 
	We will first "polish@@classBuchi" the "residual class" $\resClass{q}$ of a fixed state $q$.
	Consider the restriction $\localAutResBuchi{q}'$ of $\localAutResBuchi{q}$ to transitions labelled with $\NLettersBuchi{[q]}$, and take $S_{[q]}$ to be a final SCC of $\localAutResBuchi{q}'$; without loss of generality we assume that $q \in S_{[q]}$.

	\AP Now consider the automaton $\A'$ obtained from $\A$ by removing states in $[q] \setminus S_{[q]}$, and ""redirecting@@buchi"" transitions that go to $[q] \setminus S_{[q]}$ in $\A$ to transitions towards $q$ producing "priority"~$0$.
	\AP Note that this transformation ""preserves the residuals@@transform"": if $p_1\re{a}p_2$ in $\A$ and  $p_1\re{a}p_2'$ in $\A'$, then $p_2\eqRes{\A}p_2'$. Also, either $\sizeAut{\A'}<\sizeAut{\A}$, or $\A$ is left unchanged.
	We now prove that it preserves the language.

	\begin{claim}
		Automaton $\A'$ "recognises" the objective $W$.
	\end{claim}

	\begin{subproof}
		Let $w\in \SS^\oo$.		
		Suppose first that the "run over $w$" in $\A'$ eventually does not take "redirected@@buchi" transitions. Then, this run contains a suffix that is also a "run" in $\A$. As the transformation "preserves the residuals@@transform", $w$ is "accepted by@@aut" $\A'$ if and only if it is "accepted by@@aut" $\A'$.

		Suppose now that the "run over $w$" in $\A'$ takes infinitely many "redirected transitions@@buchi".
		Such a run in $\A'$ is of the form
		\[
			q_{\mathsf{init}} \lrp{w_0} q \lrp{w_1} p_1 \re{a_1:0} q \lrp{w_2} p_2 \re{a_2:0} q  \lrp{w_3} p_3 \re{a_3:0}\dots,
		\]
		where for all $i$ the transition $p_i \re{a_i:0} q$ is a "redirected@@buchi" one, meaning that in $\A$, reading $a_i$ from $p_i$ leads to $[q] \setminus S_{[q]}$.
		Note that $w$ is "accepted by@@aut" $\A'$, so we should prove that $w \in \L(\A)=W$.
		Observe that, for $i\geq 1$, $w_ia_i \in \locAlphBuchi{q}^*$ and reading $w_ia_i$ in $\A$ takes $q$ to $[q] \setminus S_{[q]}$, and thus by definition of $S_{[q]}$, it holds that $w_ia_i \notin \NLettersBuchi{[q]}^*$, so $w_i a_i \in \NLettersBuchi{[q]}^*\BLettersBuchi{[q]}^+\NLettersBuchi{[q]}^*$.
		We conclude that $w_1 a_1 w_2 a_2 \dots \in \BuchiC{\BLettersBuchi{[q]}}{\locAlphBuchi{q}} = q^{-1}W$ hence $w \in W$.
	\end{subproof}

	It follows that the residual class of $q$ in $\A'$ is $[q]_{\A'} = [q]_{\A} \cap Q'=S_{[q]}$.
	We now prove that it is "polished@@classBuchi".

%\begin{claim}\label{cl:nq-prio-1}
%	Paths in $\A$ from $S_{[q]}$ to $S_{[q]}$ labelled by words in $\NLettersBuchi{[q]}^*$ have "priority" $1$.
%\end{claim}
%\begin{subproof}
%	Take such a path $p \lrp w p'$, and assume for contradiction that it has "priority" $0$.
%	By definition of $S_{[q]}$, there is a path $p' \lrp {w'} p$ with $w' \in \NLettersBuchi{[q]}$.
%	But then $(ww')^\omega$ belongs to $q^{-1} W$ and to $\NLettersBuchi{[q]}^\omega$, which contradicts the fact that $\locLangBuchi{W}{u} = \BuchiC{\BLettersBuchi{[u]}}{\locAlphBuchi{u}}$.
%\end{subproof}

	\begin{claim}
		The residual class $[q]$ is "polished@@classBuchi" in $\A'$.
	\end{claim}

	\begin{subproof}
		Let $q_1,q_2 \in S_{[q]}$.
		By definition of $S_{[q]}$ there is $w \in \NLettersBuchi{[q]}$ such that $q_1 \lrp {w} q_2$ in $\A$. As this path avoids $[q] \setminus S_{[q]}$ in $\A$, it also belongs to $\A'$. We show that it produces exclusively priorities $1$.
		%Suppose by contradiction that it encounters a  "priority" $0$.
		By definition of $S_{[q]}$, there is a path $q_2 \lrp {w'} q_1$ with $w' \in \NLettersBuchi{[q]}$.
		Therefore $(ww')^\omega$ does not belongs to $\locLangBuchi{W}{q}$, so the path $q_1 \lrp {w} q_2$ cannot produce "priority"~$0$, which proves the first point in the definition of a "polished class@@buchi".
		
		Now take $q' \in S_{[q]}$ and $u \in \NLettersBuchi{[q]}^*$.
		Then by definition of $S_{[q]}$, reading $u$ from $q'$ in $\A$ leads back to $S_{[q]}$. By the previous argument, the minimal "priority" in this path is $1$.
		Again, this path avoids $[q] \setminus S_{[q]}$ in $\A$, so it also belongs to $\A'$.
	\end{subproof}

	Thus we have obtained an automaton $\A'$ for $W$ in which the class $\resClass{q}$ is "polished@@classBuchi".
	Since there are finitely many "residual classes", and the obtained automaton $\A'$ is strictly smaller than $\A$, we can repeat the process ("normalising" the automaton after each iteration, which does not increase the size) until obtaining an automaton in which all the classes are "polished@@classBuchi". (We remark that we do not claim that classes $\resClass{p}\neq \resClass{q}$ that were "polished@@classBuchi" in $\A$ will remain "polished@@classBuchi" in $\A'$. Nevertheless, the process reaches a fixpoint in which all classes are "polished@@classBuchi".)
\end{proof}

We will later on use the following property, which is our main reason for introducing "polished@@autBuchi" automata.

\begin{lemma}[Connection via losing words]\label{lemma-warm:connection-losing-words}
	Let $q\eqResState{\A}q'$ be two different "recurrent" equivalent states of a "polished automaton@@buchi" $\A$. Then there is a word $u\in \locAlphBuchi{q}^+$ such that $q\lrp{u}q'$ and $u^\oo \notin \Lang{\initialAut{\A}{q}}$.
\end{lemma}

\begin{proof}
	As the "automaton" $\A$ is "polished@@autBuchi", there is a word $u\in \NLettersBuchi{[q]}^+$ such that $q\lrp{u}q'$. This word satisfies the desired requirement.
\end{proof}

\subparagraph*{Uniform behaviour of \texorpdfstring{$0$}{0}-transitions. Proof of Lemma \texorpdfstring{\ref{lemma-warm:congruence-0-letters}}{4.19}.}
Let $\A$ be a "polished@@autBuchi" and "normalised" deterministic Büchi automaton "recognising" $W$, and suppose by contradiction that there are two states $q_1\eqResState{\A}q_2$ and a letter $a\in \SS$ such that $q_1\re{a:0}p_1$ and $q_2\re{a:1}p_2$.
By "normality", there is a word $w\in \SS^*$ such that $p_2\lrp{w:1}q_2$.
In particular, since $q_2 \rp{aw} q_2$, the "class@@res" $\resClassState{q_2}$ is "recurrent@class" in $\A$.
Note that $aw\in \locAlphBuchi{q_1}^*$ and $(aw)^\oo \notin q_1^{-1}W$. 
Let $q'$ be the state such that $p_1\lrp{w}q'$, note that $q'\in \resClassState{q_1}$.
Let $u_0$ be such that $q_\init \lrp {u_0} q_1$.
See Figure~\ref{fig-warm:non-congruent-situation} for an illustration of the situation.

Since $(aw)^\omega \notin q_1^{-1}W$, and $q_1 \lrp {aw:0} q'$, it cannot be that case that $q'=q_1$.
Hence, Lemma~\ref{lemma-warm:connection-losing-words} gives a word $u\in \locAlphBuchi{q}^+$ such that $q'\lrp{u\phantom{.}}q_1$ and $u^\oo\notin \Lang{\initialAut{\A}{q}}$.
Together, the facts that 
\begin{itemize}
	\item $(aw)^\oo\notin \Lang{\initialAut{\A}{q}}$,
	\item $u^\oo\notin \Lang{\initialAut{\A}{q}}$, and
	\item $(awu)^\oo\in \Lang{\initialAut{\A}{q}}$
\end{itemize}
prove that Eve wins the game on the right of Figure~\ref{fig-warm:non-congruent-situation}, but not "positionally@@win".

\begin{figure}
	%\centering
	\hspace{1mm}
	\begin{minipage}[c]{0.58\textwidth} 
		\includegraphics[scale=1.3    ]{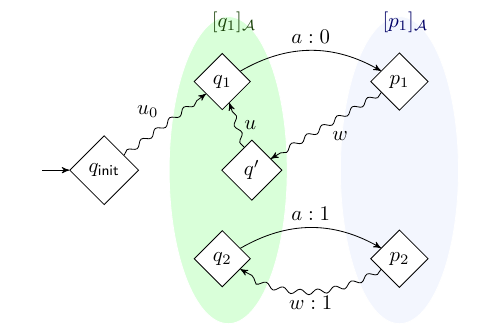}
	\end{minipage} 
	\hspace{0mm}
	\begin{minipage}[c]{0.38\textwidth} 
		\includegraphics[scale=1.4]{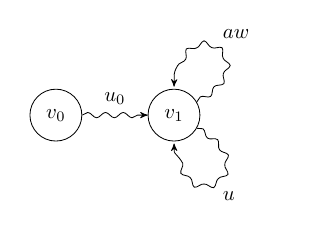}
	\end{minipage} 
	\caption{On the left, the situation in the proof of  Lemma~\ref{lemma-warm:congruence-0-letters}. We have the equivalences $q_1\eqRes{\A}q_2\eqRes{\A}q'$ and $p_1\eqRes{\A}p_2$. On the right, a game where "Eve" can produce a word in $\Lang{\A}$ but not "positionally@@win".}
	\label{fig-warm:non-congruent-situation}
\end{figure}

Thus we proved that if $W$ is a "Büchi recognisable" "positional" objective, it can be "recognised" by a "B\"uchi automaton" in which  "$0$-transitions" "behave uniformly".
By extending the technique from this section, we will show in Section~\ref{subsec:from-HP-to-signature} that this property (and its generalisation to all even priorities) also holds for automata using higher "priorities".

\subsection{Objectives recognised by \texorpdfstring{coB\"uchi}{coBüchi} automata: Total order given by safe languages}\label{subsec-warm:coBuchi}

We now consider "coB\"uchi recognisable" objectives.
Our analysis is based on "history-deterministic" automata and techniques from~\cite{AK22MinimizingGFG}.

Most of the section is devoted to the case of "prefix-independent" objectives, for which we propose a characterisation of "positionality" (Proposition~\ref{prop-warm:char-coBuchi-PI}). 
At the end of the section, we comment on how to extend this characterisation to any "coB\"uchi recognisable" objective (Proposition~\ref{prop-warm:char-coBuchi-all-lang}).
In the spirit of this warm-up, we omit proofs of sufficiency.

\paragraph*{Prefix-independent \texorpdfstring{coB\"uchi}{coBüchi} recognisable objectives}

Let us start by introducing some terminology relative to "coB\"uchi automata". %recognising "prefix-independent" languages. 
Our analysis requires considering "automata" that are not necessarily "deterministic", however, we have a fine control of the non-determinism that will appear.
First, all automata in this subsection will be "history-deterministic". 
Moreover, we can suppose that they are "deterministic over" transitions producing "priority" $2$ by the following result of Kuperberg et Skrzypczak~\cite{KS15DeterminisationGFG} (this property is sometimes called safe determinism).

\begin{lemma}[\cite{KS15DeterminisationGFG}]\label{lemma-prelim:coBuchi-det-over-2}
	Every "history-deterministic" "coB\"uchi automaton" contains an "equivalent@@aut" "subautomaton" that is "deterministic over" $2$-transitions, which can be computed in polynomial time.
\end{lemma}

\subparagraph{Safe languages and safe components.} \AP Consider a (possibly "non-deterministic") "coB\"uchi automaton" $\A$.
We define the ""$({<}2)$-safe language@@coB"" of a state $q$ (or just \emph{safe language})  as the set of finite or infinite words such that, when read from $q$, "priority" $1$ can be avoided, that is:
\[ \intro*\safeCoB{q} = \{w\in \SS^*\cup \SS^\oo \mid \text{ there is a run } q\lrp{w:2} \; \}. \]

\begin{remark}
	Two "safe languages" coincide if and only if their restrictions to finite (resp. infinite) words coincide. Indeed, an infinite word $w\in \SS^\oo$ belongs to $\safeCoB{q}$ if and only all its finite prefixes do.
\end{remark}

\AP We write $ q \intro*\leqCoB q'$ if $\safeCoB{q}\subseteq \safeCoB{q'}$; this defines a partial preorder on states of~$\A$. 
\AP We will sometimes use the term ""safe path@@coB"" to refer to paths in $\A$ that do not produce "priority"~$1$.
The use of the notation ``$(<2)$-safe'' will be justified by the generalisation of this notion to any "parity automaton".
The next lemma follows directly from the definition of "safe language@@coB".

\begin{lemma}[Monotonicity with respect to safe languages]\label{lemma-warm:monotonicity-1-safe}
	Let $\A$ be a "coB\"uchi automaton" which is "deterministic over" $2$-transitions, and let $q,q'$ be two states such that $q\leqCoB q'$. Let $u$ be a finite word in $\safeCoB q$ and write $q\lrp{u:2}p$. There is a unique path $q'\lrp{u:2}p'$ and $p\leqCoB p'$.
\end{lemma}

\AP A ""$({<}2)$-safe component@@coB"" (or just \emph{safe component}) of $\A$ is a set of states forming a "strongly connected component" in the "subautomaton" obtained by removing from $\A$ all transitions labelled $1$.
By Lemma~\ref{lemma-prelim:coBuchi-det-over-2}, we can suppose that these subautomata are "deterministic".
Also, note that if $\A$ is in "normal form", transitions between different "safe components@@coB" produce "priority" $1$, that is, states connected by a "safe path@@coB" are in the same "safe component".

%\begin{remark}\label{rmk-warm:accepting-in-safe-component}
%	A "run" of a "coB\"uchi automaton" is "accepting@@run" if and only if it eventually remains in a "safe component@@coB".
%\end{remark}

\subparagraph{Statement of the characterisation of positionality.}
We are now ready to state a characterisation of positionality of "prefix-independent" "coB\"uchi recognisable" objectives.

\begin{proposition}[Positionality for prefix-independent coB\"uchi recognisable objectives]\label{prop-warm:char-coBuchi-PI}
	A "prefix-independent" "coB\"uchi recognisable" objective is "positional" if and only if it can be "recognised" by a "deterministic" "coB\"uchi automaton" satisfying that within each "safe component@@coB", states are totally ordered by inclusion of "safe languages@@coB".	
\end{proposition}

Before going on with the proof, we discuss two examples.

\begin{example}\label{ex-warm:coBuchiNoConcave}
	In Figure~\ref{fig-warm:coBuchi-PI-not-concave} we represent a "coB\"uchi automaton" over $\SS = \{a,b,c\}$ "recognising" the following objective:
	\[ W = \text{ Words containing finitely many factors in } {c(a^*c b^*)^+c}. \]
	This is an example of an objective that is not "concave", as by shuffling words $a(ccaa)^\oo\notin W$ and $(bbcc)^\oo\notin W$ we can obtain $(abcc)^\oo\in W$. However, we show that it satisfies the hypothesis of Proposition~\ref{prop-warm:char-coBuchi-PI}, so it is "positional".	
	This automaton has a single "safe component@@coB". 
	The inclusions of the "safe languages@@coB" follows from the fact that the transitions are monotone: for every letter $\aa\in \SS$, if $q_i\re{\aa:2}q_{j}$ and $i\leq i'$, then $q_{i'}\re{\aa:2}q_{j'}$ with $j\leq j'$.
\end{example}
\begin{example}\label{ex-warm:coBuchi}
	Let $\SS = \{a,b,c\}$ and 
	\[ W = \finOften(ac) \vee \finOften(bb). \]
	We give a "coB\"uchi automaton" "recognising" $W$ and satisfying the hypothesis of Proposition~\ref{prop-warm:char-coBuchi-PI} in Figure~\ref{fig-warm:coBuchi-PI-ac-bb}. This automaton has two "safe components@@coB": $S_1 = \{q_1,q_2\}$ and $S_2 = \{p_1,p_2\}$. The states of each component are totally ordered by inclusion of "safe languages@@coB", as we have $q_1\lCoB q_2$ and $p_1\lCoB p_2$.
	Therefore, $W$ is "positional".
\end{example}
\begin{figure}
	\centering
	\includegraphics[scale=1.5]{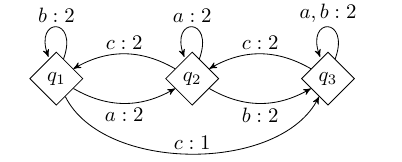}
	\caption{"Deterministic" "coB\"uchi automaton" "recognising" the objective $W$ from Example~\ref{ex-warm:coBuchiNoConcave}.}
	\label{fig-warm:coBuchi-PI-not-concave}
\end{figure}

\begin{cfigure}
	\centering
	\includegraphics[scale=1.5]{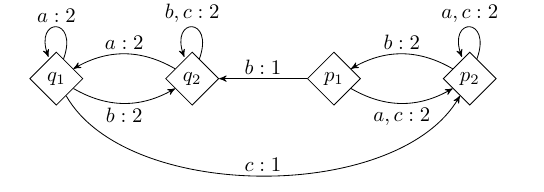}
	\caption{"Deterministic" "coB\"uchi automaton" "recognising" the objective $W$ from Example~\ref{ex-warm:coBuchi}. This automaton has two "safe components@@coB": $S_1 = \{q_1,q_2\}$ and $S_2 = \{p_1,p_2\}$, and the states of each of them are totally ordered by inclusion of "safe languages@@coB".}
	\label{fig-warm:coBuchi-PI-ac-bb}
\end{cfigure}

\begin{remark}
	In his PhD thesis~\cite[Section~6.2]{Kop08Thesis}, Kopczyński introduced a notion of \emph{monotonic automata} over finite words, and showed that if $L\subseteq \SS^*$ is recognised by such an automaton, then the ("prefix-independent") "objective" $\SS^\oo \setminus L^\oo$ is "positional"~\cite[Prop~6.6]{Kop08Thesis}.
	It turns out that these correspond exactly to the "objectives" characterised in Proposition~\ref{prop-warm:char-coBuchi-PI}.
\end{remark}

At the level of intuition, starting from a deterministic coB\"uchi automaton $\A$ "recognising" a "positional" objective $W$, our proof of the necessity in Proposition~\ref{prop-warm:char-coBuchi-PI} proceeds as follows:
\begin{itemize}
\item We turn $\A$ into a "history-deterministic" "safe centralised@@coB" automaton (see definition below) using the minimisation technique of Abu Radi and Kupferman~\cite{AK22MinimizingGFG}.
\item Using "positionality", we prove that $\leqCoB$ defines a total order on each "safe component".
\item Exploiting the total order, we are able to re-determinise $\A$.
\end{itemize}

\subparagraph{Safe centralisation and safe minimality.} 
Let $\A$ be a (possibly "non-deterministic") "coB\"uchi automaton" with only one "residual class@@state".
\AP We say that $\A$ is ""safe centralised@@coB"" if, for every pair of states $q_1,q_2$, %the inclusion $\safeCoB{q_1}\subseteq \safeCoB{q_2}$ 
if $q_1\leqCoB q_2$, then $q_1$ and $q_2$ are in the same "safe component@@coB".
\AP We say that $\A$ is ""safe minimal@@coB"" if there are no two different states with the same "safe language@@coB".

The next lemma is a consequence of~\cite{AK22MinimizingGFG}.
We present a self-contained proof that will be generalised in Lemma~\ref{lemma-nec:safe-centralisation}.

\begin{lemma}\label{lemma-warm:safe-centralisation}	Any "prefix-independent" "coB\"uchi recognisable" language can be "recognised" by a  "history-deterministic" "coB\"uchi automaton" that is "safe centralised@@coB" and "safe minimal".	
\end{lemma}

%In order to simplify the proof of this lemma, we introduce "$1$-saturated" automata (already considered in~\cite{AK22MinimizingGFG}).

\begin{proof}[Proof of Lemma~\ref{lemma-warm:safe-centralisation}]
	Let $\A$ be a "normalised", "history-deterministic", "deterministic over" $2$-transitions automaton "recognising" $W$. %, obtained by Lemma~\ref{lemma-warm:one-saturation}.
	
	\AP First, we show that we can assume $\A$ to be ""$1$-saturated"", that is, for all pairs of states $q,q'$ and letters $a\in \Sigma$, the transition $q \re {a:1} q'$ appears in $\A$.
	%\AP We consider the "coB\"uchi automaton" obtained by adding all possible transitions of the form $q \re {a:1} q'$; note that this transformation preserves "determinism over" $2$-transitions.

	\begin{claim}\label{cl:cl1-safesaturation}%{customclaim}
		The "coB\"uchi automaton" obtained by adding all possible transitions of the form $q \re {a:1} q'$ to $\A$ recognises~$W$. Moreover, it is "history-deterministic" and "deterministic over" $2$-transitions.
	\end{claim}
	
	\begin{subproof}
		Let $\A'$ be the automaton obtained adding all $1$-transitions.
		To show that $\Lang{\A'}\subseteq W$, we remark that an "accepting@@run" "run over" a word $w$ in $\A'$ eventually only reads transitions with "priority" $2$, so it eventually coincides with a run in $\A$. We conclude by "prefix-independence".
	For the other inclusion, and the fact that $\A'$ is "history-deterministic", it suffices to use the same "resolver" as in $\A$.
	It is straightforward that $\A'$ is "deterministic over" $2$-transitions.
	\end{subproof}

	In the following, we assume that $\A$ is "$1$-saturated".
	\AP We say that a "safe component@@coB" $S$ is ""redundant@@coB"" if there is $q \in S$ and $q' \notin S$ such that $q \leqCoB q'$.

	\begin{claim}\label{cl:cl1-safe-centralisation}%{customclaim}
		Let $S$ be "redundant@@coB" and consider the automaton $\A'$ obtained from $\A$ by deleting $S$.
		Then $\A'$ is "history-deterministic" and "recognises" $W$.
	\end{claim}
	
	\begin{subproof}
		Clearly $\L(\A') \subseteq W$. 
		We will describe a "sound resolver" proving that $\L(\A') = W$ and that $\A'$ is "history-deterministic".
		Let $q \in S$ and $q' \notin S$ such that $q \leqCoB q'$.
		For each $p \in S$, pick $u\in \SS^*$ such that $q \lrp{u:2} p$, and let $f(p)$ be such that $q' \lrp{u:2} f(p)$; this is well defined since $q \leqCoB q'$.
		Note that we have $p \leqCoB f(p)$ by Lemma~\ref{lemma-warm:monotonicity-1-safe}.
		By "normality" of $\A$, there is a returning path $f(p) \lrp{w:2} q'$ and thus $f(p)$ is in the same "safe component" as $q'$, so it does not belong to $S$.
		We extend $f$ to all of $Q$ by setting it to be the identity over $Q \setminus S$.

		Take a "sound resolver" $(q_0,\resolv)$ in $\A$, let $w\in \SS^\oo$, and write
		\[		\rho=q_0 \re{w_0} q_1 \re{w_1} \dots\]
		for the run in $\A$ "induced by@@resolver" $\resolv$ over $w$.
		We build a "resolver" $(q_0',\resolv')$ in $\A'$ satisfying the property that the "run induced over $w$", $\rho'=q'_0 \re{w_0} q'_1 \re{w_1} \dots$ is such that for each $i$, $q_i \leqCoB q'_i$.
		We let $q'_0=f(q_0)$, and assume $\rho'$ constructed up to $q'_i$, and $q_i\leqCoB q_i'$.		
		If there is a state $q'_{i+1} \notin S$ such that $q'_i \re{w_i:2} q'_{i+1}$ then we take this one,  which satisfies $q_{i+1} \leqCoB q_{i+1}'$ by Lemma~\ref{lemma-warm:monotonicity-1-safe}. Otherwise, take the transition $q'_{i} \re{w_i:1} f(q_{i+1})$ (which exists by "$1$-saturation"). %In particular, if $w\in \safeCoB{q_i}$, the "induced run" from $q_i$ does not produce "priority" $1$.	
		Therefore, if $\rr$ is "accepting@@run",  there is a suffix $w_iw_{i+1} \dots \in \safeCoB{q_i} \subseteq \safeCoB{q'_i}$, so the run $q'_i \re {w_i:2} q'_{i+1}\re {w_{i+1}:2}\dots$ in $\A'$ is "safe@@runCoB", and $\rr'$ is "accepting@@run" too.
	\end{subproof}
	
	Using Claim~\ref{cl:cl1-safe-centralisation}, we successively remove "redundant@@coB" "safe components@@coB" until obtaining a "safe centralised@@coB" automaton.
	
	Finally, to obtain a "safe minimal" automaton it suffices to merge states with the same "safe language". That is, we define a "$1$-saturated" automaton that has for states the classes $[q]_2$ of states of $\A$, and transitions $[q]_2 \re{a:2} [p]_2$ whenever for some (or equivalently, for all) state $q'\in [q]_2$ there is transition $q'\re{a:2}p'$ in $\A$, with $p'\in [p]_2$. It is not difficult to check that the obtained automaton "recognises" $W$, is "history-deterministic" and remains "safe centralised@@coB".
\end{proof}

\subparagraph{Total order in each safe component.} The intuitive idea on why having states of a same "safe component@@coB" totally ordered by $\leqCoB$ is necessary for positionality is the same than in the case of "closed objectives" (Lemma~\ref{lemma-warm:total-order-residuals-nec}): if $q$ and $q'$ are incomparable, there are two words $w,w'$ that produce "priority" $1$ from one state but not from the other.
In a game, if "Eve" has not kept track of where we are in the automaton, she will not know what is the best option between $w$ and $w'$.
However, an issue arises when turning this idea into an actual proof: one needs to build two full differentiating runs from $q$ and $q'$; producing "priority" $1$ just once does not suffice.
"Safe centrality@@coB" will come in handy for this purpose.

By definition, if $q\nleqCoB q'$, there is a word which produces "priority" $1$ when read from $q'$, and stays in the corresponding "safe component@@coB" when read from $q$.
The following lemma exploits "safe centrality@@coB" to extend those runs to synchronise them in the same state, while the run starting from $q$ remains safe.
For the purpose of the warm up, we only prove it assuming $\A$ is "deterministic"; extending it to the "history-deterministic" case requires some additional technicalities that will be dealt with in Section~\ref{subsec:from-HP-to-signature}.

\begin{lemma}[Synchronisation of separating runs]\label{lemma-warm:syncr-separating-runs}
	Let $\A$ be a "normalised", "safe centralised@@coB" and "safe minimal" "deterministic" "coB\"uchi automaton" with a single "residual class". Let $q$ and $q'$ be two states such that $q\nleqCoB q'$ and $p$ be any state in the "safe component@@coB" of~$q$.
	There is a word $w\in \SS^*$ such that $q\lrp{w:2} p$ and $q'\lrp{w:1}p$.
\end{lemma}
\begin{proof}
	Since $\safeCoB{q}\nsubseteq\safeCoB{q'}$, there is a word $w_1\in \SS^*$ such that $q\lrp{w_1:2}q_1$ and $q'\lrp{w_1:1}q_1'$.
	By "normality", note that $q_1$ is in the same "safe component@@coB" as $q$.
	If we have again that $q_1\nleqCoB q_1'$, we can find a word $w_2$ with the same properties.
	While the non-inclusion of "safe languages@@coB" is satisfied, repeating the argument yields two runs:
	\begin{alignat*}{3}
		&q\lrp{w_1:2} \,& q_1\lrp{w_2:2} \,& q_2\lrp{w_3:2} \dots \\
		&q'\lrp{w_1:1} \,& q_1'\lrp{w_2:1} \,& q_2'\lrp{w_3:1} \dots
	\end{alignat*}
	We claim that the process should stop after finite time, meaning that for some $i$, we have $q_i\leqCoB q_i'$.
	Otherwise, we would obtain two infinite runs over $w_1w_2w_3\dots\in \SS^\oo$, one of them "accepting@@run" and the other "rejecting@@run",
	contradicting the fact that $\A$ has a single "residual class@@state".
	
	Let $i$ be the step in which $q_i\leqCoB q_i'$. First, we note that both states are in the "safe component@@coB" of $q$: state $q_i$ is in there because there is a path $q\lrp{2}q_i$, and by "safe centrality@@coB", $q_i'$ must also be in the same "safe component@@coB".
	Let $q_{\max}$ be a state in this "safe component@@coB" maximal amongst states such that $q_i'\leqCoB q_{\max}$. Let $u\in \SS^*$ be a word labelling a path $q_i\lrp{u:2} q_{\max}$ (which exists by definition of "safe components@@coB").
	By Lemma~\ref{lemma-warm:monotonicity-1-safe}, maximality of $q_{\max}$, and "safe minimality", we also have $q_i'\lrp{u:2} q_{\max}$.
	Finally, if suffices to take a word $u'\in \SS^*$ labelling a path $q_{\max}\lrp{u':2} p $ and define $w = w_1w_2\dots w_iuu'$.
\end{proof}

We may derive the sought total orders.

\begin{lemma}[Total order in each safe component]\label{lemma-warm:total-order-coBuchi}
	Let $\A$ be a "deterministic" "coB\"uchi automaton" "recognising" a "prefix-independent" "positional" objective $W$. 
	Suppose that $\A$ is "safe centralised@@coB" and "safe minimal".
	Let $q$ and $q'$ be two different states in the same "safe component@@coB".
	Then, either $q\lCoB q'$ or $q'\lCoB q$.	
\end{lemma}
\begin{proof}
	By "safe minimality", $q\nlCoB q'$ implies $q\nleqCoB q'$.
	Suppose by contradiction that $q\nleqCoB q'$ and $q'\nleqCoB q$. Let $p$ be a state in this "safe component@@coB", and let $u,u'\in \SS^*$ be such that $p\lrp{u:2}q$ and $p\lrp{u':2}q'$.
	By Lemma~\ref{lemma-warm:syncr-separating-runs}, there are $w,w'\in \SS^\oo$ such that:
	
	%\begin{center}
		\begin{tabular}{l l}
			\centering
			$q\lrp{w:2}p$, & $q\lrp{w':1}p$,\\ 
			$q'\lrp{w:1}p$, & $q'\lrp{w':2}p$. 
		\end{tabular} 
	%\end{center}
	
	\vspace{2mm}
The situation is depicted in Figure~\ref{fig-warm:totalOrder-coBuchi}. 
We obtain that:
\begin{itemize}
	\item $(w'u)^\oo\notin W$,
	\item $(wu')^\oo\notin W$, 
	\item $(wu'w'u)^\oo\in W$.
\end{itemize}

\begin{figure}
	\centering
	\includegraphics[scale=1.5]{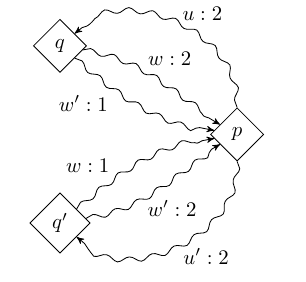}
	\caption{Situation occurring in the proof Lemma~\ref{lemma-warm:total-order-coBuchi}.}
	\label{fig-warm:totalOrder-coBuchi}
\end{figure}

Thus consider the "game" in which "Eve" controls a vertex with two self loops labelled $w'u$ and $wu'$; she can "win" by alternating both loops, but fails to "win positionally".
\end{proof}

\subparagraph{Determinisation.}\label{par-warm:determinisation} We have stated the last two lemmas of the previous paragraph for "deterministic" automata, in order to simplify the presentation.
However, "safe centralisation@@coB" only yields "history-deterministic" automata.
We now explain how to use the total order from Lemma~\ref{lemma-warm:total-order-coBuchi} to determinise "HD" "coB\"uchi automata" "recognising" "positional" languages.

Let $\A$ be a "normalised", "history-deterministic" and "deterministic over" $2$-transitions "coB\"uchi automaton" recognising a "prefix-independent" "positional" language. Order~$\leqCoB$ is total over each "safe component@@coB", by Lemma~\ref{lemma-warm:total-order-coBuchi}.
We show how to rearrange the $1$-transitions in order to define an "equivalent@@aut" "deterministic" "automaton" $\A'$ with the same structure of "safe components@@coB".

%First, we remark that we can suppose that if a transition $q\re{a:2}$ exists from state $q$, then it is the only "$a$-transition@@in  from $q$. (This property is called "$1$-homogeneity" in~\cite{AK22MinimizingGFG}, it is easily derived from the proof of Lemma~\ref{lemma-warm:safe-centralisation}.)

Let $\{S_1,S_2,\dots, S_k\}$ be the "safe components@@coB" of $\A$, enumerated in an arbitrary order.
If a word $w\in \SS^\oo$ is "accepted by@@aut" $\A$, there is a "run over" $w$ that eventually stays in one of the components $S_i$.
The main idea is that, when reading a word $w$, we can "resolve@@aut" the "non-determinism" by trying each "safe component@@coB" in a round-robin fashion.
If for a state $q$ and for a letter $a\in \SS$ there is a (unique) transition $q\re{a:2}$, we keep it as the only "$a$-transition@@in" from $q$.
If there is no transition $q\re{a:2}$, and $q$ belongs to $S_i$, we define a transition $q\re{a:1}q'$ towards some $q'$ in $S_{i+1}$. 
The total order in $S_{i+1}$ identifies a state in $S_{i+1}$ which is the best to go to: we define $q\re{a:1}q_{i+1}^{\max}$, where $q_{i+1}^{\max}$ is the unique maximal state of $S_{i+1}$ for the total order $\leqCoB$.
This defines a "deterministic" automaton $\A$'.

%We let $\A'$ be the obtained automaton. 
We prove that $\Lang{\A'} = \Lang{\A}$. Clearly $\Lang{\A'} \subseteq \Lang{\A}$, as $\A'$ is a "subautomaton" of the "$1$-saturation" of $\A$.
To show the other inclusion, let $w\in \Lang{\A}$. There is a "run over@@aut" $w$ in $\A$ that eventually remains in a "safe component@@coB", without loss of generality, we assume that it is~$S_1$.
Let $w'=w_1'w_2'\dots$ be a suffix of $w$ labelling a "run@@aut" in $S_1$: $q^1\re{w_1'}q^2\re{w_2'}\dots$. Consider a "run over" $w'$ in $\A'$.
If this "run" never visits $S_1$, it must be because it eventually remains in some "safe component@@coB", so $w'$ is "accepted by@@aut" $\A'$ in that case. If the run eventually arrives to~$S_1$ (at the i\ts{th} step), it will arrive to the maximal element $q_1^{\max}$.
At this point, the run in $\A$ is in $q^i\leqCoB q_1^{\max}$. Since  $w_i'w_{i+1}'\dots \in \safeCoB{q^i} \subseteq \safeCoB{q_1^{\max}}$, the run over this suffix is "safe@@runCoB" in $\A'$, and word $w$ is "accepted by" $\A'$.
 
\begin{example}
	The automaton from Figure~\ref{fig-warm:coBuchi-PI-ac-bb} has the shape we have described: transitions producing "priority" $1$ cycle between the two "safe components@@coB", and they go to the maximal state of the other component ($p_1\re{b:1}q_2$ and $q_1\re{c:1}p_2$).
\end{example} 

\paragraph*{Generalisation to non prefix-independent \texorpdfstring{coB\"uchi}{coBüchi} recognisable languages}

To remove the "prefix-independence" assumption, we work with the "localisation to residuals" of the objectives, as defined in Section~\ref{subsec-warm:Buchi}.
If $W$ is a "positional" objective, for each residual $\resClass{u}$ the objective $\locLangBuchi{W}{u}$ is  "positional".
In turns out that this property, together with the hypothesis over residuals that were already necessary for "open objectives", provides a characterisation.

\begin{proposition}\label{prop-warm:char-coBuchi-all-lang}
	Let $W\subseteq \SS^\oo$ be a "coB\"uchi recognisable" language. Then, $W$ is "positional" if and only if:
	\begin{itemize}
		\item $\ResW$ is totally ordered,
		\item $W$ is "progress consistent", and
		\item for all "residual class" $\resClass{u}$, objective $\locLangBuchi{W}{u}$ is "positional".
	\end{itemize}
\end{proposition}

This result is not fully satisfying (and hard to prove directly), as it relies on the "positionality" of languages $\locLangBuchi{W}{u}$.
As these objectives are "prefix-independent", we have a characterisation for their "positionality" (Proposition~\ref{prop-warm:char-coBuchi-PI}), and we can put them together to obtain a statement using exclusively structural properties of "parity automata".
The statement we obtain uses a hyerarchical decomposition of parity automata in three layers:
\begin{itemize}
	\item States are totally "preordered" by their "residual class" (layer 0).
	\item Within each "residual class", states are grouped into "safe components" (layer 1).
	\item Within each "safe component", states are totally ordered by inclusion of the "safe languages" (layer 2).
\end{itemize}

This hyerarchical decomposition foreshadows the definition of "structured" "signature automata" that we will use in Section~\ref{subsec:from-HP-to-signature} to derive a characterisation of "positionality" for all "$\oo$-regular" languages.

\begin{proposition}\label{prop-warm:char-coBuchi-all-aut}
	Let $W\subseteq \SS^\oo$ be a "coB\"uchi recognisable" language. Then, $W$ is "positional" if and only if:
	\begin{itemize}
		\item $\ResW$ is totally ordered,
		\item $W$ is "progress consistent", and
		\item $W$ can be "recognised" by a "deterministic" "coB\"uchi automaton" $\A$ such that, for all "residual class" $\resClassState{q}$, the "local automaton of the residual" $\localAutResBuchi{q}$ satisfies that its "safe components@@coB" are totally ordered by inclusion of "safe languages@@coB".
	\end{itemize}
\end{proposition}

We do not include a proof of this proposition; it is a special case of Theorem~\ref{th-reslt:MainCharacterisation-allItems}.

\begin{example}\label{ex-half:BFMM}
	We consider the following objective over the alphabet $\SS = \{a,b,c\}$:
	\[ W = \SS^*a^\oo \cup  \SS^*b^\oo \cup c\SS^*c\SS^\oo.\]
	This objective was studied in~\cite[Lemma~12]{BFMM11HalfPos} to show that there are "positional" objectives that are not "concave", nor "bipositional" (objectives from Examples~\ref{ex-warm:Buchi-a-aa} and~\ref{ex-warm:coBuchiNoConcave} also have this property); their proof of "positionality" is quite involved. 
	
	A "coB\"uchi automaton" "recognising" the objective $W$ is depicted in Figure~\ref{fig-half:example-BFMM}. Its "residuals" are totally ordered, and it is easy to check that it is "progress consistent".
	Moreover, all its "safe components@@coB" are trivial. Therefore, 
	Proposition~\ref{prop-warm:char-coBuchi-all-aut} implies that it is "positional".	
\end{example}
\begin{figure}
	\centering
	\includegraphics[scale=1.5]{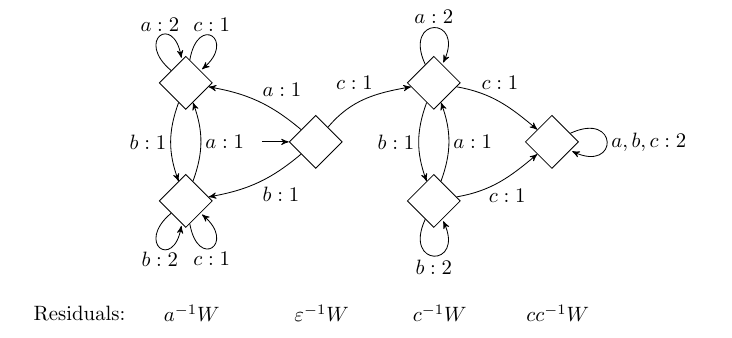}
	\caption{Automaton "recognising" the objective $W$ from Example~\ref{ex-half:BFMM}.}
	\label{fig-half:example-BFMM}
\end{figure}

\subsection{Towards objectives of higher complexity: An example}\label{subsec-warm:full-example}

We revisit the example from Figure~\ref{fig:epsilon-completable-automaton}, depicted here in Figure~\ref{fig-warm:3-priorities}. It "recognises" the "objective" of words that either contain `$a$' infinitely often, or contain no `$a$' at all and only finitely many occurrences of the factor `$bb$':

\[ W = \infOften(a) \vee (\noOcc(a) \wedge \finOften(bb)), \; \text{ over } \SS = \{a,b,c\}.\]

This "objective" is neither "B\"uchi@@rec" nor  "coB\"uchi recognisable", so none of the characterisations of this section applies to it. However, we can combine the techniques presented above to equip $\A$ with a ``nicely behaved'' total order. In the next section, we will see that this can be formalised as the fact that $\A$ is a "deterministic" "fully progress consistent" "signature automaton", so by Theorem~\ref{th-reslt:MainCharacterisation-allItems}, $W$ is "positional".

\begin{cfigure}
	\centering
	\includegraphics[scale=1.5]{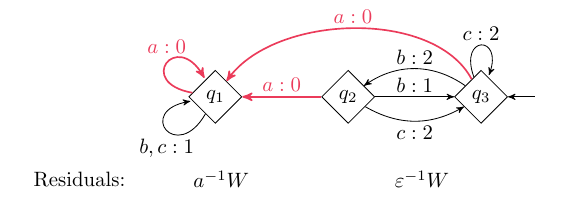}
	\caption{Automaton "recognising" the objective $W= \infOften(a) \vee (\finOften(bb) \wedge \noOcc(a))$.}
	\label{fig-warm:3-priorities}
\end{cfigure}

The objective $W$ has two "residuals": $\lquotW{\ee}$ and $\lquotW{x}$.
It is clear that $\lquotW{a} \subseteq \lquotW{\ee}$, and, since there is no transition $\resClass{a}\re{} \resClass{\ee}$, it is trivially "progress consistent".
The "states associated to@@res" $\resClass{a}$ and $\resClass{\ee}$ are, respectively, $\{q_1\}$ and $\{q_2, q_3\}$.
The "subautomaton" over $\{q_1\}$ "recognises" the "B\"uchi objective" $\Buchi{a}$, which is "positional".
Let us focus on the "subautomaton" induced by $\{q_2,q_3\}$, that coincides with $\localAutResBuchi{\ee}$.
We observe that it satisfies the hypothesis of Lemma~\ref{lemma-warm:congruence-0-letters}: transitions with "priority" $0$ "act uniformly" within $\localAutResBuchi{\ee}$; indeed, these transitions are those reading letter $a$ (in red).
Consider the restriction of this subautomaton to letters $\{b,c\}$. The obtained automaton $\localAutResBuchi{\ee}'$ is a "coB\"uchi automaton", satisfying the hypothesis of Lemma~\ref{lemma-warm:total-order-coBuchi}:  the states of the "safe component@@coB" of $\localAutResBuchi{\ee}'$ are totally ordered by inclusion of "safe languages@@coB", as $q_2 \lCoB q_3$.
In this way, we obtain a hyerarchical decomposition such that:
\begin{itemize}
	\item "Residuals" are totally ordered by inclusion, and are "progress consistent".
	\item $0$-transitions "act uniformly" within the "states of@@res" each "residual class".
	\item The "safe components@@coB" of the "coB\"uchi automaton" obtained as the restriction of each class to transitions with "priority" $\geq 1$ are totally ordered by  inclusion of "safe languages@@coB".
\end{itemize}

Generalising such hyerarchical decomposition to any "parity automaton" will be the central point of the next section.

We also note that the "$\ee$-completion" of this automaton presented in Figure~\ref{fig:epsilon-completable-automaton} follows the decomposition presented here: $\re{\ee:0}$-transitions follow the order of the residuals, and $\re{2}$-transitions the one given by the inclusion of safe languages. 
This procedure of "$\ee$-completion" will be generalised in Section~\ref{subsec:from-signature-to-HP}.

\section{Obtaining the structural characterisation of positionality}\label{sec:proof}
We now move on to the proof of Theorem~\ref{th-reslt:MainCharacterisation-allItems}.
The first step is to identify the common structural properties of "deterministic" "parity automata" "recognising" "positional" objectives.
In Section~\ref{subsec:def-signature}, we define a class of "parity automata", which we call "signature" automata, that underscores this structure.
We also introduce "full progress consistency", a further necessary condition for automata recognising "positional" languages.
To obtain the implication from \ref{item-th:halfPosFiniteEve} to \ref{item-th:signature} in Theorem~\ref{th-reslt:MainCharacterisation-allItems} we need then to show how to obtain a "fully progress consistent" "signature automaton" for a "positional" "$\oo$-regular" objective.
This is the most technical part of the proof, and it is the object of Section~\ref{subsec:from-HP-to-signature}.
We then proceed to defining "$\ee$-complete automata" and closing the cycle of implications in Section~\ref{subsec:from-signature-to-HP}.

\subsection{Signature automata and full progress consistency}\label{subsec:def-signature}

%Before finally moving on to the crucial definition of "signature automaton", 

We first introduce more precise concepts about congruences.

\subsubsection{Priority-faithful congruences and quotient automata}

\subparagraph{Priority-faithful congruences.} We recall that a "congruence" in an automaton allows us to define a "quotient@@aut" $\quotAut{\A}{}$, which is a "deterministic" "automaton structure". However, in general, no "colouring with priorities" can be defined "on top of" $\quotAut{\A}{}$ in a sensible way, as the "congruence" does not have to be compatible with the "output priorities@@aut" of the automaton. We now strengthen the definition of "congruence" for parity automata so that it will be possible to define an approximation of a correct parity condition "on top of" the "quotient automaton".

\begin{definition}[{$[0,x]$-faithful congruence}]\label{def-sig:faithful-cong}
\AP	Let $\A$ be a "parity automaton" and let $\sim$ be a "congruence" over its set of states $Q$. We say that~$\sim$ is ""$[0,x]$-faithful"" if:
	\begin{itemize}
		%\item $\sim$ "refines" relation $\eqRes{\A}$ given by equality of "residuals@@autState",
		\item for each $0\leq y\leq x$, "$y$-transitions@@out" "are uniform" over $\sim$-classes and relation $\sim$ is a "congruence for $y$-transitions", and
		\item relation $\sim$ is a "congruence for $({>}x)$-transitions".
	\end{itemize}
\end{definition}

Stated differently, a relation $\sim$ is a "$[0,x]$-faithful congruence" if, whenever there is a transition $q\re{a:y}p$ with "priority" $y\leq x$, then for every $q'\sim q$ a transition $q'\re{a}p'$ produces "priority" $y$ and goes to a state $p'\sim p$.
If there is a transition $q \re{a:y} p$ producing "priority" $y> x$, we only require that transitions $q'\re{a:>x}p'$ produce "priorities" $>x$ and go to $p'\sim p$ (but the exact "priority" produced may differ). (We remark that, although it is not explicitly imposed, "$({>}x)$-transitions@@out" "are uniform" over $\sim$-classes by the first property. Also, we recall that we assume that automata are "complete".)

\begin{remark}[{$[0,x]$-faithful congruences for deterministic automata}]
	For "deterministic" automata we can give a simpler definition. An equivalence relation $\sim$ on a "deterministic automaton" $\A$ is a "$[0,x]$-faithful congruence" if and only if, whenever $q\sim q'$, $q\re{a:y} p$ and $q'\re{a:y'} p'$, then $p\sim p'$ and $y=y'$ if $y\leq x$. 
\end{remark}

\begin{remark}
	A "$[0,x]$-faithful congruence" is "$[0,y]$-faithful" for any $y\leq x$.
\end{remark}

%\begin{remark}
%	In most cases, we will work with "faithful" congruences that moreover "refine" the "congruence" $\eqRes{\A}$ given by the equality of "residuals". However, we do not impose this property in the definition, as precisely one of the main applications of "faithful congruences" will be to help us to show that the states of the automaton $\A$ under consideration "recognise" the desired languages.
%\end{remark}

%\begin{remark}
%	If  a "parity automaton"  $\A$ be admits a "$[0,x]$-faithful congruence", then $\A$ is necessarily "${\geq}x$-homogeneous". %%%%{In the following, we will suppose for simplicity that all automata are "homogeneous" (for all priorities).}
%\end{remark}

%\begin{remark}
%	A "$[0,x]$-faithful congruence" is a "congruence for" the whole set of transitions. That is, the (deterministic) "quotient automaton" $\quotAut{\A}{}$ is well-defined.
%\end{remark}

\subparagraph{\texorpdfstring{$({\leq}x)$}{leq x}-quotient automaton.} What the definition of "$[0,x]$-faithful congruence" tells us is that transitions producing priorities $y\leq x$ are well defined in the "quotient automaton" $\quotAut{\A}{}$, in the sense that we can associate a "priority" $y$ to these transitions reliably. For transitions producing priorities $>x$, on the other hand, we only obtain the information that the "priority" produced from any state of the class will be large, but we lose some precision.

\begin{definition}[{$({\leq}x)$-quotient automaton}]\label{def:quotient-faithful}
	\AP If $\sim$ is a "$[0,x]$-faithful congruence", we can define the ""$({\leq}x)$-quotient of $\A$ by $\sim$"" to be the "parity automaton" $\intro*\autLeq{\A}{}{x}$ given by:
	
	\begin{itemize}
		\item Its set of states are the $\sim$-classes of $\A$.
		\item The initial state is $[q_{\mathsf{init}}]$, where $q_{\mathsf{init}}$ is the initial state of $\A$.
		\item For $y\leq x$, there is a transition $[q] \re{a:y} [p]$ if  $\A$ has a transition $q'\re{a:y} p'$ with $q'\in [q]$, $p'\in [p]$.
		\item There is a transition $[q] \re{a:x'}[p]$ if $\A$ has a transition $q'\re{a:>x} p'$ with $q'\in [q]$, $p'\in [p]$; where $x'=x+1$ if $x$ even, and $x'=x$ if $x$ odd.
	\end{itemize}
\end{definition}

%%%%%%%%	\AP We say that a run in $\autLeq{\A}{x}{x}$ is ""well-formed for $\autLeq{\A}{x}{x}$"" if it does not eventually consist of transitions of the form $\classSig{x}{q} \re{a:x+1}\classSig{x}{p}$.
	%	If $x$ is even, runs that are not "well-formed" are "rejecting@@run" in $\autLeq{\A}{x}{x}$.}

The automaton $\autLeq{\A}{}{x}$ is defined so transitions coming from those producing a "priority" $>x$ in $\A$ are assigned the smallest odd "priority" not smaller than $x$. This guarantees that the "projection@@quot" of "runs" that eventually only produce priorities $>x$ in $\A$ are "rejecting@@run" in $\autLeq{\A}{}{x}$. The next lemma refines this comment.

\begin{lemma}\label{lemma-sig:language-quotient-aut}
	Let $\A$ be a "parity automaton" and let $\sim$ be a "$[0,x]$-faithful congruence" over it. 
	The "$({\leq}x)$-quotient" of $\A$ by $\sim$ "recognises" the language:
	\[ \Lang{\autLeq{\A}{}{x}} = \{w\in\SS^\oo \mid w \text{ is "accepted with an even priority" } y\leq x \tin \A \}. \] 
	Moreover, if $\A$ is in "normal form", then so is $\autLeq{\A}{}{x}$.
\end{lemma}
\begin{proof}
	A "run" $\rr$ in $\A$ produces a "priority" $y\leq x$ infinitely often if and only its "projection@@quot" in $\autLeq{\A}{}{x}$ produces "priority" $y$ infinitely often, which gives the equality of the languages.
	It is a direct check that $\autLeq{\A}{}{x}$ inherits being in "normal form".
\end{proof}

%%%%%%%LE lemme dessous est interessant et je pensais qu'on voudrait l'utiliser, mais apparament on le fait pas}

%\begin{lemma}\label{lemma-sig:leq-quotients-are-faithful}
%	Let $\A$ be a "parity automaton" and $x\leq y$ be two priorities. Assume that $\sim_x$ and $\sim_y$ are two "congruences" that are, respectively,  "$[0,x]$@@faithful" and "$[0,y]$-faithful".
%	Assume that $\sim_y$ is a "refinement" of $\sim_x$. Then, the induced relation $\sim_y$ in $\autLeq{\A}{x}{x}$ is  "$[0,y]$-faithful".
%\end{lemma}

\subsubsection{Signature automata}

We give the definition of "signature automata", at the core of our characterisation.

\AP We say that a sequence of "total preorders" $\sqsubseteq_0, \sqsubseteq_1,\sqsubseteq_2,\dots, \sqsubseteq_k$ over $Q$ is a ""collection of nested total preorders"" if $\sqsubseteq_i$ "refines" $\sqsubseteq_{i-1}$, for $i>0$.
We note that, in that case, the "induced equivalence relation@@preord" $\sim_i$ also "refines" $\sim_{i-1}$.

\begin{definition}[Signature automaton]\label{def-sig:signature-aut}
	\AP Let $d\in \NN$ be a "priority". \AP A ""$d$-signature automaton"" is a "semantically deterministic" "parity automaton" $\A$
	together with a "collection of nested total preorders" $\intro*\leqSig{0}, \leqSig{1}, \leqSig{2},\dots, \leqSig{d}$ over $Q$ such that:\footnotemark{}
	\begin{enumerate}[label=\Roman*),ref=(\Roman*)]
		\item\label{item-sig:refinement-residuals} \bfDescript{Refinements of residual inclusion.} Preorder $\leqSig{0}$ "refines" the preorder $\leqResState_\A$ given by the inclusion of "residuals".	
		
		\item\label{item-sig:faithful-even} \bfDescript{Faithful partitions at even layers.} For  $0\leq x\leq d$,  $x$ even, the equivalence relation $\eqSig{x}$ is a "$[0,x]$-faithful congruence". 	
		
		\item\label{item-sig:odd-levels} \bfDescript{""$({<}x)$-safe separation"" at odd layers.} For  $2\leq x\leq d$,  $x$ even, and $q\eqSig{x-2}q'$:
		\[  q\lSig{x-1} q' \;\; \implies \;\; \text{ there is no path } q\lrp{w:\geq x} q'. \]

		\item\label{item-sig:monotonicity-even} \bfDescript{""Local monotonicity"" of $({\geq}x)$-transitions.} For an even "priority" $x\leq d$, transitions using priorities ${\geq}x$ are "monotone for $\leqSig{x}$" over each $\eqSig{x-1}$ class. That is, for $q\eqSig{x-1}q'$, if $q\leqSig{x}q'$:
		\[  q\re{a:\geq x}p \;\; \implies \;\; q'\re{a:\geq x}p', \; p\eqSig{x-1}p' \;  \tand  p\leqSig{x}p',\; \text{ for all } a\text{-transitions from } q'. \]	
		
	\end{enumerate}	
	\AP We say that $\A$ is a ""signature automaton"" if it is a "$d$-signature automaton", for $d$ the maximal "priority" appearing in $\A$.
\end{definition}

\footnotetext{For notational convenience, we let $\sim_{-2}$ be the complete relation over $\A$ throughout this definition. That is, $q\sim_{-2}p$ for all pairs of states in $Q$.}

We note that even and odd preorders play a completely different role in the previous definition. In fact, the only purpose of odd preorders is to delimit the areas in which the "local monotonicity" property will apply. Item~\ref{item-sig:odd-levels} constrains $\eqSig{x-1}$-classes to be ``sufficiently large''.

\begin{remark}\label{p4-remark:x-minus-1-congruence}
	We note that, by Item~\ref{item-sig:monotonicity-even}, for $x$ even the equivalence relations $\eqSig{x-1}$ is a "congruence for" "${\geq}x$-transitions". However, the restrictions on these odd preorders are much weaker, as we do not impose them to be "faithful".
\end{remark}

%\begin{remark}\label{p4-remark:y-minus-2-congruence}
%	We note that, for $x\leq d$, $x$ even, if $q\eqSig{x-2}q'$ and $q\re{a:\geq x}p$, for all transition $q'\re{a:\geq x}p'$ we have $p\eqSig{x-2}p'$ by "faithfulness".
%\end{remark}

\begin{example}
	Consider the automaton $\A$ from Figure~\ref{fig-warm:3-priorities} from the warm-up. This automaton has $3$ states, $q_1,q_2$, and $q_3$. It can be equipped with the structure of a "signature automaton" as follows:
	\begin{itemize}
		\item Preorder $\leqSig{0}$ is given by the inclusion of residuals: $q_1 \lSig{0} q_2,q_3$, and $q_2\eqSig{0}q_3$.
		\item Preorder $\leqSig{1}$ coincides with preorder $\leqSig{0}$.
		\item Preorder $\leqSig{2}$ is a total order: $q_1 \lSig{2} q_2\lSig{2}q_3$.
	\end{itemize}
\end{example}

Signature automata are not minimal in general, but we conjecture that by merging $\eqSig{d}$-equivalent states we should obtain a minimal automaton (see Section~\ref{subsec-concl:minimisation} for more discussions).

\subsubsection{Full progress consistency.} The existence of a "signature automaton" "recognising" an objective $W$ does not suffice to ensure "positionality" of $W$.
The problem is similar to the one we encountered when studying "open@@reg" objectives in Section~\ref{subsec-warm:open}: there are "open@@reg" objectives whose "residuals" are totally ordered but they are not "positional" (see Example~\ref{ex-warm:not-pos-open}). In that case, we needed to add the property of "progress consistency" to characterise "positionality@@half".
We generalise this notion to "signature automata" with multiple preorders.

\begin{definition}[Full progress consistency]
	\AP We say that a "signature automaton" $\A$ is ""fully progress consistent"" if, for each preorder $\leqSig{x}$, for $x$ even, and every finite word $w\in \SS^*$:
	\[ q\lSig{x} p \; \tand \; q\lrp{w:\geq x} p \;\; \implies \;\; w^\oo \in \Lang{\initialAut{\A}{q}}. \]
\end{definition}

\begin{remark}
	A "fully progress consistent" "signature automaton" is in particular "progress consistent", as the $\leqSig{0}$-preorder refines the preorder coming from the inclusion of "residuals".
\end{remark}

\subsubsection{Structured signature automata from semantic properties of languages}\label{subsubsec:struct-sig-from-semantic}

	To prove the implication \ref{item-th:halfPosFiniteEve} $\implies$ \ref{item-th:signature} from Theorem~\ref{th-reslt:MainCharacterisation-allItems}, we build a "signature automaton" from a "deterministic"  "parity automaton" $\A$ "recognising" $W$ recursively.
	In order to be able to carry out the recursion, we will in fact obtain a "signature automaton" with even stronger properties. This reinforcement of "signature automata" is done by ensuring that the preorders $\leqSig{x}$ come from semantic properties of the automaton, for which the notion of "${<}x$-safe languages" will play a major role. 
	%Some further properties are enforced, so that automata are nicely behaved and the recursive proof can be carried out \sout{not so} smoothly.
	The properties that are imposed are essentially a generalisation of the ones satisfied by the canonical "history-deterministic" coB\"uchi automaton defined by Abu Radi and Kupferman~\cite{AK22MinimizingGFG}.

	We introduce some further notation used in our semantic reinforcement of the definition of a signature automaton.

	\subparagraph{\texorpdfstring{${<}x$}{<x}-safe languages.} \AP Let $\A$ be a (possibly "non-deterministic") "parity automaton". We define the ""$({<}x)$-safe language"" of a state $q$ of $\A$ as:
	
	\[\intro*\safeSigAut{x}{q}{\A} = \{w\in \SS^*\cup \SS^\oo \mid \text{ there exists } q\lrp{w:\geq x}\}.\]
	
	We remark that $\safeSigAut{x}{q}{\A}$ is completely determined by its finite (resp. infinite) words.
	We drop the superscript $\A$ whenever the automaton is clear from the context.
	\AP A path producing no "priority" strictly smaller than $x$ is called ""$({<}x)$-safe@@path"".
	
	Next lemma simply follows from the definition.
	
	\begin{lemma}[Monotonicity of safe languages]\label{lemma-sig:monotonicity-safe-lang}
		Let $\A$ be a "parity automaton" that is "deterministic over" transitions using priorities $\geq x$. Let $q$ and $p$ be two states such that $\safeSig{x}{q}\subseteq \safeSig{x}{p}$, and let $q\lrp{u:\geq x} q'$ be a "${<}x$-safe run" over $u$ from $q$. Then, there is a unique "${<}x$-safe run" over $u$ from $p$, $p\lrp{u:\geq x} p'$, and it leads to a state $p'$ satisfying $\safeSig{x}{q'}\subseteq \safeSig{x}{p'}$. 
		%If $\A$ is moreover "homogeneous", there is a unique run over $u$ from $q$ and $p$.
	\end{lemma}

	\subparagraph{\texorpdfstring{${<}x$}{<x}-safe components.}
	\AP A ""$({<}x)$-safe component"" of $\A$ is a "strongly connected component" of the "subautomaton" obtained by removing all transitions producing a priority $<x$ from $\A$. Note that if $\A$ is in "normal form" and $x>0$, transitions changing of "$({<}x)$-safe component" produce a "priority" $<x$. That is, $q\lrp{u:\geq x} p$ implies that $q$ and $p$ are in the same "$({<}x)$-safe component".

	\begin{remark}
		The partition of $\A$ into "${<}x$-safe components" is a "refinement" of its partition into "${<}y$-safe components", for $y\leq x$.
	\end{remark}

	For the following, we fix a "parity automaton" $\A$ in "normal form" using priorities in $[0,d_{\max}]$. For each $x\in [1,d_{\max}]$, we will totally order the "${<}x$-safe components" of $\A$ in such a way that these orders successively refine each other.
	\AP For $x\in [1,d_{\max}]$ we let $\intro*\safeComp{x}_1,\dots, \safeComp{x}_{k_x}$ be the "${<}x$-safe components" of $\A$.
	For $x=1$, we let $\safeComp{1}_1 <_1 S_2^{<1} <_1\dots <_1 S_{k_1}^{<1}$ be an arbitrary order over the "${<}1$-safe components".
	Assume that an order has already been fixed at level $x-2$. Then, we fix an arbitrary total order for the "${<}x$-safe components" contained in a same "${<}(x-1)$-safe components", which yields a total order for the set of all those safe components, that refines the previous layers.
	From now on, we assume that the enumerations $\safeComp{x}_1,\dots, \safeComp{x}_{k_x}$ correspond to these orders: $\safeComp{x}_i <_x \safeComp{x}_j$ if $i<j$.

	\subparagraph{Structured signature automata.} The "preorders" of the "signature automaton" we plan to build will correspond to the following semantic properties:

	\begin{enumerate}
		\item \bfDescript{Preorder $0$ given by inclusion of residuals.} Preorder  $\leqSig{0}$ corresponds to the inclusion of "residuals": 
		\[q\leqSig{0} p \;\iff \; \Lang{\initialAut{\A}{q}}\subseteq \Lang{\initialAut{\A}{p}}.\]\label{item-struct:preorder-0}
		\vspace{-6mm}
		
		\item  \bfDescript{Odd layers correspond to safe components.} For $x\geq 2$ even, we define $\leqSig{x-1}$ by: 
		\[
			q\leqSig{x-1}p \; \iff \;{q\lSig{x-2}p} \tor [{q\eqSig{x-2}p} \tand q \in \safeComp{x}_i \tand p\in\safeComp{x}_j \text{ with } i\leq j].
		\]
		In particular, $q\eqSig{x-1}p$ if and only if ${q\eqSig{x-2}p}$ and there is a path $q\lrp{w:\geq x} p$. \label{item-struct:odd-orders}
		
		\item \bfDescript{Even preorders given by inclusion of safe languages.} For $x\geq 2$, $x$ even, we define $\leqSig{x}$ by:
		\[
			 q\leqSig{x} p \; \iff  \; {q\lSig{x-1}p} \tor [{q\eqSig{x-1}p} \tand \safeSig{x}{q}\subseteq \safeSig{x}{p}].
		\]\label{item-struct:even-preorders}
		\vspace{-6mm}
	\end{enumerate}

	These preorders already ensure some of the properties required to be a "signature automaton"; mainly, the "local monotonicity" of transitions using large priorities (Item~\ref{item-sig:monotonicity-even}), as well as the "congruence for" ${\geq}x$-transitions at $\eqSig{x}$-classes.
	Note, however, that it is not clear (and will be an important part of our proof) that $\leqSig{x}$ is total for even $x$.

	\AP Given an (even or odd) "priority"  $d\in \NN$, we say that a "parity automaton" $\A$ in "normal form" together with "nested preorders" $\sqsubseteq_0,\sqsubseteq_1,\dots,\sqsubseteq_d$ as above is a ""$d$-structured signature automaton"" if these "preorders" are "total@@preord" and moreover:

	\begin{enumerate}\setcounter{enumi}{3}
		%\item \label{item-struct:uniformity=x} \bfDescript{Uniformity of $x$-transitions.} For an even "priority" $x\leq d$, "$x$-transitions" "are uniform" over $\eqSig{x}$-classes:
%		\[ q\eqSig{x}p \tand q\re{a:x} \;\; \implies \;\; p\re{a:x} \;\;\text{ for all } a\text{-transitions from } p.\]
%		\vspace{-6mm}
		
		\item \label{item-struct:strong-congruence} \bfDescript{Strong congruence of $({\leq}x)$-priorities over even classes.} Let $0\leq x\leq d$, $x$ even. For every $y\leq x$:
		\[q\eqSig{x}q',\; q\re{a:y}p \;\tand \; q'\re{a:z}p' \; \implies \;  z=y \;\tand \; p = p'.\]
		
		\item\label{item-struct:even-classes-connected} \bfDescript{Classes at layer $x$ are $({>}x)$-connected.} For $0\leq x\leq d$ and $q\eqSig{x}q'$, we have:
		\[ q\neq q' \; \implies  \; \text{ there is a path } q\lrp{u:>x}q'.\footnotemark\]
		\vspace{-6mm}
		\footnotetext{We remark that, for $x$ odd, this property is already implied by Item~\ref{item-struct:odd-orders}.}
		\item \bfDescript{Safe centralisation.} Let $2\leq x\leq d$ be an even "priority", and let $q\eqSig{x-2}p$. Then:
		\[q\neqSig{x-1}p  \; \implies \; \safeSig{x}{q}\nsubseteq \safeSig{x}{p}.\]\label{item-struct:safe-centralisation}
		\vspace{-6mm}
	\end{enumerate}

	\AP We say that $\A$ is a ""structured signature automaton"" if it is a "$d$-structured signature automaton", for $d$ the maximal "priority" appearing in $\A$.

%\begin{remark}\label{rmk-sig:str-sig-is-homogeneous}
%	By Item~\ref{item-struct:strong-congruence}, a "$d$-structured signature automaton" is "$({\leq}d)$-homogeneous". 
%\end{remark}

\begin{remark}
	We draw the reader's attention to the fact that in Item~\ref{item-struct:strong-congruence} we do not only require $p \eqSig{x} p'$, but impose $p = p'$. This will be necessary to guarantee that the relations $\eqSig{y}$ for even priorities $y>x$ are also "congruences for" "$x$-transitions".
\end{remark}

\begin{lemma}\label{lemma-sig:structured-is-signature}
	%Let $\A$ be a "parity automaton" that is "homogeneous" and "deterministic over transitions" using priorities $\geq d$.
	A "deterministic" "$d$-structured signature automaton" is a "$d$-signature automaton".
\end{lemma}

\begin{proof}
The fact that $\leqSig{0}$ refines the inclusion of residuals is ensured by Item~\ref{item-struct:preorder-0}.
Also, the "$({<}x)$-safe separation" at odd levels (Item~\ref{item-sig:odd-levels}) is directly implied by the fact that odd layers correspond to safe components (Item~\ref{item-struct:odd-orders}).

We now show by induction on $x$ that for each $x \leq d$, $x$ even, $\eqSig{x}$ is a "$[0,x]$-faithful congruence".
Consider two states $q \eqSig{x} q'$, which rewrites as $q \eqSig{x-1} q'$ and $\safeSig{x}{q} = \safeSig{x}{p}$, and pick a transition $q \re{a:y} p$.
There are two cases.
\begin{itemize}
\item If $y \leq x$, then by Item~\ref{item-struct:strong-congruence}, we have $q' \re{a:y} p$.
\item If $y > x$, then by monotonicity of safe languages (Lemma~\ref{lemma-sig:monotonicity-safe-lang}), we have $q' \re{a:\geq x} p'$ with $\safeSig{x}{p'} = \safeSig{x}{p}$, and by induction, $p' \sim_{x-2} p$.
From Item~\ref{item-struct:safe-centralisation}, it follows that $p' \sim_{x-1} p$ and thus $p' \sim_x p$, as required.
\end{itemize}
We conclude that $\eqSig{x}$ is a "$[0,x]$-faithful congruence".

Finally, the "local monotonicity of $({\geq} x)$-transitions" follows from the fact that even preorders correspond to the inclusion of safe languages (Item~\ref{item-struct:even-preorders}) and the monotonicity of safe languages (Lemma~\ref{lemma-sig:monotonicity-safe-lang}).
%Thus there remains to prove "local monotonicity of $({\geq} x)$-transitions". We proceed by induction on even $x \leq d$.
%For $x=0$, "monotonicity@@local" is ensured by Item~\ref{item-struct:preorder-0} and Lemma~\ref{lemma-warm:total-order-residuals-nec}, so we assume $x \geq 2$.
%
%Let $q\eqSig{x-1}q'$ such that $q\leqSig{x}q'$, and consider a transition $q\re{a:\geq x}p$.
%As noted above, $\eqSig{x-2}$ is a "congruence for" $({>}x-2)$-transitions, so $q'\re{a:\geq x-1}p'$ produces a "priority" $\geq (x-1)$.
%We note that, as $\eqSig{x-2}$ is a "congruence", $p\leqSig{x-2}p'$.
%To show $p\leqSig{x}p'$, we again distinguish two cases:
%	
%	\bfDescript{(1)} If $q\lSig{x-1}q'$, then $q\in \safeComp{x}_i$ and $q\in \safeComp{x}_j$, for $i<j$. As the "priority" produced by the transitions under consideration is ${\geq x}$, we have, by definition of "safe component@@sig" and "normality" of the automaton, $q\in \safeComp{x}_i$ and $q\in \safeComp{x}_j$. Therefore, $p\lSig{x-1}p'$, so $p\leqSig{x}p'$.
%	
%	\bfDescript{(2)} If $q\eqSig{x-1}q'$, by definition $\safeSig{x}{q} \subseteq \safeSig{x}{q'}$. By the argument above, $p\eqSig{x-1}p'$ and  $\safeSig{x}{p} \subseteq \safeSig{x}{p'}$ by monotonicity of safe languages (Lemma~\ref{lemma-sig:monotonicity-safe-lang}).
\end{proof}

\subsection{From positionality to signature automata}\label{subsec:from-HP-to-signature}

This section is devoted to the proof of the implication \ref{item-th:halfPosFiniteEve} $\implies$ \ref{item-th:signature} in Theorem~\ref{th-reslt:MainCharacterisation-allItems}. 
Many of the ideas in this proof have already appeared in the warm-up section. 
However further technical issues stem from the fact that we manipulate general "parity automata".
Details for a number of proofs are relegated to Appendix~\ref{appendix:proofs-necessity}.

\begin{globalHyp*}%{Subsection~\ref{subsec:from-HP-to-signature}}
	In the whole section, $W$ stands for an "objective" that is "positional over" finite, "$\ee$-free" "Eve-games". 
	These hypotheses will not necessarily be recalled in the statements of propositions.
\end{globalHyp*}

\subsubsection{Outline of the induction}

Given a "deterministic" "parity automaton" "recognising" a "positional" objective, we will recursively define the "preorders" and equivalence relations making $\A$ a "structured signature automaton". The base case consists in showing that the  preorder $\leqSig{0}$ given by the inclusion of residuals is total, and ensuring Item~\ref{item-struct:strong-congruence} of the definition for this preorder. %We briefly discuss this at the beginning of next subsection.
For the recursion step, we suppose that we have a "deterministic" "$(x-2)$-structured signature automaton" $\A$ "recognising" $W$, for $x$ even, and we define preorders $\leqSig{x-1}$ and $\leqSig{x}$ over $\A$ as imposed by Items~\ref{item-struct:odd-orders} and~\ref{item-struct:even-preorders}. %, an apply some modification to $\A$ in order to obtain an equivalent "deterministic" "$x$-structured signature automaton". 
Then, we apply a sequence of operations, after which we obtain an "equivalent@@aut" "deterministic" "automaton", that is either "$x$-structured signature", or has strictly less states than $\A$.
In the first case, we continue to define preorders $\leqSig{x+1}$ and $\leqSig{x+2}$; in the second case, we restart the structuration procedure from the beginning, with a strictly smaller automaton. In both cases, we conclude by induction.

%First, we define $\eqSig{x-1}$ following the "$({<}x)$-safe components" of $\A$ and then, we "safe centralise@@sig" the automaton with respect to $\eqSig{x-1}$, ensuring Item~\ref{item-struct:safe-centralisation}. These are the only two properties concerning exclusively relation $\eqSig{x-1}$. The safe centralisation at this level will allow us to show that the next preorder $\leqSig{x}$ is total, and finally obtain the remaining properties. 

We conjecture that we can sequentially obtain all the preorders, without having to restart the construction at each step. However, we have not been able to overcome some technical difficulties preventing us to do so. We refer to the final subsection of Appendix~\ref{appendix:proofs-necessity} for more details.

We give a more detailed account on the specific operations we apply to obtain the different items of the definition of a "structured signature automaton" and their order:

\begin{enumerate}[label=\roman*)]
	\item \bfDescript{Relation $\eqSig{x-1}$ and safe centralisation.} We define $\eqSig{x-1}$, as determined by Item~\ref{item-struct:odd-orders}. Applying a generalisation of the procedure from~\cite{AK22MinimizingGFG}, we "$({<}x)$-safe centralise" $\A$, obtaining an equivalent automaton satisfying Item~\ref{item-struct:safe-centralisation}. The resulting automaton is no longer "deterministic", but it is "history-deterministic" and has a very restricted and controlled amount of "non-determinism".
	
	\item \bfDescript{Total order in safe components.} % and preorder $\leqSig{x}$.}
	 We prove that the states of each "$({<}x)$-safe component" are totally ordered by inclusion of "$({<}x)$-safe languages" (for which we rely on the "safe centralisation@@sig" hypothesis). % -- for which we need to use the safe centralisation hypothesis.
	This shows that the preorder $\leqSig{x}$ given by the inclusion of "safe languages@@sig" (Item~\ref{item-struct:even-preorders}) is total. 
	
%	\item \bfDescript{Monotonicity of $(\geq x)$-transitions.}  We show the "monotonicity of $(\geq x)$-transitions" for preorder $\leqSig{x}$, that is, Item~\ref{item-sig:monotonicity>x}.
	
	\item \bfDescript{Re-determinisation.} We determinise automaton $\A$, while preserving previously obtained properties. For this, the fact that $\leqSig{x}$ is total will be key.
	
	\item \bfDescript{Uniformity of $x$-transitions.} Finally, we show that either $\A$ already satisfies Items~\ref{item-struct:strong-congruence} and~\ref{item-struct:even-classes-connected}, or we can trim the automaton to an "equivalent@@aut" strictly smaller one.
\end{enumerate}

Moreover, we show that all these transformations can be performed in polynomial time.

This establishes that an objective $W$ that is "positional over" finite, "$\ee$-free" "Eve-games" can be "recognised" by a "deterministic" "structured signature" automaton.
At the end of the section, we show that such an automaton must be "fully progress consistent" (Lemma~\ref{lemma-nec:full-prog-cons}).

\subsubsection{Constructing structured signature automata for positional languages}\label{subsubsec:constructing-signature-automaton}

Let $\A = (Q, \SS, q_{\mathsf{init}}, [0,d_{\max}], \DD, \parity)$ be a "deterministic" "parity automaton" "recognising" $W$, and suppose that $W$ is "positional over" finite, "$\ee$-free" "Eve-games". 
We assume that $\A$ is in "normal form".
%We denote $d_{\max}$ the maximal "priority" appearing in $\A$.

%\begin{globalHyp*}{Subsection~\ref{subsubsec:constructing-signature-automaton}}
	In this subsection, we will apply successive transformations to the automaton $\A$, ensuring an increasing list of properties. 
	At the beginning of each paragraph, we clearly state the properties that are assumed.
	We allow ourselves to omit these hypotheses in the statements of propositions inside the paragraphs.
%\end{globalHyp*}

\paragraph*{Base case: \texorpdfstring{Preorder $\leqSig{0}$}{Preorder 0}.}

We define $q\leqSig{0} p$ if $\Lang{\initialAut{\A}{q}}\subseteq \Lang{\initialAut{\A}{p}}$, as imposed by Item~\ref{item-struct:preorder-0}.
In Lemma~\ref{lemma-warm:total-order-residuals-nec}, we showed that "positionality" of $W$ implies that this order is total. However, in our proof we used infinite, and not necessarily "$\ee$-free" games. It is not difficult to modify the proof to adapt to this set of minimal hypotheses, using "$\oo$-regularity" of $W$. We give all details in Appendix~\ref{appendix:proofs-necessity} (Lemma~\ref{lemma-app:total-order-residuals}).
Items~\ref{item-struct:odd-orders},~\ref{item-struct:even-preorders}, and~\ref{item-struct:safe-centralisation} are trivially satisfied.
Therefore, it suffices to show that we can obtain an automaton such that $\eqSig{0}$ is a "strong congruence for transitions" producing "priority" $0$ (Item~\ref{item-struct:strong-congruence}), and that $\eqSig{0}$-equivalent states can be connected by paths producing priority ${>}x$ (Item~\ref{item-struct:even-classes-connected}). For this, we apply exactly the same method presented in Section~\ref{subsec-warm:Buchi}: we obtain a "polished automaton@@buchi" and show that it satisfies the desired properties. This proof will be covered in the recursive step; the case $x=0$ does not present any particularity.

\vskip 1pt

%\begin{globalHyp*}{the rest of Subsection~\ref{subsubsec:constructing-signature-automaton}}
	Moving on to the inductive step, for the rest of the subsection, we let $x$ be an even priority such that $2<x\leq d_{\max}$ and assume that $\A$ is a "deterministic" "$(x-2)$-structured signature automaton".

\paragraph*{Safe centrality and relation \texorpdfstring{$\eqSig{x-1}$}{x-1}}
\AP We say that an automaton with a "preorder" $\leqSig{x-2}$  is ""$({<}x)$-safe centralised"" if $\eqSig{x-2}$-equivalent states that are comparable for the inclusion of "$({<}x)$-safe languages" are in the same "$({<}x)$-safe component".
%
%
%\begin{remark}\label{rmk-nec:safe-component-iff-safe-path}
%	We recall that, if $\A$ is in "normal form", two states $q,p$ are in the same $({<}x)$-safe component if and only if there is a path $q\lrp{w:\geq x} p$.
%\end{remark}
\begin{remark}
	For automata in "normal form" "$({<}x)$-safe centrality"  can be stated as:
	if $q\eqSig{x-2} p$ and there is no "$({<}x)$-safe path" connecting $q$ and $p$, then $\safeSig{x}{q} \nsubseteq \safeSig{x}{p}$.
\end{remark}

\begin{restatable}[{$({<}x)$-safe centralisation}]{lemma}{lemnecSafeCentralisation}\label{lemma-nec:safe-centralisation}\RestateRemark
	There exists a "$(x-2)$-structured signature" automaton $\A'$ "equivalent to" $\A$ which is:
	\begin{itemize}
		\item "deterministic over" transitions with "priority" different from $x-1$,
		\item "homogeneous",
		\item "history-deterministic", and
		\item "$({<}x)$-safe centralised". 
	\end{itemize}  
	Moreover, $\A'$ can be obtained in polynomial time from $\A$ and $\sizeAut{\A'}\leq \sizeAut{\A}$.
\end{restatable}

The proof of this lemma is a refinement of the corresponding result for "coB\"uchi automata" presented in the warm-up (Lemma~\ref{lemma-warm:safe-centralisation}): we saturate $\eqSig{x-2}$-classes of the original automaton $\A$ with "$(x-1)$-transitions@@out", and then remove "redundant@@safeCentr" "$({<}x)$-safe components" recursively until obtaining a "$({<}x)$-safe centralised" automaton.  We include all details in Appendix~\ref{appendix:proofs-necessity} (page~\pageref{lemma-nec:safe-centralisation-pageApp}).

Lemma~\ref{lemma-nec:safe-centralisation} allows us to define $\eqSig{x-1}$ satisfying all required properties:
for $q\eqSig{x-2}p$, we define $q\leqSig{x-1} p$ if and only if $q\in \safeComp{x}_i$ and $p\in \safeComp{x}_j$ with $i\leq j$, where $\safeComp{x}_i$ are the "$({<}x)$-safe components" of $\A$ enumerated following the order described in Section~\ref{subsec:def-signature}.
By definition,  Item~\ref{item-struct:odd-orders} is satisfied, and by "$({<}x)$-safe centralisation" of $\A$, so is Item~\ref{item-struct:safe-centralisation}.

\paragraph*{Preorder \texorpdfstring{$\leqSig{x}$}{x}: Total order given by safe languages}

In all this paragraph we assume that $\A$ is an automaton as obtained in the previous paragraph, that is: it has "nested preorders" defined up to $\leqSig{x-1}$ making it a "$(x-2)$-structured signature automaton" and satisfying Items~\ref{item-struct:odd-orders} and~\ref{item-struct:safe-centralisation} for relation~$\eqSig{x-1}$. Moreover, it is  "history-deterministic", "homogeneous", and the only non-determinism of $\A$ appears in "$(x-1)$-transitions@@out". %(in particular, it has a single initial state).

%\begin{remark}\label{rmk-nec:resolver-in-quotient}
%	As $\eqSig{x-2}$ is a "$[0,x-2]$-faithful congruence", we can consider the quotient automaton $\autLeq{\A}{x-2}{x-1}$. 
%\end{remark}

We define "preorder" $\leqSig{x}$ as imposed by Item~\ref{item-struct:even-preorders}:
\[  q\leqSig{x} p \; \iff q \lSig{x-1} p \tor [ q \eqSig{x-1} p \tand  \; \safeSig{x}{q}\subseteq \safeSig{x}{p} ],\]
and recall that it follows that $q \eqSig x p $ if and only if $q \eqSig{x-1} p$ and there is a "$(<x)$-safe path" from $q$ to $p$.

\begin{remark}
	Using Item~\ref{item-struct:strong-congruence} for priorities $y \leq x-2$ and Lemma~\ref{lemma-sig:monotonicity-safe-lang} for transitions with priority $\geq x$, we get that
	relation $\eqSig{x}$ is a "congruence for" transitions with a priority different from $x-1$.
	Moreover, over each $\eqSig{x-1}$-class, transitions with priority ${\geq}x$ are "monotone for" $\leqSig{x}$.
\end{remark}

Our objective is now to show that $\leqSig{x}$ is "total@@order" over each $\eqRes{x-1}$-class. The proof of this statement uses the same ideas as the corresponding result from the warm-up (Lemma~\ref{lemma-warm:total-order-coBuchi}). In particular, the main technical point resides in proving that, for two states $q_1 \nleqSig{x} q_2$, we can force to produce priority $x-1$ from $q_1$ while remaining "${<}x$-safe@@path" from $q_2$, and then resynchronise both paths in a same $\eqSig{x}$-class. This result, stated in Lemma~\ref{lemma-nec:syncr-separating-runs}, is the analogue to Lemma~\ref{lemma-warm:syncr-separating-runs} from the warm-up. For its proof we strongly rely on the "$({<}x)$-safe centralisation" of $\A$ and the fine control of its non-determinism. 

%\begin{lemma}[Synchronisation of separating runs]\label{lemma-nec:syncr-separating-runs-general}
%	 Let $q,q'$ be states such that  $q' \nleqSig{x} q$ (that is, either $q\lSig{y} q'$, or these states are not in the same $\eqSig{y-1}$-class). For all $p\eqSig{y} q$, there is a word $w\in \SS^*$ such that $q\lrp{w:y-1} p$ and $q'\lrp{w:y}p$.
%\end{lemma}

In the proofs of Lemmas~\ref{lemma-nec:syncr-separating-runs} and~\ref{lemma-nec:total-order-safe} we will reason at the level of $\eqSig{x}$-classes. As we only suppose that $\A$ is "$(x-2)$-structured", we do not have the "uniformity@@trans" of $x$-transitions over $\eqSig{x}$-classes yet. Lemma~\ref{lemma-nec:existence-uniform-words} below provides a weaker version of this uniformity that will suffice for the arguments in the upcoming lemmas. 

%This result will be reused for the proof of Lemma~\ref{lemma-nec:uniformity-x-transitions}.

\AP We say that a word $w$ ""produces priority $y$ uniformly"" in a class $\classSig{x}{q}$ if for every $q'\in \classSig{x}{q}$ all runs from $q'$ are of the form $q'\lrp{w:y}$.
In that case, we write $\classSig{x}{q}\lrp{w:y}$.
\AP We say that such a word "produces priority $x$ uniformly in" $\classSig{x}{q}$ ""leading to@@class"" $\classSig{x}{p}$ if for every $q' \in \classSig{x}{q}$ we have $q' \lrp{w:y} p'$ with $p' \in \classSig{x}{p}$.
In that case, we write $\classSig{x}{q} \lrp{w:y} \classSig{x}{p}$.

%\begin{remark}\label{rmk-nec:uniformity-of-x-1-priorities}
%	Since $\eqSig{x}$-classes are defined by the equality of "${<}x$-safe languages" (and by "$[0,x-2]$-faithfulness"), if $q\lrp{w:x-1}$, then $w$ "produces priority $x-1$ uniformly" in the class~$\classSig{x}{q}$.
%\end{remark}

We note that whenever $\A$ contains a path $q\lrp{w:\geq x} p$, a "run over" $w$ is unique, as $\A$ is "homogeneous" and its restriction to transitions coloured with priorities $\geq x$ is "deterministic". 

The proof of the next lemma combines "normality" of $\A$ with ideas appearing in the proof of Claim~\ref{claim-warm:char-super-words} from the warm-up; all the details can be found in Appendix~\ref{appendix:proofs-necessity} (page~\pageref{lemma-nec:uniformity-x-transitions-pageApp}).

\begin{restatable}[Existence of uniform words]{lemma}{lemNecExistenceUniformWords}\label{lemma-nec:existence-uniform-words}\RestateRemark
	Let $p$ and $q$ be two states from the same "$({<}x)$-safe component".
	There is a word $w \in \SS^*$ "producing priority $x$ uniformly" in $\classSig{x}{q}$ "leading to@@class" $\classSig{x}{p}$.%; that is, for any $q' \in \classSig{x}{q}$, we have $q' \lrp {w:x} p'$ with $p' \in \classSig{x}{p}$.
\end{restatable}

We next state the result that allows us to synchronise runs in a same $\eqSig{x}$-class.
Its proof is analogous to that of Lemma~\ref{lemma-warm:syncr-separating-runs} and can be found in Appendix~\ref{appendix:proofs-necessity}.

We let $\resolv$ be a "sound resolver" for $\A$, and assume that all states can be reached by a "run induced by@@resolv" this "resolver".
We recall that we write $q\lrpAllResolver{\resolv}{w:y}p$ if, for every word $u_0\in \SS^*$ such that the "induced run of $\resolv$ over $u_0$" arrives to $q$, the "induced run of $\resolv$ over" $u_0w$ ends in $p$ and produces $y$ as minimal "priority" in the part of the run corresponding to $w$.
Recall also that we write $\classSig{x}{q} \lrpAllResolver{\resolv}{w:y} \classSig{x}{p}$ if for any $q' \in \classSig{x}{q}$ we have $q' \lrpAllResolver{\resolv}{w:y} p'$ for some $p' \in \classSig{x}{p}$.

\begin{restatable}[Synchronisation of separating runs]{lemma}{lemNecSynchronisation}\label{lemma-nec:syncr-separating-runs}\RestateRemark
	%Let $\A$ be a "$x$-structured signature automaton", and 
	Suppose that $q \eqSig{x-1} q'$ and $q \nleqSig{x} q'$ and let $p\in \classSig{x-1}{q}$.
	There is a word $w\in \SS^+$ such that $\classSig{x}{q} \lrpAllResolver{\resolv}{w:x-1} \classSig{x}{p}$ and $\classSig{x}{q'} \lrpAllResolver{\resolv}{w:x} \classSig{x}{p}$.
\end{restatable}

We can now deduce that $\leqSig{x}$ is "total@@order" over each $\eqRes{x-1}$-class. 

%In general, $\eqSig{x}$ is not a "congruence for" "$(x-1)$-transitions";  and "$x$-transitions" do not "act uniformly" over $\eqSig{x}$-classes.
%However, using Lemmas~\ref{lemma-nec:syncr-separating-runs} and~\ref{lemma-nec:existence-uniform-words} yields words that preserve the congruence and "produce priority $x$ uniformly";
%this will allow us to reason at the level of $\eqSig{x}$-classes.

\begin{lemma}[Total order in $({<}x)$-safe components]\label{lemma-nec:total-order-safe}
	Let $q,q'\in Q$ be two states such that $q\eqRes{x-1}q'$.
	% be two $\eqRes{x-1}$-equivalent states (i.e. they are in the same "$({<}x)$-safe component"). 
	Then, either $q\leqSig{x} q'$ or $q'\leqSig{x} q$.
\end{lemma}
\begin{proof}
	Suppose by contradiction that $\safeSig{x}{q} \nsubseteq \safeSig{x}{q'}$ and $\safeSig{x}{q'} \nsubseteq \safeSig{x}{q}$. 
	Let $p$ be a state in $\classSig{x-1}{q}=\classSig{x-1}{q'}$, and apply Lemma~\ref{lemma-nec:existence-uniform-words} to obtain words $u,u'\in \SS^*$ such that $\classSig{x}{p}\lrp{u:x}\classSig{x}{q}$ and $\classSig{x}{p}\lrp{u':x}\classSig{x}{q'}$.
	
	By Lemma~\ref{lemma-nec:syncr-separating-runs}, there are words $w,w'\in \SS^\oo$ such that:
	
	%\begin{center}
	\begin{tabular}{l l}
		\centering
		$\classSig{x}{q}\lrpAllResolver{\resolv}{w:x}\classSig{x}{p}$, & $\classSig{x}{q}\lrpAllResolver{\resolv}{w':x-1}\classSig{x}{p}$,\\ 
		$\classSig{x}{q'}\lrpAllResolver{\resolv}{w:x-1}\classSig{x}{p}$, & $\classSig{x}{q'}\lrpAllResolver{\resolv}{w':x}\classSig{x}{p}$. 
	\end{tabular} 
	%\end{center}
	
	\vspace{2mm}
	The situation is analogous to the one depicted in Figure~\ref{fig-warm:totalOrder-coBuchi} in the warm-up. 
	We obtain that:
	\begin{itemize}
		\item $(w'u)^\oo\notin \Lang{\initialAut{\A}{q}}$,
		\item $(wu')^\oo\notin \Lang{\initialAut{\A}{q}}$, 
		\item $(wu'w'u)^\oo\in \Lang{\initialAut{\A}{q}}$.
	\end{itemize}

	Let $u_0\in \SS^*$ be a word such that the "run induced by $\resolv$" over $u_0$ ends in $q$ (it exists, as we have supposed that all states are "reachable using $\resolv$").
	It suffices to consider the "game" where there is path labelled $u_0$ leading to  a vertex controlled by "Eve" with two self loops; one of them producing $w'u$ and the other $wu'$. By the previous remarks, she can "win" such game by alternating both loops, but she cannot "win positionally".
\end{proof}

\paragraph*{Re-obtaining determinism}
In this paragraph we assume that $\A$ is a "parity automaton" "recognising" $W$  equipped with "nested total preorders" defined up to $\leqSig{x}$ with all properties obtained until now:
\begin{itemize}
	\item it is a "$(x-2)$-structured signature automaton",
	\item preorder $\leqSig{x-1}$ satisfies properties from Items~\ref{item-struct:odd-orders} and~\ref{item-struct:safe-centralisation} from the definition of a "structured signature automaton", %,~\ref{item-struct:even-classes-connected},
	\item preorder $\leqSig{x}$ satisfies the property from Item~\ref{item-struct:even-preorders} from the definition of a "structured signature automaton", 
	\item it is "deterministic over" transitions with priorities different from $x-1$,
	\item it is "homogeneous", and
	\item it is "history-deterministic".
\end{itemize}  
%We let $\resolv$ be a "sound resolver" for $\A$, and assume that all states are "reachable using this resolver".

We claim that we can obtain a "deterministic" "equivalent@@aut" "automaton" preserving the entire structure of "total preorders". Moreover, in the obtained automaton we guarantee that relation $\eqSig{x}$ satisfies Item~\ref{item-struct:strong-congruence} from the definition of a "structured signature automaton" for priorities $y<x$.
%We provide a high level 

	\begin{restatable}[Re-determinisation]{lemma}{lemNecDeterminisation}\label{lemma-nec:re-determinisation}\RestateRemark
		There is a "deterministic" "parity automaton" $\A'$ "equivalent to" $\A$ with "nested total preorders" defined up to $\leqSig{x}$ satisfying that:
		\begin{itemize}
			\item it is a "$(x-2)$-structured signature automaton",
			\item preorder $\leqSig{x-1}$ satisfies properties from Items~\ref{item-struct:odd-orders} and~\ref{item-struct:safe-centralisation} from the definition of a "structured signature automaton", and %,~\ref{item-struct:even-classes-connected},
			\item preorder $\leqSig{x}$ is a "congruence" and satisfies the property from Item~\ref{item-struct:even-preorders} and, for priorities $y<x$, also that from Item~\ref{item-struct:strong-congruence} .
		\end{itemize}  
		Moreover, automaton $\A'$ can be computed in polynomial time from $\A$ and $\sizeAut{\A'}\leq \sizeAut{\A}$.
	\end{restatable}
	
	The idea of the proof is a direct generalisation of the one presented in the warm-up for "coB\"uchi automata" (page~\pageref{par-warm:determinisation}): we redefine the $(x-1)$-transitions of the automaton in such a way that we ensure that a run that changes of "${<}x$-safe component" infinitely often passes through all these components in a round-robin fashion. The total order $\leqSig{x}$ allows us to identify a maximal state in each component, so we can make a deterministic choice. 
	Formal details can be found in Appendix~\ref{appendix:proofs-necessity} (page~\pageref{lemma-nec:re-determinisation-pageApp}).
	
%	\begin{proof}[Proof idea]
%	We give here a high-level idea of the proof, which generalises the one presented in the warm-up for "coB\"uchi automata" (page~\pageref{par-warm:determinisation}). Formal details can be found in Appendix~\ref{appendix:proofs-necessity}.
%		
%	The restriction of $\A'$ to priorities different from $x-1$ will be exactly the same as that of $\A$ (this part is "deterministic"). We just have to redirect its $(x-1)$-transitions in such a way that the language is preserved.
%	For this, we fix an arbitrary order the "${<}x$-safe components" of $\A$: $S_1, S_2, \dots, S_k$. 
%	If $q\re{a:x-1}$ is 
%	\end{proof}

\paragraph*{Uniformity of \texorpdfstring{$x$}{x}-transitions over \texorpdfstring{$\eqSig{x}$}{x}-classes}

We assume that $\A$ is a "deterministic" "parity automaton" "recognising" $W$ with "nested total preorders" defined up to $\leqSig{x}$ satisfying all conditions stated in Lemma~\ref{lemma-nec:re-determinisation}.
The objective of this paragraph is to obtain the remaining properties of a "$x$-structured signature automaton" (Items~\ref{item-struct:strong-congruence}
 and~\ref{item-struct:even-classes-connected}).

\begin{restatable}[Uniformity of $x$-transitions over $\eqSig{x}$-classes]{lemma}{lemNecUniformity}\label{lemma-nec:uniformity-x-transitions}\RestateRemark
	There is a "deterministic" "parity automaton" $\A'$ "equivalent to"~$\A$ such that either:
	\begin{itemize}
		\item $\A'$ is an "$x$-structured signature automaton" with $\sizeAut{\A'}\leq \sizeAut{\A}$, or
		\item $\sizeAut{\A'}< \sizeAut{\A}$.
	\end{itemize}
	In both cases, such an automaton can be computed in polynomial time from $\A$.
\end{restatable}

The proof of this lemma generalises the techniques introduced in Section~\ref{subsec-warm:Buchi} of the warm-up. Details can be found in Appendix~\ref{appendix:proofs-necessity} (from page~\pageref{lemma-nec:uniformity-x-transitions-pageApp}).
We introduce the "local automaton of a $\eqSig{x}$-class" $\classSig{x}{q}$: the automaton originated by keeping the states of $\classSig{x}{q}$ and paths connecting them producing priorities $\geq x$.
Using "positionality" and ideas analogous to those from Lemma~\ref{lemma-warm:existence-super-letter}, we show that these "local automata" admit a well-defined set of "super letters", that is, there are letters that, if read infinitely often in such a "local automaton", must produce an "accepting word".
These letters are exactly the ones carrying priority $x$ when read from $\classSig{x}{q}$ in the final automaton $\A'$. 

To obtain the uniformity of $x$-transitions, we might need to simplify the automaton: we introduce "$x$-polished automata", the target form of automata that will allow us to obtain "uniformity@@cong" of $x$-transitions.
Using the existence of "super letters", we show that we can "polish@@sig" automaton $\A$ by removing "redundant@@sig" parts of it. 
This operation might break the "normal form" of automaton $\A$,\footnotemark{} but this is not a problem, since in any case it strictly decreases the number of states of the automaton, as desired.

\footnotetext{In fact, we believe that the "polishing operation" does preserve "normality", but we have not been able to prove it.}

This ends the induction step of the proof, establishing existence of a "deterministic" "structured signature" automaton "recognising" $W$.

\subsubsection{Full progress consistency}

We show that a "structured signature automaton" "recognising" a "positional" objective must be "fully progress consistent".
Since we showed how to obtain such a "structured signature automaton" in the previous section, this ends the proof of the implication \ref{item-th:halfPosFiniteEve} $\implies$ \ref{item-th:signature} from Theorem~\ref{th-reslt:MainCharacterisation-allItems}.

\begin{lemma}[Necessity of "full progress consistency"]\label{lemma-nec:full-prog-cons}
	Let $W\subseteq \SS^\oo$ be "positional" over finite, "$\ee$-free" "Eve-games". 
	Any "structured" "signature automaton" "recognising" $W$ is "fully progress consistent".
\end{lemma}

\begin{proof}
	Suppose by contradiction that $\A$ is a "structured signature automaton" for $W$ that is not "fully progress consistent". By definition, for some priority $x$ even, there are $q\lSig{x} p$ and a word $w\in \SS^*$ such that $q\lrp{w:\geq x} p$, but $w^\oo \notin \Lang{\initialAut{\A}{q}}$.
	As $\eqSig{x}$ is a "$[0,x]$-faithful congruence", we can work with $\eqSig{x}$-classes and write $\classSig{x}{q}\lrp{w:\geq x}\classSig{x}{p}$.
	We study first the case $x>0$.	
	By Lemma~\ref{lemma-nec:syncr-separating-runs}, there is a word $u\in \SS^+$ such that $\classSig{x}{q}\lrp{u:x-1}\classSig{x}{q}$ and $\classSig{x}{p}\lrp{u:x}\classSig{x}{q}$.
	 Let $u_0\in \SS^+$ be a word reaching $q$ from the "initial state" of $\A$.\footnote{If $q$ is "initial", we omit $u_0$ and the state $v_0$ of the game from Figure~\ref{fig-nec:game-full-prog-cons} to ensure the use of an "$\ee$-free" game.} We obtain:
	\begin{itemize}
		\item $u_0w^\oo\notin W$,
		\item $u_0u^\oo \notin W$, and
		\item $u_0(wu)^\oo\in W$.
	\end{itemize}
	We consider the game depicted in Figure~\ref{fig-nec:game-full-prog-cons}. "Eve" can "win from $v_0$" by alternating loops labelled $w$ and $u$ when the play arrives to $v_{\mathsf{choice}}$. However, she cannot "win positionally" from $v_0$.	
	\begin{figure}
		\centering
		\includegraphics[scale=1.5]{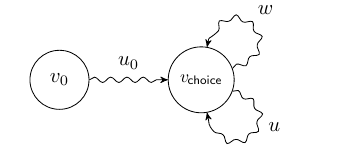}
		\caption{A game $\G$ in which "Eve" cannot "play optimally using positional strategies" if $\A$ is not "fully progress consistent", as in the proof of Lemma~\ref{lemma-nec:full-prog-cons}.}
		\label{fig-nec:game-full-prog-cons}
	\end{figure}
	
	For the case $x=0$, the proof is almost identical to that of Lemma~\ref{lemma-warm:prog-cons-nec}; we just need to ensure that the "game" in Figure~\ref{fig-warm:game-prog-cons}  (page~\pageref{fig-warm:game-prog-cons}) can be supposed finite and "$\ee$-free". Finiteness of the game can be obtained by using ultimately periodic words. To guarantee that we do not include $\ee$-transitions, if $u_0=\ee$, we remove vertex $v_0$ from the game.
\end{proof}

\subsection{From signature automata to positionality through \texorpdfstring{$\ee$}{epsilon}-complete automata}\label{subsec:from-signature-to-HP}

We now complete the equivalence of the statements from Theorem~\ref{th-reslt:MainCharacterisation-allItems}, excluding~\ref{item-th:ee-complete-forall}, by showing the implications \ref{item-th:signature} $\implies$ \ref{item-th:ee-complete} $\implies$ \ref{item-th:ee-complete-HD} $\implies$ \ref{item-th:universalGraph}.
The implication \ref{item-th:universalGraph} $\implies$ \ref{item-th:halfPosArbitrary} follows from Proposition~\ref{prop-prelim:univ-graphs} (taken from~\cite{Ohlmann23Univ})  and \ref{item-th:halfPosArbitrary} $\implies$ \ref{item-th:halfPosFiniteEve} is trivial.

\subsubsection{\texorpdfstring{$\ee$}{epsilon}-complete automata}\label{subsec:eps-complete-def}

We start with our crucial definition.

\begin{definition}
	\AP An ""$\ee$-complete"" automaton $\A$ is a "non-deterministic" "parity automaton" (with "$\ee$-transitions@@aut") with priorities ranging between 0 and $d+1$, where $d$ is even, such that
	\begin{itemize}
		\item the relations $\re{\ee:1}, \re{\ee:3},\dots,\re{\ee,d+1}$ all define total preorders, each refining the previous one;
		\item for each even $x \in \{0,2,\dots,d\}$, the relation $\re{\ee:x}$ is the strict variant of $\re{\ee:x+1}$: for any $q,q'$, it holds that $q \re{\ee:x} q'$ if and only if $q' \re{\ee:x+1} q$ does not hold.
	\end{itemize}
\end{definition}

We say that an automaton $\A$ (which will usually be taken deterministic) is ""$\ee$-completable"", if one may add $\ee$-transitions to $\A$ so as to make it "$\ee$-complete", without augmenting the language.
We say that the resulting (generally non-deterministic) automaton $\A'$ is an ""$\ee$-completion"" of $\A$; note that if $\A$ is deterministic, then $\A'$ is "history-deterministic" (it is even determinisable by pruning).
This provides the implication from~\ref{item-th:ee-complete} to~\ref{item-th:ee-complete-HD} in Theorem~\ref{th-reslt:MainCharacterisation-allItems}
We refer to Figure~\ref{fig:epsilon-completable-automaton} in Section~\ref{sec:char-half-pos} for an example.

\subsubsection{From signature automata to \texorpdfstring{$\ee$}{epsilon}-completable automata}

We now prove the implication \ref{item-th:signature} $\implies$ \ref{item-th:ee-complete} from Theorem~\ref{th-reslt:MainCharacterisation-allItems}, which can be stated as follows.

\begin{lemma}[{From \ref{item-th:signature} to \ref{item-th:ee-complete} in Theorem~\ref{th-reslt:MainCharacterisation-allItems}}]\label{lemma:signature-to-hm}
	Let $\A$ be a "fully progress consistent" "deterministic" "signature automaton".
	Then $\A$ is "$\ee$-completable".
\end{lemma}

We prove Lemma~\ref{lemma:signature-to-hm}. We refer to the discussion at the end of Section~\ref{subsec-warm:full-example} (Figure~\ref{fig-warm:3-priorities}) for an example on the ideas of this proof.
Let $\A$ be a "fully progress consistent" "deterministic" "signature" automaton with "nested preorders" $\leqSig{0},\leqSig{1},\dots, \leqSig{d}$ and let $W=\Lang{\A}$.
Consider the automaton $\A'$ obtained from $\A$ by adding, for all even priorities $x \in [0,d]$, transitions $q \re{\ee:x+1} q'$ whenever $q' \leqSig{x} q$ and $q \re{\ee:x} q'$ whenever $q' \lSig{x} q$.
Note that $\A'$ (potentially) has transitions with priorities up to $d+1$.

\begin{remark}\label{rmk-suf:eps-trans-descend}
	Note that, for $x$ even, $q \re{\ee: \geq x} p$ in $\A'$ entails $p \leqSig{x} q$. 
\end{remark}

Since by definition, $q' \lSig{x} q$ is the negation of $q \leqSig{x} q'$, it follows immediately that $\A'$ is "$\ee$-complete". Moreover, as $\A$ is a "subautomaton" of $\A'$, the inclusion $W\subseteq \Lang{\A'}$ is trivial. The difficulty lies in showing that $\Lang{\A'}\subseteq W$.

\begin{remark}\label{rmk-suf:cycles-in-HM-no-eps}
	If $q\lrp{w:x} q$ is a cycle in $\A'$ producing an even minimal priority, then $w$ is not composed exclusively of "$\ee$-letters".
\end{remark}

\AP For a priority $x$ (even or odd), we say that a transition $q \re{\ee} q'$ in $\A'$ is an ""$x$-jump"" if $q' \lSig{x} q$.
We remark that if $x'\leq x$, an "$x'$-jump" is an "$x$-jump".
%The ""biggest jump"" in a path $q \lrp{} q'$ in $\A'$ is the minimal $x$ such that an "$x$-jump" appears on the path.
We start with a useful technical lemma.

\begin{lemma}\label{lemma:paths_in_hm}
Fix a path $q' \lrp{w':\geq x} p'$ in $\A'$, with $x$ even, and consider a "run" $q \lrp{w\phantom{.}} p$ in~$\A$, where $w$ is obtained from $w'$ by removing "$\ee$-letters" (where $p=q$ if $w$ is empty).
\begin{enumerate}[label=\alph*),ref=\alph*]
	
\item\label{it-lem:no-x-jump} Assume that there is no "$x$-jump" on $q' \lrp{w':\geq x} p'$ and that $q \eqSig{x} q'$.
Then $p \eqSig{x} p'$ and $q \lrp{w:\geq x} p$ in $\A$.
Moreover, if $q' \lrp{w':x} p'$, then $q \lrp{w:x} p$ in $\A$.

\item\label{it-lem:no-x-1-jump} Assume that there is no "$(x-1)$-jump" on $q' \lrp{w':\geq x} p'$, that $q' \eqSig{x-1} q$ and $q' \leqSig{x} q$. Then $p' \eqSig{x-1} p$, $p' \leqSig{x} p$ and $q \lrp{w:\geq x} p$ in $\A$.

\item\label{it-lem:monotonicity-x-1} We have that $p'\leqSig{x-1}q'$.

\end{enumerate}
\end{lemma}

\begin{proof}
In the two first cases we deal with the case of a letter and conclude by induction.

\begin{enumerate}[label=\alph*)]
	\item There are two possibilities, depending on whether  the letter is $\ee$ or not.
	\begin{itemize}
		\item Transition $q' \re{a:\geq x} p'$ with $a \in \SS$.
		Then "$[0,x]$-faithfulness" of $\eqSig{x}$ gives $p \eqSig{x} p'$ and $q \re{a:\geq x} p$.
		Moreover if $q' \re{a:x} p'$, then by "$[0,x]$-faithfulness", $q \re {a:x} p$.
		
		\item Transition $q' \re{\ee:\geq x} p'$.
		Then $p' \leqSig{x} q'$ and since there is no "$x$-jump", $p' \eqSig{x} q'$. Thus $p=q \eqSig{x} q' \eqSig{x} p'$.
		
	\end{itemize}
	\item We distinguish the  two same cases.
	\begin{itemize}
		\item Transition $q' \re{a:\geq x} p'$ with $a \in \SS$.
		Then "local monotonicity of $({\geq}x)$-transitions" in $\A$ yields $q \re{a:\geq x} p$ in $\A$ with $p' \leqSig{x} p$. By Remark~\ref{p4-remark:x-minus-1-congruence}, $p' \eqSig{x-1} p$.
		
		\item Transition $q' \re{\ee:\geq x} p'$.
		This implies $p' \leqSig{x} q'$ and since there is no "$(x-1)$-jump", we have $p' \eqSig{x-1} q'$.
		Thus we conclude that $p' \leqSig{x} q' \leqSig{x} q=p$ and $p=q \eqSig{x-1} q' \eqSig{x-1} p'$.
	\end{itemize}
	
	\item Suppose by contradiction that $p'\gSig{x-1} q'$, and let $q_1'$ be the first state in the run such that $q'\lSig{x-1}q_1'$. We have:
	\[q' \lrp{w_1':\geq x} q_2' \re{a:\geq x} q_1' \lrp{w_2'} p',\] 
	with $q_2'\lSig{x-1}q_1'$.  As in particular $q_2'\lSig{x}q_1'$, $a\neq \ee$ (Remark~\ref{rmk-suf:eps-trans-descend}). However, this contradicts Item~\ref{item-sig:odd-levels} from the definition of "signature automaton".\qedhere
\end{enumerate}
\end{proof}

We now state the key result for deriving Lemma~\ref{lemma:signature-to-hm}.

\begin{lemma}\label{lemma-suf:even-cycles-in-hm}
Consider a cycle $q' \lrp{w':x} q'$ in $\A'$ with $x$ even, and let $w$ be obtained from~$w'$ by removing "$\ee$-letters".
Then, $w^\omega$ is accepted from $q'$ in $\A$.
\end{lemma}

\begin{proof}
	We note that by Remark~\ref{rmk-suf:cycles-in-HM-no-eps}, $w$ is not empty.
Let $y$ be minimal such that an "$y$-jump" appears on the path $q' \lrp{w':x} q'$ (and $y = d+1$ if no "$y$-jump" occurs).
\begin{itemize}
\item If $y \geq x$.
Then there is no $\re{\ee:x}$ transition on the path $q' \lrp{w':x} q'$ (otherwise it would produce an "$x$-jump").
Thus Lemma~\ref{lemma:paths_in_hm}.\ref{it-lem:no-x-jump} proves $q' \lrp{w:x} q_1$ in $\A$ with $q_1 \eqSig{x} q'$.
Then, since the $\eqSig{x}$-class is preserved, successive applications of Lemma~\ref{lemma:paths_in_hm}.\ref{it-lem:no-x-jump} give $q'\lrp{w:x} q_1 \lrp{w:x} q_2 \lrp{w:x} q_3 \lrp{w:x} \dots$ in $\A$, and thus $w^\oo$ is "accepted@@aut" from $q'$ in $\A$. 

\item If $y < x$ and $y$ odd. We show that this case cannot happen. Let $p'_1 \re {\ee} p'_2$ denote the first "$y$-jump" on the path $q' \lrp{w} q'$ in $\A'$, that is, we have
\[
q' \lrp{w'_1:> y} p'_1 \re{\ee:> y} p'_2 \lrp{w'_2:> y} q' \tin \A', \quad p_2'\lSig{y}p_1'.
\]

By Lemma~\ref{lemma:paths_in_hm}.\ref{it-lem:monotonicity-x-1}, we have that $p_1'\leqSig{y} q'$, so $p_2'\lSig{y}q'$. The existence of a path $p'_2 \lrp{w'_2:\geq y} q'$ contradicts Lemma~\ref{lemma:paths_in_hm}.\ref{it-lem:monotonicity-x-1}.

\item If $y < x$ and $y$ even.
Let $p'_1 \re {\ee} p'_2$ denote the last "$y$-jump" on the path $q' \lrp{w} q'$ in $\A'$, that is, we have
\[
	q' \lrp{w'_1:\geq y} p'_1 \re{\ee:\geq y} p'_2 \lrp{w'_2:\geq y} q' \tin \A', \quad p_2'\lSig{y}p_1',
\]
and there is no "$y$-jump" on $p'_2 \lrp{w_2} q'$. We let $w_1,w_2$ be obtained, respectively, from $w'_1$ and $w'_2$ by removing $\ee$'s.
By Lemma~\ref{lemma:paths_in_hm}.\ref{it-lem:no-x-jump}, we get that $p'_2 \lrp{w_2:\geq y} {q}$ in $\A$ for some $q \sim_{y} q'.$
As there is no "$(y-1)$-jump" in the path, by  Lemma~\ref{lemma:paths_in_hm}.\ref{it-lem:no-x-1-jump}, we get that $q \lrp{w_1:\geq y} p_1$ in $\A$ for $p_2' \lSig{y} p_1' \leqSig{y} p_1$. See Figure~\ref{fig-suf:eps-complete-lemma} for an illustration of the situation.

\begin{figure}
	\centering
	\includegraphics[scale=1.5]{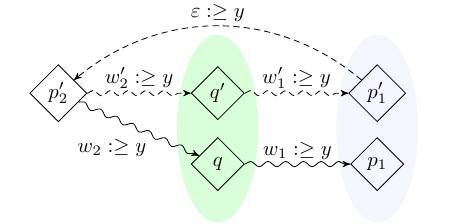}
	\caption{Situation in the third case of Lemma~\ref{lemma-suf:even-cycles-in-hm}. Dashed lines represent paths in $\A'$ and solid lines those in $\A$. States that are $\eqSig{y}$-equivalent are encircled together.}
	\label{fig-suf:eps-complete-lemma}
\end{figure}

All in all, we have obtained a path 
\[p_2' \lrp{w_2:\geq y} q \lrp{w_1: \geq y} p_1 \gSig{y} p'_2 \quad \text{ in } \A.\] 
Therefore, "full progress consistency" yields $(w_2w_1)^\omega \in \Lang{\initialAut{\A}{p_2'}}$.
As $\Lang{\initialAut{\A}{p_2'}} \subseteq \Lang{\initialAut{\A}{q'}}= \Lang{\initialAut{\A}{q}}$, we conclude that $w^\omega = (w_1w_2)^\oo \in \Lang{\initialAut{\A}{q}}$. \qedhere
\end{itemize}
\end{proof}

We are now ready conclude the proof of Lemma~\ref{lemma:signature-to-hm}.
\begin{proof}[Proof of Lemma~\ref{lemma:signature-to-hm}]
	As mentioned above, the inclusion $	W\subseteq \Lang{\A'}$ is trivial, as $\A$ is a subautomaton of $\A'$. This shows that, if the converse inclusion holds, $\A'$ is determinisable by pruning and therefore it is also "history-deterministic".
	
	We show $\Lang{\A'}\subseteq W$.
	Take an accepting run in $\A'$ over $w'\in \SS^\oo$ and decompose it as:
\[
	q_0 \lrp{w'_0} q' \lrp{w'_1:x} q' \lrp{w'_2:x} q' \lrp{w'_3:x} \cdots,
\]
where $x$ is even.
For each $i$, let $w_i$ be obtained from $w'_i$ by removing $\ee$'s (which is non-empty by Remark~\ref{rmk-suf:cycles-in-HM-no-eps}), and consider the corresponding run in $\A$:

\[
q_0 \lrp{w_0} q_1 \lrp{w_1} q_2 \lrp{w_2} q_3 \lrp{w_3} \cdots,
\]

It follows by induction that $q' \leqSig{0} q_i$, so, as order $\leqSig{0}$ "refines" the order of residuals (property~\ref{item-sig:refinement-residuals} of a "signature automaton"), words that are accepted from $q'$ in $\A$ are also "accepted@@aut" from  $q_i$. 
By Lemma~\ref{lemma-suf:even-cycles-in-hm}, it holds that for each pair of indices $j\leq j'$ we have $(w_jw_{j+1}\dots w_{j'})^\omega \in (q')^{-1}W$, so these words are also accepted from $q_i$, for all $i$.

Let $i_1,i_2,\dots$ be a sequence of indices such that $q_{i_j} = q_{i_{j+1}}$ for all $j$, and let $\tilde{q}=q_{i_1}$ be such recurring state.
Each word $w_{i_j}\dots w_{i_{j}-1}$ forms a cycle over $\tilde{q}$, that, by the previous remark, must be accepting, so the minimal priority produced on it is even. Therefore, we have found a decomposition of the run over $w$ in $\A$ of the form

\[
q_0 \lrp{w_0w_1\dots} \tilde{q} \lrp{w_{i_1}\dots w_{i_2 -1} : x_1} \tilde{q} \lrp{w_{i_2}\dots w_{i_3 -1} : x_2}  \cdots,
\]

with all $x_i$ even. We conclude that  $w_0w_1 \dots \in W$.
\end{proof}

\subsubsection{Universal graphs from \texorpdfstring{$\ee$}{epsilon}-complete automata}
As noted at the end of Subsection~\ref{subsec:eps-complete-def}, the implication~\ref{item-th:ee-complete} $\implies$ \ref{item-th:ee-complete-HD} in Theorem~\ref{th-reslt:MainCharacterisation-allItems} is immediate.
We now move on to implication \ref{item-th:ee-complete-HD} $\implies$ \ref{item-th:universalGraph}, stated as follows:

\begin{proposition}[{\ref{item-th:ee-complete-HD} $\implies$ \ref{item-th:universalGraph} in Theorem~\ref{th-reslt:MainCharacterisation-allItems}}]\label{prop:existence-univ-graph}
If there exists an "$\ee$-complete" "history-deterministic" automaton "recognising" $W$, then there exists a "well-ordered" "monotone@@univ" "$(\kappa,W)$-universal" graph for each cardinal $\kappa$.
\end{proposition}

For the rest of the section, we let $\A$ be an "$\ee$-complete" "history-deterministic" automaton "recognising" $W$, and we let $d$ be even such that $\A$ has priorities up to $d+1$, as in the above section.

\paragraph*{Closure of an \texorpdfstring{$\ee$}{epsilon}-complete automaton}

We define the order $\intro*\cleq$ over "priorities" in $[0,d+1]$ that sets $y\cleq x$ if $x$ is ``preferable'' to $y$, that is:  $1 \cleq 3 \cleq \dots \cleq d+1 \cleq d \cleq \dots \cleq 2 \cleq 0$. 
\begin{remark}\label{rmk-suf:order-priorities}
	For any pair of infinite words $w, w' \in [0,d+1]^\oo$ satisfying that for all $i$  $w_i \cleq w'_i$, it holds that:
	\[w\in \parity_{[0,d+1]} \implies  w'\in \parity_{[0,d+1]}.
	%\limsup_i(w_i) \text{ is even } \implies \limsup_i(w'_i) \text{ is even}.
	\]
\end{remark}

\AP We say that an automaton $\A$ is ""priority-closed"" if:
\begin{itemize}
	\item for any states $q,q'$, priorities $y' \cleq y$, and $a \in \SS \cup \{\ee\}$
	\[
		q \re{a:y} q' \implies q \re{a:y'} q'
	\]
	\item for any states $p,p',q,q'$ and $a \in \SS \cup \{\ee\}$,
\[
	p \re {\ee:y_1} q \re{a:y_2} q' \re{\ee:y_3} p' \implies p \re {a:\min_{\cleq}(y_1,y_2,y_3)} p'.
\] 
\end{itemize}

It is easy to turn any automaton into a "priority-closed@@aut" one.

\begin{lemma}\label{lemma:epsilon-closure}
Let $\A$ be an automaton recognising $W$.
There is an automaton $\A'$ recognising $W$ which is "priority-closed@@aut".
Moreover, if $\A$ is "history-deterministic" and "$\ee$-complete", then so is $\A'$.
\end{lemma}

\begin{proof}
Let $\A'$ be obtained by adding to $\A$ all transitions of the form $q \re {a:y'} q'$, when $q \re{a:y} q'$ is a transition in $\A$ and $y' \cleq y$, and all transitions of the form $p \re {a:\min_{\cleq}(y_1,y_2,y_3)} p'$, whenever a path $p \re {\ee:y_1} q \re{a:y_2} q' \re{\ee:y_3} p'$ appears in $\A$.
Clearly, $\A'$ is "priority-closed@@aut", $\Lang{\A} \subseteq \Lang{\A'}$ and if $\A$ is  "$\ee$-complete", then so is $\A'$. 
The fact that this operation preserves "history-determinism" is also clear, once the equality of languages is obtained.
To prove that $\Lang{\A'} \subseteq \Lang{\A}$, take an accepting "run over" $w\in\SS^\oo$ in $\A'$. We build a run over $w$ in $\A$ by replacing any newly added transition $q \re {a:y'} q'$ by $q \re{a:y} q' $,  and $p \re {a:\min_{\cleq}(y_1,y_2,y_3)} p'$ by $p \re {\ee:y_1} q \re{a:y_2} q' \re{\ee:y_3} p'$, respectively. By Remark~\ref{rmk-suf:order-priorities}, the obtained run is "accepting@@run" in $\A$.
\end{proof}

In a "priority-closed@@aut" automaton, for each "priority" $y$, transitions $\re{\ee:y}$ define a transitive relation.
If the automaton is moreover "$\ee$-complete", then for each even priority $x$, transitions of the form $\re{\ee:x+1}$ define "total preorders". 

\AP We define $q \intro*\geqHM{x} q'$ if $q \re{\ee:x+1} q'$. 
%\[  
%	\text{we denote } q \geqHM{x} q' \tif q \re{\ee:x+1} q'.
%\]
Note that, since $\A$ is "priority-closed@@aut", these preorders are "nested@@preorder": $q \geqHM{x+2} q'$ implies $q \geqHM{x} q'$.
Moreover, since $\A$ is "$\ee$-complete", for any even $x$:
\[
	q \gHM{x} q' \implies q \re{\ee:x} q'.
\]
%%%%%{Why not include a small figure here (that looks a bit like a tree, with some transitions)}
Finally, observe that, by "priority-closure", states $q,q'$ that are $\leqHM{d}$-equivalent have exactly the same incoming and outgoing transitions, and can thus be merged without altering the language (this transformation preserves "history-determinism").
Therefore, we may assume that $\leqHM{d}$ is antisymmetric and thus defines a total order on $Q$.

\AP We write $\intro*\classHM{x}{q}$ to denote the equivalence class of $q$ associated to the preorder $\leqHM{x}$. That is, $\classHM{x}{q}$ contains the states $q'$ such that $q\leqHM x q'$ and $q'\leqHM x q$.

\paragraph*{Definition of the graph}

For the remainder of the section, we fix a cardinal $\kappa$.
Let us first recall the construction of the "$(\kappa,\parity)$-universal" graph $\UPar$ for the "parity objective" over $\{0,\dots,d+1\}$ (see Example~\ref{ex-warm:graph-parity} for a proof of universality). Its vertices are of the form $(\lambda_1,\lambda_3,\dots,\lambda_{d+1}) \in \kappa^{d/2+1}$, ordered lexicographically, and its edges are given by
\[
(\lambda_1,\dots,\lambda_{d+1}) \re x (\lambda'_1,\dots,\lambda'_{d+1}) \iff \begin{cases}
	(\lambda'_1,\dots,\lambda'_{x-1})\leq (\lambda_1,\dots,\lambda_{x-1}),  &\tif x \text{ is even}, \\
	(\lambda'_1,\dots,\lambda'_{x})<(\lambda_1,\dots,\lambda_{x}), &\tow. 
\end{cases}
\]

Fix a "priority-closed@@aut" "$\ee$-complete" and "history-deterministic" automaton $\A$ with states $Q$  such that $\leqHM{d}$ defines a total order on $Q$.

\AP We define a "$\SS$-graph" $\intro*\UAut$ as follows.
Vertices of $\UAut$ are the tuples $v = (q,\lambda_{1},\lambda_3,\dots,\lambda_{d+1}) \in Q \times \kappa^{d/2+1}$.
\AP We associate to each such vertex the extended tuple 
\[\intro*\ext(v) = (\classHM{0}{q},\lambda_1,\classHM{2}{q},\lambda_3,\dots,\classHM{d-1}{q},\lambda_{d+1}).\]
 We use it to define the total order: $v\leq v'$ if $\ext(v)$ is smaller than $\ext(v')$ for the lexicographic order.
This is therefore a well-order.
%They are well-ordered lexicographically as tuples of the form $([q]_0,\lambda_1,[q]_2,\lambda_3,\dots,[q]_d,\lambda_{d})$.
Edges in $\UAut$ are given by:
\[(q,\lambda_{1},\dots,\lambda_{d+1}) \re {a} (q',\lambda'_{1},\dots,\lambda'_{d+1}) \; \iff \; \exists y \,
\left\{\begin{array}{l}
	 q \re {a:y} q' \tin \A,\;\; \tand \\[2mm]
	(\lambda_1,\dots,\lambda_{d+1}) \re {y} (\lambda'_1,\dots,\lambda'_{d+1}) \tin \UPar.
\end{array}\right.
\]

%\[	\begin{aligned}
%	(q,\lambda_{1},\dots,\lambda_{d}) \re {a} (q',\lambda_{1},\dots,\lambda_{d}) \iff &\exists y, q \re {a:y} q' \tin \A \tand \\ &(\lambda_1,\dots,\lambda_{d}) \re {y} (\lambda'_1,\dots,\lambda_{d}) \tin \UPar.
%	\end{aligned}
%\]
Paths in $\UAut$ are well-behaved with respect to $W$, as stated below.

\begin{lemma}
Let $(q,\lambda_1,\dots,\lambda_{d+1}) \lrp{w\phantom{..}}$ be an infinite path in $\UAut$. Then, $w \in q^{-1}W$.
\end{lemma}

\begin{proof}
Consider a path
\[
	(q^0,\lambda_1^0,\dots,\lambda_{d+1}^0) \re {w_0} (q^1,\lambda_1^1,\dots,\lambda_{d+1}^1) \re {w_1} \dots \tin \UAut.
\]
By definition, there are priorities $y_0,y_1,\dots$ such that
\[
	q^0 \re {w_0:y_0} q^1 \re{w_1:y_1} \dots \tin \A, \; \tand \;\; (\lambda_1^0,\dots,\lambda_{d+1}^0) \re {y_0} (\lambda_1^1,\dots,\lambda_{d+1}^1) \re {y_1} \dots \tin \UPar.
\]
Since vertices in $\UPar$ satisfy the parity objective, $\liminf(y_0y_1\dots)$ is even, thus the above run in $\A$ is accepting, and so $w_0 w_1 \dots \in (q^0)^{-1}W$.
\end{proof}

\paragraph*{Monotonicity}

"Monotonicity@@@univ" of $\UAut$ follows from the structural assumptions over $\A$.

\begin{lemma}
The graph $\UAut$ is "monotone@@univ".
\end{lemma}

\begin{proof}
Let
\[
	(q,\lambda_1,\dots,\lambda_{d+1}) \re {a} (q',\lambda'_1,\dots,\lambda'_{d+1}) > (q'',\lambda''_1,\dots,\lambda''_{d+1}) \tin \UAut.
\]
We aim to prove that $(q,\lambda_1,\dots,\lambda_{d+1}) \re {a} (q'',\lambda''_1,\dots,\lambda''_{d+1}) \tin \UAut$.
By definition of the transitions of $\UAut$, there is a priority $y$ such that $q \re{a:y} q' \tin \A$ and $(\lambda_0,\dots,\lambda_{d+1}) \re{y} (\lambda'_0, \dots, \lambda'_{d+1}) \tin \UPar$.
We remark that, by definition of the order in $\UAut$, we have that $(\classHM{0}{q''},\lambda''_1,\dots,\lambda''_{y-1},\classHM{y}{q''}) \leq (\classHM{0}{q'},\lambda'_1,\dots,\lambda'_{y-1},\classHM{y}{q'})$ (for $y$ even, similar if $y$ odd).
We distinguish four cases:
\begin{itemize}
\item If $y$ is even and $(\classHM{0}{q'},\lambda'_1,\dots,\lambda'_{y-1},\classHM{y}{q'}) = (\classHM{0}{q''},\lambda''_1,\dots,\lambda''_{y-1},\classHM{y}{q''})$.
Then in $\A$, $q \re {a:y} q' \re{\ee:y+1} q''$ thus $q \re{a:y} q''$, and in $\UAut$, $(\lambda_1,\dots,\lambda_{y-1}) \geq (\lambda'_1,\dots,\lambda'_{y-1}) = (\lambda''_1,\dots,\lambda''_{y-1})$, which concludes.

\item If $y$ is odd and $(\classHM{0}{q'},\lambda'_1,\dots,\classHM{y-1}{q'},\lambda'_y) =(\classHM{0}{q''},\lambda''_1,\dots,\classHM{y-1}{q''},\lambda''_y)$.
Then, in $\A$, $q \re{a:y} q' \re{\ee:y} q''$ thus $q \re{a:y} q''$, and in $\UAut$, $(\lambda_1,\dots,\lambda_{y})>(\lambda'_1,\dots,\lambda'_y) = (\lambda''_1,\dots,\lambda''_y)$ which concludes.

\item If for some even $x\leq y$ it holds that $(\classHM{0}{q'},\lambda'_1, \classHM{2}{q'}, \dots,\lambda'_{x-1})=(\classHM{0}{q''},\lambda''_1, \classHM{2}{q''}, \dots,\lambda''_{x-1})$ and $\classHM{x}{q''} < \classHM{x}{q'}$.
Then in $\A$, $q \re{a:y} q' \re{\ee:x} q''$ thus $q \re{a:x} q''$ and in $\UAut$, $(\lambda_1,\dots,\lambda_{x-1}) \geq (\lambda'_1,\dots,\lambda'_{x-1}) = (\lambda''_1,\dots,\lambda''_{x-1})$ thus $(\lambda_1,\dots,\lambda_{d+1}) \re x (\lambda''_1,\dots,\lambda''_{d+1})$ which concludes.

\item If for some even $x< y$ it holds that $(\classHM{0}{q'},\lambda'_1,\dots,\classHM{x}{q'})=(\classHM{0}{q''},\lambda''_1,\dots,\classHM{x}{q''})$ and $\lambda'_{x+1} > \lambda''_{x+1}$.
Then, in $\A$, $q \re{a:y} q' \re{\ee:x+1} q''$ thus $q \re{a:x+1} q''$ and in $\UAut$, $(\lambda_1,\dots,\lambda_{x+1}) \geq (\lambda'_1,\dots,\lambda'_{x+1}) > (\lambda''_1,\dots,\lambda''_{x+1})$ thus $(\lambda_1,\dots,\lambda_{d+1}) \re{x+1} (\lambda''_1,\dots,\lambda''_{d+1})$ which concludes.
\end{itemize}

The other implication $v > v' \re a v'' \implies v \re a v''$ in $\UAut$ follows exactly the same lines.
\end{proof}

\paragraph*{Universality of \texorpdfstring{$\UAut$}{U A}}

To prove Proposition~\ref{prop:existence-univ-graph}, there remains to establish "universality" of $\UAut$, which follows easily from "history-determinism" of $\A$ and "universality" of $\UPar$.

\begin{lemma}\label{lemma-suf:universality}
The graph $\UTop{\UAut}$ is "$(\kappa,W)$-universal".
\end{lemma}

\begin{proof}
We show "universality for trees" of $\UAut$ and conclude by Lemma~\ref{lemma-prelim:univ-for-trees}.
Let $T$ be a "$\SS$-tree" of size $<\kappa$ that "satisfies@@univTree" $W$.
Let $\resolv$ be a "sound resolver" for $\A$. We define in a top-down fashion a labelling $\resolv_T :T \to Q$ such that, if $w_1w_2\dots w_k$ is the labelling of the path from the "root@@univ" to a vertex $t$, then $\resolv_T(t)$ is the target state of the "run induced by@@aut" $\resolv$ in~$\A$.
In particular, $t \re{a} t' \tin T$ implies that $\resolv_T(t) \re {a:x} \resolv_T(t') \tin \A$ for some priority $x$, and, on each infinite branch $t_0 \re{a_0} t_1 \re{a_1} \dots$, the "run" $\resolv_T(t_0) \re{a_0:x_0} \resolv_T(t_1) \re{a_1:x_1} \dots$ is "accepting@@run" in~$\A$. %, so the minimal priority appearing infinitely often  is even.
Stated differently, the "$[0,d+1]$-tree" $\intro*\TPar$ obtained from $T$ by replacing each edge $t \re a t'$ with the corresponding edge $t \re x t'$ such that $\resolv_T(t) \re{a:x} \resolv_T(t')$, "satisfies@@univ" the "parity objective".

By "$(\kappa,\parity)$-universality" of $\UPar$,  there exists a "morphism@@univ" $\intro*\phiPar : \TPar \to \UPar$.
As $\TPar$ has the same set of vertices than $T$, $\phiPar$ defines a mapping from $T$ to $\UPar$.
We consider the product mapping $\phi=\resolv_T \times \phiPar\colon T \to \UAut$ that sends $t\mapsto (\resolv_T(t),\phiPar(t))$.
It defines a "morphism@@univ", as for any edge $t \re a t'$ in $T$ it holds that, for some $x$, $\resolv_T(t) \re {a:x} \resolv_T(t') \tin \A$ and $\phiPar(t) \re{x} \phiPar(t') \tin \UAut$.
\end{proof}

This completes the proofs of the  equivalence of the statements from Theorem~\ref{th-reslt:MainCharacterisation-allItems} (excluding~\ref{item-th:ee-complete-forall}), providing a characterisation of "positionality" for $\oo$-regular languages.

\section{Two decision procedures}\label{sec:decision}

We now establish decidability of "positionality" of "$\oo$-regular languages" in polynomial time, stated as Theorem~\ref{th-reslt:decid-poly} and prove the equivalence of Item~\ref{item-th:ee-complete-forall} with the other items in Theorem~\ref{th-reslt:MainCharacterisation-allItems}.
We propose two decision procedures.

The first one follows our proof of Theorem~\ref{th-reslt:MainCharacterisation-allItems} and attempts to build a "deterministic" "signature" automaton from a given deterministic parity automaton.
We believe that the techniques used in such a procedure may prove interesting also in other contexts (see conclusion in Section~\ref{subsec-concl:minimisation}).

The second procedure is simpler to describe: we use the fact that any (non-deterministic) automaton recognising a "positional" language is "$\ee$-completable" (implication~\ref{item-th:halfPosArbitrary} $\implies$ \ref{item-th:ee-complete-forall}). 
However, the proof itself relies on Theorem~\ref{th-reslt:union-PI} about the closure under union (which only relies on the equivalence \ref{item-th:halfPosFiniteEve} $\iff$ \ref{item-th:halfPosArbitrary} from Theorem~\ref{th-reslt:MainCharacterisation-allItems}).

\subsection{Procedure 1: Recursive decomposition}

The first decision procedure we present consists in, given a "deterministic" "parity automaton" $\A$, applying the construction from Section~\ref{subsec:from-HP-to-signature} to decide whether $W=\L(\A)$ is "positional".
The general idea is simply to go through that construction, either ending up with a failure indicating that $W$ is not positional, or with a "deterministic" "structured signature automaton".
If such an automaton is successfully obtained, it suffices to check "full progress consistency", which can be done in polynomial time, as explained below.

\paragraph*{Complexity of building a signature automaton}

Most proofs have been already given in Section~\ref{subsec:from-HP-to-signature}; we also require the following easy lemma.

\begin{lemma}\label{lemma-nec:safe-languages-polytime}
	Let $\A$ be a "parity automaton" and $x$ a priority. Assume that $\A$ is "deterministic over" ${\geq}x$-transitions. Given two states $q, p$ in $\A$, we can decide in polynomial time whether $\safeSig{x}{q}\subseteq \safeSig{x}{p}$.
\end{lemma}

We now detail the polynomial-time procedure.
Let $\A$ be a "deterministic" "parity automaton" "recognising" $W$.

\begin{itemize}
	\item First, we check for each pair of states $q,p$ whether $\Lang{\initialAut{\A}{q}}\subseteq \Lang{\initialAut{\A}{p}}$, or $\Lang{\initialAut{\A}{p}}\subseteq \Lang{\initialAut{\A}{q}}$. If for some pair of states these languages are incomparable, then "residuals" of $W$ are not totally ordered, and we can conclude that $W$ is not "positional". 
\end{itemize}
Suppose that we have defined total preorders up to $\leqSig{x-2}$ making $\A$ a "$(x-2)$-structured signature" automaton.
\begin{itemize}
	\item We "${<}x$-safe centralise" $\A$, which can be done in polynomial time by Lemma~\ref{lemma-nec:safe-centralisation}.
	The obtained automaton is "deterministic over" ${\geq}x$-transitions.
	
	\item We compute the "${<}x$-safe components" of $\A$, which can be done by doing a decomposition in "SCCs" of $\autGeq{\A}{x}$.
	We check whether, for each "${<}x$-safe component" $S$ and state $q$, the states in $S\cap \classSig{x-2}{q}$ are "totally preordered" by inclusion of "${<}x$-safe languages", which can be done in polynomial time by Lemma~\ref{lemma-nec:safe-languages-polytime}. 
	If this is not the case, we conclude that $W$ is not "positional".
	
	\item We remove the non-determinism from $\A$ -- in polynomial time and without increasing the number of states -- by applying Lemma~\ref{lemma-nec:re-determinisation}.
	
	\item We compute (in polynomial time) the automaton $\A'$ given by Lemma~\ref{lemma-nec:uniformity-x-transitions} (see last subsection of Appendix~\ref{appendix:proofs-necessity} for details). 
	We check whether $\Lang{\A'}=W$, which can be done in polynomial time (testing equivalence of "deterministic" "parity automata"~\cite{ClarkeDK93Unified}).
	If this is not the case, we conclude that $W$ is not "positional".
\end{itemize}

After these operations, if we have not yet found that $W$ is not "positional", we obtain an "equivalent@@aut" "deterministic" automaton $\A'$ that is either "$x$-structured signature", or strictly smaller than $\A$ (as given by Lemma~\ref{lemma-nec:uniformity-x-transitions}).\footnote{As mentioned before, we conjecture that the obtained automaton $\A'$ is always "$x$-structured signature".} In the former case, we continue defining preorders $\leqSig{x+1}$ and $\leqSig{x+2}$; in the latter, we restart from the beginning.
In total, we repeat at most $d\cdot|Q|$ times a sequence of operations that take polynomial time.

\paragraph*{Checking full progress consistency}

Assume that we have a "deterministic" "structured signature automaton" "recognising" $\A$. We cannot yet conclude that $W$ is "positional", as we do not know whether $\A$ is "fully progress consistent", however, by Lemma~\ref{lemma-nec:full-prog-cons}, if $W$ is "positional" this must be the case. By Theorem~\ref{th-reslt:MainCharacterisation-allItems}, this condition is also sufficient.
We show now that we can check "full progress consistency" of $\A$ in polynomial time, finishing the proof of Theorem~\ref{th-reslt:decid-poly}. For this, we generalise the method from~\cite[Lemma~25]{BCRV22HalfPosBuchi}.

\begin{lemma}\label{lemma-half:check-ProgConsPoly}
	Let $\A$ be a deterministic "structured signature automaton". We can decide in polynomial time whether $\A$ is "fully progress consistent".
\end{lemma}

The proof crucially relies on the following lemma.

\begin{lemma}\label{lemma-nec:charact-full-prog-cons}
	A "deterministic" "structured signature automaton" $\A$ is "fully progress consistent" if and only if, for each even priority $x$ and each pair of states $q, p$ in $\A$ such that $q\lSig{x}p$ we have:
	\[  q\lrp{w:\geq x} p \; \tand \; p\lrp{w:y} p \; \implies \; y \text{ is even}. \tag{1} \label{eq:1}\]
\end{lemma}

\begin{proof}
	It follows directly from the definition that a "deterministic" "fully progress consistent" automaton satisfies this property.
	To show the converse, assume that~\eqref{eq:1} holds; we aim to prove "full progress consistency".
	Consider an even priority $x$ and a word $w\in \SS^*$ such that $q\lSig{x}p$ and $q\lrp{w:\geq x}p$, we should prove $w^\oo \in \Lang{\initialAut{\A}{q}}$.
	We take $x$ to be minimal such that $q \lSig{x} p$, and thus we have $q \eqSig{x-2} p$.
	For $x=0$, the proof is identical to the one appearing in Lemma~\ref{lemma-warm:prog-cons-nec}, we assume that $x\geq 2$.
	Therefore, there is a "${<}x$-safe path" connecting $q$ and $p$, so we have $q\eqSig{x-1}p$ (Item~\ref{item-struct:odd-orders} from the definition of "structured signature automaton"), thus since $q \lSig{x} p$ we get $\safeSig{x}{q}\subseteq \safeSig{x}{p}$ (Item~\ref{item-struct:even-preorders}).
	Consider the "run over $w^\oo$" from $q$ in $\A$:
	\[ \rr = q\lrp{w:\geq x} p \lrp{w: y_1} p_2 \lrp{w: y_2} p_3 \lrp{w: y_3} \cdots .\]
	
	Since $w \in \safeSig{x}{p}$ and $q \lSig{x} p$, Lemma~\ref{lemma-sig:monotonicity-safe-lang} yields $\safeSig{x}{p} \subseteq \safeSig{x}{p_2}$, and hence $y_1 \geq x$.
	Since $\eqSig{x-2}$ is "$[0,x-2]$-faithful" and $q \lrp{w:\geq x-2} p \eqSig{x-2} q$, it follows that $p_2 \eqSig{x-2} p$.
	Then it follows from $p \lrp{w:\geq x} p_2$ that $p \eqSig{x-1} p_2$, and thus $p \leqSig{x} p_2$.
	Applying the same reasoning by induction yields $y_i \geq x$ and $p \leqSig{x} p_i$ for all $i$, and thus $q \lSig{x} p_i$
	
	%By monotonicity of safe languages (Lemma~\ref{lemma-sig:monotonicity-safe-lang}) and induction, $p\leqSig{x}p_i$ for all $i$.
	%Also, by the inclusion $\safeSig{x}{q}\subseteq \safeSig{x}{p}$, we have that $x\leq y_i$ for all $i$.  

	Eventually, $\rr$ closes a cycle: there are $N$ and $k$ such that, for every $i\geq N$, $p_i = p_{i+k}$.
	We let $p' = p_{kN}$ and let $y$ denote the minimal priority produced by the cycle.
	Then it holds that:
	\[ q\lSig{x} p ', \;\; q\lrp{w^{kN}:\geq x} p',  \; \tand \; p' \lrp{w^{kN}:y} p'. \]
	Thus thanks to~\eqref{eq:1}, $y$ is even, and so $w^\oo = (w^{kN})^\oo \in \Lang{\initialAut{\A}{q}}$.
\end{proof}

We can now deduce Lemma~\ref{lemma-half:check-ProgConsPoly}.

\begin{proof}[Proof of Lemma~\ref{lemma-half:check-ProgConsPoly}]
	For each pair of states $q,p\in Q$ and each priority $x$, we define the languages of finite words
	\[ L_{q\re{}p}^x = \{w\in \SS^* \mid q\lrp{w:x} p \tin \A  \}, \quad \tand \quad    L_{q\re{}p}^{\geq x} = \{w\in \SS^* \mid q\lrp{w: \geq x} p \tin \A\}. \]
	
	By Lemma~\ref{lemma-nec:charact-full-prog-cons}, $\A$ is "fully progress consistent" if and only if, for each even priority $x\in [0,d]$ and each pair of states $q,p\in Q$ such that $q\lSig{x}p$:
	\[  L_{q\re{}p}^{\geq x} \; \bigcap \; \big(\bigcup_{y \text{ odd}} L_{p\re{}p}^y \big) = \emptyset.  \]
	
	We show that for all pair of states, languages $L_{q\re{}p}^{\geq x}$ and $L_{q\re{}p}^x$ are regular and we can obtain deterministic finite automata for them in polynomial time.
	This implies that we can check the emptiness of intersections above in polynomial time, concluding the proof.
	
	For $L_{q\re{}p}^{\geq x}$ the previous claim is clear: the finite automaton obtained by taking the "automaton structure" of $\autGeq{\A}{x}$ and taking $q$ and $p$ as initial and final states, respectively, is a finite automaton recognising $L_{q\re{}p}^{\geq x}$.
	
	For $L_{q\re{}p}^{x}$, we consider the automaton over finite words that has as states $(Q\times [0,d])\cup \{(q,\init)\}$, and $(q,\init)$ and $(p,x)$ as initial and final states, respectively.
	Transitions of the automaton are of the form $(q_1,x_1)\re{a}(q_2,x_2)$ if the transition $q_1\re{a:y}q_2$ in $\A$ is such that $x_2 = \min\{x_1, y\}$. 
	In words, this automaton keeps track of the run in $\A$ from $q$ and of the minimal priority produced in the way. It accepts a word if it arrives to $p$ while producing as minimal priority $x$, as we wanted.
\end{proof}

\subsection{Procedure 2: \texorpdfstring{$\ee$}{epsilon}-completion}

We now prove the following result.

\begin{theorem}\label{th-dec:ee-completion-local}
	Let $\A$ be a "non-deterministic" "parity automaton" recognising a "positional" language $W$.
	Then for each pair of states $q,q'\in Q$, and for each even priority $x$, one may add (at least) one of the transitions
	\[
		q \re{\ee:x} q' \qquad \tor \qquad q' \re{\ee:x+1} q
	\]
	without augmenting the language of $\A$.
\end{theorem}

Before proving Theorem~\ref{th-dec:ee-completion-local}, we argue that decidability of "positionality" in polynomial time (Theorem~\ref{th-reslt:decid-poly}) follows.
Let $\A_0$ be a "deterministic" "parity automaton" recognising a language $W$ and using $d$ as maximal "priority" (assumed even).
We build an "$\ee$-completion" of $\A_0$ as follows.
At each step, pick a pair of states $(q,q')$ such that neither $q \re{\ee:x} q'$ nor $q' \re{\ee:x+1} q$ belongs to the current automaton, for $x\leq d$ even.
Then try to add one of these transitions, and see if the language increases (checking whether $\Lang \A \subseteq W$ can be done in polynomial time since $W$ is "recognised" by the "deterministic" automaton $\A_0$~\cite{ClarkeDK93Unified}).
If the language does not increase for one of the two transitions, then proceed to the next step; otherwise conclude that $W$ is not "positional" thanks to Theorem~\ref{th-dec:ee-completion-local}.

After $|Q|^2d$ steps, we obtain an automaton such that for each pair of states $(q,q')$ and for each even $x$, either $q \re{\ee:x} q'$ or $q' \re{\ee:x+1} q$.
Now for each priority $y$, close the relations $\re{\ee:y}$ by transitivity, which does not augment the language.
Moreover, for priorities $y \cleq y'$ (recall that $1 \cleq 3 \cleq \dots \cleq d+1 \cleq d \cleq \dots \cleq 2 \cleq 0$) add transition $q \re{\ee:y} q'$ whenever $\A$ contains $q \re{\ee:y'} q'$; this also does not augment the language.
Then it holds that the relations $\re{\ee:1},\re{\ee:3},\dots,\re{\ee:d+1}$ are "total preorders" "refining" one another, and that for each even $x$, $\re{\ee:x}$ is the strict variant of $\re{\ee:x+1}$.
Stated differently, the obtained automaton $\A$ is an "$\ee$-completion" of $\A_0$, which implies that $W$ is "positional" thanks to \ref{item-th:ee-complete} $\implies$ \ref{item-th:halfPosArbitrary} in Theorem~\ref{th-reslt:MainCharacterisation-allItems}.

Note that on the way, we obtain the remaining implication~\ref{item-th:halfPosArbitrary} $\implies$ \ref{item-th:ee-complete-forall} from Theorem~\ref{th-reslt:MainCharacterisation-allItems}, stated as follows:

\begin{corollary}\label{cor-dec:ee-closure-nd}
Any "non-deterministic" "parity automaton" "recognising" a "positional" language is "$\ee$-completable".
\end{corollary}

We now prove Theorem~\ref{th-dec:ee-completion-local}.
The proof is inspired by that of~\cite[Theorem 4.8]{CFGO22Universal}, but it is more involved, because we now deal with "parity automata" rather than graphs (or safety automata).
This difficulty is overcome thanks to Theorem~\ref{th-reslt:union-PI}.

\begin{proof}[Proof of Theorem~\ref{th-dec:ee-completion-local}]
Fix a pair of states $q,q' \in Q$ and an even "priority" $x$.
Consider the "game" $\G$ defined as follows (see also Figure~\ref{fig:gadget-completion} below).

\begin{itemize}
\item The alphabet is $C=(\Sigma\cup\{\ee\}) \times \{0,1,\dots,d+1\} \times \{\ent,\sma,\neut\}$.
Therefore, each edge has a letter in $\Sigma\cup \{\ee\}$, a priority in $\{0,1,\dots,d+1\}$, and a type in $\{\ent,\sma,\neut\}$.
For a word $w \in C^\omega$, we write $w_\Sigma,w_\prio$ and $w_\type$ for the respective projections.
For conciseness, we generally omit the type and write edges in the game as $\re{a:y}$, just like in the automaton.
\item The set of vertices consists in two copies of $\A$ indexed by $q$ and $q'$, together with an additional vertex $q_?$.
Formally, $V=Q \times \{q,q'\} \cup \{q_?\}$.
All vertices belong to "Adam" except for $q_?$ which belongs to "Eve".
\item The edges in the copy indexed by $q$ (resp. $q'$) follow exactly the transitions in $\A$, except those leading to $q'$ (resp. $q$), which are instead redirected to $q_?$ (but keep the same letter and priority).
\item The vertex $q_?$ has exactly two outgoing edges: $q_? \re{\ee:x+1} (q,q)$ and $q_? \re{\ee:x} (q',q')$.
\item The edge $q_? \re{\ee:x+1} (q,q)$ has type $\ent$, and edges with priority $\leq x$ inside the copy indexed by $q$ have type $\sma$. All other edges are neutral.
\item The objective is
\[
	W_\G = W_\SS \cup \oddparity \cup \gtype,
\]
where $W_\SS$ is the set of words $w$ such that $w_\Sigma \in W$, $\oddparity$ is the set of words whose minimal priority appearing infinitely often is odd, and $\gtype$ is the set of words with infinitely many occurrences of type $\ent$ and finitely many occurrences of type $\sma$.
\end{itemize}

\begin{figure}
	\begin{center}
	\includegraphics[width=0.7\linewidth]{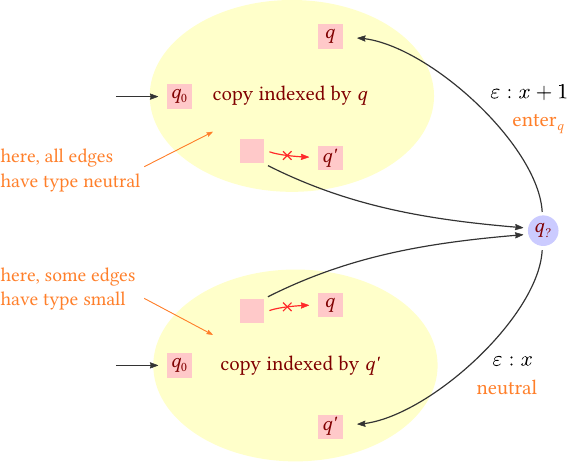}\caption{The "game" $\G$ in the proof of Theorem~\ref{th-dec:ee-completion-local}. Edges which are crossed out are those that are redirected (while keeping the same label). Note that $(q',q)$ and $(q,q')$ are not reachable; these vertices could be removed from the game.}\label{fig:gadget-completion}
	\end{center}
	\end{figure}

Clearly $W_\SS$ is "positional" and so is $\oddparity$; likewise, $\gtype$ rewrites as a "parity condition" (by mapping $\sma$ to 1, $\ent$ to 2 and $\neut$ to 3), and thus it is also "positional".
Moreover, $\oddparity$ and $\gtype$ are "prefix-independent". It follows from Theorem~\ref{th-reslt:union-PI} that $W_\G$ is positional.

\begin{claim}
"Eve" wins from $(q_\init,q)$ and from $(q_\init,q')$ in $\G$.
\end{claim}

\begin{subproof}
Consider the "strategy" which, whenever reaching $q_?$ from an edge of the form $(p,q') \re{} q_?$, reads the edge $q_? \re{\ee:x+1} (q,q)$ and whenever reaching $q_?$ from an edge of the form $(p,q) \re q_?$, reads the edge $q_? \re{\ee:x} (q',q')$.
Take an infinite path $\pi_\G$ in $\G$ from $(q_\init,q)$ which is "consistent with@@strat" the above strategy.
It is of the form
\[
	\pi_\G : (q_\init,q) \rp{w_0:y_0} (p_0,q) \re{a_0:z_0} q_? \re{\ee:x} (q',q') \rp{w_1:y_1} (p_1,q') \re{a_1:z_1} q_? \re{\ee:x+1} (q,q) \rp{w_2:y_2} \dots,
\]
where $w_0,p_0,a_0,w_1,p_1,a_1,\dots$ is either infinite (and the $w_i$'s are finite) or it ends with some $w_i$ which is infinite.
We aim to prove that $\pi_\G$ is "winning@@play", that is, its label belongs to $W_\G$.

Observe that
\[
	\pi_\A : q_\init \rp{w_0:y_0} p_0 \re{a_0:z_0} q' \rp{w_1:y_1} p_1 \re{a_1:z_1} q \rp{w_2:y_2} \dots
\]
defines a path in $\A$.
If $\pi_\A$ satisfies $W$, then $\pi_\G$ satisfies $W_\SS$ (since we only add some occurrences of $\ee$ to the projections on $\Sigma \cup \{\ee\}$).
Thus we assume that $\pi_\A$ does not satisfy $W$, which means that the "run" in the automaton is "rejecting@@run".
If $\pi_\G$ satisfies $\gtype$ then it is "winning@@play", so we assume otherwise: there are either finitely many occurrences of $\ent$ or infinitely many occurrences of $\sma$.
\begin{itemize}
\item If there are finitely many occurrences of $\ent$ in $\pi_\G$, then the "priorities" in $\pi_\G$ eventually coincide with those in $\pi_\A$, and since this is a "rejecting run" in $\A$, $\pi_\G$ satisfies $\oddparity$.
%the sequence $w_0,p_0,a_0,w_1,p_1,a_1,\dots$ is finite and therefore priorities appearing in $\pi_G$ are the same as in $\pi_\A$ plus finitely many insertions. This cannot make a rejecting run accepting, hence $\pi_\G$ satisfies $\oddparity$.
\item If there are infinitely many occurrences of $\sma$, then the minimal "priority" appearing infinitely often in $\pi_\A$ is $\leq x$, therefore adding (even infinitely many) priorities $x$ and $x+1$ does not change the fact that the run is rejecting.
Therefore $\pi_\G$ satisfies $\oddparity$. \qedhere
\end{itemize}
\end{subproof}

Since "Eve" wins and $W_\G$ is "positional", she has a "winning@@strat" "positional strategy" $\sigma$.

\begin{claim}
	If $\sigma(q_?) = q_? \re{\ee:x} q'$, then adding the transition $q \re{\ee:x} q'$ to $\A$ does not augment the language.
\end{claim}

\begin{subproof}
Let $\A'$ be the automaton obtained by adding to $\A$ the transition $q \re{\ee:x} q'$.
Consider an "accepting run" $\pi_{\A'}$ from $q_\init$ in $\A'$.
Decompose it around the occurrences of $q \re {\ee:x} q'$ as follows:
\[
	\pi_{\A'} : q_\init \rp{w_0:y_0} p_0 \re{a_0:z_0} q \re{\ee:x} q' \rp{w_1:y_1} p_1 \re{a_1:z_1} q \re{\ee:x} q' \rp{w_2:y_2} \dots,
\]
where the sequence $w_0,p_0,a_0,w_1,p_1,a_1,\dots$ is either infinite (and the $w_i$ are finite), or it is finite and ends with some $w_i$ which is infinite.

Then
\[
	\pi_{\G} : (q_\init,q') \rp{w_0:y_0} (p_0,q') \re{a_0:z_0} q_? \re{\ee:x} (q',q') \rp{w_1:y_1} p_1 \re{a_1:z_1} q_? \re{\ee:x} (q',q') \rp{w_2:y_2} \dots
\]
defines a path in $\G$ which is consistent with $\sigma$.
Therefore the label of $\pi_\G$ satisfies $W_\G$.
Note that "priorities" in $\pi_{\A'}$ and $\pi_\G$ are the same, and $\pi_{\A'}$ is accepting, therefore $\pi_\G$ does not satisfy $\oddparity$.
Moreover, $\pi_\G$ has no occurrence of $\ent$, thus it does not satisfy $\gtype$.
We conclude that $\pi_\G$ satisfies $W_\SS$, and thus $w_0a_0w_1a_1 \dots \in W$, so $\Lang {\A'} \subseteq W$.
\end{subproof}

The proof of the other case is similar to the previous one, with a slight difference in the analysis.

\begin{claim}
If $\sigma(q_?) = q_? \re{\ee:x+1} q$, then adding the transition $q' \re{\ee:x+1} q$ to $\A$ does not augment the language.
\end{claim}

\begin{subproof}
Let $\A'$ be the automaton obtained by adding to $\A$ the transition $q' \re{\ee:x+1} q$.
Consider an "accepting run" $\pi_{\A'}$ from $q_\init$ in $\A'$.
Decompose it around the occurrences of $q' \re {\ee:x+1} q$ as follows:
\[
	\pi_{\A'} : q_\init \rp{w_0:y_0} p_0 \re{a_0:z_0} q' \re{\ee:x+1} q \rp{w_1:y_1} p_1 \re{a_1:z_1} q' \re{\ee:x+1} q \rp{w_2:y_2} \dots,
\]
where the sequence $w_0,p_0,a_0,w_1,p_1,a_1,\dots$ is either infinite (and the $w_i$ are finite), or it is finite and ends with some $w_i$ which is infinite.

Then
\[
	\pi_{\G} : (q_\init,q) \rp{w_0:y_0} (p_0,q) \re{a_0:z_0} q_? \re{\ee:x+1} (q,q) \rp{w_1:y_1} p_1 \re{a_1:z_1} q_? \re{\ee:x+1} (q,q) \rp{w_2,y_2} \dots
\]
defines a path in $\G$ which is consistent with $\sigma$.
Therefore the label of $\pi_\G$ satisfies $W_\G$.
Note that "priorities" in $\pi_{\A'}$ and $\pi_\G$ are the same, and $\pi_{\A'}$ is "accepting@@run", therefore $\pi_\G$ does not satisfy $\oddparity$.

Now note that $\pi_{\A'}$ has infinitely many occurrences of priority $x+1$ and yet it is "accepting@@run", therefore it has infinitely many occurrences of priorities $\leq x$.
Thus $\pi_\G$ has infinitely many occurrences of $\sma$, so it does not satisfy $\gtype$.
We conclude that $\pi_\G$ satisfies $W_\SS$, and therefore the label $w_0a_0w_1a_1\dots$ of $\pi_{\A'}$ belongs to $W$.
\end{subproof}
This concludes the proof.
\end{proof}

\begin{remark}
It is interesting to remark that our proof relies on the fact that for any "$\omega$-regular" "positional" "objective"~$W$, the "objective" $W \cup \parity$ is "positional".
We do not know a direct proof of this fact (without using Theorem~\ref{th-reslt:union-PI} which relies on the machinery employed to prove Theorem~\ref{th-reslt:MainCharacterisation-allItems}).
Such a direct proof would give, together with Theorem~\ref{th-dec:ee-completion-local}, an easier path to the main characterisation and the polynomial time decidability (though it would fall short of establishing the 1-to-2-player lift and the finite-to-infinite lift).
\end{remark}

\section{Bipositionality of all objectives}\label{sec:bipositional}

In this section we provide a characterisation of all "bipositional" objectives, without "$\oo$-regularity" or "prefix-independence" assumptions.
This characterisation extends the result of Colcombet and Niwiński~\cite{CN06}, who showed that the only "prefix-independent" "bipositional" objective (over all game graphs) is the "parity objective".
Recently, Bouyer, Randour and Vandenhove~\cite{BRV22OmegaRegMemory} generalised that result in an orthogonal direction: they proved that the only objectives for which both players can "play optimally using" finite chromatic memory are "$\oo$-regular objectives".

\subsection{Characterisation of bipositionality and consequences}

\AP We say that an objective $W\subseteq \SS^\oo$ is ""bi-progress consistent"" if both $W$ and its complement are "progress consistent", that is, if it satisfies that for all "residual class" $\resClass{u}$ and finite word $w\in \SS^*$:
\begin{itemize}
	\item $\resClass{u} \lRes \resClass{uw} \implies uw^\oo \in W$, and
	\item $\resClass{uw} \lRes \resClass{u} \implies uw^\oo \notin W$.
\end{itemize}

\begin{theorem}[Characterisation of bipositionality]\label{th-bi:bipositional}
	An objective $W\subseteq \SS^\oo$ is "bipositional" (over all games) if and only if:
	\begin{enumerate}
		\item $W$ has a finite number of "residuals", totally ordered by inclusion, and
		\item $W$ is "bi-progress consistent", and
		\item $W$ can be "recognised" by a "parity automaton on top of" the "automaton of residuals".
	\end{enumerate}
\end{theorem}

This characterisation only holds for infinite games, as there are non "$\oo$-regular" "objectives" that are "bipositional" over finite games, as, for example, energy objectives~\cite{BFLMS08Energy} and their generalisation~\cite{Kozachinskiy24EnergyGroups}. However, we deduce from Theorem~\ref{th-half:lifts} that in the case of "$\oo$-regular" objectives these conditions do also characterise "bipositionality" over finite games.

\begin{corollary}[Bipositionality over finite games for $\oo$-regular objectives]
	An "$\oo$-regular" objective $W\subseteq \SS^\oo$ is "bipositional" over finite games if and only if it satisfies the three conditions from Theorem~\ref{th-bi:bipositional}.
\end{corollary}

\paragraph*{Consequences: Lifts and decidability}
For "$\oo$-regular objectives", we can directly lift the corollaries of Theorem~\ref{th-reslt:MainCharacterisation-allItems} obtained for "positionality" to "bipositionality". For non-$\oo$-regular ones, the finite-to-infinite lift does not hold, as commented above. On the other hand, a combination with a recent result from Bouyer, Randour and Vandenhove~\cite[Theorem~3.8]{BRV22OmegaRegMemory} implies that the 1-to-2-player lift holds for any objective.

\AP We say that an objective $W\subseteq \SS^\oo$ is ""bipositional over (finite) Eve and Adam-games"" if both $W$ and its complement $\SS^\oo\setminus W$ are "positional" over (finite) "Eve-games". 
%and its complement $\SS^\oo\setminus W$ is "positional" over (finite) "Adam-games".

\begin{corollary}[1-to-2 player lift of bipositionality]\label{cor:1-to-2-lift-bipositionality}
	An objective $W\subseteq \SS^\oo$ is "bipositional" (over all games) if and only if it is "bipositional over Eve and Adam-games".
\end{corollary}

We note that a 1-to-2 player lift was obtained for objectives that are "bipositional" over finite game graphs by Gimbert and Zielonka~\cite{GimbertZielonka2005Memory, Zielonka05Invitation} (even in the more general setting of qualitative objectives). 
However, their proof consisted in an induction over the size of the game graph, so it does not generalise to infinite games. 
Indeed, as remarked above, "bipositionality" over finite and infinite graphs behaves in a completely different manner.
In this respect, Corollary~\ref{cor:1-to-2-lift-bipositionality} and the result of Gimbert and Zielonka are incomparable.

\begin{corollary}[Finite-to-infinite lift of bipositionality for $\oo$-regular objectives]
	An "$\oo$-regular" objective $W\subseteq \SS^\oo$ is "bipositional" (over all games) if and only if it is "bipositional over finite Eve and Adam-games".
\end{corollary}

We obtain decidability for "bipositionality" in polynomial time from its counterpart in the case of "positionality" (Theorem~\ref{th-reslt:decid-poly}). We observe that Theorem~\ref{th-bi:bipositional} provides an alternative way to check "bipositionality". 

\begin{corollary}[Decidability of bipositionality]
	Given a "deterministic" "parity automaton" $\A$, we can decide in polynomial time whether $\Lang{\A}$ is "bipositional".
\end{corollary}

\paragraph*{An example}
\begin{example}[Parity over occurrences]
	We let $\SS = [0,d]$ and let $\intro*\WOccParity$ be the language of words such that the minimal "priority" appearing on them is even:
	\[ \WOccParity = \{w\in [0,d]^\oo \mid \min (w) \text{ is even}\}. \]

	An automaton "recognising" $W$ is depicted in Figure~\ref{fig-bi:aut-occParity}. It has one state per "residual", which are totally ordered, and it is immediate to check that it is "bi-progress consistent". Therefore, $\WOccParity$ is a "bipositional" objective.
\end{example}

\begin{figure}
	\centering
	\includegraphics[scale=1.5]{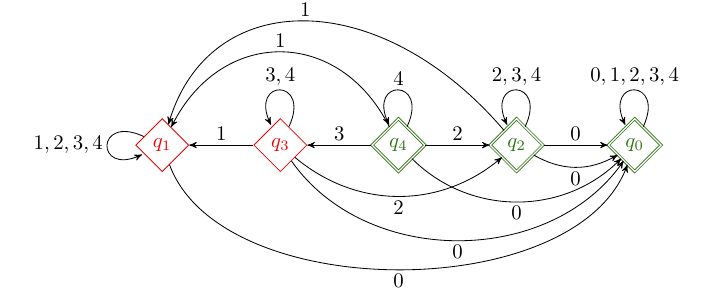}
	\caption{Automaton "recognising" $\WOccParity$, for $d=4$. The initial state is $q_4$.  This automaton is in fact what is sometimes called a weak automaton: runs that finally end in an even state are "accepting@@run", and those ending in an odd state are "rejecting@@run".}
	\label{fig-bi:aut-occParity}
\end{figure}

Some more complex examples can be generated by, for example, adding some output priorities to the automaton above without breaking the "bi-progress consistency" condition. However, the combination of being recognisable by the "automaton of residuals" with "bi-progress consistency" greatly restricts the possibilities of generating examples of "bipositional" languages. We wonder whether a more precise characterisation of languages satisfying these three properties can be obtained.

\subsection{Proof of the characterisation}

\paragraph*{Necessity of the conditions}

The necessity of the total order over the residuals of $W$ is given by Lemma~\ref{lemma-warm:total-order-residuals-nec}, and the necessity of "bi-progress consistency" is given by Lemma~\ref{lemma-warm:prog-cons-nec}.
The following lemma provides the necessity of the last condition of Theorem~\ref{th-bi:bipositional}.
It can be obtained by instantiating the first item of~\cite[Theorem~3.6]{BRV22OmegaRegMemory} for the case of "bipositional" objectives. 

\begin{lemma}[{\cite[Theorem~3.6]{BRV22OmegaRegMemory}}]
	If $W\subseteq \SS^\oo$ is "bipositional over Eve and Adam-games", then $W$ is "$\oo$-regular" and can be "recognised" by a "parity automaton" on top of the "automaton of residuals".
\end{lemma}

\paragraph*{Sufficiency of the conditions}

The sufficiency of conditions of Theorem~\ref{th-bi:bipositional} can be shown by providing "well-ordered@@univ" "monotone@@univ" "universal" graphs for $W$ and its complement. 
An even simpler option -- now that we have already done such construction for "positional" objectives -- is to provide "fully progress consistent" "signature automata" recognising these languages, and then use the characterisation of "positionality" given by Item~\ref{item-th:signature} from Theorem~\ref{th-reslt:MainCharacterisation-allItems}.

We show that the "parity automaton on top" of the "automaton of residuals" is a "fully progress consistent" "signature automaton". Since by hypothesis $\ResW$ is totally ordered, the order $\leqSig{0}$ given by inclusion of "residuals" satisfies the first requirement of the definition of "signature automaton". As the "residual classes" of this automaton are singletons, we can define all relations $\eqSig{x}$ to be the equality relation too, so this automaton satisfies all the requirements to be a "signature automaton".
Moreover, as $W$ is "progress consistent" and the $\eqSig{x}$-classes of the automaton are singletons, it is also "fully progress consistent". 
The argument is symmetric for $\SS^\oo\setminus W$.

\section{Positionality of closed and open objectives}\label{sec:open-closed}

\subsection{Closed objectives}

We recall that an objective $W$ is "closed" if 
\[ W =\Safe{L} = \{w \mid w \text{ does not contain any prefix in } L\}, \]
for some language of finite words $L$.
"Positional" "closed objectives" were first characterised by Colcombet, Fijalkow and Horn~\cite{ColcombetFH14PlayingSafe}, as those that have a totally ordered set of residuals. 
However, as they already remarked, this characterisation only holds for finite branching "game graphs".
%This fact tells us that we cannot hope to have a finite-to-infinite lift for general "closed objective" (as the one presented for $\oo$-regular ones in Corollary~\ref{th-half:lifts}).

We now give a characterisation of "positionality" over all game graphs for "closed objectives".
Namely, a "closed objective" $W$ is "positional" if and only if $\ResW$ is well-ordered by inclusion (Theorem~\ref{th-open:pos-closed}) (in fact, the well-foundedness of $\ResW$ is a necessary condition for finite memory determinacy of any objective.)
%Moreover, we prove (Lemma~\ref{lemma-open:well-found-nec}) that the well-foundedness of $\ResW$ is a necessary condition for finite memory determinacy of any objective.
This (and its generalisation to memory) was already observed in~\cite{CO25LMCS}.
\subsubsection{Well-foundedness of residuals}

Next example, taken from~\cite{ColcombetFH14PlayingSafe},  shows that total order over residuals does not suffice to ensure "positionality" of arbitrary "closed objectives".

\begin{example}[Outbidding game~\cite{ColcombetFH14PlayingSafe}: Total order does not suffice]\label{ex:outbidding-game}
	Let $\SS = \{a,b,c\}$ and $L$ be the language of finite words:
	\[ L = \{ w\in \SS^* \mid \text{ for some } u\in \SS^* \text{ with } \countLetters{u}{a}\leq \countLetters{u}{b}, uc \text{ is a prefix of } w \},\]
	where $\intro*\countLetters{u}{x}$ is the number of occurrences of $a$ in $u$.
	We consider the "closed objective" $W = \Safe{L}$. 
	The "residuals" of $W$ are totally ordered by inclusion:
	\[ \emptyset=\lquotW{c} < \dots < \lquotW{(a^n)} <\cdots <  \lquotW{a} <  \lquotW{\ee} <  \lquotW{b} < \cdots  .\]
	However, $W$ is not "positional", as witnessed by the game in Figure~\ref{fig-open:outbidding}.
\end{example}
\begin{figure}
	\centering
	\includegraphics[scale=1.5]{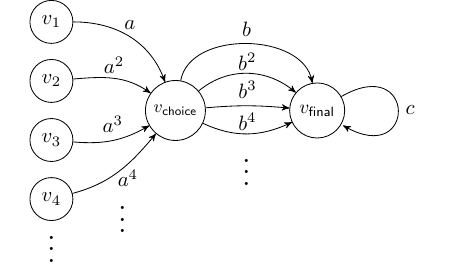}
	\caption{Outbidding game from Example~\ref{ex:outbidding-game}. First, a sequence $a^n$ is produced, for some $n\in \NN$. In order to win, "Eve" needs to answer with $b^m$, with $m>n$. Therefore, she can win from any vertex in the game, but no "positional strategy" guarantees the victory from all states.}
	\label{fig-open:outbidding}
\end{figure}

\begin{lemma}[Necessity of the well-order of residuals]
	Let $W\subseteq \SS^\oo$ be an "objective" that is "positional" over "$\ee$-free" "Eve-games". Then, $\ResW$ is well-ordered by inclusion.
\end{lemma}

We have already seen (Lemma~\ref{lemma-warm:total-order-residuals-nec}), that if $W$ is "positional", then $\ResW$ is totally ordered. We need to prove that $\ResW$ is well-founded. %(In fact, we show a stronger result: if the memory requirements of $W$ are finite, then $\ResW$ is well-founded.)

\begin{lemma}[Well-foundedness of residuals necessary for finite memory]\label{lemma-open:well-found-nec}
	Let $W\subseteq \SS^\oo$ be an "objective" such that "Eve" can "play optimally using positional strategies" (or using finite memory strategies) over "$\ee$-free" "Eve-games". Then, $\ResW$ is well-founded (for the order given by inclusion of "residuals").
\end{lemma}
\begin{proof}
	Suppose by contradiction that there is an infinite strictly decreasing sequence of "residuals":
	\[ \lquotW{u_1} \supsetneq \lquotW{u_2} \supsetneq  \cdots , \; u_i\in \SS^*.\]
	(We suppose without loss of generality that $\ee\neq u_i$ for all $i$.)
	Let $w_i\in \SS^\oo$ such that $w_i\in \lquotW{u_i} \setminus \lquotW{u_{i+1}}$.
	We consider the game -- similar to the outbidding game from Figure~\ref{fig-open:outbidding} -- in which a word $u_i$ labels a path from a vertex $v_i$ to $v_{\mathsf{choice}}$, for each $i$. From this latter vertex, "Eve" can choose a between paths labelled by $\{w_i\mid i\in \NN\}$.
	"Eve" can win be answering $w_i$ to $u_i$. However, any "positional@@strat" (or finite memory) strategy will only consider a finite number of responses $w_{j_1},\dots,w_{j_n}$. Therefore, such a strategy is "losing@@strat" from $v_K$ for $K> \max \{j_t \mid 1\leq t\leq n\}$.	
\end{proof}

\subsubsection{Characterisation for closed objectives}

\begin{theorem}[Positional closed objectives]\label{th-open:pos-closed}
	Let $W\subseteq \SS^\oo$ be a "closed objective". Then, $W$ is "positional" (over all "game graphs") if and only if $\ResW$ is well-ordered by inclusion.
\end{theorem}
\begin{proof}
	We have already shown that this condition is necessary. To prove sufficiency, we give, for each cardinal $\kk$, a "$(\kk,W)$-universal" "well-ordered@@univ" "monotone@@univ" graph. We conclude by Proposition~\ref{prop-prelim:univ-graphs}.
	
	Let $U$ be the "$\SS$-graph" that has as vertices $\ResW\setminus \{\emptyset\}$, ordered by inclusion. By hypothesis, this is a well-order. For each $a\in \SS$, we let 
	\[ \lquotW{u} \re{a} \lquotW{u'} \;\; \text{ iff } \; \; \lquotW{u'}\leq \lquotW{(ua)}. \]
	We note that if $ua$ is already losing ($\lquotW{(ua)} =\emptyset$), then transition $\lquotW{u} \re{a} \lquotW{u'}$ does not appear in $U$.
	By Lemma~\ref{lemma-prelim:monotonicity-residuals}, graph $U$ is "monotone@@univ". 
	The hypothesis of "closeness@@obj" of $W$ is fundamentally used in next claim.
	
	\begin{claim}
		A vertex $\lquotW{u}$ of $U$ "satisfies@@univ" $W$ if and only if $\lquotW{u} \subseteq \lquotW{\ee}=W$.
	\end{claim}
	\begin{subproof}
		Let $L\subseteq \SS^*$ such that $W = \Safe{L}$. Let 
		$\lquotW{u_1} \re{a_1} \lquotW{u_2} \re{a_2} \dots$
		be a path in $U$ from $u_1=u$. By induction we obtain $\emptyset\neq \lquotW{u_i} \subseteq \lquotW{(u_1a_1\dots a_{i-1})}$.
		Therefore, for all $i$, $ua_1\dots a_i \notin L$, so, by definition of $W$, the infinite word $ua_1a_2\dots$ belongs to~$W$.
	\end{subproof}
	
	We show that $U$ is "$(\kk,W)$-universal for trees", for every cardinal $\kk$, and conclude by Lemma~\ref{lemma-prelim:univ-for-trees}.
	Let $T$ be a "$\SS$-tree" which "satisfies@@tree" $W$. For each node $t\in T$, let $\phi(t) = \lquotW{u_t}$ be the minimal "residual" such that $t$ "satisfies@@univ" $\lquotW{u_t}$. In particular, for the "root@@univ" $t_0$, $\phi(t_0)$ "satisfies@@univ" $W$ by the previous claim.
	We claim that $\phi$ is a "morphism@@univ". Indeed, if $t\re{a} t'$ in $T$ and $t$ "satisfies@@univ" $\lquotW u$, then $t'$ "satisfies@@univ" $\lquotW{(ua)}$. Therefore, $\lquotW{u_t'}\leq \lquotW{(u_ta)}$, so $ \phi(t) = \lquotW{u_t} \re a \lquotW{u_{t'}} = \phi(t')$ is an edge in $U$.
\end{proof}

\subsection{Open objectives}
We recall that an objective $W$ is "open" if 
\[ W =\Reach{L} = \{w \mid w \text{ contains some prefix in } L\}, \]
for some $L\subseteq \SS^*$.

\subsubsection{Reset-stability}
In Section~\ref{subsec-warm:open}, we showed that "positional" "$\oo$-regular" "open" objectives are exactly those with residuals totally ordered and that are "progress consistent". However, for "non-$\oo$-regular" objectives, these conditions do not suffice, even if residuals are well-ordered.

\begin{example}[Progress consistency does not suffice]
	Let $\SS = \NN$ and  let $W$ be the set of sequences that are not strictly-increasing:
	\[ W = \{ a_1a_2\dots \in \NN^\omega \mid a_{i+1} \leq a_i \text{ for some } i \ \}. \]
	
	This objective is "open", as $W=\Reach{\text{two consecutive non-increasing numbers}}$.	
	Its residuals are:
	\[ \lquotW{\ee} < \lquotW 0 < \lquotW 1 < \lquotW 2 < \cdots < \lquotW{(00)} = \SS^\oo.\]
	Therefore, $\ResW$ is well-ordered. Moreover, $W$ is "progress consistent": any repetition of factors induces a non-strict inequality $<$ between consecutive letters.
	
	However, we claim that $W$ is not "positional". Consider the game in Figure~\ref{fig-open:non-reset-stab}.
	"Eve" can win this game: no matter what is the vertex $v_i$ chosen by "Adam", she can first move one position to the right, producing $i$, and then go down producing letter $1$. This ensures two consecutive non-increasing numbers.
	However, she cannot "win positionally". Indeed, if such a strategy tells her to go always to the right, the sequence produced will be strictly increasing. If she choses to go down in vertex $v_i$, Adam can win by initialising the play in that vertex.
\end{example}
\begin{figure}
	\centering
	\includegraphics[scale=1.5]{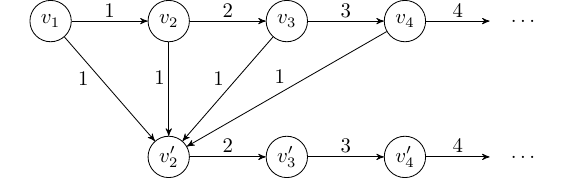}
	\caption{Game in which Eve wins if she does not produce a strictly-increasing sequence of numbers. She can win from every vertex, but not positionally.}
	\label{fig-open:non-reset-stab}
\end{figure}

\begin{definition}[Reset-stability]
	\AP We say that an "objective" $W\subseteq \SS^\oo$ is ""reset-stable"" if, for each sequence of finite words $u_1,u_2,u_3,\dots\in \SS^+$ and of residuals $\lquotW{s_0}, \lquotW{s_1}, \lquotW{s_2},\dots \in \ResW$:
%	\begin{itemize}
%		\item $\lquotW{s_0} = \lquotW{\varepsilon}$,
%	\end{itemize}
	\[\lquotW{s_{i}} \subsetneq \lquotW{(s_{i-1}u_{i})} \;\text{ for all } i\geq 1 \;\; \implies \;\; u_1u_2u_3\dots \in \lquotW{s_0}.\]
\end{definition}

\AP An intuitive idea of "reset-stability" is the following. Consider the (potentially infinite) "automaton of residuals" of $W$, which inherits the order over residuals. Add to it all $\ee$-transitions going backwards: $\lquotW s \re{\ee} \lquotW{s'}$ for $\lquotW{s'}<\lquotW{s}$. When a run takes an $\ee$-transition, we say that it ""makes a reset"".
What "reset-stability" tells us is that any run making infinitely many resets must be accepting.  The words $u_1,u_2,\dots$ in the definition above correspond to fragments where no "reset@@make" takes place, and $\lquotW{s_i}$ is the "residual" where we land after the i\ts{th} "reset@@make".
(See also the notion of "$0$-jumps" in the proof of Lemma~\ref{lemma:paths_in_hm}).

\begin{remark}
	If $W$ is "reset-stable", it is "progress consistent". The converse holds if $\ResW$ is finite and totally ordered.
\end{remark}

We note that all "closed objectives" are "reset-stable".

\begin{lemma}[Necessity of reset-stability]\label{lemma-open:reset-stab-nec}
	Let $W\subseteq \SS^\oo$ be a "positional" "objective" over "$\ee$-free" "Eve-games". Then, $W$ is "reset-stable".
\end{lemma}
\begin{proof}
	Suppose by contradiction that $W$ is not "reset-stable". That is, there are $u_1,u_2,\dots \SS^+$   and $\lquotW{s_0},\lquotW{s_1} \dots$ such that $\lquotW{s_{i}} < \lquotW{(s_{i-1}u_{i})}$, but  $u_1u_2\dots \notin \lquotW{s_0}$.
	
	Let $w_{i}\in \SS^\omega$ such that $w_i\in \lquotW{(s_{i-1}u_{i})} \setminus \lquotW{s_{i}}$.
	We consider the game pictured in Figure~\ref{fig-open:proof-reset-stab} (to ensure it to be "$\ee$-free", we just remove vertex $v_i$ if $s_i=\ee$). 
	
	\begin{figure}
		\centering
		\includegraphics[scale=1.5]{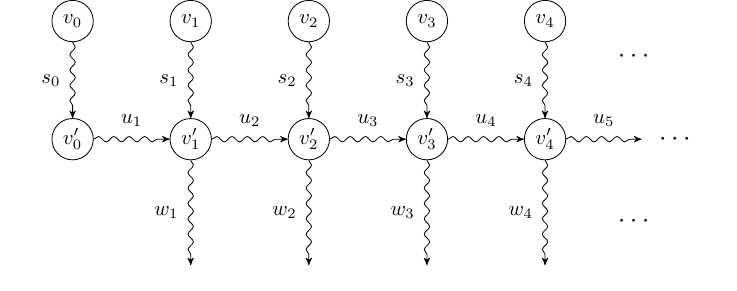}
		\caption{Game in the proof of Lemma~\ref{lemma-open:reset-stab-nec}. Eve can win from every $v_i$, but not positionally.}
		\label{fig-open:proof-reset-stab}
	\end{figure}

	"Eve" can win $\G$ from any vertex $v_i$ (and from $v_i'$ if $s_i=\ee$), as she can produce the word $s_iu_{i+1}w_{i+1}$, which belongs to $W$, as we have taken $w_{i+1}\in \lquotW{(s_{i}u_{i+1})}$. 
	However, we show that no "positional strategy" ensures to win from all these vertices. 
	We distinguish two cases. If this "strategy" takes the path $v_i'\lrp{u_{i+1}} v_{i+1}'$ for all $i$, then it is not "winning from" $v_0$, as by hypothesis $u_1u_2\dots \notin \lquotW{s_0}$.
	If on the contrary this strategy takes a path $v_i'\lrp{w_{i}}$, then it is not "winning from" $v_i$.
\end{proof}

\subsubsection{Characterisation for open objectives}

\begin{theorem}[Positional open objectives]\label{th-open:char-open}
	Let $W\subseteq \SS^\oo$ be an "open objective". Then, $W$ is "positional" (over all "game graphs") if and only if:
	\begin{itemize}
		\item $\ResW$ is well-ordered by inclusion, and
		\item $W$ is "reset-stable".
	\end{itemize}
\end{theorem}
\begin{proof}
	The necessity of the conditions has already been established in Lemmas~\ref{lemma-open:well-found-nec} and~\ref{lemma-open:reset-stab-nec}. To prove the sufficiency, we give, for each cardinal $\kk$ a "well-ordered@@univ" "monotone@@univ" graph that is "$(\kappa,W)$-universal for trees", and conclude by Proposition~\ref{prop-prelim:univ-graphs} and Lemma~\ref{lemma-prelim:univ-for-trees}.
	
	We let $U$ be the "$\SS$-graph" having as set of vertices $(\ResW \setminus \{\emptyset\}) \times \kappa$, ordered lexicographically. This graph is "well-ordered", as by hypothesis so is $\ResW$. The edges are given by:
	\[ (\lquotW u , \lambda) \, \re a \,(\lquotW{u'} , \lambda') \quad\tif \quad \left\{\begin{array}{lc}
		\lquotW{u'} = \lquotW{(ua)} \tand \lambda' < \lambda,& \tor \\[2mm]
		\lquotW{u'} \subsetneq \lquotW{(ua)},& \tor \\[2mm]
		 \lquotW{u} = \SS^\oo.&
	\end{array}\right. \]
	By Lemma~\ref{lemma-prelim:monotonicity-residuals}, this graph is "monotone@@univ". We show its "$(\kappa,W)$-universality for trees". The next claim, which relies in the "reset-stability" hypothesis, provides the key ingredient for this.
	We let $L$ be the language of finite words such that $W=\Reach{L}$.
	\begin{claim}
		For each ordinal $\lambda < \kappa$ and each "residual" $\lquotW u$, the vertex $(\lquotW u, \lambda)$ "satisfies@@univ" $\lquotW u$ in $U$ (i.e. for all paths from $(u^{-1}W, \lambda)$ with label $w$, it holds that  $uw\in W$).
	\end{claim}
	\begin{subproof}
		Let $\rr= (\lquotW{u_0},\lambda_0)\re{a_0} (\lquotW{u_1},\lambda_1)\re{a_1} \cdots $ be an infinite path from $(\lquotW u, \lambda)$ in $U$.
		%, and consider its projection over the (infinite) "automaton of residuals" $\autRes{W}$. Whenever $\rr$ contains a transition $(\lquotW{u_i} , \lambda_i) \, \re{a_i} \,(\lquotW{u_{i+1}} , \lambda_{i+1})$ with $\lquotW{u_{i+1}} = \lquotW{(u_ia_i)}$, this transition exists in $\autRes{W}$.
		If $\lquotW{u_{i+1}} \subsetneq \lquotW{(u_ia_i)}$, we say that the transition $\re{a_i}$ "makes a reset".
		By induction, we obtain that $\lquotW{u_i} \subseteq \lquotW{(ua_1\dots a_{i-1})}$. We distinguish two cases:
		(1) If $\rr$ "makes infinitely many resets", then we conclude by "reset-stability". (2) If $\rr$ "makes finitely many resets", then eventually $\lambda_{i+1}<\lambda_{i}$ for all $i$, unless $\lquotW{u_i}=\SS^\oo$. We conclude that eventually $\lquotW{u_i}=\SS^\oo \subseteq \lquotW{(ua_1\dots a_{i-1})}$. Therefore, $ua_1\dots a_{i-1}\in L$, so $ua_1a_2\dots\in W$.
	\end{subproof}
	Let $T$ be a "$\SS$-tree" whose "root@@univ" "satisfies@@univ" $W$.
	We give a "morphism@@univ" $\phi\colon T \to U$, which we decompose in $\phi_1\colon T \to \ResW \setminus \{\emptyset\}$ and $\phi_2\colon T \to \kappa$.
	For each $t\in T$, let $u_t$ be the word labelling the path from the "root" $t_0$ to $t$. We let $\phi_1(t) = \lquotW{u_t}$ for each $t$. 
	We define $\phi_2$ by transfinite induction. 
	By hypothesis, each branch eventually contains vertices $t$ such that $u_t\in L$ (that is, $\lquot{u_t} = \SS^\oo$). For all these vertices, we let $\phi_2(t) = 0$.
	The tree obtained by removing these vertices, named $T_1$, does not have any infinite branch.
	For an ordinal $\lambda<\kappa$, let $T_\lambda$ be the set of nodes for which we have not defined $\phi_2$ at step $\lambda$ of the induction.
	For each leaf $t$ of $T_\lambda$, we let $\phi_2(t) = \lambda$.
	
	This mapping has the two following properties:
	\begin{itemize}
		\item if $t\re a t'$ in $T$, then $\phi_1(t') = \lquotW{(u_ta)}$, and
		\item  if $t\re a t'$ in $T$, either $\phi_2(t')<\phi_2(t)$, or $u_{t'}\in L$.
	\end{itemize} 
	This ensures that $\phi = (\phi_1,\phi_2)$ is a "morphism@@univ", concluding the proof.
\end{proof}

\begin{example}[Positional open objective]\label{ex:open-infinite-positional}
	Let $\SS = \NN$ and  let $W$ be the set of sequences that start by $123\dots n$, and eventually decrease. Formally:
	\[ W = \{ a_1a_2\dots \in \NN^\omega \mid \text{there is } j \tst a_{i} = i \text{ for } i<j \tand a_j<j\}. \]
	
	This objective is "recognised" by the infinite reachability automaton depicted in Figure~\ref{fig-open:aut-reset-stable}.
	Its residuals are well-ordered and it is "reset-stable", as after any reset we necessarily produce an word in $W$. Therefore, $W$ is "positional".
	
	We remark that this objective is not "bipositional", as $\Res{\SS^\oo\setminus W}$ is not well-founded. This contrast with the case of "$\oo$-regular open" objectives, for which all "positional" "open objectives" are "bipositional" (Corollary~\ref{cor-warm:open-bipositional}).
\end{example}
\begin{figure}
	\centering
	\includegraphics[scale=1.5]{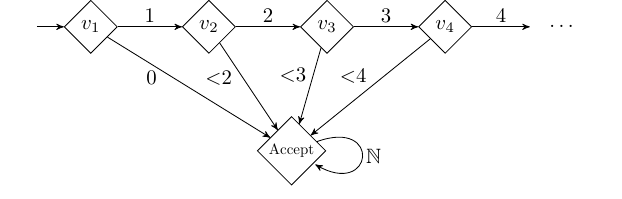}
	\caption{Automaton recognising the objective $W$ from Example~\ref{ex:open-infinite-positional}.}
	\label{fig-open:aut-reset-stable}
\end{figure}

\subsection{1-to-2-player lift and addition of neutral letters}

\begin{corollary}[1-to-2-player lift for open and closed objectives]\label{cor-open:lift}
	Let $W\subseteq \SS^\oo$ be an "open" or "closed" "objective". If $W$ is "positional" over "$\ee$-free" "Eve-games", then $W$ is "positional" over all "game graphs".
\end{corollary}

We also obtain from our proofs that the Neutral letter conjecture (Conjecture~\ref{conj-half:neutral-letter}) holds for "open" and "closed" objectives.
\begin{corollary}[Closure under addition of neutral letters]
	Let $W\subseteq \SS^\oo$ be an "open" or "closed" "objective". If $W$ is "positional", then $\addNeutral{W}$ is "positional".
\end{corollary}
\begin{proof}
	In the proof of Theorems~\ref{th-open:pos-closed} and~\ref{th-open:char-open}, we obtained the positionality of objectives by providing "well-ordered@@univ" "monotone@@univ" "$(\kappa,W)$-universal" graphs. Proposition~\ref{prop-reslt:ohlmann-neutral-letters} allows us to conclude.
\end{proof}

We have therefore obtained the 1-to-2-player lift for $\oo$-regular objectives, as well as "open" and "closed" ones.
However, the 1-to-2-player lift does not hold for arbitrary objectives. 
A counter-example appears in~\cite[Section~7]{GK22Submixing}, moreover, the objective presented there is in $\intro*\SigmaTwo$, that is, a countable union of "closed" objectives.
Another counter-example, discussed in~\cite{Kop08Thesis} and~\cite[p.236]{Vandenhove23Thesis}\footnotemark{} is:
\[ \mathbf{MP}^\QQ = \{w\in \{0,1\}^\oo \mid \liminf_n \frac 1 n \sum_{i=0}^n w_i \;\; \text{ is rational}\}. \]
\footnotetext{In his PhD~\cite{Vandenhove23Thesis}, Vandenhove discusses the 1-to-2-player lift for finite "game graphs", but this counter-example also applies to infinite "game graphs".}

We provide here yet a different example, which is conceptually simpler but of higher topological complexity than the one in~\cite{GK22Submixing}.

\begin{proposition}[No general 1-to-2-player lift]\label{prop:no-general-lift}
	There is an objective $W\subseteq \SS^\oo$  that is "positional" over "Eve-games", but is not "positional" over all "game graphs".
\end{proposition}
\begin{proof}
	Let $\SS$ be any infinite alphabet, and let $W$ be the following objective:
	\[ \intro*\WInfMuller = \{ w\in \SS^\oo \mid \text{ the set of letters occurring infinitely often in } w \text{ is finite}\}. \]
	To show that it is not "positional", we consider the game in Figure~\ref{fig-open:game-infinite-Muller}, in which "Eve" controls a single vertex, from which she can send the token to any vertex $v_i$ controlled by "Adam". 	
	\begin{figure}
		\centering
		\includegraphics[scale=1.5]{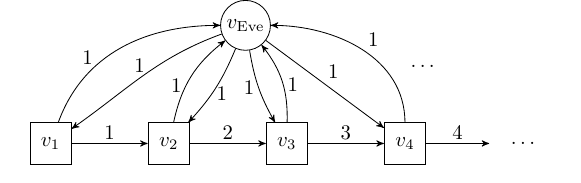}
		\caption{Game in which Eve wins if only finitely many letters are produced infinitely often. She can win by sending the token further and further away, but she cannot "win positionally".}
		\label{fig-open:game-infinite-Muller}
	\end{figure}
We claim that "Eve" can win this game from every vertex by using the following "strategy": she keeps track of the maximal index $i_{\max}$ such that the play has passed trough vertex $v_{i_{\max}}$. Whenever Adam sends back the token to $v_{\Eve}$, she will go to vertex $v_{i_{\max}}$. This "strategy" ensures that only letter $1$ will be produced infinitely often.

However, "Eve" cannot win using a "positional@@strat" (or even finite memory) strategy. Such a strategy will only consider finitely many edges $v_{\Eve} \re 1 v_i$. Let $v_k$ be the maximal such vertex. "Adam" can win against such strategy by producing longer and longer paths, and then sending back the token to the vertex controlled by "Eve": $v_i\lrp{i(i+1)\dots} v_{k+j}\re 1 v_{\Eve}$. In this way, all numbers greater that $k$ will be produced infinitely often.

We show that $\WInfMuller$ is positional over "Eve-games". 
\begin{claim}
	For every "Eve-game" $\G$ with "winning condition" $\WInfMuller$ and every fixed vertex $v_0$, if "Eve" wins from $v_0$, she can "win from $v_0$" using a "positional strategy".
\end{claim}
\begin{subproof}
	Assume that "Eve" "wins from $v_0$".
	A "strategy" from $v_0$ is just an infinite path from $v$. Consider such a path. If this path does not visit a same vertex twice, it is already a positional strategy. On the contrary, let $v_k$ be the first vertex that repeats. Consider the first two occurrences of $v_{\mathsf{rep}}$:
	\[ v_0 \re{a_0} v_1 \re{a_1} \cdots v_{k} \re{a_k} \cdots v_{k+j} \re{a_{k+j}} v_k. \]
	Then, the "positional strategy" indicating to take the edge $v_i\re{a_i}v_{i+1}$ for $i\leq k+j$ is "winning@@strat".
\end{subproof}

We use this claim to prove that $\WInfMuller$ is ("uniformly") "positional" over "Eve-games". 
Let $\G$ be an "Eve-game" with "winning condition" $\WInfMuller$. By "prefix-independence" of $\WInfMuller$, we can suppose without loss of generality that "Eve" "wins from" every vertex in $\G$. Let $v_1,v_2, \dots$ a (potentially transfinite) enumeration of vertices in $\G$. We will define a positional strategy $\strat\colon V \to E$ by transfinite induction. At step $\lambda$, let $\G_\lambda$ be the game obtained by removing vertices for which $\strat$ has already been defined. Let $i$ be the minimal index such that $v_i$ appears in $\G_\lambda$. By the previous claim, "Eve" has a positional strategy in $\G_\lambda$ that wins from $v_i$. Let $V_\lambda$ be the vertices reachable from $v$ by using this strategy, and for $v\in V_\lambda$ let $\strat(v)$ be the edge indicated by such strategy. We let $V_\lambda'$ be the vertices in $\G_\lambda\setminus V_\lambda$ from which "Eve" can reach $V_\lambda$, and fix a positional strategy doing so. We let $\strat(v')$ being given by this strategy for these vertices.
It is immediate that $\strat$ is a "positional strategy" in $\G$ that "wins from" all vertices.
\end{proof}

%%%%%\AP We say that an objective $W$ is ""positional from fixed initial vertices"" if for every game $\G$ and vertex $v$, if "Eve" can "win" $\G$ "from $v$@@win", then she can "win from $v$" using a "positional strategy". 
%
%\begin{lemma}\label{lemma-uniform-pos-for-prefix-indep}
% If $W$ is ""positional from fixed initial vertices"" over "Eve-games", then $W$ is "positional" over Eve-games.
%\end{lemma}
%}

\section{Conclusions}\label{sec:conclusions}

The results presented in Section~\ref{sec:char-half-pos} effectively address most open questions regarding "positionality" in the context of "$\oo$-regular languages".
Yet, it would be reasonable to seek a ``better'' characterisation. 
One drawback of our approach is its conceptual complexity and its exclusive focus on automata; the characterisation is based on syntactic and combinatorial properties of "parity automata", rather than on intrinsic language-theoretical properties of the languages they "recognise".
In this respect, the insights gained about "positionality" are somewhat limited.
Therefore, we believe that there is still room for improvement and for a deeper understanding of the class of "positional" "$\oo$-regular" objectives.

To conclude, we discuss further research directions extending our results.
We start by discussing the follow-up work~\cite{CO25memory}, generalising the results of this paper to objectives requiring memory.

\subsection{Follow-up work: Positionality and memory of \texorpdfstring{$\BCSigma$ languages}{objectives recognised by infinite deterministic parity automata}}

There are two natural orthogonal directions to extend our results: to consider broader classes of objectives and to characterise their memory requirements rather than just "positionality". %Some results in these directions appear in the paper~\cite{CO25memory}.

In Section~\ref{sec:open-closed} we have presented "positionality" results for some non-$\oo$-regular objectives, namely for "closed" and "open" objectives.
Going further in that direction, the goal is to develop characterisations for more complex objectives defined by topological properties, mainly, higher classes in the Borel hierarchy. 
\AP A natural first step is to look at objectives in $\intro*\BCSigma$: the class of boolean combinations of objectives in $\SigmaTwo$ (countable unions of closed objectives), or equivalently, objectives "recognised" by infinite "deterministic" parity automata~\cite[Sect.~5]{Thomas1991AutomataOI}.
In~\cite{CO25memory}, we provide a characterisation of objectives in $\BCSigma$ for which "Eve" can "play optimally" with $\leq k$ states of memory. %in games using this objective.
In particular, for $k=1$, this characterises positionality, allowing us to prove "Kopczyński's conjecture" (Conjecture~\ref{conj:Kopcz-Union}) for the class of $\BCSigma$ objectives.
The simplest class in the Borel hierarchy for which we do not have a satisfactory characterisation of positionality is $\Delta_3^0 = \Sigma_3^0 \cap \Pi_3^0$, for which we do not know whether "Kopczyński's conjecture" holds.

Concerning 1-to-2-player lifts for "positionality" (for one player), the results from this paper are closed to the theoretical limit.
Indeed, we have shown that the 1-to-2-player lift holds for $\oo$-regular, open and closed objectives (Theorem~\ref{th-half:lifts} and Corollary~\ref{cor-open:lift}), 
and it is known that no such lift holds for $\SigmaTwo$-objectives~\cite[Section~7]{GK22Submixing}.
%The versions of the 1-to-2-players lift in the case of chromatic memory remains open (see~\cite[Conjecture~9.1.1]{Vandenhove23Thesis}).

\subsection{Minimisation and canonisation of parity automata}\label{subsec-concl:minimisation}

For the proof of Theorem~\ref{th-reslt:MainCharacterisation-allItems} (Section~\ref{subsec:from-HP-to-signature}), we have introduced new notions concerning congruences for parity automata, as well as different transformations of automata.
%that aim to remove redundant states and to put automata in some form that make them suitable for our study. 
This transformations exhibit a flavour similar to what one might expect from minimisation algorithms, as their main purpose is to remove redundant states from automata. 
We believe that this techniques will be valuable for the study of $\oo$-automata in other contexts, as they allow for a fine analysis of the structure of the automata.

Also, a key ingredient in our proof is the use of "history-deterministic" automata, in particular, a generalisation of the minimisation algorithm for "history-deterministic" "coB\"uchi automata" introduced by Abu Radi and Kupferman~\cite{AK22MinimizingGFG}. We see this fact as further evidence that "history-deterministic" automata provide canonicity properties.% of high value in the theoretical study of automata.

%In fact, one of the main techniques we used ("safe centralisation" of automata) is just a generalisation of the minimisation algorithm for "history-deterministic" "coB\"uchi automata" introduced by Abu Radi and Kupferman~\cite{AK22MinimizingGFG}.

In fact, we can derive actual concrete statements about the minimisation of automata from our results:

\begin{proposition}
	Given a "deterministic parity automata" "recognising" a "bipositional" language $L$, we can compute in polynomial time the size of a minimal deterministic (resp. "history-deterministic") parity automata for $L$.	
\end{proposition}
\begin{proof}
	A necessary condition for a language $L$ to be "bipositional" is that it must be recognised by a "parity automaton on top" of the "automaton of residuals" (Theorem~\ref{th-bi:bipositional}).
	Therefore, a minimal "deterministic" or "HD" automaton for $L$ will have as many states as "residuals" of the language.
	To compute the number of "residuals" of the language, it suffices to determine the number of equivalence classes of states recognising the same language in the input deterministic automaton. For this, it suffices to do an equivalence check for each pair of states, and each of them takes polynomial time~\cite{ClarkeDK93Unified}.
\end{proof}

\begin{proposition}
	Deterministic (resp. "history-deterministic") "B\"uchi@@aut" and "coB\"uchi automata" "recognising" "positional" languages can be minimised in polynomial time.
\end{proposition}
\begin{proof}
	In the case of "B\"uchi automata", a "positional" language can be "recognised" by the "automaton of residuals" (Propositions~\ref{prop-warm:char-Buchi-all}), which is necessarily minimal amongst "history-deterministic" automata.
	
	In the case of "coB\"uchi automata", we obtained that the minimal "history-deterministic" automaton of Abu Radi and Kupferman, computable in polynomial time, can be taken "deterministic" for "positional" languages (Section~\ref{subsec-warm:coBuchi}).% (Proposition~\ref{prop-warm:char-coBuchi-all-aut}).
\end{proof}

We conjecture that our methods may lead to similar results in the more general case of "parity automata", and that minimal automata for "positional" languages can be obtained just by merging states of "signature automata".

\begin{conjecture}
	"Deterministic" and "history-deterministic" "parity automata" "recognising" "positional" languages can be minimised in polynomial time.
	Moreover, "history-deterministic" "parity automata" for this class of languages are not more succinct than "deterministic" ones.
\end{conjecture}

%We believe that we can even describe the minimal automaton that we should obtain, by refining the definition of "reduced signature automaton" given in Section~\ref{subsec:def-signature}.
%\AP We say that a "signature automaton" is ""strongly reduced@@sig"" if it is a "structured signature automaton" and, for every state $q$ and even priority $x$, if no transition producing a priority $\geq x$ leaves $q$, then the $\eqSig{x}$-class of $q$ is trivial: $\classSig{x}{q} = \{q\}$.
%
%\begin{conjecture}
%	A "strongly reduced@@sig" "signature automaton" "recognising" a language $L$ has a  minimal number of states amongst "history-deterministic" "parity automata" "recognising" $L$.
%	
%	Moreover, if $L$ is "positional", it can be "recognised" by a "deterministic" "strongly reduced@@sig" "signature automaton".
%\end{conjecture}
%
%We are highly confident that this result holds. One should be able to prove that "strongly reduced@@sig" "signature automata" are minimal by generalising the methods from~\cite{AK22MinimizingGFG}. Nevertheless, such a proof runs the risk of being highly technical.

\subsection{Algorithms for \texorpdfstring{$\omega$}{omega}-regular games}\label{subsec-concl:algorithms}
A major algorithmic problem is the resolution of games on graphs, that is, deciding whether "Eve" has a "winning strategy" in a given game.
A common approach to solving games with "$\oo$-regular" objectives is to reduce them to parity games: Build the product between the game graph and a "parity automaton" recognising the objective, and then apply a parity game solver to it. 
Using state-of-the-art quasipolynomial-time parity game solvers~\cite{CJKLS22}, this leads to a complexity of roughly $(|\G||\A|)^{\log d}$, where $d$ is the number of priorities used by $\A$.
As discussed in the previous section, if $\A$ recognises a "positional" language, we have provided a polynomial-time algorithm reducing its size (to possibly a minimal one).

Moreover, if the objective $W=\Lang{\A}$ is "positional", an alternative method is to apply a value iteration algorithm directly on the game. For this, we need a $(|G|,W)$-universal graph $\U$ (which exists by~\cite[Theorem~3.1]{Ohlmann23Univ}), which leads to an algorithm working in time $\O(|G|\|\U|)$~\cite{CFGO22Universal}.
Moreover, one of the appealing aspects of this approach is the possibility to establish lower bounds for the size of universal graphs, which has been succesfully used in the past for setting the limits of various families of algorithms~\cite{CFJLP19,CFGO22Universal}.

In this work, we have shown how to obtain a "universal graph" from a "signature automaton". A possible research direction would be to optimize this construction and determine the size of universal graphs for positional $\oo$-regular objectives.
%For parity automata $\A$ recognising "positional" objectives, 
In this way, we could expect to:
\begin{itemize}
	\item decrease the complexity of solving positional $\oo$-regular games to about $|G|^{(1+\log d)}|\A|$, 
	\item obtain matching lower bounds for $(n, \Lang{\A})$-universal graphs.
\end{itemize}

%\bibliographystyle{alpha}
%\bibliographystyle{plainurl}
%\bibliography{references.bib}

   \printbibliography
\appendix

\section{Full proofs for Section \texorpdfstring{\ref{subsec:from-HP-to-signature}}{5.2}}\label{appendix:proofs-necessity}

In this Appendix, we include full proofs for all the propositions and lemmas appearing in Section~\ref{subsec:from-HP-to-signature}.
To help formalise them, we first introduce some technical artillery that will come in handy to deal with the structure of total preorders of "signature automata".
We begin by introducing "nice transformations" of automata equipped with a "priority-faithful" relation in Section~\ref{subsec-app:nice-transforms}.
Then, in Section~\ref{subsec-app:full-proofs} we give the full details of the induction constructing a "structured signature automaton" for a "positional" language.

\subsection{Nice transformations of automata}\label{subsec-app:nice-transforms}

To recursively build a "structured signature automaton", we will apply a sequence of transformations to a given "$d$-signature automaton" $\A$, by removing states and adding or redirecting edges in such a way that relations $\eqRes{x}$ are preserved in a strong sense formalised in this section.
% Before presenting the proofs of the results from Section~\ref{subsec:from-HP-to-signature},  we introduce some notations and prove technical lemmas that will come in handy to reason about ("structured") "signature automata".

\begin{globalHyp*}%{Subsection~\ref{subsec-app:nice-transforms}}
	We recall that automata are assumed to be "complete" and "semantically deterministic". 
\end{globalHyp*}

%\paragraph*{Congruence properties of signature automata}
%
%The following lemma directly follows from a simple induction and Items~\ref{item-struct:uniformity=x} and~\ref{item-sig:monotonicity>x} of the definition of "signature automaton".
%
%\begin{lemma}\label{lemma-tech:congruence-signature-aut}
%	Let $\A$ be a "$d$-signature automaton" and let $0\leq x\leq d$ be an even priority. Let $w\in \SS^*$ be a finite word.
%	If $q\lrp{w:\geq x} p$ (resp. $q\lrp{w:x} p$) is a path from $q$ producing as minimal priority $y\geq x$ (resp. $x$), for any other path from $q$ over $w$, $q\lrp{w} p'$ the minimal priority produced is $\geq x$ (resp. $x$), and moreover, $p\eqSig{x}p'$.
%\end{lemma}

\paragraph*{Automata with a common subautomaton}

Let $\A$ be a "semantically deterministic" automaton with states $Q_\A$, and let $\eqRes{\A}$ be the "congruence" given by the equality of "residuals". We note that if $\B$ is an automaton with states $Q_\B\subseteq Q_\A$, relation $\eqRes{\A}$ induces an equivalence relation over $Q_\B$ (which, in general, is not a "congruence" nor coincides with the equality of "residuals" of $\B$).

\begin{lemma}[Automata preserving the structure of residuals]\label{lemma-tech:subautomaton-in-common}
	Let $\A$ be a "semantically deterministic" "parity automaton" with states $Q_\A$ and let $\B$ be a "parity automaton" with states $Q_\B\subseteq Q_\A$. Assume that $\eqRes{\A}$ is a "congruence" over $\B$ and that $\quotAut{\B}{\A} = \quotAut{\A}{\A}$. Let $\A'$ be a "subautomaton" of both $\A$ and $\B$.
	If a word $w\in \SS^\oo$ admits an "accepting run" in $\B$ that eventually remains in $\A'$, then $w$ is also "accepted by" $\A$.
\end{lemma}

\begin{proof}
	Let 
	\[q_0\re{w_0}q_1\re{w_1}q_2\re{w_2}\dots\re{w_{k-1}} q_k\re{w_k}q_{k+1}\re{w_{k+1}} \dots\]
	 be an "accepting run" over $w$ in $\B$ such that the suffix from $q_k$ is contained in $\A'$ (meaning that both the states and the transitions used are part of $\A'$). 
	 We consider the "projection@@quot" of the prefix of size $k$ of the run in the "quotient automaton" $\quotAut{\B}{\A} = \quotAut{\A}{\A}$.  By Lemma~\ref{lemma-prelim:induced-run-quotient}, there is a "run over" $w_0\dots w_{k-1}$ in $\A$, $p_0\re{w_0}p_1\re{w_1}\dots p_k$ whose "projection@@quot" over the "quotient automaton" coincides with the previous one. Therefore, $p_k\eqRes{\A} q_k$, that is, $\Lang{\initialAut{\A}{p_k}} = \Lang{\initialAut{\A}{q_k}}$. Since $w_{k+1}w_{k+2}\dots$ admits an "accepting run" from $q_k$ in $\A$, it also admits an "accepting run" from $p_k$, and $w$ is "accepted@@aut" by~$\A$.
\end{proof}

\paragraph*{Nice transformations of automata}

\AP For $x\in \NN$, we denote by $\intro*\autGeq{\A}{x}$ the "subautomaton" of $\A$ induced by the set of transitions using a priority $\geq x$.

\begin{definition}[{Nice transformation at level $x$}]\label{def:nice-transform}
	Let $\A$ be a "semantically deterministic" "parity automaton" over $\SS$ with states $Q$, let $x$ be a "priority", and let  $\sim$ be a "$[0,x-1]$-faithful congruence" over $\A$. 
	Let $\A'$ be a "parity automaton" over $\SS$ with states $Q'\subseteq Q$. We denote $\sim$ the induced relation over $Q'$.
	We say that $\A'$ is a ""$\sim$-nice transformation of $\A$ at level $x$"" if:
	\begin{itemize}
		%\item $\A'$ is a "simple nice transformation" of $\A$,
		\item $\sim$ is a "$[0,x-1]$-faithful congruence" over $\A'$ and $\autLeq{\A}{}{x-1} = \autLeq{\A'}{}{x-1}$ (see Definition~\ref{def:quotient-faithful}),	
		\item $\eqRes{\A}$ is a "congruence" over $\A'$ and $\quotAut{\A'}{\A} = \quotAut{\A}{\A}$, and
%		\item $\autGeq{\A'}{x+1}$ is a "subautomaton" of $\autGeq{\A}{x+1}$.
%	\end{itemize}
%	\AP We say that the transformation is ""exact"" if, moreover:
%	\begin{itemize}	
		\item $\autGeq{\A'}{x+1}$ coincides with the "subautomaton" of $\autGeq{\A}{x+1}$ induced by the states in $Q'$.
	\end{itemize}	
	 
\end{definition}

Intuitively, if $\A'$ is a "nice transformation" of $\A$ at level $x$, it means that the only relevant modifications applied to $\A$ concern "$x$-transitions@@out". The structure of the "quotient automaton@@leq" for priorities ${<}x$ is left unchanged, and so is the acceptance of runs that eventually only produce priorities ${>}x$. 

\begin{remark}
	We note that if $\sim$ is an equivalence relation that "refines" $\eqRes{\A}$, then the second item of Definition~\ref{def:nice-transform} is implied by the first one. This is in particular the case if $\sim$ is an equivalence relation $\eqSig{x}$ of a "$d$-signature automaton", for $0\leq x\leq d$.
\end{remark}

\begin{lemma}[Preservation of classes and priorities in nice transformations]\label{lemma-tech:paths-in-exact-transform}
	Let $\A$ be a "semantically deterministic" "parity automaton" equipped with  a "$[0,x-1]$-faithful congruence" $\sim$. Suppose that $\A$ is "deterministic over" ${>}x$-transitions.
	Let $\A'$ be a "$\sim$-nice transformation of $\A$ at level $x$". If $q\sim q'$ are two states of $\A'$ such that there is path $q\lrp{w:y}p$ in $\A$, then a path $q'\lrp{w:y'}p'$ in $\A'$ satisfies:
	\begin{itemize}
		\item $p\sim p'$, 
		\item if $y<x$, then $y'=y$, and
		\item if $y\geq x$, then $y'\geq x$.
	\end{itemize}
\end{lemma}

%\begin{note}
%	We remark that we suppose that these paths start in the same state, and not just in two $\sim$-equivalent states.
%\end{note}

\begin{proof}
	The equivalence $p\sim p'$ follows from the fact that $\sim$ is a "congruence" in $\A'$ and $\quotAut{\A'}{\A} = \quotAut{\A}{\A}$.
	For $y\leq x-1$ or $y'\leq x-1$, the equality $y'=y$ follows from the equality of the "${\leq}(x-1)$-quotient automaton". 
	This directly implies the third item.
\end{proof}

\begin{lemma}\label{lemma-tech:faithful-congruences-determines-priority}
	Let $\A$ be a "$\ee$-completion" that admits  a "$[0,x]$-faithful congruence"~$\sim$. If a word $w\in \SS^\oo$ admits a run such that the minimal priority produced infinitely often is $y\leq x$, then the  minimal priority produced infinitely often by any "run over $w$" is $y$. In particular, if $y$ is odd, $w$ "is rejected with priority" $y$.
\end{lemma}

\begin{proof}
	Let $q_0\re{w_1:y_1}q_1\re{w_2:y_2}\dots $ and $q_0'=q_0\re{w_1:y_1'}q_1'\re{w_2:y_2'}\dots $ be two runs over the same word in $\A$ ($q_0=q_0'$ being the "initial state" of $\A$). Since $\sim$ is a congruence, we obtain by induction that $q_i\sim q_i'$ for every $i$. Moreover, as, for $y\leq x$,  $y$-transitions "act uniformly" over $\sim$-classes, each time that $y_i\leq x$, we have that $y_i'=x_i$. 
\end{proof}

%\begin{lemma}
%	Let $\A$ be a "semantically deterministic" "parity automaton" equipped with  a "$[0,x-1]$-faithful congruence" $\sim$, and let $\A'$ be a "$\sim$-nice transformation of $\A$ at level $x$". If $\A$ is "homogeneous", then $\sim$ is moreover "$[0,x]$-faithful" in both $\A$ and $\A'$.
%\end{lemma}
%\begin{proof}
%	content
%\end{proof}

\begin{lemma}[Nice transformations preserve acceptance of most words]\label{lemma-tech:nice-transformations}
	Let $\A$ be a "semantically deterministic" "parity automaton" equipped with  a "$[0,x-1]$-faithful congruence" $\sim$.
	Let $\A'$ be a "$\sim$-nice transformation of $\A$ at level $x$". We have:
	\begin{itemize}
		\item A word $w\in \SS^\oo$ can be "accepted with an even priority" $y<x$ in $\A'$ if and only if $w$ can be "accepted with priority" $y$ in $\A$.
		\item A word $w\in \SS^\oo$ is "rejected with an odd priority" $y<x$ in $\A'$ if and only if $w$ is "rejected with priority" $y$ in $\A$.
		\item If a word $w\in \SS^\oo$ can be "accepted with an even priority" $y>x$ in $\A'$, then it is "accepted by@@aut" $\A$. % can be "accepted with a priority" $y'>x$ in $\A$.
	\end{itemize}	
If moreover $\A$ is "homogeneous" and "deterministic over" transitions using priorities $> x$, we have:
\begin{itemize}
	\item If there is a "rejecting@@run" run over $w\in \SS^\oo$  in $\A'$ producing as minimal priority $y>x$, then $w$ is "rejected@@aut" by $\A$. 
\end{itemize}
\end{lemma}

\begin{proof}
	Let $w\in \SS^\oo$ be a word "accepted with a priority" $y<x$ in $\A$ (resp. $\A'$). Then, by Lemma~\ref{lemma-sig:language-quotient-aut}, $w$ is "accepted by" the "quotient automaton@@leq" $\autLeq{\A}{}{x-1} = \autLeq{\A'}{}{x-1}$. Again by  Lemma~\ref{lemma-sig:language-quotient-aut}, $w$ is "accepted by" $\A'$ (resp. $\A$). The second item is obtained using the same argument, combined with Lemma~\ref{lemma-tech:faithful-congruences-determines-priority}.
	
	The third item is directly implied by Lemma~\ref{lemma-tech:subautomaton-in-common}, as $\autGeq{\A'}{x+1}$ is a "subautomaton" of $\autGeq{\A}{x+1}$.
	%If $w\in \SS^\oo$ can be "accepted with a priority" $y>x$ in $\A'$, then, 
	%by Lemma~\ref{lemma-tech:faithful-congruences-determines-priority}, no "run" over $w$ produces a priority  $\leq x$ infinitely often. by Lemma~\ref{lemma-tech:subautomaton-in-common}, $\A$ accepts $w$. Moreover, it cannot accept $w$ with a priority $<x$, as the first item of this lemma would lead to a contradiction.  allows us to conclude.
	
	For the last item, let $q_0\lrp{w_0}q'\lrp{w'}$ be a "rejecting@@run" "run over" $w$ in $\A'$ such that the suffix $q'\lrp{w'}$ does not produce any priority $\leq x$ (that is, it is contained in $\autGeq{\A'}{x+1}$). By determinism and "homogeneity", this is the only "run over" $w'$ from $q'$ in $\A$, and therefore $w'\notin\Lang{\initialAut{\A}{q'}}$. We conclude using the equality $\quotAut{\A'}{\A} = \quotAut{\A}{\A}$.
\end{proof}

%\begin{lemma}[Nice transformation preserve the normal form]\label{lemma-tech:nice-transform-normal-form}
%	Let $\A$ be a "semantically deterministic" "parity automaton" equipped with  a "$[0,x-1]$-faithful congruence" $\sim$ that is moreover "homogeneous" and has the following property: for every pair of states $q\sim q'$, $q\neq q'$, there is a path $q\lrp{w:\geq x} q'$.
%	Let $\A'$ be a "$\sim$-nice transformation of $\A$ at level $x$".
%	Then, if $\A$ is in "normal form", so is $\A'$.
%\end{lemma}
%\begin{proof}
%	We use the characterisation provided by Theorem~\ref{th-p1-norm:}; we suppose without loss of generality that the minimum priority used by $\A$ is $0$.
%\end{proof}

\subsection{From positionality to structured signature automata: Full proofs for Section \texorpdfstring{\ref{subsec:from-HP-to-signature}}{5.2}}\label{subsec-app:full-proofs}

We provide all the technical details for the proofs of the propositions appearing in Section~\ref{subsec:from-HP-to-signature}. 
We first state some useful simple lemmas.

\begin{lemma}\label{lemma-app:positional-implies-pos-quotient}
	Let $W\subseteq \SS^\oo$ be "positional" over finite, "$\ee$-free" "Eve-games". Then, for every word $u\in\SS^*$,  objective $\lquotW{u}$ is "positional" over finite, "$\ee$-free" "Eve-games".
\end{lemma}
\begin{proof}
	Any game with vertices $V$ witnessing non-positionality of $\lquotW{u}$ can be turned into a game witnessing non-positionality of $W$ by adding, for every $v\in V$, a fresh vertex $v_u$ and a path $v_u\lrp{u}v$.
\end{proof}

\begin{lemma}\label{lemma-app:simple-normal-form}
	To check if a "parity automaton" is in "normal form", it suffices to verify that, if $q$ and $p$ are in a same "positive@@SCC" (resp. "negative@@SCC") "SCC" and there is a transition $q\re{x} p$ producing "priority" $x>0$ (resp. $x>1$), then there are two paths  $p\lrpE q$ producing as minimal "priority" $x$ and $x-1$, respectively.
\end{lemma}
\begin{proof}
	We do the proof for the case of a "positive SCC".
	Assume that there is a path $q=q_0 \re{x_1} q_1 \re{x_2} \dots \re{x_n} q_n = p$ with $x = \min x_i >0$. By hypothesis, for $1\leq i\leq n$, there are paths $q_i \lrp{x_i-1} q_{i-1}$. Concatenating them, we obtain a path $p\lrp{x-1}q$. Iterating this process, we can obtain loops $q\lrpE p \lrpE q$ producing as minimal "priority" any number in $[0,x-1]$. To obtain a path $p \lrp{x}q$, we just use the existence of paths $q_i\lrp{x_i} q_{i-1}$.
\end{proof}

\subsubsection{Base case: Total order of residual classes}

\begin{lemma}[Total order of residual classes]\label{lemma-app:total-order-residuals}
	Let $W\subseteq \SS^\oo$ be an "$\oo$-regular" objective that is "positional" over finite, "$\ee$-free" "Eve-games". Then, $\Res{W}$ is totally ordered by inclusion.
\end{lemma}
\begin{proof}
	The proof is almost identical to that of Lemma~\ref{lemma-warm:total-order-residuals-nec}.
	We show the contrapositive.
	Assume that $W$ has two incomparable "residuals", $\lquotW{u_1}$ and $\lquotW{u_2}$. 
	We consider first the case $u_1\neq \ee$ and $u_2\neq \ee$.
	Take $w_1\in \lquot{u_1}{W}\setminus\lquot{u_2}{W}$ and $w_2\in \lquot{u_2}{W}\setminus\lquot{u_1}{W}$. Thanks to $\oo$-regularity, we can take these words of the form $w_i = u_i'(w_i')^\oo$, with $u_i',w_i'\in \SS^+$, for $i=1,2$.
	We have
	
	\begin{tabular}{l l}
		\centering
		$u_1w_1\in W$, & $u_1w_2\notin W$,\\ 
		$u_2w_1\notin W$, & $u_2w_2\in W$. 
	\end{tabular} 
	
	\vspace{2mm}
	Consider the "$\ee$-free", finite, "Eve-game" $\G$ represented in Figure~\ref{fig-appendix:game-residuals}.
	"Eve" "wins" $\G$ from $v_1$ and $v_2$: if a "play" starts in $v_i$, for $i=1,2$, she just has to take the path labelled $u_i'(w_i')^\oo$ from $v_{\mathsf{choice}}$. However, she cannot win from both $v_1$ and $v_2$ using a "positional strategy". Indeed, such "positional strategy" would choose one path $v_{\mathsf{choice}} \lrp{u_i(w_i')}$, and the "play" induced when starting from $v_{1-i}$ would be "losing@@play".
	
	\begin{figure}
		\centering
		\includegraphics[scale=1.5]{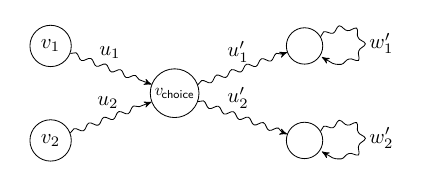}
		\caption{A game $\G$ in which "Eve" cannot "play optimally using positional strategies" if $\Res{W}$ is not totally ordered, as in the proof of Lemma~\ref{lemma-app:total-order-residuals}. }
		\label{fig-appendix:game-residuals}
	\end{figure}

	Finally, we take care of the case in which $\resClass{u_1} = \{\ee\}$ (symmetric for $u_2$). In that case, we cannot take $u_1\neq \ee$. We remark that, since $\resClass{u_1} \neq \resClass{u_2}$, we can take $u_2\neq \ee$.
	We consider the game from Figure~\ref{fig-appendix:game-residuals} in which we simply remove vertex $v_1$. This game is "$\ee$-free", and "Eve" can win from both $v_2$ and $v_{\mathsf{choice}}$, but not "positionally@@win".
\end{proof}

\begin{globalHyp*}%{Subsection~\ref{subsec-app:nice-transforms}}
	In all the rest of the subsection, we assume that $x\geq 2$.
\end{globalHyp*}

\subsubsection{Safe centrality and relation \texorpdfstring{$\eqSig{x-1}$}{x-1}. Proof of Lemma \texorpdfstring{\ref{lemma-nec:safe-centralisation}}{5.16}}

In this paragraph we give a proof of Lemma~\ref{lemma-nec:safe-centralisation}. We assume that $x\geq 2$ is an even priority and $\A$ is a "deterministic" "$(x-2)$-structured signature automaton" with initial state $q_\init$.

\lemnecSafeCentralisation* \label{lemma-nec:safe-centralisation-pageApp}
%\begin{lemmaRest}{lemma-nec:safe-centralisation}{$({<}x)$-safe centralisation}\label{lemma-nec:safe-centralisation-app}
%	There exists a "$(x-2)$-structured signature" automaton $\A'$ "equivalent to" $\A$ which is:
%	\begin{itemize}
%		\item "deterministic over" transitions with "priority" different from $x-1$,
%		\item "homogeneous",
%		\item "history-deterministic", and
%		\item "$({<}x)$-safe centralised". 
%	\end{itemize}  
%	Moreover, $\A'$ can be obtained in polynomial time from $\A$ and $\sizeAut{\A'}\leq \sizeAut{\A}$.
%\end{lemmaRest}

\subparagraph{Hypothesis.} During the proof, we will lose the "determinism" of $\A$. However, in all the subsection we will maintain the three first required properties. In the statements of all lemmas, $\A$ will stand for an "$(x-2)$-structured signature" automaton that is:
\begin{itemize}
	\item "deterministic over" transitions with "priority" different from $x-1$,
	\item "homogeneous", and
	\item "history-deterministic".
\end{itemize}

%Our objective is to obtain, in polynomial time and without increasing the number of states, an "equivalent automaton" satisfying the same list of properties that is moreover "$({<}x)$-safe centralised".

\subparagraph{Saturation.}
\AP We say that an automaton $\A'$ is ""$(x-1)$-saturated"" if for every state $q$ and letter $a\in \SS$, if a transition $q\re{a:x-1} p$ exists in $\A'$, then $q\re{a:x-1} p'$ appears in $\A'$ for all $p'\eqSig{x-2} p$.\footnotemark{}
\AP The "$(x-1)$-saturation" of $\A$ is the automaton obtained by adding all those transitions.

\footnotetext{We note that this definition slightly differs from the definition of "$1$-saturation" used in the warm-up (Section~\ref{subsec-warm:coBuchi}). In particular, the definition of the warm-up does not preserve "homogeneity". We allow ourselves these small disagreements of definitions for the sake of clarity in the presentation in each respective subsection.}

\begin{remark}\label{rmk-tech:saturation-homogeneous}
	The "$(x-1)$-saturation" of $\A$ is "homogeneous" and "deterministic over" transitions with "priority" different from $x-1$.
\end{remark}

\begin{lemma}\label{lemma-tech:saturation-normal-form}
	The "$(x-1)$-saturation" of $\A$ is in "normal form".
\end{lemma}
\begin{proof}
	We use the characterisation of "normal form" given in Lemma~\ref{lemma-app:simple-normal-form}. Let $\A'$ be the "$(x-1)$-saturation" of $\A$. 
	The property of Lemma~\ref{lemma-app:simple-normal-form} is satisfied for transitions already appearing in $\A$, as $\A$ is assumed to be in "normal form".
	Let  $q\re{a:x-1}p$ be a transition added by the saturation process, and let $q\re{a:x-1} p'$ be a transition in $\A$ with $p\eqSig{x-2} p'$. By Item~\ref{item-struct:even-classes-connected} of the definition of "structured signature automaton", there is a path $p\lrp{u: >x-2} p'$ in $\A$. By "normality" of $\A$, there are also paths $p'\lrp{u_1:x-1} q$ and $p'\lrp{u_2:x-2} q$. We obtain the two desired paths in $\A'$:
	\[p \lrp{u:>x-2} p' \lrp{u_1:x-1}q \;,  \quad \tand \quad  p \lrp{u:>x-2} p' \lrp{u_2:x-2}q.\qedhere\]
\end{proof}

The following lemma states that "$(x-1)$-saturation" is a "$\eqSig{x-2}$-nice transformation at level $x-1$", so Lemmas~\ref{lemma-tech:nice-transformations} and~\ref{lemma-tech:paths-in-exact-transform} can be applied. We recall that $\eqSig{x-2}$ "refines" $\eqRes{\A}$, so congruence of $\eqRes{\A}$ is implied by that of $\eqSig{x-2}$.

\begin{lemma}\label{lemma-app:x-1-saturation-same-quotient}
	The  "$(x-1)$-saturation" of $\A$  is a "$\eqRes{x-2}$-nice transformation of $\A$ at level $x-1$".
%	\begin{itemize}
%		\item $\eqSig{x-2}$ is a "$[0,x-2]$-faithful congruence" in $\A'$ and $\autLeq{A'}{x-2} = \autLeq{A}{x-2}{x-2}$, 
%		%\item $\eqRes{\A}$ coincides over $\A$ and over $\A'$, and
%		\item $\autGeq{\A}{x} = \autGeq{\A'}{x}$.
%	\end{itemize}% and $\autGeq{\A'}{x} = \autGeq{\A}{x}$.
\end{lemma}
\begin{proof}
	Let $\A'$ be the "$(x-1)$-saturation" of $\A$.
	It is immediate that $\autGeq{\A}{x} = \autGeq{\A'}{x}$. Moreover, the restriction of $\A$ and $\A'$ to transitions using priorities $\leq x-2$ coincides, so for each $0\leq y\leq x-2$, relation $\eqRes{x-2}$ is a "congruence for" "$y$-transitions" in $\A'$ (and therefore these transitions "act uniformly" by "determinism"). The "congruence for" transitions using priority ${>}(x-2)$ is preserved as we have only added $(x-1)$-transitions that go to the same $\eqSig{x-2}$-class.
	As $\eqSig{x-2}$ "refines" $\eqRes{\A}$, the latter relation is also a "congruence" in $\A'$ and $\quotAut{\A}{\A} = \quotAut{\A'}{\A}$.
\end{proof}

\begin{lemma}\label{lemma-app:x-1-saturation}
  The  "$(x-1)$-saturation" of $\A$ "recognises" $\Lang{\A}$. Moreover, it is "history-deterministic", "homogeneous" and "deterministic over transitions" using "priorities" different from $x-1$.
\end{lemma}
\begin{proof}
	Let $\A'$ be the  "$(x-1)$-saturation" of $\A$. We have already noted that it is "homogeneous" and "deterministic over transitions" using "priority" different from $x-1$ (Remark~\ref{rmk-tech:saturation-homogeneous}).
	%To show $\Lang{\A'}= \Lang{\A}$ we use Lemma~\ref{lemma-tech:nice-transformations}:
	If $w\in \SS^\oo$ is "accepted by" $\A'$ (resp. by $\A$), it is either "accepted with" an even priority $y<x-1$ or $y>x-1$. In the first case, since  $\A'$ is a "$\eqRes{x-2}$-nice transformation at level $x-1$", Lemma~\ref{lemma-tech:nice-transformations} allows us to conclude. In the second case, it suffices to apply Lemma~\ref{lemma-tech:subautomaton-in-common}.

	"History-determinism" of $\A'$ is clear: one can use a "resolver@@aut" for $\A$.
\end{proof}

\subparagraph{Redundant safe components.}
From now on, we suppose that $\A$ is "$(x-1)$-saturated".

\AP We say that a "$({<}x)$-safe component" $S$ of $\A$ is ""redundant@@sig"" if there is $q\in S$ and $q'\eqSig{x-2}q$, $q'\notin S$, such that $\safeSig{x}{q}\subseteq \safeSig{x}{q'}$. %We say in this case that $q'$ ""witnesses@@redundant"" that $S$ is redundant.
We note that, by "normality" of $\A$, there are no "$({\geq}x)$-transitions" entering in $S$; that is, there are no transitions $p\re{a:\geq x} q$ with $p\notin S$ and $q\in S$.

\begin{remark}
	Automaton $\A$ is "$({<}x)$-safe centralised" if and only if it does not contain any "redundant@@sig" "$({<}x)$-safe component".
\end{remark}

\begin{lemma}\label{lemma-tech:find-redundant-poly}
	If $\A$ contains some "redundant@@sig" "$({<}x)$-safe component", we can find one of them in polynomial time.
\end{lemma}
\begin{proof}
	The computation of the "$({<}x)$-safe components" of $\A$ can be done by simply a decomposition in "SCC" of $\autGeq{\A}{x}$.
	For each pair of states $q\eqSig{x-2}q'$ in different "$({<}x)$-safe components" we just need to check the inclusion $\safeSig{x}{q}\subseteq \safeSig{x}{q'}$, which can be done in polynomial time.
\end{proof}

\begin{lemma}
	Let $S$ be a "redundant@@sig" "$({<}x)$-safe component" of $\A$, and let $S'$ be a different "$({<}x)$-safe component" such that there are $q_0\in S$ and  $q_0'\in S'$, with $q_0\eqSig{x-2}q_0'$ and $\safeSig{x}{q_0}\subseteq \safeSig{x}{q_0'}$.
	Then, for each $q\in S$ there is $q'\in S'$ such that $q\eqSig{x-2}q'$ and $\safeSig{x}{q}\subseteq \safeSig{x}{q'}$.
\end{lemma}
\begin{proof}
	For each $q\in S$, pick $u\in \SS^*$ such that $q_0\lrp{u:\geq x} q$. We let $q'$ be such that $q_0'\lrp{u} q'$. Since $u\in \safeSig{x}{q_0}\subseteq \safeSig{x}{q_0'}$, this latter path produces priority $\geq x$ and, by "normality", $q'$ is in $S'$. 
	As $\eqSig{x-2}$ is a "congruence for" "$({\geq}x-2)$-transitions", $q'\eqSig{x-2}q$. By monotonicity for safe languages (Lemma~\ref{lemma-sig:monotonicity-safe-lang}), we also have $\safeSig{x}{q}\subseteq \safeSig{x}{q'}$.
\end{proof}

\subparagraph{Removing redundant safe components.}
For now on, fix $S$ to be a "redundant@@sig" "$({<}x)$-safe component" of $\A$, and $S'$ a different "$({<}x)$-safe component" as in the previous lemma.
\AP For each $q\in S$, we let $\intro*\pick(q)\in S'$ such that $q\eqSig{x-2}\pick(q)$ and $\safeSig{x}{q}\subseteq \safeSig{x}{\pick(q)}$. 
We extend $\pick$ to all of $Q$ by setting it to be the identity over $Q\setminus S$.
%Therefore, for all states $q$, $q\eqSig{x-2}\pick(q)$ and $\safeSig{x}{q}\subseteq \safeSig{x}{\pick(q)}$.

We define the automaton $\A'$ as follows:
\begin{itemize}
	\item The set of states is $Q' = Q\setminus S$. %is obtained by removing $S$ from $\A$.
	\item The "initial state" is $\pick(q_\init)$.
	\item \AP For each $p\in Q'$, if $p\re{a: y} q$ is a transition in $\A$, we let $p\re{a:y}\pick(q)$ in $\A'$.
\end{itemize}

We note that, if $q\notin S$, all transitions $p\re{a: y} q$ in $\A$ are left unchanged. In particular, the "$({\geq}x)$-transitions" of $\A'$ are the restriction of those appearing in $\A$, $\A'$ is "$(x-1)$-saturated", and "$({<}x)$-transitions" in $\A$ not entering in $S$ appear in $\A'$ too.
\AP We say that transitions $p\re{a: y} q$ of $\A$ such that $q\in S$ have been ""redirected in $\A'$@@sc"". 

%\begin{remark}
%	We note that $|\A'|<|\A|$.
%\end{remark}

\begin{lemma}\label{lemma-tech:safe-centr-is-nice-transform}
	Automaton $\A'$ is a "$\eqSig{x-2}$-nice transformation of $\A$ at level $x-1$".
\end{lemma}
\begin{proof}
	We have that $\autGeq{\A'}{x}$ is the "subautomaton" of $\autGeq{\A}{x}$ induced by states in $Q'$. We show that $\eqSig{x-2}$ is "$[0,x-2]$-faithful" in $\A'$. 
	Transitions that have not been "redirected@@sc" satisfy the congruence requirements, as they satisfy them in $\A$.
	Let $p\re{a:y}q$ and $p'\re{a:y'}q_0'$ be two transitions in $\A'$ such that $y\leq x-2$, $p\eqSig{x-2} p'$ and such that the second transition has been redirected from $p'\re{a:y'}q'$ (the first transition is possibly a redirected one too).  
	By the congruence property in $\A$, we have that $y'=y$ and $q\eqSig{x-2}q'$. Since $q'\eqSig{x-2}q_0'$, we conclude by transitivity.
	The equality $\autLeq{\A'}{x-2}{x-1}=\autLeq{\A}{x-2}{x-1}$ simply follows from the fact that "redirected@@sc" transitions have been defined preserving the $\eqRes{x-2}$-classes.
	
	As $\eqSig{x-2}$ "refines" $\eqRes{\A}$, the latter relation is also a "congruence" in $\A'$ and $\quotAut{\A}{\A} = \quotAut{\A'}{\A}$. 
\end{proof}

\begin{lemma}[Correctness of the removal of redundant components]\label{lemma-tech:equality-languages-safe-cent}
	For every state $q'\in Q'$, we have $\Lang{\initialAut{\A'}{q'}} = \Lang{\initialAut{\A}{q'}}$.
	In particular, these automata are "equivalent@@aut".
	Moreover, automaton $\A'$ is "deterministic over" transitions with "priority" different from $x-1$, "homogeneous" and "history-deterministic". 
\end{lemma}
\begin{proof}
	The fact that $\A'$ is "deterministic over" transitions with "priority" different from $x-1$ and "homogeneous"  is immediate from its definition. 
	We show the equality of languages for the "initial state". The proof is identical for a different state.
	
	The inclusion $\Lang{\A'}\subseteq \Lang{\A}$ directly follows from Lemma~\ref{lemma-tech:nice-transformations}.
	
	We describe a "sound resolver" witnessing $\Lang{\A}\subseteq \Lang{\A'}$ and "history-determinism". Take a "sound resolver" $\resolv$ in $\A$, let $w\in \SS^\oo$, and write
	\[		\rho=p_0 \re{w_0} p_1 \re{w_1} \dots\]
	for the run in $\A$ "induced by@@resolver" $\resolv$ over $w$.
	We will build a "resolver" $\resolv'$ in $\A'$ satisfying the property that the "run induced over $w$@@res", $\rho'=p'_0 \re{w_0} p'_1 \re{w_1} \dots$ is in one of the following (non-excluding) cases:
	\begin{enumerate}[label=\alph*),ref=\alph*]
		\item\label{it-lem-app:CaseA} produces priorities $<x-1$ infinitely often, 
		\item\label{it-lem-app:CaseB} eventually produces only priorities $\geq x$,
		\item \label{it-lem-app:CaseC} $p_i \eqSig{x-2} p'_i$ and $\safeSig{x}{p_i}\subseteq \safeSig{x}{p_i'}$ for every $i$ sufficiently large.
	\end{enumerate}

	\begin{claim}
		A "resolver" $\resolv'$ satisfying the property above accepts all words in $\Lang{\A}$.
	\end{claim}
	\begin{subproof}
		Suppose that $w\in \Lang{\A}$, that is, the run $\rr$ "induced by@@resolv" $\resolv$ is "accepting@@run". Let $\rr'$ be the "run induced over $w$@@res" by $\resolv'$ in $\A'$. We distinguish two cases, according to the priorities produced by the run $\rr$ in $\A$:
		
		If $\rr$ produces priorities $<x-1$ infinitely often. Then Lemma~\ref{lemma-tech:nice-transformations} allows us to conclude.

		If $\rr$ eventually only produces priorities $\geq x-1$.
		%Then, as $\rr$ is accepting, in fact it eventually only produces priorities $\geq x$.  
		Then, the two first items of Lemma~\ref{lemma-tech:nice-transformations} tells us that $\rr'$ eventually only produces priorities $\geq x-1$ too (so we are not in Case~\ref{it-lem-app:CaseA}).
		We show that $\rr'$ eventually only produces priorities $> x-1$; the last item of Lemma~\ref{lemma-tech:nice-transformations} allows us to conclude (we recall that $\A'$ is a "$\eqSig{x-2}$-nice transformation at level $x-1$").
		If we are in Case~\ref{it-lem-app:CaseB}, this property is trivially satisfied. Suppose that we are in Case~\ref{it-lem-app:CaseC}, and let $k>0$ be such that the suffix of $\rr$ from $q_k$ only produces priorities $\geq x$ and such that $\safeSig{x}{p_i}\subseteq \safeSig{x}{p_i'}$ for $i\geq k$. Therefore, there is a "run over $w_kw_{k+1}\dots$" from $p_k'$ producing exclusively priorities $\geq x$. By "determinism over transitions" with priority $\geq x$, this run is the one "induced by" $\resolv'$.
	\end{subproof}
	
	We finally show how to construct a "resolver" with this property.
	We let $p'_0=f(p_0)$, and assume $\rho'$ constructed up to $p'_i$ satisfying that $p'_i\eqSig{x-2} p_i$.
			
	If there is a transition $p'_i \re{w_i:y} p'_{i+1}$ with $y\neq x-1$, then we take this one (there is no other option),  which satisfies $p_{i+1} \eqSig{x-2} p_{i+1}'$ by Lemma~\ref{lemma-tech:safe-centr-is-nice-transform}.
	Moreover, if $y\geq x$ and $\safeSig{x}{p_i}\subseteq \safeSig{x}{p_i'}$, then $\safeSig{x}{p_{i+1}}\subseteq \safeSig{x}{p_{i+1}'}$ by Lemma~\ref{lemma-sig:monotonicity-safe-lang}.
	If there is a transition $p'_i \re{w_i:x-1}$,  we take $p_{i+1}'=f(p_{i+1})$ (this transition exists in $\A'$ by "$x-1$-saturation", as $p'_i\eqSig{x-2} p_i$).
	
	We show that this "resolver" satisfies the desired property. Suppose that we are not in the two first cases, that is,  $\rr'$ eventually only produces priorities $\geq x-1$, and it produces priority $x-1$ infinitely often. Take a suffix $p'_k \re{w_k:y_k} p'_{k+1} \re{w_{k+1}:y_{k+1}} \dots$ of $\rr'$ such that no priority $<x-1$ is produced and such that $y_k=x-1$. Then, by definition of the transitions using $x-1$ chosen by the "resolver", $p_{k+1}'=f(p_{k+1})$, so $\safeSig{x}{p_{k+1}}\subseteq \safeSig{x}{p_{k+1}'}$. We conclude by induction, as transitions taken by the resolver using priorities $\geq x-1$ preserve the inclusion of "$({<}x)$-safe languages".
\end{proof}

\subparagraph{Transformation preserves being a structured signature automaton.}
To be able to finish the proof of Lemma~\ref{lemma-nec:safe-centralisation}, we just need to show that $\A'$ is a "$(x-2)$-structured signature automaton". We give some technical lemmas that will help us show this.

\begin{lemma}\label{lemma-tech:paths-x-1-transfer-in-safe-centralisation}
	Let $q$ and $p$ be two states of $\A'$. There is a path $q\lrp{w:x-1}p$ in $\A$ if and only if there is a path $q\lrp{w:x-1}p$ in $\A'$.
\end{lemma}
\begin{proof}
	If a path $q\lrp{w:x-1}p$ appears in $\A'$, the very same path also exists in $\A$.
	
	Suppose now that a path $\rr = q\lrp{w:x-1}p$ exists in $\A$.
	Let $S$ be the "${<}x$-safe component" that has been removed from $\A$. If the path $\rr$ does not cross $S$, then it also appears in $\A$. Suppose that it enters in $S$. We remark that, by "normality" and by the definition of "safe component@@sig", each time that $\rr$ enters or exists $S$, it produces priority $x-1$. We consider the last time that $\rr$ enters and exists $S$:
	\[ q\lrp{u_1:\geq x-1} q_1 \lrp{u_2:\geq x} q_2 \re{a:x-1} q_3 \lrp{u_3:\geq x-1}  p, \]
	with $w=u_1u_2au_3$, $q_3\notin S$, and the path $q_3 \lrp{u_3}  p$ does not enter $S$ (so it also appears in~$\A'$). Consider any "run over" $u_1u_2$ from $q$ in $\A'$: 
	\[ q\lrp{u_1:\geq x-1} q_1' \lrp{u_2:\geq x-1} q_2'.\]
	As $\eqSig{x-2}$ is a "$[0,x-2]$-faithful congruence", we have that $q_2\eqSig{x-2}q_2'$. As $\A'$ is "$x-1$-saturated", there is a transition $q_2'\re{a:x-1} q_3$. Therefore, we obtain in $\A'$ the path:
	\[ q\lrp{u_1:\geq x-1} q_1' \lrp{u_2:\geq x-1} q_2' \re{a:x-1} q_3 \lrp{u_3:\geq x-1}  p.\qedhere\]
\end{proof}

\begin{lemma}\label{lemma-tech:paths-transfer-in-safe-centralisation}
	Let $q$ and $p$ be two states of $\A'$ and let $y$ be any "priority". There is a path $q\lrp{w:y}p$ in $\A$ if and only if there is a path $q\lrp{w':y}p$ in $\A'$.
\end{lemma}
\begin{proof}
	If $y\geq x$, $q$ and $p$ are in the same "$({<}x-1)$-safe component". Since the "$({<}x-1)$-safe components" in $\A'$ are safe components in $\A$, the result is clear in this case.
	
	If $y=x-1$, the result is assured by the previous Lemma~\ref{lemma-tech:paths-x-1-transfer-in-safe-centralisation}.

	Assume $y<x-1$. %, and let $y'=y-1$ if $y$ is odd, and $y'=y$ if it is even.
	Suppose that there is a path $q\lrp{w:y}p$ in $\A$ (the proof is analogous if we take this path in $\A'$), and let  $q\lrp{w:y'}p'$ be the "run over" $w$ from $q$ in $\A'$.
	As $\A'$ is a "$\eqSig{x-2}$-nice transformation", by Lemma~\ref{lemma-tech:paths-in-exact-transform}, we have that $y'=y$ and $p\eqSig{x-2}p'$. As $\A$ satisfies Item~\ref{item-struct:even-classes-connected} from the definition of a "structured signature automaton", there is a path $p'\lrp{u:\geq x-1}p$ in $\A$. By Lemma~\ref{lemma-tech:paths-x-1-transfer-in-safe-centralisation}, such a path also exists in $\A'$, so we can take $w'=wu$, giving us a path $q\lrp{w:y}p'\lrp{u:>y}p$.
\end{proof}

Previous lemma tells us, in particular, that for every "$({<})y$-safe component" $\safeComp{y}_i$ of $\A$, the intersection of $\safeComp{y}_i$ with $Q'$ constitute the states of a "$({<})y$-safe component" in $\A'$. 
Therefore, automaton $\A'$ inherits the decomposition in "$({<})y$-safe component" $\safeComp{y}_1,\dots, \safeComp{y}_{k_y}$ from $\A$, for each $y$; we simply remove those components whose intersection with $Q'$ is empty.

\begin{lemma}\label{lemma-tech:safe-centr-preserves-structured}
	Automaton $\A'$ is a "$(x-2)$-structured signature automaton".
\end{lemma}
\begin{proof}
	We go through all the conditions of the definition of a "$(x-2)$-structured signature automaton".
	We recall that, by Lemma~\ref{lemma-tech:safe-centr-is-nice-transform}, the relation $\eqSig{x-2}$ is "$[0,x-2]$-faithful congruence"~in~$\A'$.
	\begin{description}
		\item[Normal form.] We check that $\A'$ satisfies the hypothesis of the characterisation from Lemma~\ref{lemma-app:simple-normal-form}. We let $q'\re{a:y} p'$ be a transition in $\A'$. If it is not a "redirected transition@@safec", it exists in $\A$, so we can conclude by "normalisation" of $\A$ and Lemma~\ref{lemma-tech:paths-transfer-in-safe-centralisation}. 
		Assume that $q'\re{a:y} p'$ is a transition that has been "redirected from@@safec" $q'\re{a:y}p$. In particular, $p\eqRes{x-2}p'$ and $y<x$. By Item~\ref{item-struct:even-classes-connected} of the definition of "structured signature automaton" applied to $\A$, there is a path $p'\lrp{w:>y}p$ in $\A$, and by "normality", there is a returning paths $p\lrp{u_1:y}q$ and $p\lrp{u_2:y-1}q$. Again, Lemma~\ref{lemma-tech:paths-transfer-in-safe-centralisation} allows us to find the desired returning paths in $\A'$.
		
		\item[Item~\ref{item-struct:preorder-0}.] By Lemma~\ref{lemma-tech:equality-languages-safe-cent}, the residuals in $\A'$ correspond to those in $\A$, so "preorder" $\leqSig{0}$ correspond to their inclusion in $\A'$ too.
		
		\item[Item~\ref{item-struct:odd-orders}.] By the previous remarks, for all $y$, the "$({<}y)$-safe components" of $\A'$ are obtaining by taking the intersection with those in $\A$. Therefore, odd preorders $\leqSig{y-1}$ correspond to the order of "$({<}y)$-safe components" on $\A'$.
		
		\item[Item~\ref{item-struct:even-preorders}.] As $\A'$ is a "$\eqSig{x-2}$-nice transformation at level $x-1$", for every $y<x-1$ and state $q'$ in $\A'$, $\safeSigAut{y}{q'}{\A'} = \safeSigAut{y}{q'}{\A}$. Therefore, the preorders at even levels $\leqSig{y}$ correspond to the inclusion of safe languages in $\A'$, as they do in $\A$.
		
		%\item[Item~\ref{item-struct:uniformity=x}.] The $\eqSig{y}$-classes of $\A'$ are the intersection of those of $\A$ with $Q'$. The "uniformity@@trans" of "$y$-transitions", for $y\leq x-2$ even, is therefore inherit from $\A$.
		
		\item[Item~\ref{item-struct:strong-congruence}.] Let $q$ and $q'$ be two states in $\A'$ such that $q\eqSig{y} q'$, for $y\leq x-2$ even, and let $q\re{a:y'}p$ for $y'\leq y$ and $q'\re{a:z}p'$. 
		As $\A'$ is a "nice transformation", $y'=z$.
		If the first of these transitions is not a "redirected@@safec" one, then, by "strong congruence" of ${\leq}(x-2)$-priorities in $\A$, neither is the second one, and $p=p'$. 
		Assume that these transitions have been "redirected@@safec" from, $q\re{a:y'}p_1$ and $q'\re{a:y'}p'_1$. Then, as Item~\ref{item-struct:strong-congruence} is satisfied in $\A$, $p_1=p'_1$, so $p=\pick(p_1)=\pick(p_1')=p'$.
		
		\item[Item~\ref{item-struct:even-classes-connected}.] Directly follows from Lemma~\ref{lemma-tech:paths-transfer-in-safe-centralisation} and the fact that $\A$ satisfies this property.
		
		\item[Item~\ref{item-struct:safe-centralisation}.] Follows from the fact that if $q\neqSig{y-1}p$ in $\A'$, then $q\neqSig{y-1}p$ in $\A$; and the equality $\safeSigAut{y}{q}{\A'} = \safeSigAut{y}{q}{\A}$.\qedhere
 	\end{description}
\end{proof}

\subparagraph{Obtaining Lemma \texorpdfstring{\ref{lemma-nec:safe-centralisation}}{5.16}.}
We have all the necessary elements to deduce Lemma~\ref{lemma-nec:safe-centralisation}. Using Lemma~\ref{lemma-tech:find-redundant-poly}, we can decide whether $\A$ contains a "redundant@@safec" "$({<}x)$-safe component" in polynomial time.
If it contains none, $\A$ is already "$({<}x)$-safe centralised".
While we can find "redundant@@safec" safe components, we remove them applying the transformation described above. This transformation can clearly be done in polynomial time, and by Lemmas~\ref{lemma-tech:equality-languages-safe-cent} and~\ref{lemma-tech:safe-centr-preserves-structured}, the obtained automaton "recognises" the correct language and preserves all the hypothesis assumed in the induction.

\subsubsection{Existence of uniform words and synchronising separating runs}

We now provide proofs for Lemmas~\ref{lemma-nec:existence-uniform-words} and~\ref{lemma-nec:syncr-separating-runs}.
In Section~\ref{subsec:from-HP-to-signature}, we derived totality of the order $\leqSig{x}$ in each $\eqSig{x-1}$-component from these lemmas (c.f. Lemma~\ref{lemma-nec:total-order-safe}).

\subparagraph{Hypothesis.} In all the subsection we assume that $x$ is an even priority and $\A$ is a "$(x-2)$-structured signature automaton" with initial state $q_\init$ that is moreover:
\begin{itemize}
	\item "deterministic over" transitions with "priority" different from $x-1$,
	\item "homogeneous", 
	\item "history-deterministic", and
	\item "$({<}x)$-safe centralised".
\end{itemize}  

\subparagraph{Words producing priority \texorpdfstring{$x$}{x} uniformly.}

\lemNecExistenceUniformWords*\label{lemma-nec:existence-uniform-words-pageApp}
%\begin{lemmaRest}{lemma-nec:existence-uniform-words}{Existence of uniform words}\label{lemma-nec:existence-uniform-words-app}
%	Let $p$ and $q$ be two states from the same "$({<}x)$-safe component".
%	There is a word $w \in \SS^*$ "producing priority $x$ uniformly" in $\classSig{x}{q}$ "leading to@@class" $\classSig{x}{p}$.%; that is, for any $q' \in \classSig{x}{q}$, we have $q' \lrp {w:x} p'$ with $p' \in \classSig{x}{p}$.
%\end{lemmaRest}

\begin{proof}%[Proof of Lemma~\ref{lemma-nec:existence-uniform-words}]
	The fact that for such paths we must have $p_1\eqSig{x}p_2$ follows from the monotonicity of "safe languages@@sig" (Lemma~\ref{lemma-sig:monotonicity-safe-lang}) and the fact that $({\geq}x)$-transitions preserve $\eqSig{x-1}$-classes.
	
	We let $\{q_1,q_2,\dots,q_k\}$ be an enumeration of the states of $\classSig{x}{q}$.
	\begin{claim}%{customclaim}
		For each $q_i\in \classSig{x}{q}$ there is a word $u_i\in \SS^*$ such that $q_i\lrp{u_i:x}q_i$.
	\end{claim}
	\begin{subproof}
		Since $q_i$ and $q$ are in the same "$({<}x)$-safe component", there is a word $u_i'$ such that $q_i\lrp{u_i':\geq x} q$. By "normality", there is a word $u_i''$ such that $q\lrp{u_i'':x}q_i$. We just take $u_i = u_i'u_i''$.
	\end{subproof}
	We will  define $k$ finite words $w_1,w_2\dots, w_k\in \SS^*$ satisfying:
	\begin{itemize}
		\item For $q'\in \classSig{x}{q}$ and for every $i\leq k$, $q' \lrp{w_1w_2\dots w_{i}:\geq x} \classSig{x}{q}$.
		\item $q_j \lrp{w_1w_2\dots w_{i}: x} \classSig{x}{q}$ for every $j\leq i$.
	\end{itemize}
	In order to obtain these properties, we just define recursively $w_1=u_1$ and $w_i =u_j$, for $u_j$ as given by the previous claim,  if:
	\[  q_i \lrp{w_1w_2\dots w_{i-1}:\geq x} q_j.\]
	%	Thus, for $i\in{1,\dots,k}$, reading $w_1w_2\dots w_i$ from $q_i$, produces priority $x$.
	%	This implies that $w_1w_2\dots w_i$ produces "priority" $x$ when read from any $q_j$ for $j\leq i$.
	Finally, we let $w=w_1w_2\dots w_kw$, which first produces priority $x$ when read from any state of $\classSig{x}{q}$, and then goes to the $\eqSig{x}$-class of $p$ producing priorities ${\geq}x$.	
\end{proof}

\subparagraph{Resolvers implemented by finite memories.}
For the upcoming proofs, we need to introduced the notion of memories for "resolvers".
\AP Let $\A = (Q,\Sigma, q_\init, \DD, \colAut)$ be a "non-deterministic automaton". A ""memory structure@@aut"" for $\A$ is a tuple $(M,m_\init,\intro*\transMem,\nextmoveResolver)$, where $M$ is a set of memory states, $m_\init\in M$ is an initial state, $\transMem: M \times \DD \to M$ is an update function and $\intro*\nextmoveResolver\colon Q\times M \times \SS \to \DD$.
\AP It ""implements@@resolver"" a "resolver@@aut" $\resolv$ if  for all $a\in \SS$, $\resolv(\ee,a) = \nextmoveResolver(q_\init,m_\init,a)$ and for all $\rr\in \DD^+$ ending in $p$, $\resolv(\rr,a) = \nextmoveResolver(p, \transMem(m_\init, \rr), a)$.

\begin{lemma}[\cite{Boker2013NondetUnknownFuture}]\label{lemma-p0-prelim:finite-memory-resolvers}
	Every "history-deterministic@@aut" "$\ee$-completion" admits a "sound resolver@@aut" "implemented by@@resolver" a "finite memory structure@@resolv".
\end{lemma}

%We recall (Lemma~\ref{lemma-p0-prelim:finite-memory-resolvers} from Section~\ref{sec-p0-prelim:automata}) that any "history-deterministic" parity automaton admits a "sound resolver" "implemented by@@resMem" a finite memory.
In the rest of the subsection we fix a "sound resolver" $\resolv$ for~$\A$ "implemented by@@resolvMem" a "memory structure@@aut" $\M=(M, m_\init, \transMem)$.
For simplicity, we assume that every pair of a state and a memory state $(q,m)$ is ""reachable using $\resolv$@@memory"", that is, there is some word $w\in \SS^*$ such that $ \rr = q_\init\lrpResolver{\resolv}{w} q$ and $\transMem(m_\init, \rr) = m$. It is easy to see that we can get rid of this assumption in the upcoming proof just by ignoring pairs $(q,m)$ that are not "reachable@@res".

\AP For $(q,m)\in Q\times M$, we let $(q,m)\intro*\lrpResolverMem{\resolv}{w} (q',m')$ be the (unique) "run induced by $\resolv$" from $q$ when the memory structure is in state $m$.
We extend notations of the form $(q,m)\lrpExistsResolver{\resolv}{w:x} q'$ in the natural way; the previous one means that there exists $u_0\in \SS^*$ such that the induced run of $\resolv$ is $\rr = q_\init\lrpResolver{\resolv}{u_0} q$, $\transMem(m_\init, \rr) = m$ and $q_\init\lrpResolver{\resolv}{u_0w} q'$, producing priority $x$ in the second part of this run.

As $\A$ is "deterministic over" transitions using priorities $\geq x$, we may omit the subscript $\resolv$ in paths producing no priority $<x$.

\subparagraph{Synchronising separating runs.}
\lemNecSynchronisation*\label{lemma-nec:syncr-separating-runs-pageApp}
%\begin{lemmaRest}{lemma-nec:syncr-separating-runs}{Synchronisation of separating runs}\label{lemma-nec:syncr-separating-runs-app}
%	%Let $\A$ be a "$x$-structured signature automaton", and 
%	Suppose that $q \eqSig{x-1} q'$ and $q \nleqSig{x} q'$ and let $p\in \classSig{x-1}{q}$.
%	There is a word $w\in \SS^+$ such that $\classSig{x}{q} \lrpAllResolver{\resolv}{w:x-1} \classSig{x}{p}$ and $\classSig{x}{q'} \lrpAllResolver{\resolv}{w:x} \classSig{x}{p}$.
%\end{lemmaRest}

\begin{proof}
	We first show that we can force to produce priority $x-1$ from $\classSig{x}{q}$, while remaining safe from $\classSig{x}{q}$.
	\begin{claim}\label{cl-app:producing-x-1}%{customclaim}
		 There is a word $u\in \SS^+$ such that for all $s\in \classSig{x}{q}$:
		  \[s \lrpAllResolver{\resolv}{u:x-1} \,, \quad \tand\quad \classSig{x}{q'}\lrpAllResolver{\resolv}{u: x}\classSig{x}{p}.\]
	\end{claim}
	\begin{subproof}
		By definition of the preorder $\leqSig{x}$, there is a word $u_1\in \SS^+$ such that for all $s\in \classSig{x}{q}$ and $s'\in \classSig{x}{q'}$, $s \lrpAllResolver{\resolv}{u_1:x-1}$ and $s'\lrpResolver{\resolv}{u_1:\geq x}$.
		By "normality", we can extend this word so that $q'\lrpResolver{\resolv}{u_1':\geq x} q'$; by monotonicity of safe languages all runs from $\classSig{x}{q'}$ reading $u_1'$ go back to $\classSig{x}{q'}$.
		Applying  Lemma~\ref{lemma-nec:existence-uniform-words}, we obtain $u_2$ that "produces priority $x$ uniformly" in $\classSig{x}{q'}$ and goes to $\classSig{x}{p}$.
		We take $u=u_1'u_2$, which satisfies $\classSig{x}{q'} \lrpResolver{\resolv}{u: x}\classSig{x}{p}$, and it produces at least one occurrence of priority $x-1$ from every state $s\in\classSig{x}{q}$.
		By "$[0,x-2]$-faithfulness", a run $s \lrp{u}$ only produces priorities $\geq x-1$, which concludes.
	\end{subproof}
	
	%Next claim states a result that is quite close to our objective. It differentiates from it in two aspects: (1) we state it for a single state $q$ and a single "memory state", not for a full $\eqSig{x}$-class; (2) we only require the state $p'$ to be $\eqSig{x-2}$-equivalent to $q$, instead of $\eqSig{x-1}$-equivalent.
	
	\begin{claim}\label{cl-app:synchronising-one-state}%{customclaim}
		Let $p' \eqSig{x-2} q$ and let $m\in M$ be a memory state. Then there is a word $w_{q,m}$ such that:
		\[ (q,m)\lrpResolverMem{\resolv}{w_{q,m}:\geq x-1} \classSig{x}{p'}, \, \tand \;\;  \classSig{x}{p'}\lrpAllResolver{\resolv}{w_{q,m}: x}\classSig{x}{p'}.\]
	\end{claim}
	\begin{subproof}
		We distinguish two cases. First, assume that $\safeSig{x}{q}\subseteq \safeSig{x}{p'}$. In this case, by "${<}x$-safe centrality" of $\A$, $q$ and $p'$ are in the same "${<}x$-safe component" so $q\eqSig{x-1}p'$ (and $q\leqSig{x}{p'}$).
		Let $p'_{\max}$ be a state in $\classSig{x-1}{p'}$ such that $p'\leqSig{x}{p'_{\max}}$, and maximal with this property. Let $w_1\in \SS^*$ be a word such that $q\lrp{w_1:\geq x} p'_{\max}$ (which exists because these two states are in the same "${<}x$-safe component").
		By monotonicity of "safe languages", $p'\lrp{w_1:\geq x} p''$ with $p'_{\max}\leqSig{x}p''$. By maximality of $p'_{\max}$, we must have $p''\eqSig{x}{p'_{\max}}$.
		Finally, let $w_2\in \SS^*$ such that $p''\lrp{w_2:x} p'$ "producing priority $x$ uniformly in" the class $\classSig{x}{p''}$ (which exists by Lemma~\ref{lemma-nec:existence-uniform-words}). We obtain $q\lrp{w_1:\geq x} p'_{\max} \lrp{w_2:x} \classSig{x}{p'}$ and $\classSig{x}{p'}\lrp{w_1:\geq x} \classSig{x}{p'_{\max}} \lrp{w_2:x} \classSig{x}{p'}$ as required.
		
		Assume now that $\safeSig{x}{q}\nsubseteq \safeSig{x}{p'}$. In that case, we can find a word $w_1\in \safeSig{x}{p'}\setminus \safeSig{x}{q}$. By Lemma~\ref{lemma-nec:existence-uniform-words}, we may assume that it "produces priority $x$ uniformly" from $\classSig{x}{p'}$ and comes back to this class.
		Moreover, by "faithfulness", it cannot produce priorities $\leq x-2$ and respects the $\eqSig{x-2}$-classes. Thus:
		\[ (q,m)\lrpResolver{\resolv}{w_1: x-1} (q_1,m_1), \text{ with } q_1\eqSig{x-2}p' \,, \tand \quad \classSig{x}{p'}\lrp{w_1:x} \classSig{x}{p'}. \] 
		If $\safeSig{x}{q_1}\subseteq \safeSig{x}{p'}$, we can conclude by using the first case.
		While we do not have this inclusion, we build, using the argument above, a sequence of words $w_1,w_2,\dots $ such that: 
		\[ (q_i,m_i)\lrpResolver{\resolv}{w_i: x-1} (q_{i+1},m_{i+1}) \,, \tand \classSig{x}{p'}\lrp{w_i:x} \classSig{x}{p'}. \] 
		This sequence cannot be infinite.
		If it were the case, resolver $\resolv$ would "induce@@resolv" a "rejecting run" over $w_1w_2\dots$ from $(q,m)$, and an "accepting@@run" from $p'$.
		This is a contradiction, as the equivalence $q\eqSig{x-2}p$ implies $\Lang{\initialAut{\A}{q}}=\Lang{\initialAut{\A}{p'}}$ (since $\eqSig{x-2}$ refines $\eqRes{\A}$).
		Therefore, for some $k$ we must have $\safeSig{x}{q_k}\subseteq \safeSig{x}{p'}$ and we can extend the path as wanted using the first case.
	\end{subproof}

	We may finally deduce the result of Lemma~\ref{lemma-nec:syncr-separating-runs} from the two previous claims.
	First, we read the word $u$ from Claim~\ref{cl-app:producing-x-1}, which forces to produce priority $x-1$ from any state in $\classSig{x}{q}$. We now show how to use Claim~\ref{cl-app:synchronising-one-state} to redirect each state, one by one, to the class $\classSig{x}{p}$.
	
	We let $(q_1,m_1),\dots, (q_k,m_k)$ be an enumeration of all states such that there exists $s\in \classSig{x}{q}$ with $s \lrpExistsResolver{\resolv}{u:x-1} (q_i,m_i)$.
	We note that, by "$[0,x-2]$-faithfulness", $q_i\eqSig{x-2} q\eqSig{x-2} p$.
	We recursively build a sequence of $k$ words, $w_1,\dots,w_k\in \SS^*$ by setting:
	\[ (q_i,m_i)\lrpResolverMem{\resolv}{w_1\dots w_{i-1}:\geq x-1} (q_i',m_i') \lrpResolverMem{\resolv}{w_i:\geq x-1} \classSig{x}{p} \;   \tand \;\;  \classSig{x}{p}\lrpAllResolver{\resolv}{w_i: x}\classSig{x}{p}. \]
	We can indeed do this by letting $w_i = w_{q_i',m_i'}$ as given by Claim~\ref{cl-app:synchronising-one-state}, as by "$[0,x-2]$-faithfulness" $q_i'\eqSig{x-2}p$.
	
	The word $w_1\dots w_i$ satisfies that, for $j\leq i$,  $(q_j,m_j)\lrpResolverMem{\resolv}{w_1\dots w_{i}:\geq x-1} \classSig{x}{p}$.
	We conclude the proof of the lemma by putting $w'=uw_1\dots w_k$.
\end{proof}

\subsubsection{Re-determinisation}

We give the proof of Lemma~\ref{lemma-nec:re-determinisation}, that is, we show that we can obtain an "equivalent@@aut" "deterministic" automaton from $\A$ while preserving all the obtained structure of "total preorders" satisfying the conditions of a "structured signature automaton".

\subparagraph{Hypothesis.} We assume that $\A$ is a "parity automaton" "recognising" $W$ with "nested total preorders" defined up to $\leqSig{x}$ such that:
\begin{itemize}
	\item it is a "$(x-2)$-structured signature automaton",
	\item preorder $\leqSig{x-1}$ satisfies properties from Items~\ref{item-struct:odd-orders} and~\ref{item-struct:safe-centralisation} from the definition of a "structured signature automaton", %,~\ref{item-struct:even-classes-connected},
	\item preorder $\leqSig{x}$ satisfies the property from Item~\ref{item-struct:even-preorders} from the definition of a "structured signature automaton",
	\item it is "deterministic over" transitions with priorities different from $x-1$,
	\item it is "homogeneous", and
	\item it is "history-deterministic".
\end{itemize}  

\subparagraph{Obtaining a deterministic automaton.}

\lemNecDeterminisation*\label{lemma-nec:re-determinisation-pageApp}
%\begin{lemmaRest}{lemma-nec:re-determinisation}{Re-determinisation}\label{lemma-nec:re-determinisation-app}
%	There is a "deterministic" "parity automaton" $\A'$ "equivalent to" $\A$ with "nested total preorders" defined up to $\leqSig{x}$ satisfying that:
%	\begin{itemize}
%		\item it is a "$(x-2)$-structured signature automaton",
%		\item preorder $\leqSig{x-1}$ satisfies properties from Items~\ref{item-struct:odd-orders} and~\ref{item-struct:safe-centralisation} from the definition of a "structured signature automaton", and %,~\ref{item-struct:even-classes-connected},
%		\item preorder $\leqSig{x}$ is a "congruence" and satisfies the property from Item~\ref{item-struct:even-preorders} and, for priorities $y<x$, also that from Item~\ref{item-struct:strong-congruence} .
%	\end{itemize}  
%	Moreover, automaton $\A'$ can be computed in polynomial time from $\A$ and $\sizeAut{\A'}\leq \sizeAut{\A}$.
%\end{lemmaRest}

An intuitive idea for the construction of $\A'$ was given in Section~\ref{subsec:from-HP-to-signature}. We formalise it and prove its correctness now.

Automaton $\A'$ is obtained by keeping all the structure of $\A$, except for $(x-1)$-transitions; for each state $q$ and letter $a$ such that some transition $q\re{a:x-1}p$ appears in $\A$, we will redefine this transition as $q\re{a:x-1}p'$ for some $p'\eqSig{x-2}p$ as determined next.

For each $\eqSig{x-1}$-class $\classSig{x-1}{q}$ of $\A$, we pick a state in the class that is maximal for $\leqSig{x}$. We let $\intro*\pickMax(q)$ be that state. That is, for two states $q_1\eqSig{x-1}q_2$:
\begin{itemize}
	\item $\pickMax(q_1)=\pickMax(q_2)$, 
	\item $\pickMax(q_1)\in \classSig{x-1}{q_1}$, and
	\item $q_1\leqSig{x} \pickMax(q_1)$.
\end{itemize}

Recall the total order over the "${<}x$-safe components" of $\A$ given by $\safeComp{x}_1,\safeComp{x}_2,\dots, \safeComp{x}_{k_x}$.
Let $q\in \safeComp{x}_i$ and $a\in \SS$ such that $q\re{a:x-1}p$ appears in $\A$.
If it exists, we let $\intro*\inext$ be the maximal $0\leq i_\mnext<i$ such that there is some $p'\in \classSig{x-2}{p}$, $p'\in \safeComp{x}_{i_\mnext}$. If index $\inext$ is not defined, we let it be the maximal index $i\leq \inext\leq k_x$ with the previous property.
We fix a state $p_{q,a}\in \safeComp{x}_{\inext}$ with $p_{q,a}\in \classSig{x-2}{p}$.
We let the "$a$-transition" from $q$ in $\A'$ be $q\re{a:x-1} \pickMax(p_{q,a})$.
This completes the description of $\A'$. It is indeed "deterministic", as $\A$ is "homogeneous" and "deterministic over" transitions with priorities different from $x-1$.

We can find a maximal state $\pickMax(q)$ for $\leqSig{x}$ in the class $\classSig{x-1}{q}$ in polynomial time, as the comparison of "${<}x$-safe languages" can be done in polynomial time~(Lemma~\ref{lemma-nec:safe-languages-polytime}).
Therefore, automaton $\A'$ can be built in polynomial time. 

\begin{lemma}\label{lemma-app:redeterm-nice-transform}
	Automaton $\A'$ is a "$\eqSig{x-2}$-nice transformation of $\A$ at level $x-1$".
\end{lemma}
\begin{proof}
		This is clear, as the restriction of $\A'$ to transitions using a "priority" different from $(x-1)$ coincides with that of $\A$, and every transition $q\re{a:x-1}p'$ in $\A'$ comes from a transition $q\re{a:x-1}p$ in $\A$ with $p\eqSig{x-2}p'$.
\end{proof}

\begin{lemma}\label{lemma-app:redeterm-correctness}
	For every state $q\in Q$, we have $\Lang{\initialAut{\A'}{q}} = \Lang{\initialAut{\A}{q}}$.
	In particular, automaton $\A'$ "recognises" the language $\Lang{\A}$.
\end{lemma}
\begin{proof}
	For simplicity, we give the proof just for the initial state; the proof being identical for any other state.
	%For any other state it suffices to fix a word $u_q$ such that $q_\init \lrpResolver{\resolv}{u_q} q$, for a resolver $\resolv$ in $\A$.
	
	Let $w\in \SS^\oo$. If the minimal "priority" produced infinitely often by the "run over $w$" in $\A'$ is $y<x-1$ or $y>x-1$, then $w$ is "accepted by" $\A'$ if and only if $w$ is "accepted by" $\A$, by Lemma~\ref{lemma-tech:nice-transformations} and the fact that $\A'$ is a "$\eqSig{x-2}$-nice transformation of $\A$ at level $x-1$".
	
	Assume that the minimal "priority" produced infinitely often by the "run over $w$" in $\A'$ is $x-1$ (so it is "rejecting@@run"), and suppose by contradiction that $w$ is "accepted by" $\A$. By Lemma~\ref{lemma-tech:nice-transformations}, an "accepting run" over $w$ in $\A$ cannot produce a priority $y<x$ infinitely often. Therefore, it eventually remains in a "${<}x$-safe component" $\safeComp{x}_{i_\A}$. Let $\rr$ be such an "accepting run", and let $\rr'$ be the "run over $w$" in $\A'$. We represent them as:
	\[ \rr = q_0 \lrp{u\phantom{..}} q_N\re{w_N:\geq x} q_{N+1} \re{w_{N+1}:\geq x}\dots \;, \quad \rr' = q_0' \lrp{u\phantom{..}} q_N' \lrp{w_N: x_N'} q_{N+1}' \re{w_{N+1}: x_{N+1}'}\dots, \]
	where $u$ is the prefix of size $N$ of $w$, $q_0=q_0'=q_\init$, and we suppose that $q_k\in \safeComp{x}_{i_\A}$ for all $k\geq N$.
	As $\A'$ is a "$\eqSig{x-2}$-nice transformation at level $x-1$", we have that $q_k\eqSig{x-2} q_k'$ for all $k$.
	Let $k_1<k_2<k_3 \dots$ be the positions greater than $N$ where $x_{k_j}'=x-1$, and let $i_1,i_2,\dots$ be the indices of the "${<}x$-safe components" such that $q_{k_j+1}\in \safeComp{x}_{i_j}$, that is, when taking the transition $q_{k_j}\re{w_{k_j}:x-1}$ we land in $\safeComp{x}_{i_j}$.
	
	\begin{claim}%{customclaim}
		Eventually, $i_j = i_\A$.
	\end{claim}
	\begin{subproof}		
		Consider a transition $q_{k_j}' \re{w_{k_j}:x-1} q_{k_j+1}'$, and suppose first that $i_\A< i_{j-1}$. 
		We claim that $i_\A\leq i_j < i_{j-1}$. This would end the proof, as we obtain a strictly decreasing sequence of indices bounded by $i_\A$.
		In order to determine $i_j$, we need to look at the definition of $\inext$.	
		As $q_k\eqSig{x-2}q_k'$ for all $k$, there are always states in $\classSig{x-2}{q_{k_j+1}'}$ in some "${<}x$-safe components" with an index  $i_\A\leq i < i_{j-1}$.
		Thus, we obtain the desired result by definition of $\inext$.

		If $i_{j-1}\leq i_\A$, by definition of $\inext$, there is be a sequence of decreasing indices $i_{j-1}>i_j>i_{j+1}>\dots$ until no $\eqSig{x-2}$-equivalent state appears in a strictly smaller safe component. 
		By the same argument as before, there is always a $\eqSig{x-2}$-equivalent state in $\safeComp{x}_{i_\A}$, so eventually $i_\A\leq i_{j'}$, and either this is an equality, or we reduce to the previous case.
	\end{subproof}
	
	Let $j$ be the first position such that $i_j = i_\A$, and consider transitions $q_{k_j} \re{w_{k_j}:\geq x} q_{k_j+1}$ and $q_{k_j}' \re{w_{k_j}:x-1} q_{k_j+1}'$ in $\rr$ and $\rr'$, respectively.
	By definition of "$x-1$-transitions" of $\A'$, the state we go to in $\rr'$ is $q_{k_j+1}'= \pickMax(q_{k_j+1})$.
	As we have chosen $\pickMax(q_{k_j+1})$ maximal in its $\eqSig{x-1}$-class, we have $q_{k_j+1}\leqSig{x} q_{k_j+1}'$, so we have the inclusion of "${<}x$-safe languages" between these states.
	Therefore, if there is a "${<}x$-safe run" over $w'$ from $q_{k_j+1}$ in $\A$, there is also such a "safe run@@sig" over $w'$ from $q_{k_j+1}'$ in $\A'$. 
	This contradicts the fact that the run $\rr'$ produces priority $x-1$ infinitely often, while the run $\rr$ is "${<}x$-safe@@run" from $q_{k_j+1}$, concluding the proof.	
\end{proof}

To finish the proof of Lemma~\ref{lemma-nec:re-determinisation}, we just need to show that $\A'$ preserves all the properties	of the preorders induced by $\A$. 
As in the previous section, to obtain "normality" of the automaton and Item~\ref{item-struct:even-classes-connected}, we rely on a technical lemma that tells us that we can connect states in the same $\eqSig{x-2}$-component as desired.
	
\begin{lemma}
	Let $q\eqRes{x-2}p$ be two different states in $Q$. There is a word $w\in \SS^*$ and a path $q\lrp{w:> x-2} p$ in $\A'$
\end{lemma}
\begin{proof}
	Let $i_p$ be the index such that $p\in \safeComp{x}_{i_p}$. If $q$ belongs to this same "safe component@@sig", we can connect both states by a path producing priorities $\geq x$.
	If not, by "${<}x$-safe centrality" of $\A$, there is a word $w_1\in \safeSig{x}{p}\setminus \safeSig{x}{q}$. We let $q\lrp{w_1:x-1}q_1$ and $p\lrp{w_1:\geq x} p_1$. We have that $q_1\eqSig{x-2} p_1$ and $p_1\in \safeComp{x}_{i_p}$. While $q_j\notin \safeComp{x}_{i_p}$, we extend this run in a similar way.
	Using the same argument as in the proof of the previous lemma, by definition of $\inext$ each time that the run from $q$ sees a priority $x-1$ it decreases the index of its "safe component@@sig", so eventually it must land in $\safeComp{x}_{i_p}$.
\end{proof}

\begin{lemma}\label{lemma-app:connecting-paths-in-redeterminisation}
	Let $q\lrp{w:y}p$ be a path in $\A$. There is a path $q\lrp{w':y} p$ in $\A'$ connecting the same states and producing the same minimal priority.
\end{lemma}
\begin{proof}
	If $y\geq x$ we can just take $w=w'$.
	Suppose $y<x$ and consider the run $q\lrp{w:y'}p'$ in $\A'$. By Lemma~\ref{lemma-tech:paths-in-exact-transform}, $p\eqSig{x-2}p'$. Also, $y=y'$ (if $y<x-1$, this is given by Lemma~\ref{lemma-tech:paths-in-exact-transform}, if $y=x-1$, by the definition of $\A'$).
	By the previous lemma, we can extend this run to $p'\lrp{w_2:\geq x-2}p$, and take $w' =ww_2$.	
\end{proof}

\begin{lemma}%\label{lemma-app:redeterm-correctness}
	Automaton $\A'$, with the preorders $\leqSig{0},\dots, \leqSig{x}$ inherited from $\A$ satisfies:
	\begin{itemize}
		\item it is a "$(x-2)$-structured signature automaton",
		\item preorder $\leqSig{x-1}$ satisfies properties from Items~\ref{item-struct:odd-orders} and~\ref{item-struct:safe-centralisation} from the definition of a "structured signature automaton", and %,~\ref{item-struct:even-classes-connected},
		\item preorder $\leqSig{x}$ satisfies the property from Item~\ref{item-struct:even-preorders} and, for priorities $y<x$, also that from Item~\ref{item-struct:strong-congruence}.
	\end{itemize}  
\end{lemma}
\begin{proof}
	We start verifying the properties for the preorders $\leqSig{x-1}$ and $\leqSig{x}$.
	We note that the "${<}x$-safe components" of $\A'$ exactly correspond to those in $\A$, and that for every $q\in Q$, $\safeSigAut{x}{q}{\A}=\safeSigAut{x}{q}{\A'}$. The fact that $\leqSig{x-1}$ satisfies Items~\ref{item-struct:odd-orders} and~\ref{item-struct:safe-centralisation}, and that  $\leqSig{x}$ satisfies Item~\ref{item-struct:even-preorders} follows immediately. 
	We check that relation $\eqSig{x}$ satisfies Item~\ref{item-struct:strong-congruence} for priorities $y<x$. For $y\leq x-2$, this follows from the fact that $\eqSig{x}$ refines $\eqSig{x-2}$, and the latter relation satisfies Item~\ref{item-struct:strong-congruence}. For $y=x-1$, if $q\re{a:x-1} p$ and $q'\re{a:x-1} p'$, with $q\eqSig{x}q'$, by definition of the "$x-1$-transitions" in $\A'$, $p = \pickMax(p) = \pickMax(p') = p'$.
	
	Checking that $\A'$ is a "$(x-2)$-structured signature automaton" poses no difficulty. It can be done in an analogous way as it was done in the proof of Lemma~\ref{lemma-tech:safe-centr-preserves-structured}; by applying Lemma~\ref{lemma-app:connecting-paths-in-redeterminisation} to obtain "normality" of $\A'$ and property from Item~\ref{item-struct:even-classes-connected}.	
\end{proof}

\subsubsection{Uniformity of \texorpdfstring{$x$}{x}-transitions over \texorpdfstring{$\eqSig{x}$}{x-1}-classes}

We finally show how to transform $\A$ into an "equivalent@@aut" automaton that is either "$x$-structured signature", or strictly smaller. The techniques presented here generalise those applying to B\"uchi automata appearing  in Section~\ref{subsec-warm:Buchi} of the warm-up.

\subparagraph{Hypothesis.} In all this subsection we suppose that $\A$ is "deterministic" "parity automaton" "recognising" $W$ with "nested total preorders" defined up to $\leqSig{x}$ such that:
\begin{itemize}
	\item it is a "$(x-2)$-structured signature automaton",
	\item preorder $\leqSig{x-1}$ satisfies properties from Items~\ref{item-struct:odd-orders} and~\ref{item-struct:safe-centralisation} from the definition of a "structured signature automaton", and %,~\ref{item-struct:even-classes-connected},
	\item preorder $\leqSig{x}$ is a "congruence" and satisfies the property from Item~\ref{item-struct:even-preorders} and, for priorities $y<x$, also that from Item~\ref{item-struct:strong-congruence} .
\end{itemize}

Our objective is to prove, under this list of hypothesis, that we can either obtain an "equivalent@@aut" "deterministic" "$x$-structured signature automaton", or reduce the number of states~of~$\A$.

\lemNecUniformity*\label{lemma-nec:uniformity-x-transitions-pageApp}
%\begin{lemmaRest}{lemma-nec:uniformity-x-transitions}{Uniformity of $x$-transitions over $\eqSig{x}$-classes}\label{lemma-nec:uniformity-x-transitions-app}
%	There is a "deterministic" "parity automaton" $\A'$ "equivalent to"~$\A$ such that either:
%	\begin{itemize}
%		\item $\A'$ is an "$x$-structured signature automaton" with $\sizeAut{\A'}\leq \sizeAut{\A}$, or
%		\item $\sizeAut{\A'}< \sizeAut{\A}$.
%	\end{itemize}
%	In both cases, such an automaton can be computed in polynomial time from $\A$.
%\end{lemmaRest}

We remark that $\eqSig{x}$ already satisfies most desired properties of "monotonicity@@cong"; only the "uniformity for" $x$-transitions is missing.

\begin{lemma}\label{lemma-app:monotonicite-preorder-x}
	The relation $\eqSig{x}$ is a "$[0,x-1]$-faithful congruence". Moreover, over each $\eqSig{x-1}$-class, transitions using priorities ${\geq}x$ are "monotone for" $\leqSig{x}$. 
\end{lemma}
\begin{proof}
	The "$[0,x-2]$-faithfulness" follows from the fact that $\eqSig{x}$ refines $\eqSig{x-2}$ and $\A$ is "$(x-2)$-structured signature".
	The "uniformity@@cong" of $(x-1)$-transitions over $\eqSig{x}$-classes is given by the fact that $\eqSig{x}$-equivalent states have the same "${<}x$-safe language", combined with the "uniformity@@cong" of ${<}(x-1)$-transitions.
	
	The fact that $\eqSig{x}$ is a "congruence for" "$(x-1)$-transitions" follows from Item~\ref{item-struct:strong-congruence} of the definition of a "structured signature automaton" (we recall that $\eqSig{x}$ satisfies this property by hypothesis).
	The "congruence for" ${\geq}x$-transitions and the monotonicity of ${\geq}x$-transitions for $\leqSig{x}$ at each $\eqSig{x-1}$-class follow from the monotonicity of "${<}x$-safe languages" (Lemma~\ref{lemma-sig:monotonicity-safe-lang}).
\end{proof}

\subparagraph{Polished automata.} We generalise the notion of "polished automata@@buchi" from Section~\ref{subsec-warm:Buchi} to our current setting.
Recall that $\classSig{x}{q}$ is the class of $q$ for the equivalence relation $\eqSig{x}$.

\begin{definition}[Polished classes and automata]
	%Let $\A$ be a "parity automaton" and let $\sim$ be a "congruence" over $\A$. Let $x$ be a "priority" and $q$ a state of $\A$.
	\AP We say that the class $\classSig{x}{q}$ is ""$x$-polished@@class""~if:
	\begin{itemize}
		\item Words producing priority $x$ "act uniformly" in $\classSig{x}{q}$. That is, if $q_1,q_2\in \classSig{x}{q}$ and $q_1\lrp{w:x}$, then $q_2\lrp{w:x}$.
		%\item Incoming  "$x$-transitions" go to a same state: if $q_1,q_2\in \classSig{x}{q}$, $p_1\re{a:x}q_1$ and $p_2\re{a:x}q_2$, then $q_1=q_2$.
		\item For every $q_1,q_2\in \classSig{x}{q}$, $q_1\neq q_2$, there is a path $q_1\lrp{w:>x}q_2$ producing exclusively priorities $>x$ joining $q_1$ and $q_2$.
	\end{itemize}
	\AP We say that the automaton $\A$ is ""$x$-polished@@aut"" if all its $\eqSig{x}$-classes are "$x$-polished@@class".
\end{definition}

\begin{remark}
	We remark that, as $x\geq 2$ and we assume that the automaton $\A$ is in "normal form", all non-trivial $\eqSig{x}$-classes are "recurrent@@class": if a $\eqSig{x}$-class is not trivial, there is a cycle visiting all the states of the class.
	Therefore, we do not need to take care of "transient classes" (as it was the case in Lemma~\ref{lemma-warm:polished-automaton} from the warm-up).
\end{remark}

\begin{lemma}
	We can decide whether $\A$ is "$x$-polished@@aut" in polynomial time.
\end{lemma}
\begin{proof}
	As $\eqSig{x}$ is a "$[0,x-1]$-faithful congruence" (Lemma~\ref{lemma-app:monotonicite-preorder-x}), we just need to check the first property for letters, which can be done in linear time in $|\Sigma||\A|$.
	
	For the second property, we just need to check whether, for each $q\in Q$, the "subautomaton" induced by $\classSig{x}{q}$ and transitions with priority ${>}x$ is strongly connected.
\end{proof}

\paragraph*{Case 1: \texorpdfstring{$\A$}{A} is already \texorpdfstring{$x$}{x}-polished} Assume that $\A$ is "$x$-polished@@aut". In this case, it is almost an "$x$-structured signature automaton". We just need to ensure that if $q\eqSig{x} q'$, two transitions $q\re{a:x}p$ and $q'\re{a:x}p'$ go to a same state $p=p'$.

We remark that $\eqSig{x}$ already satisfies most desired properties of "monotonicity@@cong"; only the "uniformity for" $x$-transitions is missing.

%\begin{lemma}\label{lemma-app:x-faithful-if-polished}
%	If $\A$ is "$x$-polished@@aut", relation $\eqSig{x}$ is a "$[0,x]$-faithful congruence".
%\end{lemma}
%\begin{proof}
%	Directly follows from Lemma~\ref{lemma-app:monotonicite-preorder-x} and the fact that $\A$ is "$x$-polished@@aut".
%\end{proof}

In order to obtain the "strong congruence" of $x$-transitions (Item~\ref{item-struct:strong-congruence}), we redirect some $x$-transitions of $\A$.
For each $\eqSig{x}$-class $\classSig{x}{q}$, pick an arbitrary state $\pickP(q)\in \classSig{x}{q}$. (Formally, $\intro*\pickP\colon Q \to Q$ such that $\pickP(q)=\pickP(q')$ if $q\eqSig{x}q'$). We let $\A'$ be the automaton obtained as follows:
\begin{itemize}
	\item The states of $\A'$ are the same than those in $\A$.
	\item Transitions using priorities different from $x$ are those in $\A$.
	\item If $q\re{a:x}p$, we let $q\re{a:x}\pickP(p)$ in~$\A'$.
\end{itemize}

It is immediate to check that $\A'$ is a "$\eqSig{x}$-nice transformation of $\A$ at level $x$". Moreover, $\autGeq{\A}{x+1} = \autGeq{\A'}{x+1}$. These remarks directly give:

\begin{lemma}
	$\A'$ is "$x$-polished@@aut".
\end{lemma}

\begin{lemma}\label{lemma-app:same-paths-if-polished}
	There is a path $q\lrp{w:y}p$ in $\A$ if and only if there is a path $q\lrp{w':y}p$ in~$\A'$.
\end{lemma}
\begin{proof}
	We suppose that there is $q\lrp{w:y}p$ in $\A$ (the converse proof is symmetric).
	If $y>x$, we have that $q\lrp{w:y}p$, as $\autGeq{\A}{x+1} = \autGeq{\A'}{x+1}$. 
	If $y\leq x$, as $\A'$ is a "nice transformation at level $x$", we have that $q\lrp{w:y}p'$ in $\A'$, with $p\eqSig{x}p'$. As $\A'$ is "$x$-polished@@aut", there is a path $p'\lrp{w_2:>x} p$. We conclude by taking $w'=ww_2$.
\end{proof}

\begin{lemma}
	Automaton $\A'$ is "equivalent to@@aut" $\A$, and it is an "$x$-structured signature automaton".
\end{lemma}
\begin{proof}
	The fact that $\Lang{\A} = \Lang{\A'}$ follows easily using that $\A'$ is a "$\eqSig{x}$-nice transformation of $\A$ at level $x$" and applying Lemma~\ref{lemma-tech:nice-transformations}.
	Lemma~\ref{lemma-app:same-paths-if-polished} implies that $\A'$ is in "normal form".
	Verifying that $\A'$ is an "$x$-structured signature automaton" is just a routine check, using that $\A'$ is "$x$-polished@@aut" and Lemma~\ref{lemma-app:same-paths-if-polished}.
\end{proof}

\paragraph*{Case 2: Polishing a class}

Assume now that there is a class $\classSig{x}{q}$ that is not "$x$-polished@@class" in $\A$. We show that we can remove some states from this class, obtaining an strictly smaller "equivalent@@aut" automaton.

\subparagraph{Local languages and local automata.}
\AP We define the $x$-local alphabet at $\classSig{x}{q}$ by
\[
\intro*\locAlph{x}{q} = \{w\in \SS^+ \mid \classSig{x}{q} \lrp{w:\geq x} \classSig{x}{q} \text{ and for any proper prefix } w' \text{ of } w,  \classSig{x}{q} \lrp{w:\geq x} \classSig{x}{p}\neq \classSig{x}{q} \}.
\]

We remark that, as $\eqSig{x}$ is a "congruence for" "${\geq}x$-transitions" (Lemma~\ref{lemma-app:monotonicite-preorder-x}),  $\locAlph{x}{q}$ is well-defined and the notation $\classSig{x}{q} \lrp{w:\geq x}$ can be used. A word $w\in \SS^*$ belongs to $\locAlph{x}{q}^*$ if and only if it connects states in the class $\classSig{x}{q}$. Elements in $\locAlph{x}{q}$ are those that do not pass twice through this class.
Note that $\locAlph{x}{q}$ is a "prefix code", and therefore it is "uniquely decodable" (even if, in general, it is infinite).

\AP Seeing words in $\locAlph{x}{q}^\omega$ as words in $\SS^\omega$, define the localisation of $W$ to~$\classSig{x}{q}$ to be the objective
\[
\intro*\locLang x q = \{ w \in \locAlph{x}{q}^\omega \mid w \in \Lang{\initialAut{\A}{q}}\}.
\]
Observe that, as $\eqSig{x}$ "refines" $\eqRes{\A}$, this last definition does not depend on the choice of $q$ and $\locLang x q$ is "prefix-independent". 
Moreover, $\locLang x q$ is "positional" over finite, "$\ee$-free" "Eve-games": any "$\locLang x q$-game" in which "Eve" could not "play optimally using positional strategies" would provide a counterexample for the "positionality" of $\Lang{\initialAut{\A}{q}}$, which is "positional" if $W$ is (Lemma~\ref{lemma-app:positional-implies-pos-quotient}).

\AP The ""local automaton of the class $\classSig{x}{q}$"" is the automaton $\intro*\localAutSig{x}{q}$ defined as:
\begin{itemize}
	\item The set of states is $\classSig{x}{q}$.
	\item The initial state is arbitrary.
	\item For $w\in \locAlph{x}{q}$, $q_1\re{w:y}q_2$ if $q_1\lrp{w:y}q_2$ in $\A$ (we must have $y\geq x$).
\end{itemize}

\subparagraph{Super words and super letters for local languages.}
We recall some terminology introduced in the warm-up. Assume that $L$ is a "prefix-independent" language.
\AP We say that $u\in \SS^+$ is a ""super word"" for $L$ if, for every $w\in \SS^\oo$, if $w$ contains $u$ infinitely often as a factor, then $w\in L$. 
\AP If $s$ is a letter, we say that it is a ""super letter"".

\AP For $q$ a state and $x$ an even priority, we let $\intro*\BLetters{x}{q}\subseteq \locAlph{x}{q}$ be the set of "super letters for" $\locLang{x}{q}$, and we write $\intro*\NLetters{x}{q} = \locAlph{x}{q}\setminus \BLetters{x}{q}$. We refer to $\NLetters{x}{q}$ as the set of ""neutral letters@@sig"" of $\locAlph{x}{q}$ (for $\locLang{x}{q}$).

\begin{lemma}[Super words and uniformity]\label{lemma-app:super-words-in-local-languages}
	A word $w\in \locAlph{x}{q}^\oo$ is a "super word" for $\locLang x q$ if and only if $w$ "produces priority $x$ uniformly" in $\classSig{x}{q}$, that is, for all $q'\in \classSig{x}{q}$,  $q'\lrp{w:x}\classSig{x}{q}$.
\end{lemma}
\begin{proof}
	By "normality" of $\A$, if there is $q_1\in[q]$ such that $q_1\lrp{w:>x}q_2$, there is a word $w'\in \SS^*$ labelling a returning path $q_2\lrp{w':x+1}q_1$. Therefore, $(ww')^\oo\notin \Lang{\initialAut{\A}{q}}$, so $w$ is not a "super word" for $\locLang x q$.
	The converse implication is clear, since each time word $w$ is read, the minimal priority produced by the automaton is $x$.
\end{proof}

In particular, using previous lemma, we can detect the set of "super letters" $\BLetters{x}{q}$ in polynomial time.

\subparagraph{Super words of positional languages.}
The use of the  hypothesis of "positionality" of $W$ for proving Lemma~\ref{lemma-nec:uniformity-x-transitions} resides in the next fundamental result.

\begin{lemma}[Neutral letters do not form super words]\label{lemma-app:super-words-contain-super-letter}
	Let $w\in \locAlph{x}{q}^+$ be a "super word" for $\locLang{x}{q}$. Then, $w$ contains some "super letter".
\end{lemma}
\begin{proof}
	If $w$ is already a letter in $\locAlph{x}{q}$, we are done. If not, let $w=w_1w_2$ be any non-trivial decomposition into smaller words $w_1,w_2\in \locAlph{x}{q}^+$. We show that either $w_1$ or $w_2$ are "super words" for $\locLang{x}{q}$. This allows us to finish the proof, as we can recursively chop $w$ into strictly smaller "super words" until obtaining a "super letter".

	Suppose by contradiction that neither $w_1$ or $w_2$ are "super words". Then, by Lemma~\ref{lemma-app:super-words-in-local-languages}, there are states $q_1$ and $q_2$ such that $q_1\lrp{w_1:>x}q_1'$ and $q_2\lrp{w_2:>x}q_2'$.
	By "normality", we obtain returning paths $q_1'\lrp{u_1:x+1}q_1$ and $q_2'\lrp{u_2:x+1}q_2$.
	Therefore, $(w_1u_1)^\oo\notin W$ and $(w_2u_2)^\oo\notin W$.
	We consider the "game" $\G$ with "winning condition" $\Lang{\initialAut{\A}{q}}$ consisting in a vertex $v$ with self loops $u_1w_1$ and $u_2w_2$ (see Figure~\ref{fig-warm:game-chop-super-words} from the warm-up).
	"Eve" can  "win" game $\G$, as alternating the two self loops she produces the word $(u_1w_1w_2u_2)^\oo$, which belongs to $\Lang{\initialAut{\A}{q}}$ since $w_1w_2$ is a "super word".
	However, "positional strategies" in this game produce either $(w_1u_1)^\oo$ or $(w_2u_2)^\oo$, both losing.
	This contradicts the "positionality" of $\Lang{\initialAut{\A}{q}}$, and therefore, that of $W$ (Lemma~\ref{lemma-app:positional-implies-pos-quotient}).
\end{proof}

\subparagraph{Polishing a \texorpdfstring{$\eqSig{x}$}{x}-class of \texorpdfstring{$\A$}{A}.}
We show how to "polish" the class $\classSig{x}{q}$ of $\A$. This process has the property that, either $\classSig{x}{q}$ is already "polished@@classSig", or the obtained automaton $\A'$ has strictly less states than $\A$, as desired. %Therefore, by successively applying it while some class is not "polished", we obtain a "$x$-polished automaton" in polynomial time.

Assume that the class $\classSig{x}{q}$ is not "$x$-polished".
Consider the restriction of $\localAutSig{x}{q}$ to transitions labelled with $\NLetters{x}{q}$, which we denote $\localAutSig{x}{q}'$. \AP Take $S_{\classSig{x}{q}}$ to be a ""final SCC"" of $\localAutSig{x}{q}'$ (that is, one without edges leading to states not on itself).

Fix a state $q_0\in \classSig{x}{q}$.
Consider the automaton $\A'$ obtained from $\A$ by removing states in $\classSig{x}{q} \setminus S_{\classSig{x}{q}}$, and redirecting transition that go to $\classSig{x}{q} \setminus S_{\classSig{x}{q}}$ in $\A$ to transitions towards $q_0$. 
For these redirected transitions, we keep the same priority if it is $\leq x$, and set it to $x$ otherwise. Formally:

\begin{itemize}
	\item The set of states of $\A'$ is $Q' = Q \setminus \left( \classSig{x}{q} \setminus S_{\classSig{x}{q}} \right)$.
	\item The initial state is $q_\init$, or $q_0$ if $q_\init\in \classSig{x}{q} \setminus S_{\classSig{x}{q}}$.
\end{itemize}
For $q'\in Q'$:
\begin{itemize}
	\item If $q'\re{a:y}p$ in $\A$ and $p\notin \classSig{x}{q}$, then $q\re{a:y}p$ in $\A'$.
	\item If $q'\re{a:y}p$ in $\A$,  $p\in \classSig{x}{q}\setminus S_{\classSig{x}{q}}$, and $y\leq x$, then $q'\re{a:y}q_0$ in $\A'$.
	\item If $q'\re{a:y}p$ in $\A$,  $p\in \classSig{x}{q} \setminus S_{\classSig{x}{q}}$, and $y> x$, then $q'\re{a:x}q_0$ in $\A'$.
\end{itemize}

\AP For transitions in the two latter cases, we say that $q'\re{a:y}q_0$ has been ""redirected@@polish"" from $q'\re{a:y}p$.

\begin{remark}\label{rmk-app:nothing-changes-if-polished}
	If $\classSig{x}{q} = S_{\classSig{x}{q}}$, then $\A' = \A$.
\end{remark}

The following lemma will be use to show that we can compute $\A'$ in polynomial time, to prove the correctness of $\A'$, and to obtain that $\classSig{x}{q}$ is "$x$-polished@@class" in $\A'$.

\begin{lemma}\label{lemma-app:neutral-letters-in-polish}
	Let $q_1,q_2\in S_{\classSig{x}{q}}$, and let $w\in \NLetters{x}{q}^*$ labelling a path $q_1\lrp{w:y}q_2$ in $\A$.
	Then, $y>x$.
\end{lemma}
\begin{proof}
	The fact that $y\geq x$ simply follows from the fact that $\NLetters{x}{q}\subseteq \locAlph{x}{q}$, which, by definition, contains words connecting the states in $\classSig{x}{q}$ producing no priority ${<}x$.
	
	Suppose by contradiction that $y=x$. Then, by the same argument as in the proof of Lemma~\ref{lemma-nec:existence-uniform-words} (see also Claim~\ref{cl:existence-super-words}),
	there is $w'\in \NLetters{x}{q}^*$ "producing priority $x$ uniformly" in $\classSig{x}{q}$ and coming back to this class; that is, for every $q'\in \classSig{x}{q}$, $q'\lrp{w':x}\classSig{x}{q}$.
	Therefore, by Lemma~\ref{lemma-app:super-words-in-local-languages}, $w'$ is a "super word". 
	By Lemma~\ref{lemma-app:super-words-contain-super-letter}, $w'$ must contain a "super letter", a contradiction, as $w'\in \NLetters{x}{q}^*$ and $\locAlph{x}{q}$ is "uniquely decodable".
\end{proof}

\begin{lemma}
	Automaton $\A'$ can be computed in polynomial time from $\A$.
\end{lemma}
\begin{proof}
	To obtain $S_{\classSig{x}{q}}$, we first build a finite representation of the restriction of $\localAutSig{x}{q}$ to "neutral letters@@polishSig" (we recall that, in general, $\localAutSig{x}{q}$ might have an infinite number of transitions).
	One way of doing that is to build the following graph $G$: for each pair of states $q_1,q_2\in \classSig{x}{q}$ and each $y> x$, we put an edge $q_1\re{y} q_2$ if there is a path from $q_1$ to $q_2$ producing $y$ as minimal priority and not passing trough another state in $\classSig{x}{q}$.	
	By Lemma~\ref{lemma-app:neutral-letters-in-polish}, $S_{\classSig{x}{q}}$ is a subgraph of $G$. To obtain the states in $S_{\classSig{x}{q}}$ we just need to perform a decomposition in "SCCs" of $G$ and take a "final SCC" of it.\footnotemark
\end{proof}
\footnotetext{If $W$ is not "positional", the procedure described here does provide a set of states $S_{\classSig{x}{q}}$, but it might lead to an incorrect automaton $\A'$. If our objective is to decide the "positionality" of $W$, at the end of the procedure we need to check the equality $\Lang{\A} = \Lang{\A'}$; if it does not hold, we can conclude that $W$ is not "positional".}

We consider $\A'$ equipped with the preorders $\leqSig{0},\dots,\leqSig{x}$ inherited from $\A$.

\begin{lemma}\label{lemma-app:polish-is-nice-transform}
	Automaton $\A'$ is a "$\eqSig{x}$-nice transformation of $\A$ at level $x$".
\end{lemma}
\begin{proof}
	We first note that, by Lemma~\ref{lemma-app:monotonicite-preorder-x}, $\eqSig{x}$ is "$[0,x-1]$-faithful" in $\A$, so it makes sense to speak of a "$\eqSig{x}$-nice transformation at level $x$". %As always, $\eqSig{x}$ refines $\eqRes{\A}$.
	
	Automaton $\autGeq{\A'}{x+1}$ coincides with the "subautomaton" of $\autGeq{\A}{x+1}$ induced by states in $Q'$.
	Indeed, let $q_1',q_2'\in Q'$ and $q_1'\re{a:>x}q_2'$ in $\A$. As these states are in $Q'$, $q_2'\notin \classSig{x}{q} \setminus S_{\classSig{x}{q}}$, so the transition has not been "redirected@@polish", and it appears in $\A'$. Conversely, all transitions producing a priority $>x$ in $\A'$ appear in $\A$.

	We show that $\eqSig{x}$ is "$[0,x-1]$-faithful" in $\A'$ and $\autLeq{\A}{x-1}{x-1}=\autLeq{\A}{x-1}{x-1}$. 
	Let $p_1,p_2\in Q'$ such that $p_1\eqSig{x}p_2$, and let  $p_1\re{a:y_1'}q_1'$ and $p_2\re{a:y_2'}q_2'$ be two transitions in $\A'$. 
	Transitions that have not been "redirected@@polish" satisfy the congruence requirements, as they satisfy them in $\A$. 
	Assume that the first of these transitions have been redirected from  $p_1\re{a:y_1}q_1$ in $\A$. We have that $q_1'=q_0\eqSig{x} q_1$, so, $q_1'\eqSig{x}q_2'$ by the congruence property in $\A$.
	If $y_1'<x$, then $y_1=y_1'$ and the "$y_1$-uniformity" of transitions in $\A$ yields $y_1'=y_2'$.
	Therefore, we also have that $y_1'\geq x$ if and only if $y_2'\geq x$
	This gives both  the "$[0,x-1]$-faithfulness" in $\A'$ and the equality of the "quotient@@leq" automata.
	%In the same way, we obtain that "redirected@@polish" transitions producing priorities ${\geq}x$ also preserve $\eqSig{x}$-classes, so $\eqSig{x}$ is a "congruence for" ${\geq}x$-transitions.
	
	As $\eqSig{x}$ "refines" $\eqRes{\A}$, the latter relation is also a "congruence" in $\A'$ and $\quotAut{\A}{\A} = \quotAut{\A'}{\A}$. 
\end{proof}

\begin{lemma}[Correctness of the polishing operation]\label{lemma-tech:equality-languages-polish}
	Automaton $\A'$ "recognises" $\Lang{\A}$.
	%For every state $q'\in Q'$, we have $\Lang{\initialAut{\A'}{q'}} = \Lang{\initialAut{\A}{q'}}$.
	%In particular, these automata are "equivalent@@aut".
\end{lemma}
\begin{proof}
	%For simplicity, we assume that $q'=q_\init$ is the initial state of $\A'$ and that it coincides with the initial state in $\A$. The proof is analogous for any other state.	
	Let $w\in \SS^\oo$. If $w$ is "accepted@@byPriority" or "rejected with a priority" $y<x$ or $y>x$ in $\A'$, by Lemma~\ref{lemma-tech:nice-transformations}, $w\in \Lang{\A}$ if and only if $w\in \Lang{\A'}$.
	Suppose then that $w$ is  "accepted with priority" $x$ in $\A'$. Let $\rr'$ be the "run over" $w$ in $\A'$. If $\rr'$ eventually does not take any "redirected transition@@polish", then it is eventually a run in $\A$, and we can conclude by Lemma~\ref{lemma-tech:subautomaton-in-common}. Suppose that $\rr'$ takes "redirected transitions@@polish" infinitely often; moreover, eventually all such transitions produce priority $x$. We decompose $\rr'$ as follows:
	\[ \rr' = q_\init \lrp{w_0} p_0'\re{a_0:x} q_0 \lrp{w_1:\geq x} p_1'\re{a_1:x} q_0 \lrp{w_1:\geq x} p_2'\re{a_2:x} q_0 \lrpE{} \dots ,\]
	where no priority $<x$ appears after $p_0'$, each  transition $p_i'\re{a_1:x} q_0$ is a "redirected@@polish" one, and no "redirected@@polish" transition appears in  paths $q_0 \lrp{w_i:\geq x}$, in particular, these paths appear in $\A$.

	\begin{claim}
		For each $i\geq 1$, the word $w_ia_i$ belongs to $\locAlph{x}{q}^+$ and is a "super word for"~$\locLang{x}{q}$.
	\end{claim}
	\begin{subproof}
		The word $w_ia_i$ connects two states in $\classSig{x}{q}$ in $\A'$ producing no priority $<x$. Since $\A'$ is a "$\eqSig{x}$-nice transformation at level $x$",
		word  $w_ia_i$ also connects states in $\classSig{x}{q}$ in $\A$, without producing priorities $<x$. Therefore, it belongs to $\locAlph{x}{q}^+$.
		
		Consider the path $q_0\lrp{w_i}p_i'\re{a_i}q'$ in $\A$. 	
		As we suppose that transition $p_i'\re{a_i}q_0$ has been "redirected@@polish" in $\A'$, $q'\notin S_{\classSig{x}{q}}$. Then, since $S_{\classSig{x}{q}}$ is a "final SCC" of the restriction of $\localAutSig{x}{q}$ to "$\NLetters{x}{q}$-transitions@@in", $w_ia_i$ contains some factor that is a letter in $\locAlph{x}{q}\setminus\NLetters{x}{q}$. Such a factor is a "super letter", so $w_ia_i$ is a "super word".
	\end{subproof}

	Consider the "run over" $w$ in $\A$, that we divide following the decomposition of $\rr'$:
	\[ \rr = q_\init \lrp{w_0} p_0\re{a_0:\geq x} q_1 \lrp{w_1:\geq x} p_1\re{a_1:\geq x} q_2 \lrp{w_2:\geq x} p_2\re{a_2:\geq x} q_3 \lrpE{} \dots .\]
	As $\A'$ is a "$\eqSig{x}$-nice transformation at level $x$", $q_i\eqSig{x}q_0$ for all $i$, and $\rr$ does not produce any priority $<x$ from $p_0$.  
	By Lemma~\ref{lemma-app:super-words-in-local-languages}, as $w_ia_i$ is a "super word", the path $q_i\lrp{w_ia_i} q_{i+1}$ produces priority $x$.
	Therefore, $w$ is "accepted by@@aut" $\A$.
\end{proof}

\begin{lemma}[Polishing polishes]\label{lemma-tech:[q]-is-polished}
	The class $\classSig{x}{q}$ is "$x$-polished@@class" in $\A'$.
\end{lemma}
\begin{proof}
	Let $q_1,q_2\in Q'$ be two states in the class $\classSig{x}{q}$. Assume that, for a word $w\in \SS^*$, the path $q_1\lrp{w:x}p_1$ produces priority $x$. As $\eqSig{x}$ is "$[0,x-1]$-faithful", $q_2\lrp{w:\geq x}p_2$. Suppose by contradiction that this latter path produces exclusively priorities $>x$. Then, this path also exists in $\A$, and by "normality" (of $\A$), there is a returning path $p_2\lrp{w':x+1}q_2$.
	We obtain therefore a path $q_1\lrp{ww':x}\classSig{x}{q}$. However, in $\A'$, $\classSig{x}{q} = S_{\classSig{x}{q}}$, so, by Lemma~\ref{lemma-app:neutral-letters-in-polish}, $ww'$ contains a "super letter", so $(ww')^\oo\in \Lang{\initialAut{\A}{q}}$, contradicting the fact that there is a cycle $q_2\lrp{ww':x+1} q_2$.
	
	The second property of the definition of an "$x$-polished class" is satisfied in $\A'$, as we have "redirected@@polish" all "$x$-transitions" entering in $\classSig{x}{q}$ to the state $q_0$.
	
	We show the third item. Let $q_1,q_2\in \classSig{x}{q}$. Since  $\classSig{x}{q} = S_{\classSig{x}{q}}$ in $\A'$, there is a path $q_1\lrp{w\phantom{.}} q_2$  for some $w\in \NLetters{x}{q}^*$. By Lemma~\ref{lemma-app:neutral-letters-in-polish}, this path produces exclusively priorities~${>}x$.
\end{proof}

This lemma allows us to conclude. We have obtained a deterministic automaton $\A'$ that is "equivalent@@aut" to $\A$. We claim that $\sizeAut{\A'}<\sizeAut{\A}$. 
Indeed, if this was not the case, we would have that $\classSig{x}{q} = S_{\classSig{x}{q}}$, so, by Remark~\ref{rmk-app:nothing-changes-if-polished}, $\A= \A'$. By the previous Lemma~\ref{lemma-tech:[q]-is-polished} this implies that $\classSig{x}{q}$ was already "$x$-polished@@class" in $\A$, a contradiction.

\subparagraph{Discussion: Why not just continue polishing?} We have just showed a method to "$x$-polish@@class" a given class of $\A$. The natural continuation would be to "polish@@sig" the rest of classes, until obtaining an "$x$-polished automaton", and then apply the first case.
The main difficulty is that the polishing operation we have presented might break the "normality" of $\A$. "Normality" of automata is key in all the process
(see for example Lemma~\ref{lemma-app:super-words-in-local-languages}), so we cannot guarantee to be able to continue polishing the classes of $\A'$.
We would need to be able to either show that $\A'$ is in "normal form" (for example, by having an analogous to Lemma~\ref{lemma-app:same-paths-if-polished}), or to show that we can "normalise" $\A'$ while maintaining the properties of being an "$(x-2)$-structured signature automaton".
We have not succeeded in ensuring these properties, although we believe that it should be possible to~do~so.

\end{document}